\documentclass[preprint,11pt]{elsarticle}
\usepackage[a4paper, total={6.15in, 10in}]{geometry}
\usepackage[utf8]{inputenc}
\PassOptionsToPackage{hyphens}{url}\usepackage[colorlinks=true]{hyperref}
\usepackage[svgnames,x11names,table]{xcolor}
\usepackage{tabularray}

\usepackage[T1]{fontenc}
\usepackage{lmodern}

\usepackage{subfiles}

\usepackage{graphicx}
\usepackage{subcaption}
\usepackage{tikz}
\usepackage{tikz-cd}
\usepackage{circuitikz}
\usetikzlibrary{petri,positioning, calc,decorations.pathmorphing,shapes,cd,arrows,decorations.pathreplacing}

\usepackage{caption}
\usepackage[final]{listings}
\usepackage{multirow}

\usepackage{url}
\usepackage{siunitx}
\usepackage{amssymb,amsthm,mathrsfs,nccmath}
\usepackage{centernot}
\theoremstyle{definition}
\newtheorem{notation}{Notation}
\newtheorem{definition}{Definition}
\newtheorem{proposition}{Proposition}
\newtheorem{theorem}{Theorem}
\newtheorem{corollary}{Corollary}
\newtheorem{lemma}{Lemma}
\newtheorem{remark}{Remark}

\usepackage[nounderscore]{syntax}

\usepackage{fancyhdr}
\fancypagestyle{pprintTitle}{%
    \fancyhf{}
      \fancyfoot[l]{%
        \color[gray]{0.80}\tiny
        © 2026. Licensed under CC BY-NC-ND 4.0.
        https://creativecommons.org/licenses/by-nc-nd/4.0/
	 }

}
\newcommand{\arrowflow}{
\tikz \draw[-{to}] (-1pt,0) -- (1pt,0);
}
\newcommand{\invertedarrowflow}{
\tikz \draw[-{to}] (1pt,0) -- (-1pt,0);
}
\newcommand{\flow}[4]{ 
	\draw[#3] (#1) to [#4] node [pos=0.5] {\arrowflow} (#2);
}
\newcommand{\flowdiag}[5]{ 
	\draw[#3] (#1) to [#4] node [#5] {\arrowflow} (#2);
}

\tikzset{fill fraction/.style n args={1}{path picture={
 \fill[black] (path picture bounding box.south west) rectangle
 ($(path picture bounding box.north west)!0.5!(path picture bounding box.north
 east)$);}}
}
\newcommand{\qmatch}[4]{ 
\node[draw,fill=white,fill fraction={black}{0.5},inner sep=1.2pt,minimum size=3pt,label=left:{\scriptsize #4}] (#1) at (#2,#3) {};
}
\newcommand{\dmatch}[5]{ 
\node[circle,draw,fill=white,fill fraction={black}{0.5},inner sep=1.2pt,minimum size=3pt,label=#5:{\scriptsize #4}] (#1) at (#2,#3) {};
}
\newcommand{\qinplain}[4]{ 
\node[draw=black,fill=white,inner sep=0pt,minimum size=3pt,label=\scriptsize #4] (#1) at (#2,#3) {};
}
\newcommand{\qin}[4]{ 
\flow{{#2,#3}}{$(#2,#3)+(0.8,0)$}{dashed}{};
\qinplain{#1}{#2}{#3}{#4};
}
\newcommand{\dinplain}[4]{ 
\node[circle,draw=black,fill=white,inner sep=0pt,minimum size=3pt,label=left:\scriptsize #4] (#1) at (#2,#3) {};
}
\newcommand{\din}[4]{ 
\flow{{#2,#3}}{$(#2,#3)+(0.8,0)$}{}{};
\dinplain{#1}{#2}{#3}{#4};
}
\newcommand{\qoutplain}[4]{ 
\node[fill=black,inner sep=0pt,minimum size=3pt,label=\scriptsize #4] (#1) at (#2+0.8,#3) {};
}
\newcommand{\qout}[4]{ 
\flow{{#2,#3}}{$(#2,#3)+(0.8,0)$}{dashed}{};
\qoutplain{#1}{#2}{#3}{#4};
}
\newcommand{\doutplain}[4]{ 
\node[circle,fill=black,inner sep=0pt,minimum size=3pt,label=right:\scriptsize #4] (#1) at (#2+0.8,#3) {};
}
\newcommand{\dout}[4]{ 
\flow{{#2,#3}}{$(#2,#3)+(0.8,0)$}{}{};
\doutplain{#1}{#2}{#3}{#4};
}

\newcommand{\forkplain}[2]{ 
\draw[fill=gray!30] (#1+1,#2+0.25) node[anchor=north]{}
  -- (#1+1.25,#2+0.5) node[anchor=north]{}
  -- (#1+1.25,#2) node[anchor=south]{}
  -- cycle;
}
\newcommand{\fork}[5]{ 
\forkplain{#1}{#2};
\flow{{#1,#2+0.25}}{{#1+1,#2+0.25}}{dashed}{};
\flow{{#1+1.25,#2+0.4}}{{#1+1.25+0.95,#2+0.4}}{dashed}{};
\flow{{#1+1.25,#2+0.1}}{{#1+1.25+0.95,#2+0.1}}{dashed}{};
\node[draw=black,fill=white,inner sep=0pt,minimum size=3pt] (p) at (#1,#2+0.25) {};
\node[draw=black,fill=black,inner sep=0pt,minimum size=3pt] (q) at (#1+2.25,#2+0.4) {};
\node[draw=black,fill=black,inner sep=0pt,minimum size=3pt] (r) at (#1+2.25,#2+0.1) {};
}
\newcommand{\joinplain}[2]{ 
\draw[fill=gray!30] (#1+1.25,#2+0.25) node[anchor=north]{}
  -- (#1+1,#2+0.5) node[anchor=north]{}
  -- (#1+1,#2) node[anchor=south]{}
  -- cycle;
}
\newcommand{\join}[5]{ 
\joinplain{#1}{#2};
\flow{{#1+1.25,#2+0.25}}{{#1+1.25+0.95,#2+0.25}}{dashed}{};
\flow{{#1,#2+0.4}}{{#1+1,#2+0.4}}{dashed}{};
\flow{{#1,#2+0.1}}{{#1+1,#2+0.1}}{dashed}{};
\node[draw=black,fill=black,inner sep=0pt,minimum size=3pt] (p) at (#1+2.25,#2+0.25) {};
\node[draw=black,fill=white,inner sep=0pt,minimum size=3pt] (q) at (#1,#2+0.4) {};
\node[draw=black,fill=white,inner sep=0pt,minimum size=3pt] (r) at (#1,#2+0.1) {};
}

\newcommand{\computonPrimitive}[5]{ 
\draw[draw=black,fill=white] (#1,#2) rectangle ++(#3,#4) node[pos=0.5]{#5};
}
\newcommand{\computonComposite}[4]{ 
\draw[draw=black,fill=black!2] (#1,#2) rectangle ++(#3,#4);
}

\title{Compositional Separation of Control Flow and Data Flow}
\date{}

\begin{document}

\hypersetup{
  urlcolor=Blue4,
  citecolor=Blue4,
  linkcolor=Blue4
}

\begin{frontmatter}
\author[]{Damian Arellanes \\ Lancaster University, United Kingdom}
\begin{abstract}
Every Model of High-Level Computation (MHC) has an underlying composition mechanism for combining simple computing devices into more complex ones. Composition can be done by (explicitly or implicitly) defining control flow, data flow or any combination thereof. Control flow specifies the order in which individual computations are activated, whereas data flow defines how data is exchanged among them. Unfortunately, traditional MHCs either mix data and control or only consider one dimension explicitly, which makes it difficult to reason about data flow and control flow separately. Reasoning about these dimensions orthogonally is a crucial desideratum for optimisation, maintainability and verification purposes. In this paper, we introduce a novel MHC that explicitly treats data flow and control flow as separate dimensions, while providing modularity. As the model is rooted in category theory, it provides category-theoretic operations for compositionally constructing sequential, parallel, branching or iterative composites. Compositionality entails that a composite exhibits the same properties as its respective constituents, including separation of concerns and modularity. We conclude the paper by demonstrating how our proposal can be used to model high-level computations in two different application domains: software engineering and artificial intelligence. 
\end{abstract}
\end{frontmatter}

\section{Introduction}
\label{sec:Introduction}

In the context of theoretical computer science, \emph{compositionality} refers to the property of Models of High-Level Computation (MHCs) that allows the inductive definition of complex computing devices from simpler ones \cite{tripakis_modular_2013,baez_open_2020,arellanes_models_2025,basu_rigorous_2011}. MHCs raise the level of abstraction of their classical, low-level counterpart (e.g., Turing Machines) by giving a birds-eye-view of multiple interacting devices which can individually correspond to low- or even high-level computations \emph{per se}. As such devices are treated as black boxes, their internal details are irrelevant. What matters is how to compose/glue them into high-level abstractions, e.g., a sequential process to trigger the computation of devices A and B, in that order. Apart from moving up the ladder of abstraction, another characteristic defining difference with respect to their classical counterpart is that MHCs can be open in the sense they can compute on data streams coming from the external world, which make them suitable to be useful in the actual construction or simulation of complex computing systems \cite{arbab_composition_2006}. Examples of MHCs include component models \cite{lau_introduction_2017}, workflow languages \cite{van_der_aalst_application_1998} and process algebras \cite{best_box_2002}. 

When composition is done algebraically, the resulting computation structures (known as \emph{composites}) exhibit the same characteristics as their constituents \cite{lau_introduction_2017}. Algebraic compositionality can be realised by the composition of \emph{control flow} \cite{arellanes_evaluating_2020} or \emph{data flow} \cite{arbab_reo_2004}. Control flow defines the order in which individual computing devices are computed, whereas data flow defines how data is passed among them. Traditionally, MHCs do not support algebraic composition and they allow the definition of computations in which data follows control. This sort of coupling makes it difficult (i) to (formally) reason about computation order and data production/consumption separately and (ii) to explicitly distinguish between control and data dependencies \cite{message_programming_2013}. Consequently, it is hard to (1) verify these dimensions independently \cite{varea_dual_2006,clarke_model_2008}, or (2) modify/optimise/transform control flow without affecting data flow (or vice versa)  \cite{arellanes_decentralized_2023,vanderbauwhede_separation_2007}. As separating control from data addresses (1) for increased reliability and (2) for enhanced maintenance, enabling such a separation within the foundational semantics of any MHC is a crucial desideratum. Facilitating this in a compositional manner is far from trivial but it can provide additional benefits such as modularity for reusing high-level computations at scale (i.e., functional scalability \cite{arellanes_evaluating_2020}) and compositional verification towards soundness-by-construction \cite{bordis_correctness-by-construction_2023}. For example, we can verify termination compositionally through the analysis of control flow only, i.e., without considering data flow at all. 

While \emph{control-based composition} approaches define explicit control flow for the coordination of computing devices, \emph{data-based composition} defines implicit control in the collaborative exchange of data \cite{arellanes_analysis_2018}. Thus, the notion of control flow has higher precedence than that of data flow because it is always present in any composition mechanism (and not the other way round).\footnote{A potential explanation is that the notion of control flow is tightly linked to the arrow of time, whereas data is just a piece of information that is sent from one place to another in some specific time-dependent order.} In fact, it is possible to compose complex computations by control flow only and without the need of passing data at all (cf. actuator composition \cite{arellanes_algebraic_2018}).\footnote{Even if we argue that control is a piece of information/data, the act of sending it from one computing device to another is still governed by the grandiose, apparently unavoidable, notion of control flow (cf. interleaving semantics for global execution traces in concurrent systems \cite{sebesta_concepts_2018}).} More generally, in every model of computation, control flow is always present no matter if it is implicit or explicit. For instance, a finite state automaton can describe the finite control of a Turing machine when data is abstracted away. Here, control structures (e.g., sequencing, branching or looping) are not explicitly given but they rather emerge from state transitions \cite{arellanes_models_2025}. 

As control flow is ever present, we believe that the right way of constructing complex, high-level computational behaviours is through a control-based composition mechanism that does not neglect the role of data passing. Accordingly, in this paper we propose a model in which fundamental units of composition (known as \emph{computons}) are passive open systems able to interact with their environment via an interface which consists of input and output ports. As a port is a structural construct that can exclusively buffer either data or control, computons exchange data and control separately. Our model is compositional in the sense computons can be inductively composed into larger ones via well-defined control-based composition operations. Such operations are rooted in category theory and allow the formal construction of sequential, parallel, branching or iterative computons from simpler ones. As a result of composition semantics, composites preserve the structure of the composed entities, so they also separate data and control and their port-based interface is inductively constructed from the composed computons. Remarkably, unlike existing compositional approaches, any two computons can always be composed sequentially or in parallel, regardless of the data they require or produce. Rather than focusing on specific operational semantics, the focus of this paper is to provide formal category-theoretic operators, built upon colimit semantics, for the compositional construction of computing devices (i.e., computons) that separate data flow and control flow. Although different ways of expressing computon behaviour are possible, because the model is independent of concrete execution semantics, in this paper we use the \emph{token game} from the theory of classical P/T Petri nets for this purpose, as shown in Figure~\ref{fig:intro}.

\vspace{-2pt}

\begin{figure}[!h]
\centering
\begin{tikzpicture}[scale=0.83]
\node[rectangle, draw, align=center] (box) at (0.8,2.4) {\scriptsize Category-Theoretic\\\scriptsize Composition\\\scriptsize Semantics};
\draw [decorate,decoration={brace,amplitude=5pt,mirror,raise=4ex}] (-0.22,3.3) -- (-0.22,1.5) node[midway,xshift=-5em]{\scriptsize Computons};

\draw[-{Latex[length=-0.1pt 4.9]}, line width=2.2pt] (0.8,1.45) -- (0.8,0.65);

\node[rectangle, draw, align=center] (box) at (0.8,0) {\scriptsize Operational\\\scriptsize Semantics};
\draw [decorate,decoration={brace,amplitude=5pt,mirror,raise=4ex}] (0.25,0.6) -- (0.25,-0.6) node[midway,xshift=-5em]{};
\node at (-2.45,0){\scriptsize P/T Petri Nets};
\end{tikzpicture}
\caption{This paper is primarily focused on (category-theoretic) composition semantics to study control-based composition in its pure form, without adhering to specific execution details. Separating composition from operation allows us to describe computon behaviour through different execution models. Although here we use the well-known token game from classical P/T Petri nets, other formalisms can be used for the same purpose such as natural semantics \cite{kahn_natural_1987}, small-step semantics \cite{plotkin_origins_2004} or even timed Petri nets \cite{ramchandani_analysis_1974}, just to name a few.}
\label{fig:intro}
\end{figure}

By focusing on the fundamentals of control-based composition while capturing the core mechanisms that underlie all MHCs (i.e., control flow and data flow), the model we propose is intended to be a universal meta-framework for the formal reasoning of every possible high-level computation. This generality, coupled with the fact that we do not subscribe to concrete execution semantics or implementation details, allows the model's foundations to serve as the basis for the future development of concrete specification languages, implementation frameworks or even standard tools that can be used by software engineers to compositionally construct complex software systems that exhibit separation of control and data flow. For example, the proposed model can be implemented via the Alloy specification language \cite{jackson_alloy_2019} so as to support model checking. If less assurances are preferred or visual modelling is required, mainstream languages can be used instead such as Java or JavaScript. 

\vspace{2.5pt}

Offering generality also enables the combination of distinct models of computation. For example, the behaviour of a (non-composite) computon can be expressed as a Turing machine, whereas the behaviour of another (non-composite) can be defined as a (typed) lambda abstraction. Then, using the proposed sequencing operator, we can form a composite computon for executing the Turing machine before supplying its output to the lambda abstraction in order to perform a reduction operation.\footnote{Of course, in practice, the Turing machine's output might need to be encoded to make it compatible with the input expected by the lambda abstraction, either by a run-time environment or by introducing an intermediate computon for doing the encoding.} If a parallelising operator is used instead, then these two computing devices can be triggered at the same time.

\vspace{2.5pt}

Apart from enabling generality, orthogonalising data and control and separating operational from composition semantics, our model facilitates modularity and encapsulation as a result of realising compositionality. Encapsulation allows treating computons unifiedly, and is realised by the fact that a composite defines explicit control and data flow structures that can only be accessed through a well-defined interface. Thus, a computon can be perceived as an encapsulated black-box that can only interact with other computons via its visible ports. Hiding internal structure in this way enable us to build computons of considerable complexity.

\vspace{2.5pt}

As control is what governs computation, one of the primary objectives of the computon model is to facilitate the compositional formation of explicit control flow structures. As such structures are orthogonal to data, data flow can be ``disregarded'' to independently reasoning about execution order, i.e., unlike models where control is implicit, execution order does not need to be ``discovered''. In Section \ref{sec:computons}, we formalise functors for this purpose, which can be useful for static analysis or verification purposes.

\vspace{2pt}

The rest of the paper is structured as follows. Section \ref{sec:computons} presents the definition of computons by treating them as set-valued functors in a category we introduce which we refer to as the \emph{category of computons and computon morphisms}. Although this paper is mainly focused on the compositional construction of such entities and not on their operational semantics, Section \ref{sec:operational-semantics} discusses the notion of computon execution via the well-known \emph{token game} that has been used in the context of classical P/T Petri nets for many years. Sections \ref{sec:trivial-computons} and \ref{sec:primitive-computons} present the most elementary classes of computons that serve as building blocks for constructing complex composites. Section \ref{sec:composite-computons} describes formal category-theoretic operations to form sequential, parallel, branching or iterative composites. Section \ref{sec:applications} provides two applications of the proposed model in the domains of software engineering and artificial intelligence, by discussing compositional system construction and showing how the separation of control and data flow can be exploited for model transformation purposes. Section \ref{sec:related-work} presents and analyses related work, and Section \ref{sec:conclusions} outlines the conclusions and future directions of our work. A mapping from Petri net syntax to the graphical notation we use to represent computons is provided in \ref{sec:appendix-mapping}.
\section{Computons} \label{sec:computons}

Intuitively, a \emph{computon} is a bipartite graph with two types of nodes: \emph{computation units} and \emph{ports}.\footnote{The word \emph{computon} derives from the Latin root for computation (i.e., \emph{computatio}) and the Greek suffix \emph{-on}. In Physics, such a suffix is traditionally used to designate subatomic particle names \cite{arellanes_composition_2024}.} A computation unit is a construct that receives information in ports, performs some computation and produces new information in other ports. A port that is connected between two computation units is known as \emph{internal port} (or \emph{i-port}), whereas a port that is exclusively connected to or from a computation unit is called \emph{external} (or \emph{e-port}). As ports and computation units are connected via edges, edges represent information flow ranging from control signals to complex data values.\footnote{In the context of resource theories \cite{coecke_mathematical_2016}, data ports correspond to resource wires.}

The interface of a computon towards the outside world is determined by its collection of e-ports each being an external control inport (\emph{ec-inport}), an external control outport (\emph{ec-outport}), an external data inport (\emph{ed-inport}) or an external data outport (\emph{ed-outport}). An ec-inport is where control flow originates, an ec-outport is where control flow terminates, an ed-inport stores data coming from the external world whereas an ed-outport stores data resulting from the computon's operation. A port that is ec-inport and ec-outport is called \emph{ec-inoutport}. Similarly, if it is both ed-inport and ed-outport, it is called \emph{ed-inoutport}. 

Figure \ref{fig:port-taxonomy} presents the naming system we employ to derive port names where it is clear that, just as there are \emph{ec-ports} and \emph{ed-ports}, there are also internal control ports (\emph{ic-ports}) and internal data ports (\emph{id-ports}). 

\begin{figure}[!h]
\centering
\begin{tikzpicture}
\node(e) at (0,2.75) {e};
\node(ec) at (1,3.5) {c};
\node(eci) at (2.5,4) {in};
\node(eco) at (2.5,3.5) {out};
\node(ecio) at (2.5,3) {inout};
\node(ecip) at (4,4) {port};
\node(ecop) at (4,3.5) {port};
\node(eciop) at (4,3) {port};
\node(ed) at (1,2) {d};
\node(edi) at (2.5,2.5) {in};
\node(edo) at (2.5,2) {out};
\node(edio) at (2.5,1.5) {inout};
\node(edip) at (4,2.5) {port};
\node(edop) at (4,2) {port};
\node(ediop) at (4,1.5) {port};

\node(ic) at (1,1) {c};
\node(icp) at (4,1) {port};
\node(i) at (0,0.5) {i};
\node(id) at (1,0) {d};
\node(idp) at (4,0) {port};

\node[opacity=0.3] (ecipn) at (6,4) {ec-inport};
\node[opacity=0.3] (ecopn) at (6,3.5) {ec-outport};
\node[opacity=0.3] (eciopn) at (6,3) {ec-inoutport};
\node[opacity=0.3] (edipn) at (6,2.5) {ed-inport};
\node[opacity=0.3] (edopn) at (6,2) {ed-outport};
\node[opacity=0.3] (ediopn) at (6,1.5) {ed-inoutport};
\node[opacity=0.3] (icpn) at (6,1) {ic-port};
\node[opacity=0.3] (idpn) at (6,0) {id-port};

\draw (id) -- (idp);
\draw (ic) -- (icp);
\draw (edio) -- (ediop);
\draw (edo) -- (edop);
\draw (edi) -- (edip);
\draw (ecio) -- (eciop);
\draw (eco) -- (ecop);
\draw (eci) -- (ecip);
\draw (e) -- (ec);
\draw (e) -- (ed);
\draw (ec) -- (eci);
\draw (ec) -- (eco);
\draw (ec) -- (ecio);
\draw (ed) -- (edi);
\draw (ed) -- (edo);
\draw (ed) -- (edio);
\draw (i) -- (ic);
\draw (i) -- (id);  
\end{tikzpicture}
\caption{Taxonomy of port types. The letter \emph{e} stands for external, \emph{i} for internal, \emph{c} for control and \emph{d} for data. External ports serve as inputs, outputs or both for facilitating interaction with the external world. Internal ports lack this characteristic as they solely mediate communication within the internals of a computon. For the sake of conciseness, we write dashes only after \emph{c} or \emph{d}, rather than writing them for each level of the hierarchy. For example, the path \emph{e-c-out-port} is abbreviated as \emph{ec-outport}. In total, there are eight possible paths, i.e., eight possible port types.}
\label{fig:port-taxonomy}
\end{figure}

The dichotomy between data and control in our proposal entails that a computon is a unit of control-driven computation wherein control signals and data values travel independently via edges. Ports are differentiated by colours which, in practice, can be abstract data types such as the type of booleans or the type of integers. In this paper, in order to provide a general ``type distinction framework'', ports are deliberatively coloured with natural numbers. Although non-concrete sets could be used in principle, $\mathbb{N}$ is particularly convenient because it simplifies the presentation and interpretation of our approach by offering familiar, ready-to-use elements (i.e., colours) that can later be associated with concrete types. For example, in one type system, $2$ may represent the type of booleans, $3$ the type of floats and $4$ the type of strings; whereas in another system, such numbers may denote characters, strings and integers, correspondingly. Such a description is straightforward precisely because we work with $\mathbb{N}$; using a fully abstract set would make this kind of explanations considerably less intuitive. In any case, we require a designated colour for control ports so we use $0$ for that purpose. Any other natural number is used for data ports, as shown in Figure \ref{fig:computon-syntax-tables}(a).

\begin{figure}[!h]
\centering
\subcaptionbox{Syntax for ports and flows.}
{
\begin{tblr}{
    colspec = {|c|c|c|c|c|c|},
    cell{6}{4} = {gray!10!white},cell{6}{5} = {gray!10!white},cell{6}{6} = {gray!10!white},
    cell{11}{4} = {gray!10!white},cell{11}{5} = {gray!10!white},cell{11}{6} = {gray!10!white}
  }
 \hline
 & & \small\emph{Syntax} & \small\emph{Colour Set} & \small\emph{Incoming Edges} & \small\emph{Outgoing Edges} \\
 \hline
 \SetCell[r=4]{m}\small Control & \small\emph{ec-inport} & \tikz{\node[draw=black,fill=white,inner sep=0pt,minimum size=3pt] (q0) at (0,0) {};} & \small${\{0\}}$ & \small Never & \small Always \\
 \cline{2-6}
 & \small\emph{ec-outport} & \tikz{\node[fill=black,inner sep=0pt,minimum size=3pt] (q0) at (0,0) {};} & \small${\{0\}}$ & \small Always & \small Never \\ 
 \cline{2-6}
 & \small\emph{ec-inoutport} & \tikz{\node[draw,fill=white,fill fraction={black}{0.5},inner sep=1.2pt,minimum size=3pt] (q) at (0,0) {};} & \small${\{0\}}$ & \small Never & \small Never \\ 
 \cline{2-6}
 & \small\emph{ic-port} & \tikz{\node[draw,fill=white,fill fraction={black}{0.5},inner sep=1.2pt,minimum size=3pt] (q) at (0,0) {};} & \small${\{0\}}$ & \small Always & \small Always \\ 
 \cline{2-6}
 & \small\emph{control flow} & \tikz{\flow{{0,0}}{{0.8,0}}{dashed}{};} &  & &  \\ 
 \hline
 \SetCell[r=4]{m}\small Data & \small\emph{ed-inport} & \tikz{\node[circle,draw=black,fill=white,inner sep=0pt,minimum size=3pt] (q0) at (0,0) {};} & \small${\mathbb{N}^+}$ & \small Never & \small Always \\
 \cline{2-6}
 & \small\emph{ed-outport} & \tikz{\node[circle,fill=black,inner sep=0pt,minimum size=3.4pt] (q0) at (0,0) {};} & \small${\mathbb{N}^+}$ & \small Always & \small Never \\ 
 \cline{2-6}
 & \small\emph{ed-inoutport} & \tikz{\node[circle,draw,fill=white,fill fraction={black}{0.5},inner sep=1.2pt,minimum size=3pt] (q) at (0,0) {};} & \small${\mathbb{N}^+}$ & \small Never & \small Never \\ 
 \cline{2-6}
 & \small\emph{id-port} & \tikz{\node[circle,draw,fill=white,fill fraction={black}{0.5},inner sep=1.2pt,minimum size=3pt] (q) at (0,0) {};} & \small${\mathbb{N}^+}$ & \small Always & \small Always \\ 
 \cline{2-6} 
 & \small\emph{data flow} & \tikz{\flow{{0,0}}{{0.8,0}}{}{};} & & &  \\
 \hline
\end{tblr}
}
\subcaptionbox{Syntax for computons. Here, ${n_1,\ldots,n_{10} \in \mathbb{N}^+}$ and the $\lambda$-box and triangles denote computation units. The other box corresponds to a composite computon whose internals depend on the composition operator being used. More details on such operators are provided in Section \ref{sec:composite-computons}.}[\textwidth]
{
\begin{tabular}{ |c|c| } 
		\hline         
         & \small\emph{Syntax} \\
		 \hline
		 \small\emph{Trivial Computon} &    	
    	\begin{tikzpicture}    	 
\node(del) at (0,0){};
\qmatch{q1}{1.2}{2.2}{}
\node[label={$\vdots$}] (dots1) at (1.2,1.35) {};
\qmatch{qi}{1.2}{1.3}{}
\dmatch{d1}{1.2}{0.9}{$n_1$}{left};
\node[label={$\vdots$}] (dots3) at (1.2,0.01) {};
\dmatch{dj}{1.2}{0}{$n_2$}{left};
\node(del2) at (2.4,0){};
    	\end{tikzpicture}
    	\\ \hline
		 \small\emph{Functional Computon} &    	
    	\begin{tikzpicture}
	    	\computonPrimitive{1}{0}{1}{1.9}{$\lambda$}
	    	
      \qin{1q0}{0.2}{1.7}{}
			\din{1i1}{0.2}{1.1}{$n_3$};
			\node[label={$\vdots$}] (1i2) at (0.6,0.25) {};
			\din{1in}{0.2}{0.2}{$n_4$};

			\qout{1q1}{2}{1.7}{}
			\dout{1o1}{2}{1.1}{$n_5$};
			\node[label={$\vdots$}] (1o2) at (2.4,0.25) {};
			\dout{1om}{2}{0.2}{$n_6$};
    	\end{tikzpicture}
    	\\ \hline
		 \small\emph{Fork Computon} &    	
    	\begin{tikzpicture}	    
			\fork{0}{0}{}{}{}
    	\end{tikzpicture}
    	\\ \hline
    	\small\emph{Join Computon} &    	
    	\begin{tikzpicture}	    
			\join{0}{0}{}{}{}
    	\end{tikzpicture}
    	\\ \hline
    	\small\emph{Composite Computon} &
    	\begin{tikzpicture}
    		\node at (-0.3,0){};
	    	\computonComposite{1.14}{-0.35}{1}{2.25}
	    	
    		\qin{1q0}{0.34}{1.7}{}
    		\node[label={$\vdots$}] (1oy) at (0.74,0.9) {};
    		\qin{1q1}{0.34}{0.9}{}
			\din{1i1}{0.34}{0.6}{$n_7$};
			\node[label={$\vdots$}] (1i2) at (0.74,-0.2) {};
			\din{1in}{0.34}{-0.15}{$n_8$};

			\qout{1q2}{2.14}{1.7}{}
			\node[label={$\vdots$}] (1ox) at (2.54,0.9) {};
			\qout{1q3}{2.14}{0.9}{}
			\dout{1o1}{2.14}{0.6}{$n_9$};
			\node[label={$\vdots$}] (1o2) at (2.54,-0.2) {};
			\dout{1om}{2.14}{-0.15}{$n_{10}$};
    	\end{tikzpicture}
    	\\ \hline
	\end{tabular}
}
\caption{Graphical syntax of the computon model.}
\label{fig:computon-syntax-tables}
\end{figure}

A glance at Figure \ref{fig:computon-syntax-tables}(a) reveals that control ports are associated with square shapes, whereas data ports are displayed as circles. As control ports are always zero-labelled, we will omit their colour for clarity purposes, and we just display the colour of data ports. Figure \ref{fig:computon-syntax-tables}(a) shows that e-inports and e-outports are depicted on white and black backgrounds, respectively. As e-inoutports are a combination of e-inports and e-outports, their background is half black and half white. It is important to mention that e-inports do not have any incoming edges but just outgoing ones, while e-outports only have incoming edges. Ec-inoutports have no edges at all and ic-ports have adjacent edges on both ends. So, even if they share the same graphical representation, ic-ports and ec-inoutports can be distinguished by their connected edges (the same is true for id-ports and ed-inoutports). We decide to use the same syntax for them because, intuitively, both i-ports and e-inoutports receive and forward information.

Figure \ref{fig:computon-syntax-tables}(b) displays the rest of the syntax we will be using throughout this paper to discuss the computon model. We acknowledge that the semantics of these diagrams is not obvious at this stage so we refer the reader to Sections \ref{sec:trivial-computons}, \ref{sec:primitive-computons} and \ref{sec:composite-computons} for further details on this matter. The use of boxed diagrams with port-based interfaces is getting increasingly popular in the literature on compositionality. Like existing notations, our graphical syntax maintains a clear distinction between computation units and ports, and differentiates between inputs (i.e., e-inports), outputs (i.e., e-outports) and input-outputs (i.e., i-ports). The difference lies in the support to distinguishing among different types of (high-level) computing devices which we refer to as computons (viz., trivial, functional, fork, join and composites). Distinguishing between them is of vital importance for immediately recognising individual computational behaviour while providing a clear visualisation of the structural parts of a composite. Structural clarity would be less evident if all computon types used identical syntax. 

Another key difference with respect to existing notations is that our syntax offers a clear syntactic separation between control and data ports and between control and data flows; thus, emphasising the explicit separation of control flow and data flow provided by our model. Furthermore, we introduce syntactic constructs for expressing input-output ports that are never connected to computation units, which we refer to as e-inoutports. 

A glance at Figure \ref{fig:computon-syntax-tables}(b) reveals that a collection of e-inoutports gives rise to what we call trivial computons (see Section \ref{sec:trivial-computons}). This figure also shows that functional computons are able to receive data and exactly one control signal before producing further data and a new control signal (see Section \ref{sec:primitive-computons}). Figure \ref{fig:computon-syntax-tables}(b) also shows that fork and join computons do not require or produce any data, but just control (see Section \ref{sec:primitive-computons}). All these properties can be immediately devised by just looking at our graphical syntax, without the need of delving into formal definitions. Trivial, functional, fork and join computons can be used to form even more complex entities which we refer to as composite computons (see Section \ref{sec:composite-computons}). 

Given the above reasons, we believe that our graphical syntax is more suitable than existing ones for discussing compositional construction over a wide range of diverse computons that explicitly separate control flow and data flow. Our syntax is indeed ready to be used by a visual programming language built on top of the computon model theory, which we intend to develop in the near future. 

In Section \ref{sec:primitive-computons}, we will see that functional, fork and join computons belong to the class of \emph{primitive computons}. We say they are primitive because each of them captures a (potentially low-level) computation in a single unit that can be composed into high-level computations (i.e., composite computons) via the operators we propose in Section \ref{sec:composite-computons}. Although different operational rules can be given to computation units due to the separation of execution from composition semantics, the expectation is that a unit must encapsulate a halting computation that can only be triggered when all the unit's e-inports have information. As a unit always has at least one ec-inport and ed-inports are optional, execution is necessarily driven by control. Upon termination, a unit must produce information in all its e-outports (which always include ec-outports), no matter whether computation is functional or relational. This operation strictly guarantees that composites do not diverge and that the intended semantics of sequencing, parallelising, branching and iteration are consistently preserved. Although we adhere to this strict operation scheme in Section \ref{sec:operational-semantics}, the computon model is flexible enough to accommodate other execution semantics for additional features such as time, costs, weights, pre-/post-conditions, probabilistic randomness or any combination thereof. 

Regardless of the chosen operation scheme, the internal structure of a computation unit must never be made visible from a high-level perspective, in order to support modularity of primitive computons and, by extension, of the composite computons built upon them. Composition does not depend on such internal details, but only on defined interfaces. Computon interfaces use natural numbers as an abstraction of data types to specify a contract that tells what inputs to supply and what outputs to expect from computation, without exposing internal workings such as internal control or data routing.

The way control and data are consumed within a computation unit depends on the concrete interpretation given to computation units themselves. For example, if a unit is a Turing machine for realising an $n$-ary computable function, the unit shall have at least one ec-inport (for receiving control signals) and exactly $n$ ed-inports (one for each argument). As previously stated, even if all the ed-inports have data, it is expected that the unit will not be triggered until receiving a control signal in each ec-inport. Once all e-inports have information, the corresponding Turing machine can be executed on the $n$ data inputs so input control signals can be safely disregarded.\footnote{Recall that any Turing machine can use special delimiters in the tape to process multiple inputs.} Upon termination, a control signal can be assigned to each unit's ec-outport to indicate that the machine's computation has terminated and that output data (stored in ed-outports) is ready to be consumed by other units.\footnote{A Turing machine's output can be read from the tape upon termination.} Consumption order is dictated by control flow within composite computons.

\subsection{Formal Definition}

Formally, a computon is a functor from \textbf{Comp} to \textbf{Set} (see Definition \ref{def:computon}) where \textbf{Set} is the category of finite sets and total functions and \textbf{Comp} is the free category generated by the following diagram:\footnote{Following the notation of function abstraction in Lambda Calculus, we use $\lambda$ to denote computons.}\,\footnote{Lifting conditions can be established to enforce (non-)emptiness, injectivity, surjectivity and even uniqueness, among others. Although such constraints can be straightforwardly applied on \textbf{Comp}, we do not present them for the clarity of argument. We simply assume that computons are well-defined if and only if the conditions imposed by Definition \ref{def:computon} are met (e.g., the surjectivity of the colouring function). A similar approach has been used in the context of databases for constraining database schemas \cite{spivak_database_2014}.}

\[
\begin{tikzcd}
 & O \arrow[dl, "\sigma"'] \arrow[dr, "t"] & &  \\
U & & P \arrow[r, "c"] & \Sigma \\
 & I \arrow[ul, "\tau"]\arrow[ur, "s"'] & &
\end{tikzcd}
\]

which consists of five objects and twelve morphisms (including identity and composite morphisms given by trivial paths and path concatenation, respectively).

\begin{definition} [Computon] \label{def:computon}
A computon $\lambda$ is a functor ${\textbf{Comp} \rightarrow \textbf{Set}}$ that maps:
\begin{itemize}
\item $U$ to a (possibly empty) set $\lambda(U)$ of computation units,
\item $P$ to a (non-empty) set $\lambda(P)$ of ports,
\item $O$ to a (possibly empty) set $\lambda(O)$ of edges,
\item $I$ to a (possibly empty) set $\lambda(I)$ of edges,
\item $\Sigma$ to a (non-empty) set ${\lambda(\Sigma) \subset \mathbb{N}}$ of colours where ${0\in\lambda(\Sigma)}$,
\item $\sigma$ to a surjective function $\lambda(\sigma):\lambda(O) \twoheadrightarrow \lambda(U)$ that specifies the outgoing edges of each computation unit,
\item $\tau$ to a surjective function $\lambda(\tau):\lambda(I) \twoheadrightarrow \lambda(U)$ that specifies the incoming edges of each computation unit,
\item $t$ to a function ${\lambda(t):\lambda(O) \rightarrow \lambda(P)}$ that specifies the incoming edges of each port,
\item $s$ to a function ${\lambda(s):\lambda(I) \rightarrow \lambda(P)}$ that specifies the outgoing edges of each port, and
\item $c$ to a surjective function ${\lambda(c): \lambda(P) \twoheadrightarrow \lambda(\Sigma)}$ that assigns to each port a colour
\end{itemize}
such that ${\sigma\restriction_{(c\circ t)^{-1}(0)}}$ is surjective, ${\tau\restriction_{(c\circ s)^{-1}(0)}}$ is surjective and there is:
\begin{itemize}
\item an identity function $1_{\lambda(x)}$ in \textbf{Set} for each object $x$ of $\textbf{Comp}$, 
\item a composite function $\lambda(g) \circ \lambda(f)$ in \textbf{Set} for each pair $(f,g)$ of composable morphisms in $\textbf{Comp}$,
\item at least one port ${p \in {[\lambda(P) \setminus Im(\lambda(s))]}}$ with ${\lambda(c)(p)=0}$ and
\item at least one port ${q \in {[\lambda(P) \setminus Im(\lambda(t))]}}$ with ${\lambda(c)(q)=0}$.
\end{itemize}
\end{definition}

As a computon $\lambda$ is a set-valued functor, it can be expressed in the form of a tuple ${(U,P,I,O,\Sigma,\sigma,\tau,t,s,c)}$. Without loss of generality, we took the liberty of simplifying the expression in order to reduce clutter, e.g., we write $U$ for $\lambda(U)$. For the rest of the paper, the reader must bear in mind that each component of $\lambda$ is an actual set or a function, not an object or a morphism in \textbf{Comp}. To distinguish between computons, we use natural numbers as subscripts which carry over computon components. If the symbol for a computon has no subscript, the computon components have no subscript either.

In Definition \ref{def:computon}, we write ${\sigma\restriction_{(c\circ t)^{-1}(0)}}$ and ${\tau\restriction_{(c\circ s)^{-1}(0)}}$ to express the restriction of the functions $\sigma$ and $\tau$ to the fibers ${(c\circ t)^{-1}(0)}$ and ${(c\circ s)^{-1}(0)}$, respectively. The surjectivity condition on these two functions ensures that every computation unit (if any) has at least one incoming edge and at least one outgoing edge connected from and to a $0$-coloured port, respectively. As every function we deal with is total and $\sigma$ and $\tau$ are surjective in general, every edge always runs from a port to a computation unit or viceversa. That is, a computon has neither dangling edges nor dangling computation units. When the set of units is empty, a computon is necessarily made up of coloured ports only. In Section \ref{sec:trivial-computons}, we will see that such a class of entities, referred to as trivial computons, is needed for the coherence of our theory.

\begin{definition}[Computon Interface] \label{def:computon-interface}
The interface of a computon $\lambda$ towards the external world is a tuple $(P^+,P^-)$ where ${P^+:=P\setminus Im(t)}$ is the set of e-inports of $\lambda$ and ${P^-:=P\setminus Im(s)}$ is the set of e-outports of $\lambda$. A port ${p \in Im(s) \cap Im(t)}$ is called an i-port of $\lambda$.
\end{definition}

\begin{notation}
Given a computon $\lambda$, $C^+$ denotes its set of ec-inports, $C^-$ its set of ec-outports, $D^+$ its set of ed-inports and $D^-$ its set of ed-outports. These sets are defined as follows:
\[
{C^\square := \{p \in P^\square \mid c(p) = 0\}} \text{ with } \square \in \{+,-\}
\]
\[
{D^\square := \{p \in P^\square \mid c(p) > 0\}} \text{ with } \square \in \{+,-\}
\]
\end{notation}

As the last two conditions of Definition \ref{def:computon} state that computons must have at least one ec-inport and at least one ec-outport, it trivially follows that ${C^+ \neq\emptyset\neq C^-}$. Data ports are optional so $D^+$ and $D^-$ can be empty. Fork computons are an example where ${D^+=\emptyset=D^-}$. Particularly, Figure \ref{fig:computon-syntax-tables}(b) shows that they always possess only one element in ${P^+}$ and exactly two elements in ${P^-}$. It also shows that this sort of computons have no i-ports either, i.e., ${Im(s)\cap Im(t)=\emptyset}$. For more technical details on fork computons, see Section \ref{sec:primitive-computons}.

Notice in Definition \ref{def:computon-interface} that the sets $P^+$ and $P^-$ are not necessarily disjoint so a port can be e-inport and e-outport at the same time. If ${p \in C^+ \cap C^-}$, then $p$ is an \emph{ec-inoutport}. If ${p \in D^+ \cap D^-}$, it is an \emph{ed-inoutport}. Figure \ref{fig:computon-syntax-tables}(b) shows that trivial computons have all e-inoutports, i.e., ${P=P^+\cap P^-}$. For more technical details on them, see Section \ref{sec:trivial-computons}.

In Figure \ref{fig:computon-syntax-tables}(a), it is indicated that e-inports and e-outports only possess outgoing and incoming edges, respectively. This property follows from Proposition \ref{prop:eports-edges}. Using Definition \ref{def:computon-interface}, it is easy to additionally show that e-inoutports have no edges at all and that i-ports have both incoming and outgoing edges. Thus, even if ic-ports and ec-inoutports share the same graphical representation, they can be distinguished by their connected edges, with the same being true for id-ports with respect to id-inoutports. 

\begin{proposition}\label{prop:eports-edges}
If $\lambda$ is a computon, then ${P^+\setminus P^- = P^+ \cap Im(s)}$ and ${P^-\setminus P^+ = P^- \cap Im(t)}$.
\end{proposition}
\begin{proof}
Let ${p \in P^+\setminus P^-}$ be a port of a computon $\lambda$. Then, ${p \in P^+\setminus P^- \iff p \in P^+ \land p \notin P^-}$ ${\iff p \in P^+ \land \lnot(p \in P \land p \notin Im(s))}$ by Definition \ref{def:computon-interface} ${\iff p \in P^+ \land (p \notin P \lor p \in Im(s))}$\\${\iff p \in P^+ \land p \in Im(s)}$ because ${p \in P}$ is always true ${\iff p \in P^+ \cap Im(s)}$. The proof of ${P^-\setminus P^+ = P^- \cap Im(t)}$ is completely analogous.
\end{proof}

As control ports and data ports are identified as separate entities, information movement within a computon corresponds to either data flow or control flow. Particularly, we say that any control port is connected to or from a computation unit via a control flow edge, whereas a data port is connected analogously but with a data flow edge (see Definitions \ref{def:flow} and \ref{def:flow-edges}). The collection of ports receiving information from a computation unit $u$ and sending information to $u$ are denoted $u \bullet$ and $\bullet u$, respectively. Similarly, $\bullet p$ and $p \bullet$ denote the source and target computation units of a port $p$, respectively (see Definition \ref{def:pre-post-sets}). When there is information flow from every e-inport or i-port to some e-outport, we say that the computon is connected (see Definition \ref{def:computon-connected}). 

\begin{definition} [Information Flow] \label{def:flow}
Given a computon $\lambda$, let $p \in P$ and $u \in U$. We say there is information flow from $p$ to $u$ if there is an edge $i \in I$ such that $s(i)=p$ and $\tau(i)=u$. This is denoted ${p \xrightarrow{i} u}$. If there is an edge $o \in O$ with $\sigma(o)=u$ and $t(o)=p$, we say there is information flow from $u$ to $p$, written ${u \xrightarrow{o} p}$. We use ${p_1 \xrightarrow{\exists} p_n}$ to denote the existence of ${p_1 \xrightarrow{i_1} u_1 \xrightarrow{o_1} p_2 \xrightarrow{i_2} u_2 \xrightarrow{o_2} \cdots \xrightarrow{i_{n-1}} u_{n-1} \xrightarrow{o_{n-1}} p_n}$ for $p_1,\ldots,p_n \in P$, $u_1,\ldots,u_{n-1} \in U$, $o_1,\ldots,o_{n-1} \in O$, $i_1,\ldots,i_{n-1} \in I$ and $n\geq 2$.
\end{definition}

\begin{definition} [Control Flow and Data Flow Edges] \label{def:flow-edges}
Given a computon $\lambda$ and an edge ${e \in I \cup O}$, we say $e$ represents control flow if ${c(s(e))=0}$ or ${c(t(e))=0}$; otherwise, it represents data flow.
\end{definition}

\begin{definition}[Pre- and Post-Sets] \label{def:pre-post-sets}
For a computation unit ${u \in U}$ of a computon $\lambda$, ${\bullet u}$ and ${u \bullet}$ denote the sets ${\{p \in P \mid (\exists i \in I)(p \xrightarrow{i} u)\}}$ and ${\{p \in P \mid (\exists o \in O)(u \xrightarrow{o} p)\}}$, respectively. Similarly, for a port ${p \in P}$, ${\bullet p}$ and ${p \bullet}$ denote the sets $\{u \in U \mid (\exists o \in O)(u \xrightarrow{o} p)\}$ and $\{u \in U \mid (\exists i \in I)(p \xrightarrow{i} u)\}$, respectively.
\end{definition}

\begin{definition} [Connected Computon] \label{def:computon-connected}
We say that a computon $\lambda$ is connected if and only if for each ${p \in Im(s)\cup P^+}$ there exists some ${q \in P^-}$ such that ${p \xrightarrow{\exists} q}$ holds.
\end{definition}

A glance at Figure \ref{fig:computon-syntax-tables}(b) reveals that functional computons satisfy Definition \ref{def:computon-connected}, whereas trivial ones do not. In general, the existence of at least one e-inoutport implies that the corresponding computon is not connected, as captured by the contrapositive of Proposition \ref{prop:computon-connected-isolated-port}. When a computon is connected, there always are computation units and there is information flow from every e-outport or i-port to some e-inport (see Propositions \ref{prop:computon-connected-alwaysunits} and \ref{prop:computon-connected-reverse}). 

\begin{proposition} \label{prop:computon-connected-isolated-port}
If $\lambda$ is a connected computon, ${P^+\cap P^-=\emptyset}$.
\end{proposition}
\begin{proof}
Assume for contradiction that $\lambda$ is connected with some ${p \in P^+\cap P^-}$ so that ${p\notin Im(t)}$ and ${p\notin Im(s)}$. Using Definition \ref{def:computon-connected}, we know there must be some ${q \in P^-}$ where ${p\xrightarrow{\exists}q}$, i.e., ${s(p)=i}$ for some ${i\in I}$ as per Definition \ref{def:flow}. As this clearly contradicts ${p\notin Im(s)}$, we conclude ${P^+\cap P^-=\emptyset}$.
\end{proof}

\begin{proposition} \label{prop:computon-connected-alwaysunits}
Every connected computon has at least one computation unit.
\end{proposition}
\begin{proof}
Assume for contradiction $\lambda$ is a connected computon with ${U=\emptyset}$, meaning $\sigma$ and $\tau$ are well-defined only if ${I=\emptyset=O}$. Since $\lambda$ is connected, for each ${p \in Im(s)\cup P^+}$ there is some ${q \in P^-}$ for which ${p \xrightarrow{\exists} q}$ holds. 
\begin{itemize}
\item If ${p \in Im(s)}$, there must be some ${i\in I}$ where ${s(i)=p}$; thereby, contradicting ${I=\emptyset}$.
\item If ${p \in P^+}$, we have two possibilities:
\begin{itemize}
\item There is some ${i\in I}$ where ${s(i)=p}$ which also contradicts ${I=\emptyset}$.
\item There is no ${i\in I}$ where ${s(i)=p}$ so there is no information flow from $p$ to any other port, including e-outports; thus, contradicting the fact that $\lambda$ is connected.
\end{itemize} 
\end{itemize}
Therefore, we conclude ${U \neq \emptyset}$.
\end{proof}

\begin{proposition} \label{prop:computon-connected-reverse}
If $\lambda$ is a connected computon, for each ${q\in P^-\cup Im(t)}$ there is a port ${p \in P^+}$ where ${p \xrightarrow{\exists}q}$.
\end{proposition}
\begin{proof}
If $\lambda$ is a connected computon with ${q\in P^-\cup Im(t)}$, we know by Proposition \ref{prop:computon-connected-isolated-port} that ${q \in P^-\cap P^+}$ cannot hold. So, if ${q \in P^-}$, then ${q \notin P^+}$; in other words, ${q \notin Im(s)}$ and ${q \in Im(t)}$. So, it suffices to show ${q \in Im(t)}$ only. For this, we perform the following recursive procedure by initially considering the set ${V:=\{q\}}$ of visited ports: 

As ${q \in Im(t)}$, there is some ${o\in O}$ where ${t(o)=q}$. Applying the surjectivity of $\sigma$ and $\tau$ and the totality of $s$, we derive ${p\xrightarrow{i}u\xrightarrow{o}q}$ for some ${u \in U}$, some ${i \in I}$ and some ${p\in P \setminus V}$. If ${p \in P^+}$, then the proof is complete. Otherwise, ${p\in Im(t)}$ so we simply repeat this procedure with ${V \cup \{p\}}$ in place of $V$ and $p$ in place of $q$.

The above procedure will eventually terminate since the sets ${U, I, O}$ and $P$ are finite, and because we use ${V \subseteq P}$ to explore new ports at each step. It is guaranteed that an e-inport will be chosen for exploration at some point since every computon has at least one ec-inport and there are no dangling computation units (see Definition \ref{def:computon}).
\end{proof}

At this stage, we have provided sufficient details about the general structure of computons by treating them as set-valued functors. Defining computons in this way gives rise to a functor category which we refer to as $\textbf{Set}^\textbf{Comp}$.
\subsection{The Category of Computons}

$\textbf{Set}^\textbf{Comp}$ is a category whose objects and morphisms are computons and computon morphisms, respectively (see Definition \ref{def:computon-morphism}).

\begin{definition} [Computon morphism] \label{def:computon-morphism}
If $\lambda_1$ and $\lambda_2$ are two computons, a computon morphism ${\alpha\colon \lambda_1 \rightarrow \lambda_2}$ is a natural transformation whose components are the total functions ${\alpha_U\colon U_1 \rightarrow U_2}$, ${\alpha_P\colon P_1 \rightarrow P_2}$, ${\alpha_O\colon O_1 \rightarrow O_2}$, ${\alpha_I\colon I_1 \rightarrow I_2}$ and ${\alpha_\Sigma\colon \Sigma_1 \hookrightarrow \Sigma_2}$ such that the diagrams of Figure \ref{fig:computon-morphism-commutative-diagrams} commute and ${\vec{i}(\alpha) \cup \vec{o}(\alpha) \subseteq P_1^+ \cup P_1^-}$. Here, $\vec{i}(\alpha)$ and $\vec{o}(\alpha)$ denote ${\{p_1 \in P_1 \mid \bullet \alpha(p_1) \setminus \alpha(\bullet p_1) \neq \emptyset\}}$ and ${\{p_1 \in P_1 \mid \alpha(p_1) \bullet \setminus \alpha(p_1 \bullet) \neq \emptyset\}}$, respectively. 

\begin{figure}[!h]
\centering
\begin{tikzcd}
P_1 \arrow[r, twoheadrightarrow, "c_1"]\arrow[d, "\alpha_P"']
& \Sigma_1 \arrow[d, hook, "\alpha_\Sigma"] \\
P_2 \arrow[r, twoheadrightarrow, "c_2"']
& \Sigma_2
\end{tikzcd}
\qquad\quad
\begin{tikzcd}
I_1 \arrow[r, twoheadrightarrow, "\tau_1"]\arrow[d, "\alpha_I"']
& U_1 \arrow[d, "\alpha_U"] 
& O_1 \arrow[l, twoheadrightarrow, "\sigma_1"']\arrow[d, "\alpha_O"] \\
I_2 \arrow[r, twoheadrightarrow, "\tau_2"']
& U_2 
& O_2 \arrow[l, twoheadrightarrow, "\sigma_2"]
\end{tikzcd}
\qquad\quad
\begin{tikzcd}
I_1 \arrow[r, "s_1"]\arrow[d, "\alpha_I"']
& P_1 \arrow[d, "\alpha_P"] 
& O_1 \arrow[l, "t_1"']\arrow[d, "\alpha_O"] \\
I_2 \arrow[r, "s_2"']
& P_2 
& O_2 \arrow[l, "t_2"]
\end{tikzcd}
\caption{A computon morphism is a natural transformation $\alpha\colon\lambda_1 \rightarrow \lambda_2$.}
\label{fig:computon-morphism-commutative-diagrams}
\end{figure}
\end{definition}

\begin{notation} \label{notation:sets-im}
To simplify notation when referring to the components of a computon morphism $\alpha\colon \lambda_1 \rightarrow \lambda_2$, we write $\alpha(u)$ for $\alpha_U(u)$, $\alpha(p)$ for $\alpha_P(p)$, $\alpha(o)$ for $\alpha_O(o)$ and $\alpha(i)$ for $\alpha_I(i)$. For the rest of the paper, we also write $\alpha(A)$ to denote $Im(\alpha_P \vert_A)$ if $A \subseteq P_1$ or $Im(\alpha_U \vert_A)$ if $A \subseteq U_1$. Likewise, we use $\alpha(B)^{-1}$ to denote $\{p_1 \in P_1 \mid \alpha(p_1) \in B\}$ if $B \subseteq P_2$ or $\{u_1 \in U_1 \mid \alpha(u_1) \in B\}$ if $B \subseteq U_2$.
\end{notation}

\begin{remark} \label{remark:monomorphism} 
Naturally, composition of computon morphisms $\alpha$ and $\beta$ is defined component-wise:
\begin{equation*}
\begin{split}
(\beta_U, \beta_P, \beta_I, & \beta_O, \beta_\Sigma) \circ (\alpha_U, \alpha_P, \alpha_I, \alpha_O, \alpha_\Sigma) = (\beta_U \circ \alpha_U, \beta_P \circ \alpha_P, \beta_I \circ \alpha_I, \beta_O \circ \alpha_O, \beta_\Sigma \circ \alpha_\Sigma)
\end{split}
\end{equation*}
\end{remark}

Figure \ref{fig:computon-morphism-example}(a) describes a computon morphism ${\alpha}$ from ${\lambda_1}$ to ${\lambda_2}$. The top-level diamond specifies that ${\lambda_1}$ has ports ${p_1 \in P_1}$ and ${p_2 \in P_1}$ connected to and from a computation unit ${u_1 \in U_1}$ through the edges ${i_1 \in I_1}$ and ${o_1 \in O_1}$, respectively. Thereby, forming the information flows ${p_1 \xrightarrow{i_1} u_1 \xrightarrow{o_1} p_2}$. As $p_1$ has no incoming edges and $p_2$ has no outgoing edges, we use Definition \ref{def:computon-interface} to deduce ${P_1^+:=\{p_1\}}$ and ${P_1^-:=\{p_2\}}$. Since both $p_1$ and $p_2$ are zero-coloured and they are the only ports in $\lambda_1$, we further deduce ${P_1^+=C_1^+}$, ${P_1^-=C_1^-}$ and ${D_1^+=\emptyset=D_1^-}$. Therefore, by Definition \ref{def:flow-edges}, both $i_1$ and $o_1$ denote control flow. The top-level diagram of Figure \ref{fig:computon-morphism-example}(b) shows the graphical representation of $\lambda_1$ using computon syntax.\footnote{Recall that ports are coloured with natural numbers and edges are not coloured at all. The diagram on the right-hand side of Figure \ref{fig:computon-morphism-example} just displays port and flow labels for illustrative purposes. For now, we just display computation units as labels but, in upcoming sections, we will use specific syntax to distinguish among different types of such units.} 

\begin{figure}[!h]
\centering
\subcaptionbox{Using commutative diagrams.}
{
\begin{tikzcd}[ampersand replacement=\&]
 \& \{o_1\} \arrow[dl, twoheadrightarrow, "\sigma_1(o_1)=u_1"'] \arrow[dr, "t_1(o_1)=p_2"] \arrow[ddd, bend left=40, opacity=0.4, "\alpha_O(o_1)=o_2"'{yshift=15pt}] \& \&  \\
\{u_1\} \arrow[ddd, bend right=20, opacity=0.4, "\alpha_U(u_1)=u_2"'] \& \& \{p_1,p_2\} \arrow[r, twoheadrightarrow, "\substack{c_1(p_1)=0 \\ c_1(p_2)=0}"] \arrow[ddd, bend left=20, opacity=0.4, "\substack{\alpha_P(p_1)=p_3 \\ \alpha_P(p_2)=p_5}"] \& \{0\} \arrow[ddd, hook, bend left=20, opacity=0.4, "\alpha_\Sigma(0)=0"] \\
 \& \{i_1\} \arrow[ul, twoheadrightarrow, "\tau_1(i_1)=u_1"]\arrow[ur, "s_1(i_1)=p_1"'] \arrow[ddd, bend right=40, opacity=0.4, "\alpha_I(i_1)=i_2"{yshift=-15pt}] \& \& \\
 \& \{o_2,o_3\} \arrow[dl, twoheadrightarrow, "\substack{\sigma_2(o_2)=u_2 \\ \sigma_2(o_3)=u_3}"'] \arrow[dr, "\substack{t_2(o_2)=p_5 \\ t_2(o_3)=p_5}"] \& \&  \\
\{u_2,u_3\} \& \& \{p_3,p_4,p_5\} \arrow[r, twoheadrightarrow, "\substack{c_2(p_3)=0 \\ c_2(p_4)=0 \\ c_2(p_5)=0}"'] \& \{0\} \\
 \& \{i_2,i_3\} \arrow[ul, twoheadrightarrow, "\substack{\tau_2(i_2)=u_2 \\ \tau_2(i_3)=u_3}"]\arrow[ur, "\substack{s_2(i_2)=p_3 \\ s_2(i_3)=p_4}"'] \& \&
\end{tikzcd}  
}\hspace{-0.3cm}
\subcaptionbox{Using computon syntax: \tikz{\node[draw=black,fill=white,inner sep=0pt,minimum size=3pt,label={right:\scriptsize Ec-inport}] at (0,0) {}; \node[fill=black,inner sep=0pt,minimum size=3pt,label={right:\scriptsize Ec-outport}] at (2,0) {};} \tikz{\draw[dashed] (0.2,0) to node[pos=0.5,yshift=0]{\arrowflow} (0.7,0); \& \node[inner sep=0pt,minimum size=3pt,label={[xshift=0.8cm]right:\scriptsize Control flow edge}]{};} }
{
\vspace{1cm}
\begin{tikzpicture}
\node[draw=black,fill=white,inner sep=0pt,minimum size=3pt,label={[text opacity=0.4]left:\scriptsize $p_1$},opacity=0.4] (p1) at (0,3.5) {};
\draw[dashed,opacity=0.4] (0.05,3.5) to [] node [pos=0.5] {\arrowflow} node [pos=0.5,yshift=7] {\scriptsize $i_1$} (1,3.5);
\node[opacity=0.4] (u1) at (1.2,3.5) {\scriptsize $u_1$};
\draw[dashed,opacity=0.4] (1.35,3.5) to [] node [pos=0.5] {\arrowflow} node [pos=0.5,yshift=7, opacity=0.4] {\scriptsize $o_1$} (2.1,3.5);
\node[fill=black,inner sep=0pt,minimum size=3pt,label={[text opacity=0.4]right:\scriptsize $p_2$},opacity=0.4] (p2) at (2.15,3.5) {};

\draw[->,opacity=0.4,bend right=20] (0,3.3) to node[left]{} (-0.1,-0.3);
\draw[->,opacity=0.4,bend right=20] (0.5,3.3) to node[left]{} (0.5,0);
\draw[->,opacity=0.4] (1.2,3.3) to node[right]{} (1.2,-0.2);
\draw[->,opacity=0.4,bend left=20] (1.7,3.3) to node[right]{} (1.7,0);
\draw[->,opacity=0.4,bend left=20] (2.2,3.3) to node[right]{} (2.2,-0.3);

\node[draw=black,fill=white,inner sep=0pt,minimum size=3pt,label={[text opacity=0.4]left:\scriptsize $p_3$},opacity=0.4] (p3) at (0,-0.5) {};
\draw[dashed,opacity=0.4] (0.05,-0.5) to [] node [pos=0.5] {\arrowflow} node [pos=0.5,yshift=7, opacity=0.4] {\scriptsize $i_2$} (1,-0.5);
\node[opacity=0.4] (u2) at (1.2,-0.5) {\scriptsize $u_2$};
\draw[dashed,opacity=0.4] (1.35,-0.5) to [] node [pos=0.5] {\arrowflow} node [pos=0.5,yshift=7, opacity=0.4] {\scriptsize $o_2$} (2.1,-0.5);
\node[fill=black,inner sep=0pt,minimum size=3pt,label={[text opacity=0.4]right:\scriptsize $p_5$},opacity=0.4] (p5) at (2.15,-0.5) {};

\node[draw=black,fill=white,inner sep=0pt,minimum size=3pt,label={left:\scriptsize $p_4$}] (p4) at (0,-1) {};
\draw[dashed] (0.05,-1) to [] node [pos=0.5] {\arrowflow} node [pos=0.5,yshift=-7] {\scriptsize $i_3$} (1,-1);
\node (u3) at (1.2,-1) {\scriptsize $u_3$};
\draw[dashed] (1.4,-1) to [] node [pos=0.3,rotate=27] {\arrowflow} node [below] {\scriptsize $o_3$} (p5);
\end{tikzpicture}   
}
\caption{Example of a computon morphism ${\alpha\colon\lambda_1 \rightarrow \lambda_2}$.}
\label{fig:computon-morphism-example}
\end{figure}

The diamond at the bottom of Figure \ref{fig:computon-morphism-example}(a) specifies that $\lambda_2$ has ports ${p_3 \in P_2}$ and ${p_5 \in P_2}$ connected to and from a computation unit ${u_2 \in U_2}$ via the edges ${i_2 \in I_2}$ and ${o_2 \in O_2}$, respectively. This computon also includes the edge ${i_3 \in I_2}$ for connecting the port ${p_4 \in P_2}$ to a computation unit ${u_3 \in U_2}$ which, in turn, is connected to $p_5$ via the edge ${o_3 \in O_2}$. Thereby, forming the information flow ${p_3 \xrightarrow{i_2} u_2 \xrightarrow{o_2} p_5 \xleftarrow{o_3} u_3 \xleftarrow{i_3} p_4}$. Similar to $\lambda_1$, we observe that $p_3$ and $p_4$ have no incoming edges and that $p_5$ has no outgoing edges. So, by Definition \ref{def:computon-interface}, ${P_2^+:=\{p_3,p_4\}}$ and ${P_2^-:=\{p_5\}}$. As all ports of $\lambda_2$ are also zero-coloured, it follows that ${P_2^+=C_2^+}$, ${P_2^-=C_2^-}$ and ${D_2^+=\emptyset=D_2^-}$. By Definition \ref{def:flow-edges}, this means all the edges of $\lambda_2$ represent control flow. The bottom-level diagram of Figure \ref{fig:computon-morphism-example}(b) displays the graphical structure of $\lambda_2$ through the use of computon syntax.

The components of the computon morphism ${\alpha\colon\lambda_1 \rightarrow \lambda_2}$ are displayed in gray on Figure \ref{fig:computon-morphism-example}(a), in order to distinguish them from the diamond diagrams that define computons. Figure \ref{fig:computon-morphism-example}(b) shows that this morphism maps the unique computation unit of $\lambda_1$ to $u_2$, the sole ec-inport of $\lambda_1$ to the ec-inport $p_3$, the unique ec-outport of $\lambda_1$ to the unique ec-outport of $\lambda_2$, the edge ${i_1}$ to the edge ${i_2}$ and the edge ${o_1}$ to the edge ${o_2}$. This mapping is sound since it ensures the diagrams of Figure \ref{fig:computon-morphism-example}(a) commute, meaning that the structure of $\lambda_1$ is preserved within $\lambda_2$. 

Observe in $\lambda_2$ that $u_2$ and $u_3$ are the only computation units connected to $p_5$ so that ${\bullet p_5=\{u_2,u_3\}}$ by Definition \ref{def:pre-post-sets}. Considering ${\alpha_P(p_2)=p_5}$, as shown in Figure \ref{fig:computon-morphism-example}(a), we have ${\bullet p_5=\{u_2,u_3\}=\bullet \alpha_P(p_2)}$. A further inspection of Figure \ref{fig:computon-morphism-example}(a) reveals that the set of all computation units mapped from $u_1$ is $\{u_2\}$ because ${\alpha_U(u_1)=u_2}$  is the only mapping given by $\alpha_U$, i.e., ${\alpha_U(\{u_1\})=Im(\alpha_U\mid_{\{u_1\}})=\{u_2\}}$, according to Notation \ref{notation:sets-im}. Definition \ref{def:pre-post-sets} allows us to deduce ${\{u_1\}=\bullet p_2}$ because $u_1$ is the only computation unit connected to $p_2$. So, ${\alpha_U(\{u_1\})=Im(\alpha_U\mid_{\{u_1\}})=\{u_2\}=\alpha_U(\bullet p_2)}$. As ${\bullet\alpha_P(p_2)\setminus\alpha_U(\bullet p_2)=\{u_2,u_3\}\setminus\{u_2\}\neq\emptyset}$, we use Definition \ref{def:computon-morphism} to conclude ${p_2 \in \vec{i}(\alpha)}$. Basically, ${\vec{i}(\alpha)}$ denotes the set ports in $\lambda_1$ that are mapped to ports in $\lambda_2$ connected to computation units not included in the $\alpha$-embedding (in this case, $u_3$). The set ${\vec{o}(\alpha)}$ is similar but contains $\lambda_1$-ports that are mapped to $\lambda_2$-ports connected to computation units excluded from the $\alpha$-embedding. As we did to verify ${p_2 \in \vec{i}(\alpha)}$, we can show ${\vec{o}(\alpha)=\emptyset}$ in our example. This is because the information flow ${p_1 \xrightarrow{i_1} u_1 \xrightarrow{o_1} p_2}$ of $\lambda_1$ is entirely inserted into the information flow ${p_3 \xrightarrow{i_2} u_2 \xrightarrow{o_2} p_5}$ of $\lambda_2$, and there are no flows of the form ${p_3 \xrightarrow{} u}$ or ${p_5 \xrightarrow{} u}$ in $\lambda_2$.

The example we just described allows us to intuitively perceive a computon morphism as an embedding (or an insertion) of a computon into a (potentially more complex) one, while preserving ports (with their respective colours, incoming edges and outgoing edges) and computation units (with their respective incoming and outgoing edges). As a result of this preservation, an e-inport is mapped to an e-inport or an i-port (e.g., ${\alpha_P(p_1)=p_3 \in P_2^+}$ in Figure \ref{fig:computon-morphism-example}) --- see Propositions \ref{prop:computon-morphism-inports-outports}, \ref{prop:computon-morphism-eports-equality} and \ref{prop:computon-morphism-eports-become-iports}. Similarly, an e-outport is mapped to an e-outport or an i-port  (e.g., ${\alpha_P(p_2)=p_5 \in P_2^-}$ in Figure \ref{fig:computon-morphism-example}) --- see Propositions \ref{prop:computon-morphism-inports-outports}, \ref{prop:computon-morphism-eports-equality} and \ref{prop:computon-morphism-eports-become-iports}. While a computon morphism can map external ports to internal ones, the latter can never be mapped to external ports due to the commutative diagrams presented in Figure \ref{fig:computon-morphism-commutative-diagrams}.

\begin{proposition} \label{prop:computon-morphism-inports-outports}
If ${\alpha\colon\lambda_1 \rightarrow \lambda_2}$ is a computon morphism, ${\alpha^{-1}(P_2^+) \subseteq P_1^+}$ and ${\alpha^{-1}(P_2^-) \subseteq P_1^-}$.
\end{proposition}
\begin{proof}
By letting ${\alpha\colon\lambda_1 \rightarrow \lambda_2}$ be a computon morphism, we only prove ${\alpha^{-1}(P_2^+) \subseteq P_1^+}$ by contrapositive, since the proof of ${\alpha^{-1}(P_2^-) \subseteq P_1^-}$ is completely analogous. 

If ${p_1 \in P \setminus P_1^+}$, then ${(\exists o_1 \in O_1)[t_1(o_1)=p_1]}$ (see Definition \ref{def:computon-interface}). By commutativity, we know ${t_2(\alpha(o_1))=\alpha(t_1(o_1))=\alpha(p_1)}$ which implies ${\alpha(p_1) \notin P_2^+}$ (see Definition \ref{def:computon-interface}) and, consequently, ${p_1 \notin \alpha^{-1}(P_2^+)}$. As the implication ${p_1 \notin P_1^+ \implies p_1 \notin \alpha^{-1}(P_2^+)}$ is logically equivalent to ${p_1 \in \alpha^{-1}(P_2^+) \implies p_1 \in P_1^+}$, we conclude ${\alpha^{-1}(P_2^+) \subseteq P_1^+}$, as required.
\end{proof}

\begin{proposition} \label{prop:computon-morphism-eports-equality}
If ${\alpha\colon\lambda_1 \rightarrow \lambda_2}$ is a computon morphism, ${P_1^+\cap\vec{i}(\alpha)=\emptyset \implies \alpha^{-1}(P_2^+)=P_1^+}$ and ${P_1^- \cap \vec{o}(\alpha)=\emptyset \implies \alpha^{-1}(P_2^-)=P_1^-}$.
\end{proposition}
\begin{proof}
Let ${\alpha\colon\lambda_1 \rightarrow \lambda_2}$ be a computon morphism and assume ${P_1^+ \cap \vec{i}(\alpha)=\emptyset}$. This assumption says that if ${p_1 \in P_1^+}$ then ${p_1 \notin \vec{i}(\alpha)}$ so that ${\bullet \alpha(p_1) \setminus \alpha(\bullet p_1)=\emptyset}$ which is true when ${\bullet \alpha(p_1) = \alpha(\bullet p_1)}$. As ${\bullet p_1 = \emptyset}$ because ${p_1 \in P_1^+}$, we have ${\bullet \alpha(p_1) = \emptyset = \alpha(\bullet p_1)}$. The fact ${\bullet \alpha(p_1) = \emptyset}$ implies ${\alpha(p_1) \in P_2^+}$, i.e., ${p_1 \in \alpha^{-1}(P_2^+)}$. Thus, proving ${P_1^+ \subseteq \alpha^{-1}(P_2^+)}$. Since ${\alpha^{-1}(P_2^+) \subseteq P_1^+}$ also holds by Proposition \ref{prop:computon-morphism-inports-outports}, we conclude $P_1^+ \cap \vec{i}(\alpha)=\emptyset \implies \alpha^{-1}(P_2^+)=P_1^+$.

The proof of ${P_1^- \cap \vec{o}(\alpha)=\emptyset \implies \alpha^{-1}(P_2^-)=P_1^-}$ follows analogously.
\end{proof}

\begin{proposition} \label{prop:computon-morphism-eports-become-iports}
If $\lambda_1$ is a connected computon and ${\alpha\colon\lambda_1 \rightarrow \lambda_2}$ is a computon morphism, ${(P_1^+ \cap \vec{i}(\alpha)) \cup (P_1^- \cap \vec{o}(\alpha)) \subseteq \alpha^{-1}(Im(t_2) \cap Im(s_2))}$.
\end{proposition}
\begin{proof}
Let ${\alpha\colon\lambda_1 \rightarrow \lambda_2}$ be a computon morphism from a connected computon $\lambda_1$ to an arbitrary computon $\lambda_2$. If ${p_1 \in P_1^+ \cap \vec{i}(\alpha)}$, then there is some ${u_2 \in \bullet\alpha(p_1)\setminus\alpha(\bullet p_1)}$. By Definition \ref{def:pre-post-sets}, there must also be some ${o_2 \in O_2}$ where ${\sigma_2(o_2)=u_2}$ and ${t_2(o_2)=\alpha(p_1)}$. That is, ${\alpha(p_1) \in Im(t_2)}$.

Now, since $\lambda_1$ is a connected computon and ${p_1 \in P_1^+}$, there is some ${i_1 \in I_1}$ and some ${u_1 \in U_1}$ where ${p_1 \xrightarrow{i_1} u_1}$ holds (see Definition \ref{def:computon-connected} and Proposition \ref{prop:computon-connected-alwaysunits}). By commutativity and because ${s_1(i_1)=p_1}$, ${s_2(\alpha(i_1))=\alpha(s_1(i_1))=\alpha(p_1)}$. That is, ${\alpha(p_1) \in Im(s_2)}$.

As having ${\alpha(p_1) \in Im(t_2) \cap Im(s_2)}$ implies ${p_1 \in \alpha^{-1}(Im(t_2) \cap Im(s_2))}$, we have just proved ${P_1^+ \cap \vec{i}(\alpha) \subseteq \alpha^{-1}(Im(t_2) \cap Im(s_2))}$. The proof of ${P_1^- \cap \vec{o}(\alpha) \subseteq \alpha^{-1}(Im(t_2) \cap Im(s_2))}$ is completely analogous. 
\end{proof}

\subsection{Colimits in the Category of Computons}

Although computons are set-valued functors, general colimits in $\textbf{Set}^\textbf{Comp}$ do not always exist because morphisms need to satisfy the special conditions imposed by Definition \ref{def:computon-morphism} and there are no initial objects, i.e., $\textbf{Set}^\textbf{Comp}$ is not cocomplete. When colimits exist, they are canonically computed component-wise in $\textbf{Set}$ so they constitute formal category-theoretic operations to glue multiple computons together, according to the instructions given by the morphisms of a certain diagram. For example, a pushout construction glues computons by identifying their common parts in the form of an \emph{apex computon}. This notion is formalised in Definition \ref{def:pushout-computation}. 

\begin{definition} [Pushout Construction] \label{def:pushout-computation}
Given a span ${\lambda_1 \xleftarrow{\alpha_1} \lambda_0 \xrightarrow{\alpha_2} \lambda_2}$ of computon morphisms, the pushout of the corresponding diagram in $\textbf{Set}^\textbf{Comp}$:
\begin{center}
\begin{tikzcd}
 & \lambda_0 \arrow[dl, "\alpha_1"'] \arrow[dr, "\alpha_2"] &  \\
\lambda_1 \arrow[dr, "\beta_1"'] & & \lambda_2 \arrow[dl, "\beta_2"] \\
 & \lambda_3 & 
\end{tikzcd}
\end{center}
denoted ${(\beta_1\colon\lambda_1 \rightarrow \lambda_3,\lambda_3,\beta_2\colon\lambda_2 \rightarrow \lambda_3)}$ or ${\lambda_1+_{\lambda_0}\lambda_2}$, is obtained by computing the pushout in \textbf{Set} of each individual computon component, except $\Sigma$ which simply is the union of ${\Sigma_1}$ and ${\Sigma_2}$: 
\[
P_3 = P_1 +_{P_0} P_2
\] 
\[
U_3 = U_1 +_{U_0} U_2
\] 
\[
I_3 = I_1 +_{I_0} I_2
\] 
\[
O_3 = O_1 +_{O_0} O_2
\] 
\[
\Sigma_3 = \Sigma_1\cup\Sigma_2
\]
with $\tau_3$, $\sigma_3$, $s_3$, $t_3$ and $c_3$ being defined in the obvious way:
\[
\tau_3\colon I_1 +_{I_0} I_2 \twoheadrightarrow U_1 +_{U_0} U_2
\]
\[
\sigma_3\colon O_1 +_{O_0} O_2 \twoheadrightarrow U_1 +_{U_0} U_2
\]
\[
s_3\colon I_1 +_{I_0} I_2 \rightarrow P_1 +_{P_0} P_2
\]
\[
t_3\colon O_1 +_{O_0} O_2 \rightarrow P_1 +_{P_0} P_2
\]
\[
c_3\colon P_1 +_{P_0} P_2 \twoheadrightarrow \Sigma_1 +_{\Sigma_0} \Sigma_2
\]
\end{definition}

Unfortunately, a pushout operation cannot be computed for every span of computon morphisms. To understand why, let us recall that Definition \ref{def:computon-morphism} enforces a computon morphism to insert a computon into another only at the boundaries, so that computons can only ``interact'' through their e-ports. Four examples of valid computon morphisms are depicted in Figure \ref{fig:pushable-span-example}(a). 

\begin{figure}[!h]
\centering
\subcaptionbox{The pushout can be computed because $\alpha_1$ and $\alpha_2$ form a pushable span.}
{
\begin{tikzpicture}
\node[draw=black,fill=white,inner sep=0pt,minimum size=3pt,opacity=0.4] (p1) at (1.6,2.5) {};
\draw[dashed,opacity=0.4] (1.65,2.5) to [] node [pos=0.5] {\arrowflow} node [pos=0.5,yshift=7] {} (2.6,2.5);
\node[opacity=0.4] (u1) at (2.8,2.5) {\scriptsize $u_0$};
\draw[dashed,opacity=0.4] (3,2.5) to [] node [pos=0.5] {\arrowflow} node [pos=0.5,yshift=7, opacity=0.4] {} (3.8,2.5);
\node[fill=black,inner sep=0pt,minimum size=3pt,opacity=0.4] (p2) at (3.9,2.5) {};

\draw[->,opacity=0.4,bend right=20] (2.25,2.15) to node[left]{\scriptsize $\alpha_1$} (1.3,1.5);
\draw[->,opacity=0.4,bend left=20] (3.5,2.15) to node[right]{\scriptsize $\alpha_2$} (4.45,1.5);
\draw[->,opacity=0.4,bend right=20] (1.3,0.3) to node[left]{\scriptsize $\beta_1$} (2.25,-0.35);
\draw[->,opacity=0.4,bend left=20] (4.45,0.3) to node[right]{\scriptsize $\beta_2$} (3.5,-0.35);

\node[draw=black,fill=white,inner sep=0pt,minimum size=3pt,opacity=0.4] (p3) at (0,1.1) {};
\draw[dashed,opacity=0.4] (0.05,1.1) to [] node [pos=0.5] {\arrowflow} node [pos=0.5,yshift=7, opacity=0.4] {} (1,1.1);
\node[opacity=0.4] (u2) at (1.2,1.1) {\scriptsize $u_1$};
\draw[dashed,opacity=0.4] (1.4,1.1) to [] node [pos=0.5] {\arrowflow} node [pos=0.5,yshift=7, opacity=0.4] {} (2.2,1.1);
\node[fill=black,inner sep=0pt,minimum size=3pt,opacity=0.4] (p5) at (2.3,1.1) {};
\draw[dashed] (p3) to [] node [pos=0.5,rotate=-27] {\arrowflow} node [pos=0.5,yshift=-7] {} (1,0.6);
\node (u3) at (1.2,0.6) {\scriptsize $u_2$};
\draw[dashed] (1.4,0.6) to [] node [pos=0.3,rotate=27] {\arrowflow} node [below] {} (p5);

\node[draw=black,fill=white,inner sep=0pt,minimum size=3pt,opacity=0.4] (q3) at (3,0.9) {};
\draw[dashed,opacity=0.4] (3.05,0.9) to [] node [pos=0.5] {\arrowflow} node [pos=0.5,yshift=7, opacity=0.4] {} (4,0.9);
\node[opacity=0.4] (v2) at (4.2,0.9) {\scriptsize $u_3$};
\draw[dashed,opacity=0.4] (4.4,0.9) to [] node [pos=0.5] {\arrowflow} node [pos=0.5,yshift=7, opacity=0.4] {} (5.2,0.9);
\node[fill=black,inner sep=0pt,minimum size=3pt,opacity=0.4] (q5) at (5.3,0.9) {};

\node[draw=black,fill=white,inner sep=0pt,minimum size=3pt,opacity=0.4] (z3) at (1.6,-0.7) {};
\draw[dashed,opacity=0.4] (1.65,-0.7) to [] node [pos=0.5] {\arrowflow} node [pos=0.5,yshift=7,opacity=0.4] {} (2.6,-0.7);
\node[opacity=0.4] (u4) at (2.8,-0.7) {\scriptsize $u_4$};
\draw[dashed,opacity=0.4] (3,-0.7) to [] node [pos=0.5] {\arrowflow} node [pos=0.5,yshift=7,opacity=0.4] {} (3.8,-0.7);
\node[fill=black,inner sep=0pt,minimum size=3pt,opacity=0.4] (z5) at (3.9,-0.7) {};
\draw[dashed] (z3) to [] node [pos=0.5,rotate=-27] {\arrowflow} node [pos=0.5,yshift=-7] {} (2.6,-1.2);
\node (u5) at (2.8,-1.2) {\scriptsize $u_5$};
\draw[dashed] (3,-1.2) to [] node [pos=0.3,rotate=27] {\arrowflow} node [below] {} (z5);

\end{tikzpicture}  
}\hspace{1.3cm}
\subcaptionbox{The pushout cannot be computed because $\alpha_1$ and $\alpha_2$ do not form a pushable span.}
{
\begin{tikzpicture}
\node[draw=black,fill=white,inner sep=0pt,minimum size=3pt,opacity=0.4] (p1) at (1.6,2.5) {};
\draw[dashed,opacity=0.4] (1.65,2.5) to [] node [pos=0.5] {\arrowflow} node [pos=0.5,yshift=7] {} (2.6,2.5);
\node[opacity=0.4] (u1) at (2.8,2.5) {\scriptsize $u_0$};
\draw[dashed,opacity=0.4] (3,2.5) to [] node [pos=0.5] {\arrowflow} node [pos=0.5,yshift=7, opacity=0.4] {} (3.8,2.5);
\node[fill=black,inner sep=0pt,minimum size=3pt,opacity=0.4] (p2) at (3.9,2.5) {};

\draw[->,opacity=0.4,bend right=20] (2.25,2.15) to node[left]{\scriptsize $\alpha_1$} (1.3,1.5);
\draw[->,opacity=0.4,bend left=20] (3.5,2.15) to node[right]{\scriptsize $\alpha_2$} (4.45,1.5);
\draw[->,opacity=0.4,bend right=20] (1.3,0.3) to node[left]{\scriptsize $\beta_1$} (2.25,-0.35);
\draw[->,opacity=0.4,bend left=20] (4.45,0.3) to node[right]{\scriptsize $\beta_2$} (3.5,-0.35);

\node[draw=black,fill=white,inner sep=0pt,minimum size=3pt,opacity=0.4] (p3) at (0,1.1) {};
\draw[dashed,opacity=0.4] (0.05,1.1) to [] node [pos=0.5] {\arrowflow} node [pos=0.5,yshift=7, opacity=0.4] {} (1,1.1);
\node[opacity=0.4] (u2) at (1.2,1.1) {\scriptsize $u_1$};
\draw[dashed,opacity=0.4] (1.4,1.1) to [] node [pos=0.5] {\arrowflow} node [pos=0.5,yshift=7, opacity=0.4] {} (2.2,1.1);
\node[fill=black,inner sep=0pt,minimum size=3pt,opacity=0.4] (p5) at (2.3,1.1) {};
\draw[dashed] (p3) to [] node [pos=0.5,rotate=-27] {\arrowflow} node [pos=0.5,yshift=-7] {} (1,0.6);
\node (u3) at (1.2,0.6) {\scriptsize $u_2$};
\draw[dashed] (1.4,0.6) to [] node [pos=0.3,rotate=27] {\arrowflow} node [below] {} (p5);

\node[draw=black,fill=white,inner sep=0pt,minimum size=3pt] (q3) at (3,0.9) {};
\draw[dashed] (3.05,0.9) to [] node [pos=0.5] {\arrowflow} node [pos=0.5,yshift=7] {} (4,0.9);
\node (v2) at (4.2,0.9) {\scriptsize $u_4$};
\draw[dashed] (4.4,0.9) to [] node [pos=0.5] {\arrowflow} node [pos=0.5,yshift=7] {} (5.2,0.9);
\node[draw,fill=white,fill fraction={black}{0.5},inner sep=0pt,minimum size=3pt,opacity=0.4] (q5) at (5.3,0.9) {};
\draw[dashed,opacity=0.4] (5.35,0.9) to [] node [pos=0.5] {\arrowflow} node [pos=0.5,yshift=7, opacity=0.4] {} (6.3,0.9);
\node[opacity=0.4] (x2) at (6.5,0.9) {\scriptsize $u_3$};
\draw[dashed,opacity=0.4] (6.7,0.9) to [] node [pos=0.5] {\arrowflow} node [pos=0.5,yshift=7, opacity=0.4] {} (7.5,0.9);
\node[fill=black,inner sep=0pt,minimum size=3pt,opacity=0.4] (z5) at (7.6,0.9) {};

\node[draw=black,fill=white,inner sep=0pt,minimum size=3pt] (q3) at (0.55,-0.7) {};
\draw[dashed] (0.6,-0.7) to [] node [pos=0.5] {\arrowflow} node [pos=0.5,yshift=7] {} (1.55,-0.7);
\node (v2) at (1.75,-0.7) {\scriptsize $u_5$};
\draw[dashed] (1.95,-0.7) to [] node [pos=0.5] {\arrowflow} node [pos=0.5,yshift=7] {} (2.85,-0.7);
\node[draw=black,fill fraction={black}{0.5},inner sep=0pt,minimum size=3pt,opacity=0.4] (z3) at (2.9,-0.7) {};
\draw[dashed,opacity=0.4] (2.95,-0.7) to [] node [pos=0.5] {\arrowflow} node [pos=0.5,yshift=7,opacity=0.4] {} (3.9,-0.7);
\node[opacity=0.4] (u4) at (4.1,-0.7) {\scriptsize $u_6$};
\draw[dashed,opacity=0.4] (4.3,-0.7) to [] node [pos=0.5] {\arrowflow} node [pos=0.5,yshift=7,opacity=0.4] {} (5.1,-0.7);
\node[fill=black,inner sep=0pt,minimum size=3pt,opacity=0.4] (z5) at (5.2,-0.7) {};
\draw[dashed] (z3) to [] node [pos=0.5,rotate=-27] {\arrowflow} node [pos=0.5,yshift=-7] {} (3.9,-1.2);
\node (u5) at (4.1,-1.2) {\scriptsize $u_7$};
\draw[dashed] (4.3,-1.2) to [] node [pos=0.3,rotate=27] {\arrowflow} node [below] {} (z5);

\end{tikzpicture}  
}
{
\begin{tikzpicture}
\matrix [below, ampersand replacement=\&] at (current bounding box.south) {
\draw[dashed] (2.3,0) to node[pos=0.5,yshift=4]{\arrowflow} (2.8,0); \& \node[inner sep=0pt,minimum size=3pt,label=right:{\scriptsize Control flow edge}]{}; \& \node{}; \& \&
 \node[draw=black,fill=white,inner sep=0pt,minimum size=3pt,label=right:{\scriptsize Ec-inport}] {}; \& \node{};
 \& \node[fill=black,inner sep=0pt,minimum size=3pt,label={right:{\scriptsize Ec-outport}}] {}; \& \node{};
 \& \& \node[draw,fill=white,fill fraction={black}{0.5},inner sep=0pt,minimum size=3pt,label={right:\scriptsize Ic-port}] {};\\
};
\end{tikzpicture}
}
\caption{An example to illustrate the semantics of a pushable span of computon morphisms.}
\label{fig:pushable-span-example}
\end{figure}

Figure \ref{fig:pushable-span-example}(b) shows that, unfortunately, defining a span of valid computon morphisms is not sufficient to compute a pushout in $\textbf{Set}^\textbf{Comp}$. A pushout only exists for spans that adhere to Definition \ref{def:computon-morphisms-pushable}. 

\begin{definition} [Pushable Span] \label{def:computon-morphisms-pushable}
A span ${\lambda_1 \xleftarrow{\alpha_1} \lambda_0 \xrightarrow{\alpha_2} \lambda_2}$ of computon morphisms is pushable if ${\alpha_1(\vec{i}(\alpha_2)) \cup \alpha_1(\vec{o}(\alpha_2)) \subseteq P_1^+ \cup P_1^-}$ and ${\alpha_2(\vec{i}(\alpha_1)) \cup \alpha_2(\vec{o}(\alpha_1)) \subseteq P_2^+ \cup P_2^-}$.
\end{definition}

Basically, Definition \ref{def:computon-morphisms-pushable} says that a span ${\lambda_1 \xleftarrow{\alpha_1} \lambda_0 \xrightarrow{\alpha_2} \lambda_2}$ of computon morphisms is pushable when, for every port $p \in P_0$, the following holds:
\begin{enumerate}
\item If ${\alpha_1(p)}$ is connected to/from a computation unit of $\lambda_1$ that does not form part of the $\alpha_1$-embedding, then ${\alpha_2(p)}$ must be either e-inport or e-outport in $\lambda_2$. \label{condition1-pushable}
\item If ${\alpha_2(p)}$ is connected to/from a computation unit of $\lambda_2$ that does not form part of the $\alpha_2$-embedding, then ${\alpha_1(p)}$ must be either e-inport or e-outport in $\lambda_1$. \label{condition2-pushable}
\end{enumerate}

To understand these conditions, consider the span formed by $\alpha_1$ and $\alpha_2$ in Figure \ref{fig:pushable-span-example}(a), which embeds the apex computon into the parts highlighted in gray. Here, $\alpha_1$ maps the e-inport $p_0$ of the apex to the e-inport of the left leg. As the e-inport of the left leg is connected to the computation unit $u_2$, which does not form part of the $\alpha_1$-embedding, it follows that ${p_0\in\vec{o}(\alpha_1)}$ (see Definition \ref{def:computon-morphism}). According to the Condition \ref{condition1-pushable} presented above, $\alpha_2(p_0)$ must be either e-inport or e-outport. As in this case $\alpha_2(p_0)$ is indeed an e-inport and the other conditions of Definition \ref{def:computon-morphisms-pushable} are analogously satisfied, we conclude that our span is pushable. 

Figure \ref{fig:pushable-span-example}(b) presents a counterexample in which $\alpha_2$ does not map the e-inport of the apex computon to an external port of the right leg but to an i-port that lies between the computation units $u_3$ and $u_4$. As ${p_0\in\vec{o}(\alpha_1)}$ and $\alpha_2(p_0)$ is neither e-inport nor e-outport, we have that the Condition \ref{condition1-pushable} presented above is not satisfied. Hence, the span from Figure \ref{fig:pushable-span-example}(b), although valid, is not pushable. 

In particular, the pushout of such a span cannot be computed because the induced computon morphism $\beta_2$ violates Definition \ref{def:computon-morphism}, i.e., ${\alpha_2(p_0)\in\vec{o}(\beta_2)}$ but ${\alpha_2(p_0)}$ is neither e-inport nor e-outport. The fact a pushout construction can only be computed for pushable spans (see Proposition \ref{prop:pushout-pushable}) entails that $\textbf{Set}^\textbf{Comp}$ does not have all pushouts. Nevertheless, when such a construction exists for a span whose legs are connected computons, the result of the corresponding operation is a connected computon (see Proposition \ref{prop:pushout-connected}). 

\begin{proposition} \label{prop:pushout-pushable}
Let ${\alpha_1\colon \lambda_0 \rightarrow \lambda_1}$ and ${\alpha_2\colon \lambda_0 \rightarrow \lambda_2}$ be two computon morphisms. The pushout of ${\alpha_1}$ and ${\alpha_2}$ exists $\iff$ ${\lambda_1 \xleftarrow{\alpha_1} \lambda_0 \xrightarrow{\alpha_2} \lambda_2}$ is pushable. 
\end{proposition}
\begin{proof}
${(\implies)}$ Assuming that the pushout ${(\beta_1\colon\lambda_1 \rightarrow \lambda_3,\lambda_3,\beta_2\colon\lambda_2 \rightarrow \lambda_3)}$ of ${\alpha_1\colon \lambda_0 \rightarrow \lambda_1}$ and ${\alpha_2\colon \lambda_0 \rightarrow \lambda_2}$ exists in ${\textbf{Set}^\textbf{Comp}}$, we just prove ${\alpha_1(\vec{i}(\alpha_2)) \subseteq P_1^+ \cup P_1^-}$ since the other conditions of Definition \ref{def:computon-morphisms-pushable} follow analogously.

Supposing there is some ${p_1 \in \alpha_1(\vec{i}(\alpha_2)) \setminus (P_1^+ \cup P_1^-)}$, we know there is a port ${p_0 \in \vec{i}(\alpha_2)}$ where ${\alpha_1(p_0)=p_1}$. As the pushout ${(\beta_1,\lambda_3,\beta_2)}$ exists in $\textbf{Set}^\textbf{Comp}$, the equation ${\beta_1(\alpha_1(p_0))=}$ $\beta_1(p_1)=\beta_2(\alpha_2(p_0))$ holds. Since ${p_0 \in \vec{i}(\alpha_2)}$, there is some ${u_2 \in \bullet \alpha_2(p_0) \setminus \alpha_2(\bullet p_0)}$ so that ${\beta_1(p_1)=\beta_2(\alpha_2(p_0)) \in \beta_2(u_2) \bullet}$. As ${p_1 \notin P_1^+ \cup P_1^-}$, ${p_1 \notin \vec{i}(\beta_1)}$ (see Definition \ref{def:computon-morphism}), meaning that there is some ${u_1 \in \bullet p_1}$ where ${\beta_1(u_1)=\beta_2(u_2)}$. Using commutativity, we deduce the existence of ${u_0 \in U_0}$ such that ${\alpha_1(u_0)=u_1}$ and ${\alpha_2(u_0)=u_2}$. As this contradicts the fact ${u_2 \in \bullet \alpha_2(p_0) \setminus \alpha_2(\bullet p_0)}$, we conclude ${\alpha_1(\vec{i}(\alpha_2)) \subseteq P_1^+ \cup P_2^-}$.

${(\impliedby)}$ Assuming ${\lambda_1 \xleftarrow{\alpha_1} \lambda_0 \xrightarrow{\alpha_2} \lambda_2}$ is a pushable span of computon morphisms (as per Definition \ref{def:computon-morphisms-pushable}), we prove that the pushout ${(\beta_1\colon\lambda_1 \rightarrow \lambda_3,\lambda_3,\beta_2\colon\lambda_2 \rightarrow \lambda_3)}$ of ${\alpha_1}$ and ${\alpha_2}$ can be constructed via Definition \ref{def:pushout-computation}. For this, we first prove that ${\beta_1}$ and ${\beta_2}$ are computon morphisms. Below we provide the proof for ${\beta_1}$ only, since the other is completely analogous.

As \textbf{Set} has all pushouts, the existence of each component of ${\beta_1}$ and ${\beta_2}$ can be directly deduced. For example, the ${\beta_1}$-component embedding ports of ${P_1}$ into ${P_3}$ exists because ${P_3 = P_1 +_{P_0} P_2}$ can always be computed in \textbf{Set}. Consequently, the equations ${\beta_i \circ c_1 = c_3 \circ \beta_i}$, ${\beta_i \circ \tau_1 = \tau_3 \circ \beta_i}$, ${\beta_i \circ \sigma_1 = \sigma_3 \circ \beta_i}$, ${\beta_i \circ s_1 = s_3 \circ \beta_i}$ and ${\beta_i \circ t_1 = t_3 \circ \beta_i}$ hold for $i=1,2$ (see the commutative diagrams of Definition \ref{def:computon-morphism}). 

For the $\Sigma$-component of $\beta_1$ to be an inclusion (as required by Definition \ref{def:computon-morphism}), we simply choose the canonical function ${f:\Sigma_1\rightarrow\Sigma_1\cup\Sigma_2}$ given by ${f(x)=x}$ for all ${x \in \Sigma_1}$. We now show ${\vec{i}(\beta_1) \cup \vec{o}(\beta_1) \subseteq P_1^+ \cup P_1^-}$ holds.

If ${p_1 \in \vec{i}(\beta_1)}$ then ${\bullet \beta_1(p_1) \setminus \beta_1(\bullet p_1) \neq \emptyset}$ so there exists some ${u_3 \in \bullet \beta_1(p_1) \setminus \beta_1(\bullet p_1)}$ and no ${u_1 \in \bullet p_1}$ where ${\beta_1(u_1)=u_3}$. Since ${U_3 = U_1 +_{U_0} U_2}$ (by Definition \ref{def:pushout-computation}), there must be some ${u_2 \in U_2}$ where ${\beta_2(u_2)=u_3 \in \bullet \beta_1(p_1) \setminus \beta_1(\bullet p_1)}$ and, consequently, some ${o_3 \in O_3}$ where ${\beta_2(u_2) \xrightarrow{o_3} \beta_1(p_1)}$ (see Definition \ref{def:pre-post-sets}). As there is no ${o_1 \in O_1}$ satisfying ${\beta_1(\sigma_1(o_1)) = \sigma_3(\beta_1(o_1)) =}$ $\sigma_3(o_3)=u_3$ because there is no ${u_1 \in U_1}$ satisfying ${\beta_1(u_1)=u_3}$, ${O_3 = O_1 +_{O_0} O_2}$ implies that there must be some ${o_2 \in O_2}$ for which ${\beta_2(o_2)=o_3}$ is true. As ${\beta_2(u_2) \xrightarrow{\beta_2(o_2)} \beta_1(p_1)}$ holds, we use the commutativity property to deduce the existence of ${p_2 \in P_2}$ such that ${\beta_2(p_2)=\beta_1(p_1)}$ and $u_2 \in \bullet p_2$. As ${\beta_2(p_2) = \beta_1(p_1)}$ and ${P_3=P_1 +_{P_0} P_2}$, we use again commutativity to deduce there also is a port ${p_0 \in P_0}$ with ${\alpha_1(p_0)=p_1}$ and ${\alpha_2(p_0)=p_2}$ so that ${u_2 \in \bullet \alpha_2(p_0)}$. Since ${U_3= U_1 +_{U_0} U_2}$ and ${(\nexists u_1 \in U_1)[\beta_1(u_1)=u_3=\beta_2(u_2)]}$, we have that there is no ${u_0 \in U_0}$ where ${\alpha_2(u_0)=u_2}$. That is, ${u_2 \notin \alpha_2(\bullet p_0)}$ because ${\bullet p_0 \subseteq U_0}$. Thus, it is true that ${u_2 \in \bullet \alpha_2(p_0) \setminus \alpha_2(\bullet p_0)}$ and, therefore, that ${p_0 \in \vec{i}(\alpha_2)}$ (see Definition \ref{def:computon-morphism}).

Since ${\alpha_1(\vec{i}(\alpha_2)) \subseteq P_1^+ \cup P_1^-}$ (by Definition \ref{def:computon-morphisms-pushable}) and ${p_0 \in \vec{i}(\alpha_2)}$ (by the above), we have ${\alpha_1(p_0)=p_1 \in P_1^+ \cup P_1^-}$. A similar approach can be used to show that ${q_1 \in \vec{o}(\beta_1)}$ implies $q_1 \in P_1^+ \cup P_1^-$. So, ${\vec{i}(\beta_1) \cup \vec{o}(\beta_1) \subseteq P_1^+ \cup P_1^-}$. 

Having proved ${\beta_1}$ is a computon morphism, we now assume ${\gamma_1\colon \lambda_1 \rightarrow \lambda_4}$ and ${\gamma_2\colon \lambda_2 \rightarrow \lambda_4}$ are computon morphisms with ${\gamma_1 \circ \alpha_1 = \gamma_2 \circ \alpha_2}$, in order to show there is a unique computon morphism ${\gamma_3\colon \lambda_3 \rightarrow \lambda_4}$ such that the corresponding diagram commutes. As it is obvious that the ${\Sigma}$-component of ${\gamma_3}$ is an inclusion function because the ${\Sigma}$-components of ${\beta_i}$ and ${\gamma_i}$ also are (for $i=1,2$) and because ${\Sigma_3=\Sigma_1\cup\Sigma_2}$, we just prove ${\vec{i}(\gamma_3) \cup \vec{o}(\gamma_3) \subseteq P_3^+ \cup P_3^-}$. 

Let ${p_3 \in \vec{i}(\gamma_3)}$ so ${\bullet \gamma_3(p_3) \setminus \gamma_3(\bullet p_3) \neq \emptyset}$. As ${P_3 = P_1 +_{P_0} P_2}$, we observe ${p_3=\beta_j(p_j)}$ for some ${p_j \in P_j}$ with ${j=1,2}$. With this in mind, we perform the following operations:
\begin{align*}
\emptyset & \neq \bullet \gamma_3(p_3) \setminus \gamma_3(\bullet p_3) = \bullet \gamma_3(\beta_j(p_j)) \setminus \gamma_3(\bullet \beta_j(p_j)) \\
 &\hspace{3.31cm} = \bullet \gamma_j(p_j) \setminus \gamma_3(\bullet \beta_j(p_j)) \hspace{1cm}\text{because } \gamma_3 \circ \beta_j = \gamma_j \\
 &\hspace{3.31cm} \subseteq \bullet \gamma_j(p_j) \setminus \gamma_3(\beta_j(\bullet p_j)) \hspace{0.9cm}\text{ because } \beta_j(\bullet p_j) \subseteq \bullet \beta_j(p_j) \\
 &\hspace{3.31cm} = \bullet \gamma_j(p_j) \setminus \gamma_j(\bullet p_j) \hspace{1.7cm}\text{because } \gamma_3 \circ \beta_j = \gamma_j
\end{align*}
By the above, we deduce ${p_j \in \vec{i}(\gamma_j)}$ and, consequently, ${p_j \in P_j^+ \cup P_j^-}$ (because $\gamma_j$ is a computon morphism with ${\vec{i}(\gamma_j) \subseteq P_j^+ \cup P_j^-}$ --- see Definition \ref{def:computon-morphism}). Using the facts ${p_3=\beta_j(p_j) \in \vec{i}(\gamma_3)}$ and ${p_j \in P_j^+ \cup P_j^-}$, we further deduce ${\beta_j^{-1}(p_3) \subseteq P_j^+ \cup P_j^-}$. Hence, $p_3 \in P_3^+ \cup P_3^-$ by Proposition \ref{prop:computon-morphism-inports-outports}. The proof of $q_3 \in \vec{o}(\gamma_3) \implies q_3 \in P_3^+ \cup P_3^-$ is completely analogous.

As the $\Sigma$-component of $\gamma_3$ is an inclusion function and $\vec{i}(\gamma_3) \cup \vec{o}(\gamma_3) \subseteq P_3^+ \cup P_3^-$, we conclude that $\gamma_3$ is a computon morphism. Such a morphism is unique because each of its components are unique (by the fact that a pushout in \textbf{Set} satisfies the universal property).
\end{proof}

\begin{proposition} \label{prop:pushout-connected}
Let ${\lambda_1 \xleftarrow{\alpha_1} \lambda_0 \xrightarrow{\alpha_2} \lambda_2}$ be a pushable span of computon morphisms. If $\lambda_1$ and $\lambda_2$ are connected computons, then the pushout of $\alpha_1$ and $\alpha_2$ is a connected computon.
\end{proposition}
\begin{proof}
Let ${\lambda_1 \xleftarrow{\alpha_1} \lambda_0 \xrightarrow{\alpha_2} \lambda_2}$ be a pushable span of computon morphisms and assume ${\lambda_1}$ and ${\lambda_2}$ are connected computons. By Proposition \ref{prop:pushout-pushable}, ${({\beta_1\colon\lambda_1 \rightarrow \lambda_3},\lambda_3,{\beta_2\colon\lambda_2 \rightarrow \lambda_3})}$ can be constructed from ${\alpha_1}$ and ${\alpha_2}$.

To prove ${\lambda_3}$ is a connected computon, let ${p_3 \in Im(s_3) \cup P_3^+}$. Since ${P_3=P_1 +_{P_0} P_2}$, we know $p_3$ is identified with a port from (1) $\lambda_1$, (2) $\lambda_2$ or (3) both. Before considering these three scenarios, let ${V_i}$ be a sequence of visited e-outports of $\lambda_i$ for $i\in \{1,2\}$. The sequence $V_i$ is initially empty and is iteratively updated through the following process.

\begin{enumerate}
\item If $p_3$ is exclusively identified with a port from $\lambda_1$, i.e., there is a port ${p_1 \in P_1}$ where ${\beta_1(p_1)=p_3}$, we have that ${p_1 \in P_1^-}$ cannot hold because: \label{pushout-connected-1}
\begin{itemize}
\item If ${\beta_1(p_1)=p_3 \in Im(s_3)}$, then there is some ${i_3 \in I_3}$ where ${s_3(i_3)=\beta_1(p_1)=p_3}$. As there is no ${i_1 \in I_1}$ with ${s_1(i_1)=p_1}$ because ${p_1 \in P_1^-}$, ${I_3=I_1 +_{I_0} I_2}$ implies there must be some ${i_2 \in I_2}$ for which ${\beta_2(i_2)=i_3}$ holds. By the totality of $s_2$, there must also be some ${p_2 \in P_2}$ where ${s_2(i_2)=p_2}$. Using the commutative equations derived from $\beta_2$, we obtain: ${\beta_2(s_2(i_2))=s_3(\beta_2(i_2))=s_3(i_3)=\beta_1(p_1)=p_3}$ which violates our assumption that $p_3$ is exclusively identified with a $\lambda_1$-port.
\item If ${\beta_1(p_1)=p_3 \in P_3^+}$, we use Proposition \ref{prop:computon-morphism-inports-outports} to deduce ${p_1 \in P_1^+}$. As Proposition \ref{prop:computon-connected-isolated-port} says ${P_1^+\cap P_1^-=\emptyset}$ because $\lambda_1$ is connected, we have that ${p_1 \in P_1^+ \cap P_1^-}$ cannot hold.
\end{itemize}

Now, if ${p_1 \notin P_1^-}$, then there is some ${i_1 \in I_1}$ where ${s_1(i_1)=p_1}$, i.e., ${p_1 \in Im(s_1)}$, meaning ${p_1 \xrightarrow{\exists} q_1}$ must hold for some ${q_1 \in P_1^-}$ with ${(\nexists k)[V_1[k]=q_1]}$ (because $\lambda_1$ is a connected computon). By the commutativity property of $\beta_1$, ${\beta_1(p_1) \xrightarrow{\exists} \beta_1(q_1)}$ must also hold. We now append ${\langle q_1 \rangle}$ to $V_1$ and consider the following two cases: 
\begin{itemize}
\item ${\beta_1(q_1) \in P_3^-}$ when there is no ${p_2 \in P_2}$ such that ${\beta_1(q_1)=\beta_2(p_2)}$ (i.e., ${\beta_1(q_1)}$ is exclusively identified with a $\lambda_1$-port) or when ${\beta_1(q_1)}$ is identified only with e-outports of $\lambda_2$. As ${\beta_1(p_1)=p_3 \in Im(s_3) \cup P_3^+}$ and ${\beta_1(q_1) \in P_3^-}$, $\lambda_3$ must be a connected computon.
\item ${\beta_1(q_1) \notin P_3^-}$ when there is some ${p_2 \in P_2\setminus P_2^-}$ with ${\beta_1(q_1)=\beta_2(p_2)}$. As $\lambda_2$ is a connected computon, ${p_2 \xrightarrow{\exists} q_2}$ for some ${q_2 \in P_2^-}$ and, therefore, ${\beta_2(p_2) \xrightarrow{\exists} \beta_2(q_2)}$ must also hold by the commutativity property of $\beta_2$. Considering ${\beta_1(p_1)=p_3}$ and ${\beta_1(q_1)=\beta_2(p_2)}$, we have ${p_3 \xrightarrow{\exists} \beta_2(p_2) \xrightarrow{\exists} \beta_2(q_2)}$. If ${\beta_2(q_2) \in P_3^-}$, then $\lambda_3$ is a connected computon. If not, perform \ref{pushout-connected-1}, \ref{pushout-connected-2} or \ref{pushout-connected-3}, whichever holds for ${\beta_2(q_2)}$.
\end{itemize}
\item The proof when $p_3$ is exclusively identified with a $\lambda_2$-port is symmetric to that of (1). \label{pushout-connected-2}
\item \label{pushout-connected-3} The port $p_3$ is identified with a port from both $\lambda_1$ and $\lambda_2$, i.e., ${p_3=\beta_1(p_1)=\beta_2(p_2)}$ for some ${p_1 \in P_1}$ and some ${p_2 \in P_2}$. In this case, if ${p_1 \in Im(s_1) \cup P_1^+}$, there is some ${q_1 \in P_1^-}$ for which ${p_1 \xrightarrow{\exists} q_1}$ and ${(\nexists k)[V_1[k]=q_1]}$ because $\lambda_1$ is a connected computon; consequently, we append ${\langle q_1 \rangle}$ to $V_1$ and infer ${\beta_1(p_1) \xrightarrow{\exists} \beta_1(q_1)}$ from the commutativity property of $\beta_1$. If ${p_2 \in Im(s_2) \cup P_2^+}$, there is some ${q_2 \in P_2^-}$ such that ${(\nexists k)[V_2[k]=q_2]}$ and ${p_2 \xrightarrow{\exists} q_2}$ because $\lambda_2$ is a connected computon; hence, we append ${\langle q_2 \rangle}$ to $V_2$ and infer ${\beta_2(p_2) \xrightarrow{\exists} \beta_2(q_2)}$ from the commutativity property of $\beta_2$. 

Now, perform the following check nondeterministically for each ${j=1,2}$: If ${\beta_j(p_j) \xrightarrow{\exists} \beta_j(q_j)}$ and ${\beta_j(q_j) \in P_3^-}$, there is information flow from ${p_3=\beta_1(p_1)=\beta_2(p_2)}$ to an e-outport of $\lambda_3$, meaning $\lambda_3$ is a connected computon. If ${\beta_j(p_j) \xrightarrow{\exists} \beta_j(q_j)}$ and ${\beta_j(q_j) \notin P_3^-}$, perform \ref{pushout-connected-1}, \ref{pushout-connected-2} or \ref{pushout-connected-3}, whichever holds for ${\beta_j(q_j)}$.
\end{enumerate}

Checking from \ref{pushout-connected-1} to \ref{pushout-connected-3} is an iterative process which is repeated until yielding ${p_3 \xrightarrow{\exists} \cdots \xrightarrow{\exists} q_3}$ for the initial $p_3$ and some ${q_3 \in P_3^-}$. Termination is guaranteed because (i) the number of ports, edges and computation units of $\lambda_3$ is finite; (ii) $V_i$ contains all the visited e-outports of $\lambda_i$ so each iteration extends the information flow pipeline to a new e-outport; and (iii) $\lambda_1$ and $\lambda_2$ are both connected computons. As the iterative process will eventually devise information flow from $p_3$ to an e-outport of $\lambda_3$, we conclude $\lambda_3$ must be a connected computon.
\end{proof}

Pushouts are not the only colimits that can be computed in $\textbf{Set}^\textbf{Comp}$. Another useful operation for describing our theory of computons is that of coproduct which intuitively allows the definition of a side-by-side computon. As per Proposition \ref{prop:computon-coproduct}, this operation can always be computed in $\textbf{Set}^\textbf{Comp}$ so that such a category has all coproducts. Computing the coproduct of two connected computons results in another connected computon, as shown in Proposition \ref{prop:computon-coproduct-connected}.

\begin{proposition} \label{prop:computon-coproduct}
The coproduct ${\lambda_1 + \lambda_2}$ of computons ${\lambda_1}$ and ${\lambda_2}$ always exists in ${\textbf{Set}^\textbf{Comp}}$. 
\end{proposition}
\begin{proof}
The coproduct ${\lambda_3}$ of a computon ${\lambda_1}$ and a computon ${\lambda_2}$, written ${\lambda_1 + \lambda_2}$, is obtained by computing the following in ${\textbf{Set}}$: ${P_3=P_1 + P_2}$, ${U_3=U_1 + U_2}$, ${I_3=I_1 + I_2}$, ${O_3=O_1 + O_2}$ and ${\Sigma_3=\Sigma_1 \cup \Sigma_2}$. Particularly, the operation to obtain ${\Sigma_3}$ can always be computed since the cospan ${\Sigma_1 \hookrightarrow \Sigma_1 \cup \Sigma_2 \hookleftarrow \Sigma_2}$ of unique inclusion functions always exists in ${\textbf{Set}}$ (because ${\Sigma_1,\Sigma_2 \subset \mathbb{N}}$). This operation also satisfies the universal property in the sense that, for any set ${\Sigma_4 \subset \mathbb{N}}$ with inclusions ${\Sigma_1 \hookrightarrow \Sigma_4}$ and ${\Sigma_2 \hookrightarrow \Sigma_4}$, there is a unique inclusion ${\Sigma_1 \cup \Sigma_2 \hookrightarrow \Sigma_4}$. Consequently, the function ${c_3}$ is canonically identified with the mapping ${P_1+P_2 \twoheadrightarrow \Sigma_1 \cup \Sigma_2}$ which is surjective by the fact ${(\forall j \in \{1,2\})(\forall p \in P_j)(\exists !x \in \Sigma_j)[c_j(p)=x]}$. All the functions of ${\lambda_3}$, including ${c_3}$, are defined in the obvious way to make the corresponding squares commute. For example: ${\forall o_3 \in O_3, \sigma_3(o_3)} = 
    \begin{cases}
        (\beta_1 \circ \sigma_1)(o_1) & \text{if } o_3 = \beta_1(o_1) \text{ for some } o_1 \in O_1\\
        (\beta_2 \circ \sigma_2)(o_2) & \text{if } o_3 = \beta_2(o_2) \text{ for some } o_2 \in O_2
    \end{cases}$
where ${\beta_k\colon \lambda_k \rightarrow \lambda_3}$ is the canonical injection into the coproduct $\lambda_3$. 

The existence of each component of ${\beta_k}$ follows directly from the fact that ${\textbf{Set}}$ has all coproducts. Particularly, the ${\Sigma}$-component of ${\beta_k}$ is the unique inclusion ${\Sigma_k \hookrightarrow \Sigma_1 \cup \Sigma_2}$ while the others are obvious morphisms of the form ${A_k \rightarrow A_1 + A_2}$ such as ${U_k \rightarrow U_1 + U_2}$. Computing ${U_3}$ as the disjoint union of ${U_1}$ and ${U_2}$ implies ${\vec{i}(\beta_k) \cup \vec{o}(\beta_k) = \emptyset \subseteq P_k^+ \cup P_k^-}$.

As the coproduct of each component of $\lambda_3$ is computed in \textbf{Set} and \textbf{Set} has all coproducts, it is true that coproduct in ${\textbf{Set}^\textbf{Comp}}$ satisfies the universal property. This means that, if there is a computon ${\lambda_4}$ with morphisms ${\gamma_1\colon \lambda_1 \rightarrow \lambda_4}$ and ${\gamma_2\colon \lambda_2 \rightarrow \lambda_4}$, there is a unique morphism ${\gamma_3\colon \lambda_3 \rightarrow \lambda_4}$ such that ${\gamma_3 \circ \beta_1 = \gamma_1}$ and ${\gamma_3 \circ \beta_2 = \gamma_2}$. As the ${\Sigma}$-components of $\gamma_1$ and $\gamma_2$ are both inclusion functions it is easy to see that the ${\Sigma}$-component of ${\gamma_3}$ is the unique inclusion ${\Sigma_3 \hookrightarrow \Sigma_4}$. The other components of ${\gamma_3}$ are given in the obvious way. For example:  
\[
\forall p_3 \in P_3, \gamma_3(p_3) = 
    \begin{cases}
        \gamma_1(p_1) & \text{if } p_3 = \beta_1(p_1) \text{ for some } p_1 \in P_1\\
        \gamma_2(p_2) & \text{if } p_3 = \beta_2(p_2) \text{ for some } p_2 \in P_2
    \end{cases}
\]

To meet the rest of the requirements of Definition \ref{def:computon-morphism} for $\gamma_3$, we just now have to prove ${\vec{i}(\gamma_3) \cup \vec{o}(\gamma_3) \subseteq P_3^+ \cup P_3^-}$. Below we provide the proof of ${\vec{i}(\gamma_3) \subseteq P_3^+ \cup P_3^-}$ since the other is completely analogous.

Let ${p_3 \in \vec{i}(\gamma_3)}$ so ${\bullet \gamma_3(p_3) \setminus \gamma_3(\bullet p_3) \neq \emptyset}$. As ${P_3=P_1+P_2}$, we observe that ${p_3=\beta_n(p_n)}$ for some ${p_n \in P_n}$ (${n=1,2}$). Using a similar reasoning as the proof of Proposition \ref{prop:pushout-pushable}, we deduce ${p_n \in \vec{i}(\gamma_n)}$ which implies ${p_n \in P_n^+ \cup P_n^-}$ because ${\gamma_n}$ is a computon morphism with ${\vec{i}(\gamma_n) \subseteq P_n^+ \cup P_n^-}$ (see Definition \ref{def:computon-morphism}). As ${\vec{i}(\beta_n) \cup \vec{o}(\beta_n) = \emptyset}$, ${P_n \cap \vec{i}(\beta_n)=\emptyset=P_n \cap \vec{o}(\beta_n)}$ and, consequently, ${\beta_n^{-1}(P_3^+)=P_n^+}$ and ${\beta_n^{-1}(P_3^-)=P_n^-}$ (see Proposition \ref{prop:computon-morphism-eports-equality}). That is, ${p_n \in P_n^+ \cup P_n^-}$ $\iff$ ${p_n \in \beta_n^{-1}(P_3^+) \cup \beta_n^{-1}(P_3^-)}$. Using the fact ${p_3=\beta_n(p_n)}$, we conclude ${p_3 \in P_3^+ \cup P_3^-}$.
\end{proof}

\begin{proposition}\label{prop:computon-coproduct-connected}
The coproduct of two connected computons is a connected computon.
\end{proposition}
\begin{proof}
By Proposition \ref{prop:computon-coproduct}, we know the coproduct $\lambda_3$ of connected computons $\lambda_1$ and $\lambda_2$ exists. To prove $\lambda_3$ is also connected, consider a port ${p_3 \in Im(s_3) \cup P_3^+}$ which is necessarily identified with a port from either $\lambda_1$ or $\lambda_2$ by the definition of coproduct. Assuming ${p_3=\beta_k(p)}$ for some ${p \in P_k}$ and ${k \in \{1,2\}}$, we prove by cases:
\begin{itemize}
\item If ${p_3 \in Im(s_3)}$, there is some ${i_3 \in I_3}$ where ${s_3(i_3)=p_3=\beta_k(p)}$. By commutativity, there also is some ${i \in I_k}$ for which ${\beta_k(i)=i_3}$ and ${s_k(i)=p}$, i.e., ${p \in Im(s_k)}$. As $\lambda_k$ is a connected computon, ${p \xrightarrow{\exists} q}$ holds for some ${q \in P_k^-}$ (see Definition \ref{def:computon-connected}). Again, by commutativity, ${\beta_k(p) \xrightarrow{\exists} \beta_k(q)}$ so that ${p_3 \xrightarrow{\exists} \beta_k(q)}$. Using coproduct definition and ${q \in P_k^-}$, we deduce ${\beta_k(q) \in P_3^-}$.
\item If ${p_3 \in P_3^+}$, then ${p \in P_k^+}$ by Proposition \ref{prop:computon-morphism-inports-outports}. As $\lambda_k$ is a connected computon, we have some ${q \in P_k^-}$ for which ${p \xrightarrow{\exists} q}$. Hence, ${\beta_k(p) \xrightarrow{\exists} \beta_k(q)}$ by commutativity and, therefore, ${p_3 \xrightarrow{\exists} \beta_k(q)}$. Using coproduct definition and ${q \in P_k^-}$, we deduce ${\beta_k(q) \in P_3^-}$.
\end{itemize}
Having information flow from an arbitrary port in ${Im(s_3) \cup P_3^+}$ to an e-outport of $\lambda_3$ entails that $\lambda_3$ is a connected computon.
\end{proof}

\subsection{Control Flow and Data Flow Structures}
\label{sec:control-data-structures}

The control flow structure of a computon can be expressed as an object in the category of directed labelled graphs and graph homomorphisms, $\textbf{Set}^\textbf{Gr}$, which is formalised below.

\begin{definition}\label{def:digraph-cat}
Let $\textbf{Gr}$ be the category freely generated by the following diagram:
\[
\begin{tikzcd}
 E \arrow[r, shift left=2pt, "\vec{s}"]\arrow[r, shift right=2pt, "\vec{t}"'] & V \arrow[r, "l"] & L
\end{tikzcd}
\]
which gives rise to the functor category $\textbf{Set}^\textbf{Gr}$ of directed labelled graphs and graph homomorphisms \cite{lowe_algebraic_1993}. In this category, composition is defined component-wise and the components of every identity morphism are all identity functions (which map domain elements to themselves). Even though a directed labelled graph $G$ is a functor $\textbf{Gr}\rightarrow\textbf{Set}$, we simplify notation by writing $V$ for the set $G(V)$ of vertices, $E$ for the set $G(E)$ of edges and $L$ for the set $G(L)$ of labels. We similarly write $\vec{s}$ for the source function $G(\vec{s})$, $\vec{t}$ for the target function $G(\vec{t})$ and $l$ for the labelling function $G(l)$. Whenever there is a subindex for a graph, we use the same subindex for its components, e.g., $V_1$ for $G_1$. The notation $\vec{s}$ and $\vec{t}$ is used to avoid confusion with the $s$ and $t$ functions of a computon $\lambda$.
\end{definition}

As the control flow structure of a computon is an object in $\textbf{Set}^\textbf{Gr}$, we refer to it as a Control Flow Graph (CFG) which is constructed via the application of the functor described in Definition \ref{def:computon-cfg}.

\begin{definition}[Computon CFG]\label{def:computon-cfg}
The functor $\mathscr{C}\colon\textbf{Set}^{\textbf{Comp}}\rightarrow \textbf{Set}^{\textbf{Gr}}$ maps each computon $\lambda$ to its underlying CFG $\mathscr{C}(\lambda)$ as follows:
\begin{itemize}
\item The set $V$ of vertices of $\mathscr{C}(\lambda)$ is $U\cup\{p\in P \mid c(p)=0\}$.
\item The set $E$ of edges of $\mathscr{C}(\lambda)$ is $\{i\in I \mid c(s(i))=0\}\cup\{o\in O \mid c(t(o))=0\}$.
\item The set $L$ of labels of $\mathscr{C}(\lambda)$ is $\{0,\kappa\}$.
\item The source function $\vec{s}\colon E\rightarrow V$ of $\mathscr{C}(\lambda)$ is given by $\vec{s}(e) = 
    \begin{cases}
        s(e) & \text{if } e\in I \\
        \sigma(e) & \text{if } e\in O
    \end{cases}$
\item The target function $\vec{t}\colon E\rightarrow V$ of $\mathscr{C}(\lambda)$ is given by $\vec{t}(e) = 
    \begin{cases}
        t(e) & \text{if } e\in O \\
        \tau(e) & \text{if } e\in I
    \end{cases}$
\item The labelling function $l\colon V\rightarrow L$ of $\mathscr{C}(\lambda)$ is given by $l(v) = 
    \begin{cases}
        c(v) & \text{if } v\in P \\
        \kappa & \text{if } v\in U
    \end{cases}$
\end{itemize}
For a computon morphism ${\alpha\colon\lambda_1\rightarrow\lambda_2}$, there is a graph homomorphism ${\mathscr{C}(\alpha)\colon\mathscr{C}(\lambda_1)\rightarrow \mathscr{C}(\lambda_2)}$ such that:
\begin{itemize}
\item $\mathscr{C}(\alpha)_V\colon V_1\rightarrow V_2$ is given by $\mathscr{C}(\alpha)_V(v) = 
    \begin{cases}
        \alpha_U(v) & \text{if } v\in U_1 \\
        \alpha_P(v) & \text{if } v\in P_1 
    \end{cases}$
\item $\mathscr{C}(\alpha)_E\colon E_1\rightarrow E_2$ is given by $\mathscr{C}(\alpha)_E(e) = 
    \begin{cases}
        \alpha_I(e) & \text{if } e\in I_1 \\
        \alpha_O(e) & \text{if } e\in O_1 
    \end{cases}$
\item $\mathscr{C}(\alpha)_L\colon L_1\rightarrow L_2$ is given by $\mathscr{C}(\alpha)_L(x) = 
    \begin{cases}
        \alpha_\Sigma(x) & \text{if } x\in \Sigma_1 \\
        \kappa & \text{otherwise} 
    \end{cases}$
\end{itemize}
\end{definition}

Definition \ref{def:computon-cfg} indicates that a computon CFG is obtained by just considering control flow edges, control ports and computation units, while ignoring data elements. As computation units are not labelled within computons, all of them take the constant label $\kappa$ in the corresponding CFG. The construction given by Definition \ref{def:computon-cfg} is functorial in the sense $\mathscr{C}$ preserves the structure of $\textbf{Set}^{\textbf{Comp}}$, including composition and identities (see Proposition \ref{prop:control-flow-functorial}). 

\begin{proposition}\label{prop:control-flow-functorial}
The process of building a computon CFG is functorial, i.e., $\mathscr{C}$ is a functor.
\end{proposition}
\begin{proof}
Given Definition \ref{def:computon-cfg}, we first prove that, for every computon morphism $\alpha\colon\lambda_1\rightarrow\lambda_2$, $\mathscr{C}(\alpha)\colon\mathscr{C}(\lambda_1)\rightarrow \mathscr{C}(\lambda_2)$ is a graph homomorphism that preserves labels and oriented incidence. To verify this, we show that the following diagram commutes:
\[
\begin{tikzcd}
 E_1 \arrow[r, shift left=2pt, "\vec{s}_1"]\arrow[r, shift right=2pt, "\vec{t}_1"']\arrow[d, "\mathscr{C}(\alpha)_E"'] & V_1 \arrow[r, "l_1"]\arrow[d, "\mathscr{C}(\alpha)_V"] & L_1\arrow[d, "\mathscr{C}(\alpha)_L"] \\
 E_2 \arrow[r, shift left=2pt, "\vec{s}_2"]\arrow[r, shift right=2pt, "\vec{t}_2"'] & V_2 \arrow[r, "l_2"'] & L_2 
\end{tikzcd}
\]
For our proof, we simplify notation by omitting the $(\alpha)$ part of a $\mathscr{C}(\alpha)$-component; for example, we write $\mathscr{C}_E$ for $\mathscr{C}(\alpha)_E$. We only verify the equation $l_2\circ\vec{s}_2\circ \mathscr{C}_E=\mathscr{C}_L\circ l_1\circ\vec{s}_1$ since the equation $l_2\circ\vec{t}_2\circ \mathscr{C}_E=\mathscr{C}_L\circ l_1\circ\vec{t}_1$ can be showed analogously. By letting $e \in E_1$, we have two cases:
\begin{enumerate}
\item $e \in I_1$:
\begin{align*}
 l_2(\vec{s}_2(\mathscr{C}_E(e))) &= l_2(\vec{s}_2(\alpha_I(e))) \hspace{0.4cm}\text{by the definition of }\mathscr{C}(\alpha)_E\text{ and considering } e\in I_1 \\
 &= l_2(s_2(\alpha_I(e))) \hspace{0.4cm}\text{by the definition of }\vec{s}_2\text{ and considering }\alpha_I(e)\in I_2 \\
 &= l_2(\alpha_P(s_1(e))) \hspace{0.33cm}\text{because } s_2\circ\alpha_I=\alpha_P\circ s_1 \\
 &= c_2(\alpha_P(s_1(e))) \hspace{0.3cm}\text{by the definition of }l_2\text{ and considering }\alpha_P(s_1(e))\in P_2\\
 &= \alpha_\Sigma(c_1(s_1(e))) \hspace{0.33cm}\text{because } c_2\circ\alpha_P=\alpha_\Sigma\circ c_1 \\
 &= \mathscr{C}_L(c_1(s_1(e))) \hspace{0.33cm}\text{by the definition of }\mathscr{C}(\alpha)_L\text{ and considering }c_1(s_1(e))\in \Sigma_1 \\
 &= \mathscr{C}_L(l_1(s_1(e))) \hspace{0.4cm}\text{by the definition of }l_1\text{ and considering } s_1(e)\in P_1 \\
 &= \mathscr{C}_L(l_1(\vec{s}_1(e))) \hspace{0.4cm}\text{by the definition of }\vec{s}_1\text{ and considering } e\in I_1 
\end{align*}
\item $e \in O_1$:
\begin{align*}
 l_2(\vec{s}_2(\mathscr{C}_E(e))) &= l_2(\vec{s}_2(\alpha_O(e))) \hspace{0.4cm}\text{by the definition of }\mathscr{C}(\alpha)_E\text{ and considering } e\in O_1 \\
 &= l_2(\sigma_2(\alpha_O(e))) \hspace{0.35cm}\text{by the definition of }\vec{s}_2\text{ and considering } \alpha_O(e)\in O_2 \\
 &= \kappa \hspace{2.4cm}\text{by the definition of }l_2\text{ and considering } \sigma_2(\alpha_O(e))\in U_2 \\
 &= \mathscr{C}_L(\kappa) \hspace{1.6cm}\text{by the definition of }\mathscr{C}_L\text{ and considering } \kappa\in L_1\setminus\Sigma_1 \\
 &= \mathscr{C}_L(l_1(\sigma_1(e))) \hspace{0.4cm}\text{by the definition of }l_1\text{ and considering } \sigma_1(e)\in U_1 \\
 &= \mathscr{C}_L(l_1(\vec{s}_1(e))) \hspace{0.45cm}\text{by the definition of }\vec{s}_1\text{ and considering } e\in O_1 \\
\end{align*}
\end{enumerate}
As the above diagram commutes, it follows that $\mathscr{C}(\alpha)$ preserves labels, sources and targets. To verify $\mathscr{C}$ preserves identities, we have to show $\mathscr{C}(1_\lambda)=1_{\mathscr{C}(\lambda)}$ for any computon $\lambda$, which is trivially true since both identity computon morphisms and identity graph homomorphisms are built upon identity functions. Considering the computon morphisms ${\alpha_1\colon\lambda_1\rightarrow\lambda_2}$ and ${\alpha_2\colon\lambda_2\rightarrow\lambda_3}$, we now show that composition is preserved too.

If $v \in V_1$, then either $v \in U_1$ or $v \in P_1$. We only prove the first case since the proof of the other is symmetric:
\begin{align*}
 \mathscr{C}(\alpha_2\circ\alpha_1)_V(v) &= (\alpha_2\circ\alpha_1)_U(v) \hspace{0.85cm}\text{by the definition of }\mathscr{C}(\alpha_2\circ\alpha_1)_V\text{ and considering } v\in U_1 \\
 &= \alpha_2(\alpha_1(v)) \hspace{1.5cm}\text{by Notation }\ref{notation:sets-im} \\ 
 &= (\mathscr{C}(\alpha_2)_V\circ\alpha_1)(v) \hspace{0.3cm}\text{by the definition of }\mathscr{C}(\alpha_2)_V \text{ and considering } \alpha_1(v)\in U_2 \\
 &= (\mathscr{C}(\alpha_2)_V\circ \mathscr{C}(\alpha_1)_V)(v) \hspace{0.15cm}\text{by the definition of }\mathscr{C}(\alpha_1)_V \text{ and considering } v\in U_1
\end{align*}
A similar approach can be used to show $\mathscr{C}(\alpha_2\circ\alpha_1)_E=\mathscr{C}(\alpha_2)_E\circ \mathscr{C}(\alpha_1)_E$ and $\mathscr{C}(\alpha_2\circ\alpha_1)_L=\mathscr{C}(\alpha_2)_L\circ \mathscr{C}(\alpha_1)_L$ so $\mathscr{C}(\alpha_2\circ\alpha_1)=\mathscr{C}(\alpha_2)\circ \mathscr{C}(\alpha_1)$ in general. As $\mathscr{C}$ preserves structure, identities and composition (the proof of associativity of composition is similar to the above), we conclude that the construction presented in Definition \ref{def:computon-cfg} is functorial. 
\end{proof}

Interestingly, the control flow structure of any computon can always be embodied by some other computon; thus, giving rise to the endofunctor described in Definition \ref{def:endofunctor-control}.

\begin{definition}[Control Flow Endofunctor]\label{def:endofunctor-control}
The endofunctor ${\mathfrak{E}\colon\textbf{Set}^{\textbf{Comp}}\rightarrow \textbf{Set}^{\textbf{Comp}}}$ maps a computon $\lambda$ to a computon $\mathfrak{E}(\lambda)$ as follows:
\begin{itemize}
\item The set $\mathfrak{E}(U)$ of computation units is $U$,
\item The set ${\mathfrak{E}(O)}$ of edges is given by ${\{o \in O \mid c(t(o))=0\}}$,
\item The set ${\mathfrak{E}(I)}$ of edges is given by ${\{i \in I \mid c(s(i))=0\}}$,
\item The set ${\mathfrak{E}(P)}$ of ports is given by ${\{p \in P \mid c(p)=0\}}$,
\item The set ${\mathfrak{E}(\Sigma)}$ of colours is ${\{0\}}$,
\item The function ${\mathfrak{E}(\sigma)}$ is given by ${\sigma \restriction_{\mathfrak{E}(O)}}$,
\item The function ${\mathfrak{E}(t)}$ is given by ${t \restriction_{\mathfrak{E}(O)}}$,
\item The function ${\mathfrak{E}(\tau)}$ is given by ${\tau \restriction_{\mathfrak{E}(I)}}$,
\item The function ${\mathfrak{E}(s)}$ is given by ${s \restriction_{\mathfrak{E}(I)}}$ and
\item The function ${\mathfrak{E}(c)}$ is given by ${c \restriction_{\mathfrak{E}(P)}}$.
\end{itemize}
Given a computon morphism ${\alpha\colon\lambda_1\rightarrow\lambda_2}$, the components of ${\mathfrak{E}(\alpha)\colon\mathfrak{E}(\lambda_1)\rightarrow \mathfrak{E}(\lambda_2)}$ are defined as follows:
\begin{itemize}
\item ${\mathfrak{E}(\alpha)_U\colon\mathfrak{E}(U_1)\rightarrow \mathfrak{E}(U_2)}$ by $\mathfrak{E}(\alpha)_U=\alpha_U$,
\item ${\mathfrak{E}(\alpha)_O\colon\mathfrak{E}(O_1)\rightarrow \mathfrak{E}(O_2)}$ by $\mathfrak{E}(\alpha)_O=\alpha_O \restriction_{\mathfrak{E}(O_1)}$,
\item ${\mathfrak{E}(\alpha)_I\colon\mathfrak{E}(I_1)\rightarrow \mathfrak{E}(I_2)}$ by $\mathfrak{E}(\alpha)_I=\alpha_I \restriction_{\mathfrak{E}(I_1)}$,
\item ${\mathfrak{E}(\alpha)_P\colon\mathfrak{E}(P_1)\rightarrow \mathfrak{E}(P_2)}$ by $\mathfrak{E}(\alpha)_P=\alpha_P \restriction_{\mathfrak{E}(P_1)}$, and
\item ${\mathfrak{E}(\alpha)_\Sigma\colon\mathfrak{E}(\Sigma_1)\rightarrow \mathfrak{E}(\Sigma_2)}$ by $\mathfrak{E}(\alpha)_\Sigma=\alpha_\Sigma \restriction_{\mathfrak{E}(\Sigma_1)}$. 
\end{itemize}
\end{definition} 

\begin{remark}
A glance at Definition \ref{def:endofunctor-control} reveals that a computon ${\mathfrak{E}(\lambda)\colon\textbf{Comp}\rightarrow\textbf{Set}}$ is a subfunctor of its source computon ${\lambda\colon\textbf{Comp}\rightarrow\textbf{Set}}$. That is, there exists a natural transformation ${\mathfrak{E}(\lambda)\rightarrow\lambda}$ whose components are monic. 
\end{remark}

Although the structure of any computon can be expressed as another in terms of computation units and control flow/ports only, connectivity is not always preserved by the construction presented in Definition \ref{def:endofunctor-control}. As an example, consider the computon $\lambda$ displayed in Figure \ref{fig:endofunctor-connectivity}(a), which clearly is connected because there is a sequence of information flows from each ec-inport to the only ec-outport. Figure \ref{fig:endofunctor-connectivity}(b) shows the resulting object ${\mathfrak{E}(\lambda)}$ which, although valid, does not satisfy Definition \ref{def:computon-connected}. This follows from the fact that the upper ec-inport does not have any sequence of information flows that can take it to the only ec-outport.\footnote{From the example presented in Figure \ref{fig:endofunctor-connectivity}, we can observe that the endofunctor $\mathfrak{E}$ could be useful to detect potential non-termination. Our concept of termination is formalised in Section \ref{sec:operational-semantics}.} 

\begin{figure}[!h]
\centering
\subcaptionbox{A connected computon $\lambda$.}
{
\begin{tikzpicture}
\node[draw=black,fill=white,inner sep=0pt,minimum size=3pt,opacity=0.4] (q3) at (0.55,-0.7) {};
\draw[dashed,opacity=0.4] (0.6,-0.7) to [] node [pos=0.5] {\arrowflow} node [pos=0.5,yshift=7] {} (1.35,-0.7);
\node[opacity=0.4] (v2) at (1.55,-0.7) {\scriptsize $u_1$};
\draw[dashed,opacity=0.4] (1.75,-0.7) to [] node [pos=0.5] {\arrowflow} node [pos=0.5,yshift=7] {} (2.65,-0.7);
\node[draw=black,fill fraction={black}{0.5},inner sep=0pt,minimum size=3pt,opacity=0.4] (z3) at (2.7,-0.7) {};
\draw[dashed,opacity=0.4] (2.75,-0.7) to [] node [pos=0.5] {\arrowflow} node [pos=0.5,yshift=7,opacity=0.4] {} (3.7,-0.7);
\node[opacity=0.4] (u4) at (3.9,-0.7) {\scriptsize $u_2$};
\draw[dashed,opacity=0.4] (u4) -- (3.9, -0.3) to [] node [pos=0.5,rotate=180] {\arrowflow} node [pos=0.5,yshift=7,opacity=0.4] {} (2.7,-0.3) -- (z3);
\draw (4.1,-0.7) to [] node [pos=0.5] {\arrowflow} node [pos=0.5,yshift=7,opacity=0.4] {} (4.95,-0.7);
\node[circle,draw=black,fill fraction={black}{0.5},inner sep=0pt,minimum size=3pt] (z5) at (5,-0.7) {};
\node[opacity=0.4] (u3) at (6.1,-0.7) {\scriptsize $u_3$};
\draw (z5) to [] node [pos=0.5] {\arrowflow} node [pos=0.5,yshift=7,opacity=0.4] {} (5.9,-0.7);
\node[fill=black,inner sep=0pt,minimum size=3pt,opacity=0.4] (y) at (7.2,-0.7) {};
\draw[dashed,opacity=0.4] (6.3,-0.7) to [] node [pos=0.5] {\arrowflow} node [pos=0.5,yshift=7,opacity=0.4] {} (y);

\node[draw=black,fill=white,inner sep=0pt,minimum size=3pt,opacity=0.4] (j0) at (0.55,-1.1) {};
\draw[dashed,opacity=0.4] (j0) to [] node [pos=0.5] {\arrowflow} node [pos=0.5,yshift=7,opacity=0.4] {} (6.1,-1.1) -- (6.1,-0.8);
\end{tikzpicture}  
}\hspace{1.3cm}
\subcaptionbox{Computon ${\mathfrak{E}(\lambda)}$ which, evidently, is not connected.}
{
\begin{tikzpicture}
\node[draw=black,fill=white,inner sep=0pt,minimum size=3pt,opacity=0.4] (q3) at (0.55,-0.7) {};
\draw[dashed,opacity=0.4] (0.6,-0.7) to [] node [pos=0.5] {\arrowflow} node [pos=0.5,yshift=7] {} (1.35,-0.7);
\node[opacity=0.4] (v2) at (1.55,-0.7) {\scriptsize $u_1$};
\draw[dashed,opacity=0.4] (1.75,-0.7) to [] node [pos=0.5] {\arrowflow} node [pos=0.5,yshift=7] {} (2.65,-0.7);
\node[draw=black,fill fraction={black}{0.5},inner sep=0pt,minimum size=3pt,opacity=0.4] (z3) at (2.7,-0.7) {};
\draw[dashed,opacity=0.4] (2.75,-0.7) to [] node [pos=0.5] {\arrowflow} node [pos=0.5,yshift=7,opacity=0.4] {} (3.7,-0.7);
\node[opacity=0.4] (u4) at (3.9,-0.7) {\scriptsize $u_2$};
\draw[dashed,opacity=0.4] (u4) -- (3.9, -0.3) to [] node [pos=0.5,rotate=180] {\arrowflow} node [pos=0.5,yshift=7,opacity=0.4] {} (2.7,-0.3) -- (z3);
\node[opacity=0.4] (u3) at (6.1,-0.7) {\scriptsize $u_3$};
\node[fill=black,inner sep=0pt,minimum size=3pt,opacity=0.4] (y) at (7.2,-0.7) {};
\draw[dashed,opacity=0.4] (6.3,-0.7) to [] node [pos=0.5] {\arrowflow} node [pos=0.5,yshift=7,opacity=0.4] {} (y);

\node[draw=black,fill=white,inner sep=0pt,minimum size=3pt,opacity=0.4] (j0) at (0.55,-1.1) {};
\draw[dashed,opacity=0.4] (j0) to [] node [pos=0.5] {\arrowflow} node [pos=0.5,yshift=7,opacity=0.4] {} (6.1,-1.1) -- (6.1,-0.8);
\end{tikzpicture}  
}
{
\begin{tikzpicture}
\matrix [below, ampersand replacement=\&] at (current bounding box.south) {
\draw[dashed] (2.3,0) to node[pos=0.5,yshift=4]{\arrowflow} (2.8,0); \& \node[inner sep=0pt,minimum size=3pt,label=right:{\scriptsize Control flow edge}]{}; \& \node{}; \& \&
 \node[draw=black,fill=white,inner sep=0pt,minimum size=3pt,label=right:{\scriptsize Ec-inport}] {}; \& \node{};
 \& \node[fill=black,inner sep=0pt,minimum size=3pt,label={right:{\scriptsize Ec-outport}}] {}; \& \node{};
 \& \& \node[draw,fill=white,fill fraction={black}{0.5},inner sep=0pt,minimum size=3pt,label={right:\scriptsize Ic-port}] {};
 \& \node at (0,0){}; \& \draw (0,0) to node[pos=0.5,yshift=4]{\arrowflow} (0.5,0);
 \& \& \node[inner sep=0pt,minimum size=3pt,label=right:{\scriptsize Data flow edge}] {};
 \& \node at (0,0){}; \& \node[circle,draw,fill=white,fill fraction={black}{0.5},inner sep=0pt,minimum size=3pt,label=right:{\scriptsize Id-port}] {};\\
};
\end{tikzpicture}
}
\caption{An example to show that the endofunctor $\mathfrak{E}$ does not preserve connectivity in general. Like in Figure \ref{fig:computon-morphism-example}, for now we just display computation units as labels. In upcoming sections, we will use specific syntax to distinguish among different types of units. We use opacity to illustrate the mapping done by $\mathfrak{E}$ and, thereby, highlighting that the data elements of $\lambda$ are ``forgotten'' in ${\mathfrak{E}(\lambda)}$.}
\label{fig:endofunctor-connectivity}
\end{figure}

Despite not preserving connectivity, $\mathfrak{E}$ retains all control ports and computation units, together with their flow adjacency and port colouring. More generally speaking, such endofunctor is functorial because it preserves the structure of $\textbf{Set}^{\textbf{Comp}}$, including composition and identities. Checking functoriality can be trivially done in a similar manner as the proof of Proposition \ref{prop:control-flow-functorial}. We just need to check that the objects and morphisms in the image of $\mathfrak{E}$ are indeed computons and computon morphisms in the sense of Definitions \ref{def:computon} and \ref{def:computon-morphism}. For this, we have Propositions \ref{prop:endofunctor-control-objects} and \ref{prop:endofunctor-control-morphisms}.

\begin{proposition}\label{prop:endofunctor-control-objects}
If $\lambda$ is a computon, then $\mathfrak{E}(\lambda)$ is also a computon.
\end{proposition}
\begin{proof}
By Definitions \ref{def:computon} and \ref{def:endofunctor-control}, it is clear that ${\mathfrak{E}(\lambda)}$ must have a (possibly empty) set ${\mathfrak{E}(U)}$ of computation units, a non-empty set ${\mathfrak{E}(P)}$ of ports, a possibly empty set ${\mathfrak{E}(O)}$ of outgoing edges, a possibly empty set ${\mathfrak{E}(I)}$ of incoming edges and a non-empty set ${\mathfrak{E}(\Sigma)=\{0\}\subset \mathbb{N}}$ of colours. As ${\mathfrak{E}(s)=s \restriction_{\mathfrak{E}(I)}}$ and $s$ is total (by Definitions \ref{def:endofunctor-control} and \ref{def:computon}), ${\mathfrak{E}(s)}$ is a total function. Analogously, ${\mathfrak{E}(t)}$, ${\mathfrak{E}(\sigma)}$ and ${\mathfrak{E}(\tau)}$ are total too. 

To show ${\mathfrak{E}(\tau)}$ is surjective in addition, simply use Definition \ref{def:endofunctor-control} and the fiber of ${c\circ s}$ over $0$ to derive ${\mathfrak{E}(\tau)=\tau\restriction_{\mathfrak{E}(I)}=\tau\restriction_{\{i \in I \mid c(s(i))=0\}}=\tau\restriction_{(c\circ s)^{-1}(0)}}$. As ${\tau\restriction_{(c\circ s)^{-1}(0)}}$ is necessarily surjective according to Definition \ref{def:computon}, it follows that ${\mathfrak{E}(\tau)}$ also is. An analogous approach can prove surjectivity for ${\mathfrak{E}(\sigma)}$. Since ${\mathfrak{E}(c)(p)=0}$ for all ${p \in \mathfrak{E}(P)}$, it is obvious ${\mathfrak{E}(c)}$ is total surjective.

Now, if we consider an arbitrary unit ${u\in\mathfrak{E}(U)}$, there must be some ${i\in\mathfrak{E}(I)}$ because ${\mathfrak{E}(\tau)}$ is onto, i.e., ${i\in\{j\in I\mid c(s(j))=0\}}$ as per Definition \ref{def:endofunctor-control}. By noticing ${s\restriction_{\mathfrak{E}(I)}}$ and ${c\restriction_{\mathfrak{E}(P)}}$ respectively map input flows to control ports and control ports to $0$, we deduce ${c\restriction_{\mathfrak{E}(P)}\circ s\restriction_{\mathfrak{E}(I)}}$ must map input flows to $0$. In other words, ${i\in (c\restriction_{\mathfrak{E}(P)}\circ s\restriction_{\mathfrak{E}(I)})^{-1}(0)}$ or, in line with Definition \ref{def:endofunctor-control}, ${i\in(\mathfrak{E}(c)\circ\mathfrak{E}(s))^{-1}(0)}$. Hence, ${\mathfrak{E}(\tau)\restriction_{(\mathfrak{E}(c)\circ\mathfrak{E}(s))^{-1}(0)}}$ is a surjective function.

Finally, to show $\mathfrak{E}(\lambda)$ has at least one ec-inport and at least one ec-outport, let $q \in P^+$ with $c(q)=0$ (i.e., $q \in \mathfrak{E}(P)$) and consider $q \in P^+ \implies q \notin Im(t) \implies q \notin Im(\mathfrak{E}(t)) \implies q \in \mathfrak{E}(P)^+$. Since ${q \in \mathfrak{E}(P)^+}$ and ${\mathfrak{E}(c)(q)=c(q)=0}$, $q$ must be an ec-inport of ${\mathfrak{E}(\lambda)}$. An analogous reasoning can additionally show that $\mathfrak{E}(\lambda)$ has at least one ec-outport. 

Having satisfied all the conditions from Definition \ref{def:computon}, we conclude ${\mathfrak{E}(\lambda)}$ is a well-formed computon.
\end{proof}

\begin{proposition}\label{prop:endofunctor-control-morphisms}
If ${\alpha\colon\lambda_1\rightarrow \lambda_2}$ is a computon morphism, then ${\mathfrak{E}(\alpha)\colon\mathfrak{E}(\lambda_1)\rightarrow \mathfrak{E}(\lambda_2)}$ is a computon morphism.
\end{proposition}
\begin{proof}
Assuming ${\alpha:\lambda_1\rightarrow \lambda_2}$ is a computon morphism, we first check $\vec{i}(\mathfrak{E}(\alpha))\cup\vec{o}(\mathfrak{E}(\alpha))\subseteq \mathfrak{E}(P_1)^+\cup \mathfrak{E}(P_1)^-$. For this, let $p \in \vec{i}(\mathfrak{E}(\alpha))\cup\vec{o}(\mathfrak{E}(\alpha))$ so that $\bullet \mathfrak{E}(\alpha)(p)\setminus \mathfrak{E}(\alpha)(\bullet p)\neq\emptyset$ or $\mathfrak{E}(\alpha)(p)\bullet\setminus \mathfrak{E}(\alpha)(p\bullet)\neq\emptyset$. By Definition \ref{def:endofunctor-control}, we have $\mathfrak{E}(\alpha)(p)=\alpha(p)$ because ${p \in \mathfrak{E}(P_1)}$. So, $\bullet \alpha(p)\setminus \alpha(\bullet p)\neq\emptyset$ or $\alpha(p)\bullet\setminus \alpha(p\bullet)\neq\emptyset$. Using Definitions \ref{def:computon-morphism} and \ref{def:computon-interface}, $p\in\vec{i}(\alpha)\cup\vec{o}(\alpha)\implies p\in P_1^+\cup P_1^-$ so $p\notin Im(t_1)$ or $p\notin Im(s_1)$. As $\mathfrak{E}(t_1)=t_1 \restriction_{\mathfrak{E}(O_1)}$ and $\mathfrak{E}(s_1)=s_1 \restriction_{\mathfrak{E}(I_1)}$, $p\notin Im(\mathfrak{E}(t_1))$ or $p\notin Im(\mathfrak{E}(s_1))$. Hence, $p \in \mathfrak{E}(P_1)^+ \cup \mathfrak{E}(P_1)^-$.

Now, notice $\alpha_U$, $\alpha_P$, $\alpha_O$ and $\alpha_I$ are necessarily total functions by Definition \ref{def:computon-morphism}. As $\mathfrak{E}(\alpha)_U=\alpha_U$, $\mathfrak{E}(\alpha)_P=\alpha_P \restriction_{\mathfrak{E}(P_1)}$, $\mathfrak{E}(\alpha)_O=\alpha_O \restriction_{\mathfrak{E}(O_1)}$ and $\mathfrak{E}(\alpha)_I=\alpha_I \restriction_{\mathfrak{E}(I_1)}$, we use the definition of function restriction to deduce $\mathfrak{E}(\alpha)_U$, $\mathfrak{E}(\alpha)_P$, $\mathfrak{E}(\alpha)_O$ and $\mathfrak{E}(\alpha)_I$ must be total too. As $\mathfrak{E}(\alpha)_\Sigma$ is necessarily a function $\{0\}\rightarrow\{0\}$ given by $0 \mapsto 0$ (see Definition \ref{def:endofunctor-control}), $\mathfrak{E}(\alpha)_\Sigma$ is trivially an inclusion. Therefore, we conclude $\mathfrak{E}(\alpha)$ satisfies all the conditions from Definition \ref{def:computon-morphism}, i.e., $\mathfrak{E}(\alpha)$ is a computon morphism.
\end{proof}

Unfortunately, the data flow structure of a computon cannot be fully embodied by another computon, since Definition \ref{def:computon} enforces objects in $\textbf{Set}^{\textbf{Comp}}$ to always have control ports. Nevertheless, such a structure can still be presented as a Data Flow Graph (DFG) which is just a directed labelled graph obtained via the application of the functor from Definition \ref{def:computon-dfg}. Again, as computation units are not labelled within computons, all of them take the constant label $\kappa$ in the corresponding DFG.

\vspace{1pt}

\begin{definition}[Computon DFG]\label{def:computon-dfg}
The functor $\mathscr{D}\colon\textbf{Set}^{\textbf{Comp}}\rightarrow \textbf{Set}^{\textbf{Gr}}$ maps a computon $\lambda$ to its underlying DFG $\mathscr{D}(\lambda)$ as follows:
\begin{itemize}
\item The set $V$ of vertices of $\mathscr{D}(\lambda)$ is given by $U\cup\{p\in P \mid c(p)>0\}$.
\item The set $E$ of edges of $\mathscr{D}(\lambda)$ is given by $\{i\in I \mid c(s(i))>0\}\cup\{o\in O \mid c(t(o))>0\}$.
\item The set $L$ of labels of $\mathscr{D}(\lambda)$ is given by $\Sigma\cup\{\kappa\}$.
\item The source function $\vec{s}\colon E\rightarrow V$ of $\mathscr{D}(\lambda)$ is given by $\vec{s}(e) = 
    \begin{cases}
        s(e) & \text{if } e\in I \\
        \sigma(e) & \text{if } e\in O
    \end{cases}$
\item The target function $\vec{t}\colon E\rightarrow V$ of $\mathscr{D}(\lambda)$ is given by $\vec{t}(e) = 
    \begin{cases}
        t(e) & \text{if } e\in O \\
        \tau(e) & \text{if } e\in I
    \end{cases}$
\item The labelling function $l\colon V\rightarrow L$ of $\mathscr{D}(\lambda)$ is given by $l(v) = 
    \begin{cases}
        c(v) & \text{if } v\in P \\
        \kappa & \text{if } v\in U
    \end{cases}$
\end{itemize}
For a computon morphism $\alpha\colon\lambda_1\rightarrow\lambda_2$, we define the components of the corresponding graph homomorphism $\mathscr{D}(\alpha)\colon\mathscr{D}(\lambda_1)\rightarrow \mathscr{D}(\lambda_2)$ as follows:
\begin{itemize}
\item $\mathscr{D}(\alpha)_V\colon V_1\rightarrow V_2$ by $\mathscr{D}(\alpha)_V(v) = 
    \begin{cases}
        \alpha_U(v) & \text{if } v\in U_1 \\
        \alpha_P(v) & \text{if } v\in P_1 
    \end{cases}$
\item $\mathscr{D}(\alpha)_E\colon E_1\rightarrow E_2$ by $\mathscr{D}(\alpha)_E(e) = 
    \begin{cases}
        \alpha_I(e) & \text{if } e\in I_1 \\
        \alpha_O(e) & \text{if } e\in O_1 
    \end{cases}$
\item $\mathscr{D}(\alpha)_L\colon L_1\rightarrow L_2$ by $\mathscr{D}(\alpha)_L(x) = 
    \begin{cases}
        \alpha_\Sigma(x) & \text{if } x\in \Sigma_1 \\
        \kappa & \text{otherwise} 
    \end{cases}$
\end{itemize}
\end{definition}

Definition \ref{def:computon-dfg} indicates that a computon DFG $\mathscr{D}(\lambda)$ is constructed in a similar fashion as its counterpart. The only difference is that, rather than considering control elements, $\mathscr{D}(\lambda)$ considers data ports, data flows and computation units only. That is, Definitions \ref{def:computon-cfg} and \ref{def:computon-dfg} just differ in how vertices, edges and labels are constructed. By Proposition \ref{prop:data-flow-functorial}, Definition \ref{def:computon-dfg} also yields a functorial construction.

\begin{proposition}\label{prop:data-flow-functorial}
The process of building a computon DFG is functorial, i.e., $\mathscr{D}$ is a functor.
\end{proposition}
\begin{proof}
Definitions \ref{def:computon-cfg} and \ref{def:computon-dfg} define the labelling, source and target functions in the same way. Similarly, the components of a graph homomorphism are defined analogously. Therefore, the proof of this proposition is similar to that of Proposition \ref{prop:control-flow-functorial}.
\end{proof}
\section{Operational Semantics} \label{sec:operational-semantics}

The operational behaviour of a computon can be described via \emph{token game semantics} because the structure of a computon can be expressed as a classical P/T Petri net ${(S,T,in,out)}$ where $S$ is a set of places, $T$ is a set of transitions and ${in,out:T\rightarrow S^\oplus}$ are functions assigning to each transition its corresponding pre- and post-set.\footnote{In this paper, we use the term classical (P/T) Petri net to refer to a categorical Petri net, which are equivalent in expressive power \cite{ermel_taste_1996}.} Here, ${S^\oplus}$ denotes the free commutative monoid generated by $S$. As each element of ${S^\oplus}$ is a formal finite linear combination with natural number coefficients, $\oplus$ denotes addition of linear combinations. Given this notion, Definition \ref{def:petri-net-cat} formalises the category of Petri nets.

\begin{definition}[Category of Petri Nets]\label{def:petri-net-cat}
\textbf{Petri} is the category of Petri nets, where each morphism ${f:(S_1,T_1,in_1,out_1)\rightarrow (S_2,T_2,in_2,out_2)}$ consists of a total function ${f_T:T_1\rightarrow T_2}$ and a monoid homomorphism ${f^\oplus:S_1^\oplus\rightarrow S_2^\oplus}$ that make the following diagram commute:
\[
\begin{tikzcd}
 T_1\arrow[d, "f_T"']\arrow[rrr, shift left=2pt, "in_1"]\arrow[rrr, shift right=2pt, "out_1"'] & & & S_1^\oplus\arrow[d, "f^\oplus"] \\
T_2 \arrow[rrr, shift left=2pt, "in_2"]\arrow[rrr, shift right=2pt, "out_2"'] & & & S_2^\oplus
\end{tikzcd}
\]
As it is a sound category, $\textbf{Petri}$ satisfies the identity law and is closed under associative composition (which is defined component-wise) \cite{meseguer_petri_1990}. For monoidal structure preservation, ${f^\oplus}$ leaves the identity fixed while respecting addition of linear combinations. This homomorphism, which is the unique extension of a function ${f\colon S_1\rightarrow S_2}$, is given by ${f^\oplus(\bigoplus_{s\in S_1} n_s s)=\bigoplus_{s\in S_1}n_s f(s)}$ where ${n_s\in\mathbb{N}}$ denotes the multiplicity of a place ${s\in S_1}$. In other words, ${f^\oplus}$ respects the combinatorial structure of $S_1$ in $S_2$.
\end{definition}

Using Definition \ref{def:petri-net-cat}, we now specify a functor to map each computon to its underlying Petri net, i.e., we provide operational semantics for computons in the theory of classical Petri nets. 

\begin{definition}[Computons as Petri Nets]\label{def:functor-computon-to-petri}
The functor ${\mathcal{N}:\textbf{Set}^{\textbf{Comp}}\rightarrow\textbf{Petri}}$ sends any computon $\lambda$ to its underlying Petri net ${\mathcal{N}(\lambda)}$ such that:
\begin{itemize}
\item ${\mathcal{N}(\lambda)}$ has the set $P$ of ports as its set of places.
\item ${\mathcal{N}(\lambda)}$ has the set $U$ of computation units as its set of transitions.
\item The pre-set function ${in:U\rightarrow P^\oplus}$ is given by $in(u)=\bigoplus_{p\in\bullet u}1p$ for all $u\in U$.
\item The post-set function ${out:U\rightarrow P^\oplus}$ is given by $out(u)=\bigoplus_{p\in u\bullet}1p$ for all $u\in U$.
\end{itemize}
On mappings, $\mathcal{N}$ takes a computon morphism ${\alpha:\lambda_1\rightarrow\lambda_2}$ to a net morphism $\mathcal{N}(\alpha):\mathcal{N}(\lambda_1)\rightarrow \mathcal{N}(\lambda_2)$ as follows:
\begin{itemize}
\item ${\mathcal{N}(\alpha)_T}:U_1\rightarrow U_2$ is a function given by ${\mathcal{N}(\alpha)_T(u)=\alpha_U(u)}$ for all ${u \in U_1}$, and
\item ${\mathcal{N}(\alpha)^\oplus:P_1^\oplus\rightarrow P_2^\oplus}$ is a monoid homomorphism which leaves the identity fixed and respects the monoid operation $\oplus$ such that, for every combination ${p_1\oplus\cdots\oplus p_n}$ given by ${P_1^\oplus}$, we have $\mathcal{N}(\alpha)^\oplus(p_1\oplus\cdots\oplus p_n )=\alpha_P(p_1)\oplus\cdots\oplus\alpha_P(p_n)$.
\end{itemize}
Here, for $j=1,2$ we clearly abuse notation by the use of $U_j$ and ${P_j^\oplus}$ to respectively denote the set of transitions of $\mathcal{N}(\lambda_j)$ and the free commutative monoid on $\mathcal{N}(\lambda_j)$-places.
\end{definition}

A glance at Definition \ref{def:functor-computon-to-petri} reveals that the functor $\mathcal{N}$ preserves the structure of the category of computons in $\textbf{Petri}$, including composition of computon morphisms, the identity law and associativity of composition. The functoriality of $\mathcal{N}$ can be trivially checked by noticing that a net morphism ${\mathcal{N}(\alpha)}$ is completely built upon the $U$- and $P$-components of a computon morphism $\alpha$, and that a Petri net ${\mathcal{N}(\lambda)}$ includes all the ports (i.e., places) and computation units (i.e., transitions) from $\lambda$. Particularly, the pre- and post-set functions of ${\mathcal{N}(\lambda)}$ only consider unitary coefficients with no repeated places. So, there is a one-to-one correspondence between the input places of a transition $u$ and the ports in ${\bullet u}$, and between the output places of $u$ and the ports in ${u \bullet}$ (see Definition \ref{def:pre-post-sets}). That is, the pre- and post-set of $u$ can directly be treated as sets of places rather than bags. This formalisation can be done without any consequences since multiplicity of inputs/outputs is explicitly specified in a computon, as a result of operating directly on individual port elements and individual edges. Controlling multiplicity outside Petri nets is an example of how composition semantics can dictate operational aspects. Although we decide to use Petri nets because they fit naturally with the structure of computons, other formalisms can be used to specify operational semantics for computons (e.g., timed Petri nets \cite{ramchandani_analysis_1974}).\footnote{A same computon can behave differently depending on the chosen operational semantics.} This flexibility is due to the separation of structure/composition (in computons) from operation (in P/T Petri nets in our case). 

Although the categorical structure of ${\textbf{Set}^{\textbf{Comp}}}$ is totally preserved in ${\textbf{Petri}}$, the structure of individual objects (computons) is not fully preserved because port colouring is not taken into account. Despite of forgetting port colouring, control and data flow is implicit in the underlying net of any computon. To explicitly present control flow as a net, the proof of Proposition \ref{prop:functor-control-petri} says we can construct a functor from ${\mathcal{N}}$ that does not operate on the whole category of computons but on the image of the endofunctor $\mathfrak{E}$ from Definition \ref{def:endofunctor-control}.  

\begin{proposition}\label{prop:functor-control-petri}
There is a composite functor ${\mathcal{C}\circ \mathfrak{E}:\textbf{Set}^{\textbf{Comp}}\rightarrow\textbf{Petri}}$ to present the control flow structure of any computon as a Petri net.
\end{proposition}
\begin{proof}
The image ${\mathfrak{E}(\textbf{Set}^{\textbf{Comp}})}$ of the category ${\textbf{Set}^{\textbf{Comp}}}$ under the endofunctor $\mathfrak{E}$ is trivially a full subcategory of ${\textbf{Set}^{\textbf{Comp}}}$ (see Definition \ref{def:endofunctor-control}). Then, there is an obvious functor ${\mathcal{C}:\mathfrak{E}(\textbf{Set}^{\textbf{Comp}})\rightarrow\textbf{Petri}}$ given by the restriction of $\mathcal{N}$ to ${\mathfrak{E}(\textbf{Set}^{\textbf{Comp}})}$. That is, $\mathcal{C}$ uses the mapping from Definition \ref{def:functor-computon-to-petri} to send the control flow structure ${\mathfrak{E}(\lambda)}$ of a computon $\lambda$ to a Petri net ${\mathcal{C}(\mathfrak{E}(\lambda))}$ and a computon morphism ${\mathfrak{E}(\lambda_1)\rightarrow \mathfrak{E}(\lambda_2)}$ to a net morphism ${\mathcal{C}(\mathfrak{E}(\lambda_1))\rightarrow \mathcal{C}(\mathfrak{E}(\lambda_2))}$. 

Evidently, the composite functor ${\mathcal{C}\circ \mathfrak{E}:\textbf{Set}^{\textbf{Comp}}\rightarrow\textbf{Petri}}$ is well-defined because the image of $\mathfrak{E}$ is the domain of $\mathcal{C}$. 
\end{proof}

Unfortunately, presenting the data flow structure of a computon as a net cannot be done as elegantly as we did for Proposition \ref{prop:functor-control-petri}, since computons always require control ports and, therefore, there is no way of solely representing data flow in the form of a computon. Despite of this, Proposition \ref{prop:functor-data-petri} shows that data flow can indeed be presented as a net through a functorial construction akin to Definition \ref{def:functor-computon-to-petri}.

\begin{proposition}\label{prop:functor-data-petri}
There is a functor $\mathcal{D}:\textbf{Set}^{\textbf{Comp}}\rightarrow\textbf{Petri}$ to present the data flow structure of any computon as a Petri net.
\end{proposition}
\begin{proof}
Proving this proposition requires the construction of a functor ${\mathcal{D}:\textbf{Set}^{\textbf{Comp}}\rightarrow\textbf{Petri}}$ which, given a computon $\lambda$, defines a Petri net ${\mathcal{D}(\lambda)}$ as follows:
\begin{itemize}
\item ${\mathcal{D}(\lambda)}$ has ${\{p \in P \mid c(p)>0\}}$ as its set $S$ of places, i.e., ${S^\oplus}$ is the free commutative monoid on the set of data ports of $\lambda$.
\item ${\mathcal{D}(\lambda)}$ has the set $U$ of computation units as its set of transitions.
\item The pre-set function ${in:U\rightarrow S^\oplus}$ is given by $in(u)=\bigoplus_{p\in\bullet u\cap S}1p$ for all $u\in U$.
\item The post-set function ${out:U\rightarrow S^\oplus}$ is given by $out(u)=\bigoplus_{p\in u\bullet\cap S}1p$ for all $u\in U$.
\end{itemize}
On mappings, the functor ${\mathcal{D}}$ takes each computon morphism ${\alpha:\lambda_1\rightarrow\lambda_2}$ to a net morphism $\mathcal{D}(\alpha):\mathcal{D}(\lambda_1)\rightarrow\mathcal{D}(\lambda_2)$ as follows:
\begin{itemize}
\item ${\mathcal{D}(\alpha)_T}:U_1\rightarrow U_2$ is a function given by ${\mathcal{D}(\alpha)_T(u)=\alpha_U(u)}$ for all ${u \in U_1}$, and
\item ${\mathcal{D}(\alpha)^\oplus:S_1^\oplus\rightarrow S_2^\oplus}$ is a monoid homomorphism which leaves the identity fixed and respects the monoid operation $\oplus$ such that, for every combination ${p_1\oplus\cdots\oplus p_n}$ given by ${S_1^\oplus}$, we have $\mathcal{D}(\alpha)^\oplus(p_1\oplus\cdots\oplus p_n )=\alpha_P(p_1)\oplus\cdots\oplus\alpha_P(p_n)$.
\end{itemize}
Here, for $j=1,2$ we clearly abuse notation by the use of $U_j$ and ${S_j^\oplus}$ to respectively denote the set of transitions of $\mathcal{D}(\lambda_j)$ and the free commutative monoid on $\mathcal{D}(\lambda_j)$-places.

\vspace{0.35cm}

Consider the net morphism ${\mathcal{D}(\alpha):\mathcal{D}(\lambda_1)\rightarrow\mathcal{D}(\lambda_2)}$ and without loss of generality assume ${u_1 \in U_1}$ is a transition of $\mathcal{D}(\lambda_1)$ where ${in_1(u_1)=p_1\oplus\cdots\oplus p_n}$. As each ${p_j\in\bullet u_1\cap S_1}$ has a unitary coefficient and occurs only once, we can treat ${p_1\oplus\cdots\oplus p_n}$ and the set $\{p_1,\ldots,p_n\}$ interchangeably. Leveraging $\{p_1,\ldots,p_n\}\subseteq\bullet u_1$, we use the commutativity property of $\alpha$ to deduce ${\alpha_U(u_1)\in U_2}$ and ${\{\alpha_P(p_1),\ldots,\alpha_P(p_n)\}\subseteq\bullet\alpha_U(u_1)}$. With this in mind, we now prove ${\mathcal{D}(\alpha)^\oplus\circ in_1=in_2\circ\mathcal{D}(\alpha)_T}$:
\begin{align*}
 \mathcal{D}(\alpha)^\oplus(in_1(u_1)) &= \mathcal{D}(\alpha)^\oplus(p_1\oplus\cdots\oplus p_n)\hspace{0.4cm}\text{ since }in_1(u_1)=p_1\oplus\cdots\oplus p_n \\
 &= \alpha_P(p_1)\oplus\cdots\oplus\alpha_P(p_n)\hspace{0.25cm}\text{ by the definition of }\mathcal{D}(\alpha)^\oplus \\
 &= \{\alpha_P(p_1),\ldots,\alpha_P(p_n)\}\hspace{0.4cm}\text{ treating the combination as a set (see above)} \\
 &\subseteq \bullet\alpha_U(u_1)\hspace{2.63cm}\text{ by the commutativity of }\alpha\text{ (see above)}  
\end{align*}

Having ${\{\alpha_P(p_1),\ldots,\alpha_P(p_n)\}\subseteq\bullet\alpha_U(u_1)}$ allows us to use Definition \ref{def:pre-post-sets} to deduce ${\alpha_P(p_j)\in P_2}$ and ${(\exists i\in I_2)[\alpha_P(p_j)\xrightarrow{i}\alpha_U(u_1)]}$ for all ${j=1,\ldots,n}$. Since ${c_1(p_j)>0}$ because ${p_j\in S_1}$, the colour preservation of $\alpha$ says ${c_2(\alpha_P(p_j))>0}$ for all ${j=1,\ldots,n}$, i.e., ${\{\alpha_P(p_1),\ldots,\alpha_P(p_n)\}\subseteq S_2}$. Hence: 
\begin{align*}
 \{\alpha_P(p_1),\ldots,\alpha_P(p_n)\} &=\alpha_P(p_1)\oplus\cdots\oplus\alpha_P(p_n)\hspace{0.4cm}\text{ treating the set as a combination }\\
 &= in_2(\alpha_U(u_1))\hspace{0.4cm}\text{because }\{\alpha_P(p_1),\ldots,\alpha_P(p_n)\}\subseteq\bullet\alpha_U(u_1)\cap S_2 \\
 &= in_2(\mathcal{D}(\alpha)_T(u_1))\hspace{0.25cm}\text{ by the definition of }\mathcal{D}(\alpha)_T 
\end{align*}
Thus, showing ${\mathcal{D}(\alpha)_S^\oplus\circ in_1=in_2\circ\mathcal{D}(\alpha)_T}$ in general. Since the proof of $\mathcal{D}(\alpha)_S^\oplus\circ out_1=out_2\circ\mathcal{D}(\alpha)_T$ is completely analogous, we conclude that the functor $\mathcal{D}$ preserves transitions (i.e., computation units) and place adjacency (i.e., data port adjacency). To check $\mathcal{D}$ also preserves composition, it suffices to observe that the $T$-component of a net morphism ${\mathcal{D}(\alpha)}$ corresponds to $\alpha_U$ and that the other component is entirely built upon $\alpha_P$. Therefore, ${\mathcal{D}(\alpha_2\circ\alpha_1)=\mathcal{D}(\alpha_2)\circ\mathcal{D}(\alpha_1)}$ for any pair of computon morphisms ${\alpha_1:\lambda_1\rightarrow\lambda_2}$ and ${\alpha_2:\lambda_2\rightarrow\lambda_3}$. As ${1_{\mathcal{D}(\lambda_2)}\circ\mathcal{D}(\alpha)=\mathcal{D}(\alpha)=\mathcal{D}(\alpha)\circ 1_{\mathcal{D}(\lambda_1)}}$ holds component-wise, the identity law also follows trivially.

\vspace{0.35cm}

For transitions, associativity of composition holds because associativity is satisfied in ${\textbf{Set}^\textbf{Comp}}$ and the $T$-component of a net morphism corresponds to the $U$-component of a computon morphism. We just need to verify associativity for the corresponding monoid homomorphism by considering the net morphisms ${\mathcal{D}(\alpha_1):\mathcal{D}(\lambda_1)\rightarrow\mathcal{D}(\lambda_2)}$, ${\mathcal{D}(\alpha_2):\mathcal{D}(\lambda_2)\rightarrow\mathcal{D}(\lambda_3)}$ and ${\mathcal{D}(\alpha_3):\mathcal{D}(\lambda_2)\rightarrow\mathcal{D}(\lambda_3)}$. Assuming ${q_1\oplus\cdots\oplus q_m}$ is a linear combination from ${S_1^\oplus}$, we obtain:

\begin{align*}
 & \mathcal{D}(\alpha_3)^\oplus[(\mathcal{D}(\alpha_2)^\oplus\circ\mathcal{D}(\alpha_1)^\oplus)(q_1\oplus\cdots\oplus q_m)] \\
 =& \mathcal{D}(\alpha_3)^\oplus[(\mathcal{D}(\alpha_2\circ\alpha_1)^\oplus(q_1\oplus\cdots\oplus q_m)]\hspace{1.2cm}\text{ by composition preservation} \\ 
 =& \mathcal{D}(\alpha_3)^\oplus[(\alpha_2\circ\alpha_1)(q_1)\oplus\cdots\oplus(\alpha_2\circ\alpha_1)(q_m)]\hspace{0.13cm}\text{ by the definition of }\mathcal{D}(\alpha_2\circ\alpha_1)^\oplus \\ 
 =& \alpha_3[(\alpha_2\circ\alpha_1)(q_1)]\oplus\cdots\oplus\alpha_3[(\alpha_2\circ\alpha_1)(q_m)]\hspace{0.38cm}\text{ by the definition of }\mathcal{D}(\alpha_3)^\oplus \\
 =& (\alpha_3\circ\alpha_2)[\alpha_1(q_1)]\oplus\cdots\oplus(\alpha_3\circ\alpha_2)[\alpha_1(q_m)]\hspace{0.4cm}\text{ by associativity of computon morphisms} \\
 =& \mathcal{D}(\alpha_3\circ\alpha_2)^\oplus[\alpha_1(q_1)\oplus\cdots\oplus\alpha_1(q_m)]\hspace{1.58cm}\text{ by the definition of }\mathcal{D}(\alpha_3\circ\alpha_2)^\oplus \\
 =& \mathcal{D}(\alpha_3\circ\alpha_2)^\oplus[\mathcal{D}(\alpha_1)^\oplus(q_1\oplus\cdots\oplus q_m)]\hspace{1.45cm}\text{ by the definition of }\mathcal{D}(\alpha_1)^\oplus \\
 =& [\mathcal{D}(\alpha_3)^\oplus\circ\mathcal{D}^\oplus(\alpha_2)][\mathcal{D}(\alpha_1)^\oplus(q_1\oplus\cdots\oplus q_m)]\hspace{0.4cm}\text{ by composition preservation}
\end{align*}
Thus, showing ${\mathcal{D}(\alpha_3)^\oplus\circ[\mathcal{D}(\alpha_2)^\oplus\circ\mathcal{D}(\alpha_1)^\oplus]=[\mathcal{D}(\alpha_3)^\oplus\circ\mathcal{D}^\oplus(\alpha_2)]\circ\mathcal{D}(\alpha_1)^\oplus}$ and, hence, demonstrating that the construction given by $\mathcal{D}$ is indeed functorial. That is, $\mathcal{D}$ preserves the structure of data ports (i.e., places) and computation units (i.e., transitions) together with identities and composition, while satisfying associativity of composition. Thus, we conclude $\mathcal{D}$ is a sound categorical construction to faithfully present the data flow structure of any computon as a Petri net.  
\end{proof}

Definition \ref{def:functor-computon-to-petri} together with Propositions \ref{prop:functor-control-petri} and \ref{prop:functor-data-petri} constitute three alternative ways of studying information flow within a computon in the theory of classical P/T Petri nets. The former allows us to describe both control flow and data flow as a net, whereas the second and third one serve to respectively present control flow and data flow.\footnote{Although $\mathcal{N}$ defines execution semantics that include data flow, one might prefer to use the composite functor $\mathcal{C}\circ \mathfrak{E}$ to exclusively model execution since it is well known that control flow comprehensively captures computational behaviour.} In any case, a net's behaviour corresponds to the classical token game:
\begin{enumerate}
\item A transition ${u}$ is enabled at some state if and only if each input place of $u$ has at least one token. This rule is a consequence of having a unitary coefficient for each element of $in(u)$ and is applicable to a net under $\mathcal{N}$, $\mathcal{C}\circ \mathfrak{E}$ or $\mathcal{D}$.\label{firing-1}
\item If ${u}$ is enabled at some state, then it fires to reach a new state in which $u$ consumes exactly one token from each corresponding input place and produces exactly one token in each corresponding output place. Token consumption and token production are unitary because all the elements of $in(u)$ and $out(u)$ have unitary coefficients. This rule is applicable to a net under $\mathcal{N}$, $\mathcal{C}\circ \mathfrak{E}$ or $\mathcal{D}$.\label{firing-2}
\end{enumerate} 

In the context of Petri nets, a state is just a distribution of tokens over places given by a marking function, as formalised in Definition \ref{def:marking}.

\begin{definition}[State of a Net]\label{def:marking}
The state of a Petri net ${(S,T,in,out)}$ at some point in time is a marking function ${M\colon S\rightarrow\mathbb{N}}$ which assigns zero or more tokens to each place. We say $M$ is initial if ${M(p)>0}$ for every input place ${p\in S}$ and ${M(q)=0}$ for every non-input place ${q\in S}$. Dually, $M$ is final when ${M(r)>0}$ for every output place ${r\in S}$ and ${M(s)=0}$ for every non-output place ${s\in S}$. We write ${M_i}$ and ${M_f}$ for the initial and final states, respectively.
\end{definition}

As the functors $\mathcal{N}$, $\mathcal{C}\circ \mathfrak{E}$ and $\mathcal{D}$ map to Petri nets, Definition \ref{def:marking} can be used to determine the state of computon nets at some point in time. Passing from one marking to another constitutes an evolution of states in which control and/or data flow occurs implicitly. This operational aspect is a consequence of deriving a net from the category of computons which explicitly indicate which ports store control and which ones buffer data. Particularly, for nets under $\mathcal{N}$, if a computation unit $u$ has $m$ control ports and $n$ data ports connected to it, the corresponding transition will have $m$ input places to buffer incoming control and $n$ input places to store incoming data. Consequently, by Rule \ref{firing-1}, $u$ will fire with $m$ ``control tokens'' and $n$ ``data tokens''. Assuming $u$ has $j$ control ports and $k$ data ports connected from it, Rule \ref{firing-2} says $u$ must produce $j$ ``control tokens'' and $k$ ``data tokens'' upon firing. We use quotation marks to indicate token colouring is implicitly defined in a computon's net. \footnote{Computon nets do not explicitly colour places or tokens to distinguish between control and data. However, if token/port colouring is explicitly needed, one can define functors from $\textbf{Set}^\textbf{Comp}$ to the category of coloured Petri nets.}

Apart from controlling multiplicity of places and token colouring, separating composition from operation allows us to dictate crucial execution aspects by construction. For instance, connected computons provide their underlying nets with  enhanced coverability for better reachability analysis. Evidently, by Propositions \ref{prop:pushout-connected} and \ref{prop:computon-coproduct-connected}, this advantage extends to pushouts and coproducts, respectively. 

As there is a path from every input place to some output place in the net of a connected computon (under $\mathcal{N}$ or $\mathcal{C}\circ\mathfrak{E}$), a state marking all the input places of such a net has an option to reach a state marking all the corresponding output places (in which no further transitions can fire). Although this provides connected computons with a weak potential to complete, there is no guarantee all execution paths will lead to successful completion. For stronger termination guarantees, we need to verify additional properties such as deadlock-freedom. In Section \ref{sec:composite-computons}, we demonstrate how this additional operational aspect can be statically enforced by our composition operators, i.e., from the (structural) composition dimension and without any domain knowledge. Our precise notion of deadlock-freedom is formalised in Definition \ref{def:deadlock-freedom}.

\begin{definition}[Deadlock-Freedom]\label{def:deadlock-freedom}
A Petri net is deadlock-free if, for every marking state $M$ reachable from ${M_i}$ with ${M\neq M_f}$, there exists some enabled transition at $M$.
\end{definition}

\begin{remark}
Recall that the reachability relation on Petri net markings satisfies the reflexive property. So, every marking is trivially reachable from itself by zero transition firings. 
\end{remark}

\begin{remark}
By Definition \ref{def:marking}, a final state $M_f$ does not enable any transitions. Consequently, if the condition ${M\neq M_f}$ were removed from Definition \ref{def:deadlock-freedom}, the net corresponding to a terminating computon would necessarily exhibit deadlocks. Although $M_f$ would constitute a deadlock state under a classical notion of deadlock-freedom, Definition \ref{def:deadlock-freedom} does not classify it as such since our goal is to prevent unwanted blocking during execution only. Similar notions of deadlock-freedom have been adopted in other works \cite{van_der_aalst_application_1998,fossati_multiparty_2014}.
\end{remark}

\begin{remark}\label{rem:deadlock-freedom}
In Section \ref{sec:composite-computons}, we are interested in verifying deadlock-freedom only for nets that have an initial and a final state in the sense of Definition \ref{def:deadlock-freedom}. As Definition \ref{def:computon} enforces computons to have ec-inports and ec-outports, nets under $\mathcal{N}$ or $\mathcal{C}\circ \mathfrak{E}$ satisfy these two conditions. For $\mathcal{D}$-nets, we only check deadlock-freedom when such conditions are satisfied. This is because some of them do not necessarily have input and output places (to buffer data) by the fact Definition \ref{def:computon} does not require computons to have ed-inports or ed-outports.
\end{remark}

Enforcing deadlock-freedom at the level of composition semantics contributes towards satisfying termination-by-construction and, hence, to building large-scale complex systems with predictable behaviour and enhanced reliability. Having a weak option to complete together with deadlock-freedom constitute necessary conditions for termination in some cases, but they are not sufficient in general. Hence, it is important to reasoning about other crucial operational properties such as livelocks and boundedness. Given the compositional nature of the proposed model, such additional assurances can be verified compositionally using existing Petri net tools. 

\vspace{0.3cm}

Although this paper is primarily focused on the static (structural) dimension given by composition, this section briefly discussed how operational aspects (i.e., dynamics) can be dictated by construction. As our main focus is composition, we did not discuss algebraic operations on Petri nets (e.g., refinement). Petri nets just serve as a medium to express computon behaviour. Our overarching tenet is that operations for manipulating behaviour must be realised at the composition-level and then propagated to the operational-level through the functors we propose in this section. In Section \ref{sec:composite-computons}, we further discuss how execution can be dictated by construction. In the future, we plan to study additional advantages derived from the separation of composition from operation such as using formalisms beyond Petri nets to describe and analyse computon behaviour. Using functors on $\textbf{Petri}$ to endow computons with execution semantics entails that all the theory of categorical Petri nets \cite{baez_categories_2021} is applicable to our work.

In the future, we intend to explore computon processes in more detail, e.g., by studying computon behaviour within symmetric monoidal categories \cite{master_petri_2020}.\footnote{The exploitation of symmetric monoidal categories to study concurrent behaviours in Petri nets started long before \cite{master_petri_2020}. Early works include \cite{meseguer_petri_1990}, \cite{degano_axiomatizing_1989} and \cite{sassone_category_1995}.} For now, we would just like to highlight that in the \emph{token game for computons}, due to the inherent concurrent nature of Petri nets, ``control tokens'' can arrive before ``data tokens'' (or viceversa). By implicitly having ``control places'' separated from ``data places'' in a net under $\mathcal{N}$, a transition only fires when both ``control'' and ``data'' tokens are placed in their respective input places. This means a transition is a passive construct with blocking behaviour which implicitly synchronises data and control before firing. After firing, it produces exactly one token in each corresponding output place. 
\section{Trivial Computons} \label{sec:trivial-computons}

A \emph{trivial computon} has no computation units, no edges and no i-ports at all, but just a number of e-inoutports (see Definition \ref{def:computon-trivial}). Up to isomorphism, it is the only object in $\textbf{Set}^\textbf{Comp}$ with no computation units (see Proposition \ref{prop:computon-trivial-computationless}); consequently, it is not connected in the sense of Definition \ref{def:computon-connected} (see Proposition \ref{prop:computon-trivial-connectivity}).

\begin{definition} [Trivial Computon] \label{def:computon-trivial}
A trivial computon $\lambda$ is a computon whose diagram in \textbf{Set} has the form:
\[
\begin{tikzcd}
 & \emptyset \arrow[dl, twoheadrightarrow, "\sigma"'] \arrow[dr, "t"] & &  \\
\emptyset & & P \arrow[r, twoheadrightarrow, "c"] & \Sigma \\
 & \emptyset \arrow[ul, twoheadrightarrow, "\tau"]\arrow[ur, "s"'] & &
\end{tikzcd}
\]
\end{definition}

\begin{proposition} \label{prop:computon-trivial-computationless}
A computon has no computation units if and only if it is a trivial computon.
\end{proposition}
\begin{proof}
$(\implies)$ Let $\lambda$ be a computon with $U=\emptyset$. By the definition of empty function, we have $U=\emptyset$ only if $I = \emptyset = O$ so that $\sigma$ and $\tau$ are surjective. Definition \ref{def:computon} states that any computon is required to have at least one coloured port so $|P| \geq 1$ and $|\Sigma| \geq 1$. As $s$ and $t$ are not surjective by the definition of empty function, we have $p \in P^+ \cap P^-$ for all $p \in P$ (see Definition \ref{def:computon-interface}). Particularly, if $c(p)=0$, then $p \in C^+ \cap C^-$; otherwise, $p \in D^+ \cap D^-$. As the function $c$ is surjective by Definition \ref{def:computon}, we conclude that $\lambda$ is a trivial computon.

$(\impliedby)$ This follows directly from Definition \ref{def:computon-trivial}.
\end{proof}

\begin{proposition}\label{prop:computon-trivial-connectivity}
A trivial computon is not connected.
\end{proposition}
\begin{proof}
If $\lambda$ is a trivial computon, then $U=\emptyset$. Therefore, $\lambda$ is not connected as per Proposition \ref{prop:computon-connected-alwaysunits}.
\end{proof}

The general structure of a trivial computon with $j$ ec-inoutports and $k$ ed-inoutports is depicted in Figure \ref{fig:computon-trivial}, together with its corresponding Petri nets. Definition \ref{def:computon-unit} states that, in our theory, there is a distinguished computon of this sort consisting of a single ec-inoutport, which we refer to as the \emph{unit computon}.

\begin{figure}[!h]
\centering
\subcaptionbox{Trivial computon $\lambda$}
{
\begin{tikzpicture}
\node(del) at (0,0){};
\qmatch{q1}{1.2}{2.6}{$c(p_1)$}
\node[label={$\vdots$}] (dots1) at (1.2,1.75) {};
\qmatch{qi}{1.2}{1.7}{$c(p_j)$}
\dmatch{d1}{1.2}{1.1}{$c(d_1)$}{left};
\node[label={$\vdots$}] (dots3) at (1.2,0.21) {};
\dmatch{dj}{1.2}{0.2}{$c(d_k)$}{left};
\node(del2) at (2.4,0){};
\end{tikzpicture}
}
\hspace{0.1cm}
\subcaptionbox{$\mathcal{N}(\lambda)$}
{
\begin{tikzpicture}
\node(del) at (0,0){};
\node[place,label={[xshift=-0.4cm,yshift=-0.4cm]:\scriptsize $p_1$},minimum size=3mm] (q0) at (1.2,2.4) {};
\node[label={$\vdots$}] (dots1) at (1.2,1.55) {};
\node[place,label={[xshift=-0.4cm,yshift=-0.4cm]:\scriptsize $p_j$},minimum size=3mm] (qi) at (1.2,1.5) {};
\node[place,label={[xshift=-0.4cm,yshift=-0.4cm]:\scriptsize $d_1$},minimum size=3mm] (d1) at (1.2,0.9) {};
\node[label={$\vdots$}] (dots3) at (1.2,0.01) {};
\node[place,label={[xshift=-0.4cm,yshift=-0.4cm]:\scriptsize $d_k$},minimum size=3mm] (dj) at (1.2,0) {};
\node(del2) at (2.4,0){};
\end{tikzpicture}
}
\hspace{0.1cm}
\subcaptionbox{$\mathcal{C}(\mathfrak{E}(\lambda))$}
{
\begin{tikzpicture}
\node(del) at (0,0){};
\node[place,label={[xshift=-0.4cm,yshift=-0.4cm]:\scriptsize $p_1$},minimum size=3mm] (q0) at (1.2,2.4) {};
\node[label={$\vdots$}] (dots1) at (1.2,1.55) {};
\node[place,label={[xshift=-0.4cm,yshift=-0.4cm]:\scriptsize $p_j$},minimum size=3mm] (qi) at (1.2,1.5) {};
\node(del2) at (2.4,0){};
\end{tikzpicture}
}
\hspace{0.1cm}
\subcaptionbox{$\mathcal{D}(\lambda)$}
{
\begin{tikzpicture}
\node(del) at (0,0){};
\node[place,label={[xshift=-0.4cm,yshift=-0.4cm]:\scriptsize $d_1$},minimum size=3mm] (d1) at (1.2,0.9) {};
\node[label={$\vdots$}] (dots3) at (1.2,0.01) {};
\node[place,label={[xshift=-0.4cm,yshift=-0.4cm]:\scriptsize $d_k$},minimum size=3mm] (dj) at (1.2,0) {};
\node(del2) at (2.4,0){};
\end{tikzpicture}
}
\hspace{0.1cm}
{
\begin{tikzpicture}
\node(del) at (0,0){};
\matrix [below, ampersand replacement=\&] at (0,2.5) {
 \node[draw,fill=white,fill fraction={black}{0.5},inner sep=1.2pt,minimum size=3pt,label=right:{\scriptsize Ec-inoutport}] {}; \\
 \node[circle,draw,fill=white,fill fraction={black}{0.5},inner sep=1.2pt,minimum size=3pt,label=right:{\scriptsize Ed-inoutport}] {}; \\
};
\end{tikzpicture}
}
\caption{A trivial computon $\lambda$ and its corresponding Petri nets. By Definition \ref{def:computon-trivial}, $\lambda$ does not have any computation units but just e-inoutports; consequently, $\mathcal{N}(\lambda)$, $\mathcal{C}(\mathfrak{E}(\lambda))$ and $\mathcal{D}(\lambda)$ have places only so they are behaviourless. The net $\mathcal{N}(\lambda)$ embodies the whole structure of $\lambda$, whereas $\mathcal{C}(\mathfrak{E}(\lambda))$ and $\mathcal{D}(\lambda)$ have places for control and data only, respectively. Labels on places are just for reference purposes since we deal with non-labelled Petri nets, as discussed in Section \ref{sec:operational-semantics}.}
\label{fig:computon-trivial}
\end{figure}

\begin{remark}
By Definition \ref{def:computon}, a trivial computon $\lambda$ always includes control ports (i.e., ${j\geq 1}$) and can optionally possess data ports (i.e., ${k\geq 0}$). Therefore, ${\mathcal{N}(\lambda)}$ and ${\mathcal{C}(\mathfrak{E}(\lambda))}$ always have at least one place, whereas ${\mathcal{D}(\lambda)}$ can have no places at all.
\end{remark}

\begin{definition} [Unit Computon] \label{def:computon-unit}
The \emph{unit computon} is a trivial computon with ${|P|=|\Sigma|=1}$. We use $\Lambda$ to denote it.
\end{definition}

The existence of $\Lambda$ can be proven by the fact that, according to Definition \ref{def:computon}, the set of ports and the set of colours are never empty. By the same definition, we can observe that a computon can have no computation units and no edges at all. Proposition \ref{prop:unit-unique} uses this observation to show that $\Lambda$ is unique up to unique isomorphism. 

\begin{proposition}\label{prop:unit-unique}
$\Lambda$ is unique up to unique isomorphism.
\end{proposition}
\begin{proof}
By letting $\lambda_1$ and $\lambda_2$ be two unit computons, we construct a computon morphism ${\alpha:\lambda_1 \rightarrow \lambda_2}$. As ${U_1=U_2=O_1=O_2=I_1=I_2=\emptyset}$, the only choice we have for $\alpha_U$, $\alpha_O$ and $\alpha_I$ is the empty function. For $\alpha_P$, Definition \ref{def:computon-unit} forces us to exclusively consider the unique function given by ${\alpha_P(p_1)=p_2}$ for the unique ports ${p_1 \in P_1}$ and ${p_2 \in P_2}$. By this mapping, $p_1$ can never be in ${\vec{i}(\alpha)\cup\vec{o}(\alpha)}$ because ${\bullet\alpha(p_1)=\bullet p_2=\alpha(p_1)\bullet=p_2\bullet=\emptyset}$ (as per Definition \ref{def:computon-trivial}) so ${\vec{i}(\alpha)\cup\vec{o}(\alpha)=\emptyset \subseteq P_1^+ \cup P_1^-}$. The inclusion ${\alpha_\Sigma(0)=0}$ is the only option we have for ${\alpha_\Sigma}$, given that ${\Sigma_1=\{0\}=\Sigma_2}$ (see Definitions \ref{def:computon} and \ref{def:computon-unit}). As this construction satisfies Definition \ref{def:computon-morphism}, $\alpha$ is indeed a computon morphism. More specifically, $\alpha$ is a computon isomorphism because its inverse is necessarily constructed in the same way, and it is unique because the components of $\alpha$ are unique functions in $\textbf{Set}$ (i.e., empty functions or trivial injections). 

Given that most of the $\alpha$-components are empty functions, the only equation that needs to be verified from the commutative squares of Definition \ref{def:computon-morphism} is ${\alpha_\Sigma\circ c_1=c_2\circ\alpha_P}$. Having ${\alpha_\Sigma(c_1(p_1))=\alpha_\Sigma(0)=0=c_2(p_2)=c_2(\alpha_P(p_1))}$, we conclude that the naturality condition of the unique isomorphism $\alpha$ holds. Thus, proving that the unit computon is unique up to unique isomorphism.
\end{proof}

Since a computon is required to have at least one coloured port, $\Lambda$ can be perceived as the ``simplest'' object in $\textbf{Set}^\textbf{Comp}$, which corresponds to \tikz{\node at (0,0){};\node[draw,fill=white,fill fraction={black}{0.5},inner sep=1.2pt,minimum size=3pt] (q) at (0,0.01) {};} in graphical notation. In spite of this structural feature, $\Lambda$ is not an initial object in such a category since there are $k \geq 1$ computon morphisms from it to any other computon $\lambda$, where $k$ is the number of control ports in $\lambda$.\footnote{This can be easily proved by induction on the number of control ports of an arbitrary computon.} Under this premise, as mentioned in Section \ref{sec:computons}, $\textbf{Set}^\textbf{Comp}$ has no initial objects.

In $\textbf{Set}^\textbf{Comp}$, there are distinguished morphisms that respectively embed a trivial computon into all the e-inports or all the e-outports of some computon. These morphisms are referred to as in- and out-markers, respectively. The intuition behind these notions is captured in Definition \ref{def:computon-morphism-markers}. 

\begin{definition}[In- and Out-markers]\label{def:computon-morphism-markers}
An \emph{in-marker} of a computon $\lambda$ is a computon monomorphism $\alpha: \lambda_0 \rightarrow \lambda$ where $\lambda_0$ is a trivial computon with $\alpha(P_0)=P^+$. If $\alpha(P_0)=P^-$, then $\alpha$ is an \emph{out-marker} of $\lambda$. 
\end{definition}

\begin{notation}
For convenience, we write $\lambda^+$ and $\lambda^-$ for the respective in- and out-markers of a computon $\lambda$. We use the word ``the'' because all the in-markers of a computon are isomorphic to each other, with the same being true for the corresponding out-markers. When the context is clear, we simply use the word ``markers'' to unifiedly refer to such morphisms which by, Proposition \ref{prop:markers-always}, always exist for any computon.
\end{notation}

\begin{proposition}\label{prop:markers-always}
Every computon has markers.
\end{proposition}
\begin{proof}
The proof follows directly from the fact that every computon has at least one e-inport and at least one e-outport (see Definition \ref{def:computon}).
\end{proof}

As the image of a marker covers all the e-inports or all the e-outports of a computon, it is easy to show that every span of markers is pushable. For convenience, we capture this truth in Proposition \ref{prop:markers-pushable}. Also, as only e-ports are identified in the pushout of a span of markers, it is true that the induced morphisms of the pushout of in- or out-markers preserves e-inports or e-outports, correspondingly (see Proposition \ref{prop:markers-inout}).

\begin{proposition}\label{prop:markers-pushable}
Every span of markers is pushable.
\end{proposition}
\begin{proof}
The proof follows directly from Definitions \ref{def:computon-morphisms-pushable} and \ref{def:computon-morphism-markers}.
\end{proof}

\begin{proposition}\label{prop:markers-inout}
Let ${\square \in \{+,-\}}$ and ${j=1,2}$. If ${(\beta_1:\lambda_1 \rightarrow \lambda_3, \lambda_3, \beta_2:\lambda_2 \rightarrow \lambda_3)}$ is the pushout of a span ${\lambda_1 \xleftarrow{\lambda_1^\square} \lambda_0 \xrightarrow{\lambda_2^\square} \lambda_2}$, then ${\beta_j(\lambda_j^{\square}(P_0))=P_3^{\square}}$.
\end{proposition}
\begin{proof}
Suppose ${(\beta_1:\lambda_1 \rightarrow \lambda_3, \lambda_3, \beta_2:\lambda_2 \rightarrow \lambda_3)}$ is the pushout of a span ${\lambda_1 \xleftarrow{\lambda_1^+} \lambda_0 \xrightarrow{\lambda_2^+} \lambda_2}$ of in-marker morphisms, and assume for contradiction there is some ${p_3 \in \beta_j(\lambda_j^+(P_0)) \triangle P_3^+}$ for ${j \in \{1,2\}}$. If ${p_3 \in \beta_j(\lambda_j^+(P_0)) \setminus P_3^+}$, then there is some ${u_3 \in \bullet p_3}$ and, by pushout commutativity, some ${u_j \in \bullet p_j}$ for ${p_j \in \lambda_j^+(P_0)}$ such that ${\beta_j(u_j)= u_3}$ and ${\beta_j(p_j)= p_3}$. As Definition \ref{def:computon-morphism-markers} says ${\lambda_j^+(P_0)=P_j^+}$, we have ${p_j \in \lambda_j^+(P_0) \iff p_j \in P_j^+}$, i.e., a contradiction to the assumption ${u_j \in \bullet p_j}$. The case ${p_3 \in P_3^+ \setminus \beta_j(\lambda_j^+(P_0))}$ never holds since this would violate the commutativity property of pushout constructions.

Proving the statement ${\gamma_j(\lambda_j^-(P_4))=P_5^-}$ for the pushout ${(\gamma_1:\lambda_1 \rightarrow \lambda_5, \lambda_5, \gamma_2:\lambda_2 \rightarrow \lambda_5)}$ of a span ${\lambda_1 \xleftarrow{\lambda_1^-} \lambda_4 \xrightarrow{\lambda_2^-} \lambda_2}$ is completely analogous. Therefore, we conclude that our initial proposition is true.
\end{proof}

The existence of marker morphisms gives rise to the notion of dual computons. Informally, the dual of a computon $\lambda$ is constructed by structurally swapping e-inports with e-outports and vice versa, so the domains of $\lambda^+$ and $\lambda^-$ precisely correspond to the domains of the out- and in-markers of its dual, respectively. This is formalised in Definition \ref{def:computon-dual}.

\begin{definition}[Dual Computons]\label{def:computon-dual}
If there are spans ${\lambda_2 \xleftarrow{\lambda_2^+} \lambda_0 \xrightarrow{\lambda_3^-} \lambda_3}$ and ${\lambda_2 \xleftarrow{\lambda_2^-} \lambda_1 \xrightarrow{\lambda_3^+} \lambda_3}$, we say that $\lambda_2$ is a dual of a computon $\lambda_3$ and vice versa.
\end{definition}

A computon $\lambda$ can have multiple duals (with different structure each) because Definition \ref{def:computon-dual} only requires a dual to swap the interface of $\lambda$, without imposing any constraints on internal structure. More precisely, duals are not unique up to isomorphism, so a dual of a dual is not necessarily isomorphic to the original computon. Although uniqueness is not satisfied in general, Proposition \ref{prop:markers-connected-dual} states that for a connected computon, it is possible to have a dual with inverted information flows. Therefore, when connectivity is guaranteed, the construction given in the corresponding proof can be applied twice to obtain a connected computon that is isomorphic to the original one.

\begin{proposition}\label{prop:markers-connected-dual}
If $\lambda$ is a connected computon, there exists a connected computon which is a dual of $\lambda$.
\end{proposition}
\begin{proof}
Assuming $\lambda_2$ is a connected computon, we construct a computon $\lambda_3$ by letting ${P_2=P_3}$, ${U_2=U_3}$, ${I_2=O_3}$, ${O_2=I_3}$ and ${c_2=c_3}$. For each ${o \in O_3}$, we set ${\sigma_3(o)=\tau_2(o)}$ and ${t_3(o)=s_2(o)}$ to yield ${Im(\sigma_3)=Im(\tau_2)}$ and ${Im(t_3)=Im(s_2)}$. Similarly, for each ${i \in I_3}$, we let ${s_3(i)=t_2(i)}$ and ${\tau_3(i)=\sigma_2(i)}$ to have ${Im(s_3)=Im(t_2)}$ and ${Im(\tau_3)=Im(\sigma_2)}$. That is, $\lambda_3$ has the same coloured ports and computation units as $\lambda_2$, but with inverted information flows. So, ${p \xrightarrow{\exists} q}$ for $\lambda_2$ iff ${q \xrightarrow{\exists} p}$ for $\lambda_3$. Moreover, ${P_2^+=P_3^-}$ and ${P_2^-=P_3^+}$ because:
\[
p_2 \in P_2^+ \iff p_2 \notin Im(t_2) \iff p_2 \notin Im(s_3) \iff p_2 \in P_3^-
\]
\[
p_2 \in P_2^- \iff p_2 \notin Im(s_2) \iff p_2 \notin Im(t_3) \iff p_2 \in P_3^+
\]
By the above, ${Im(s_3)\cup P_3^+=Im(t_2)\cup P_2^-}$ is evident. As $\lambda_2$ is connected, we simply use Proposition \ref{prop:computon-connected-reverse} to deduce that, for each ${q_3\in Im(t_2)\cup P_2^-}$, there is a port ${p_3 \in P_2^+}$ where ${p_3\xrightarrow{\exists}q_3}$. Equivalently, for each ${q_3\in Im(s_3)\cup P_3^+}$, there is a port ${p_3 \in P_3^-}$ where ${q_3\xrightarrow{\exists}p_3}$. That is, by Definition \ref{def:computon-connected}, $\lambda_3$ is also connected.

To finalise our proof, we construct a trivial computon $\lambda_0$ and a trivial computon $\lambda_1$ by letting ${P_0=P_2^+}$, ${c_0=c_2\restriction P_2^+}$, ${P_1=P_2^-}$ and ${c_1=c_2\restriction P_2^-}$. As ${P_0=P_2^+=P_3^-}$, there evidently is an in-marker ${\lambda_0 \rightarrow \lambda_2}$ and an out-marker ${\lambda_0 \rightarrow \lambda_3}$. Analogously, having $P_1=P_2^-=P_3^+$ implies the existence of an out-marker ${\lambda_1 \rightarrow \lambda_2}$ and an in-marker ${\lambda_1 \rightarrow \lambda_3}$. By Definition \ref{def:computon-dual}, we conclude that $\lambda_3$ is a connected computon and a dual of $\lambda_2$, as required.
\end{proof}
\section{Primitive Computons} \label{sec:primitive-computons}

Like a trivial computon, a primitive one has no i-ports. The difference is that there is exactly one computation unit to which all ports are attached via edges (see Definition \ref{def:computon-primitive}). So, every port can be either e-inport or e-outport, never both (see Proposition \ref{prop:computon-primitive-insouts}). This implies a primitive computon is connected, i.e., it has neither dangling ports nor dangling computation units (see Proposition \ref{prop:computon-primitive-connected}). For these reasons, the underlying net of any primitive computon is deadlock-free (see Proposition \ref{prop:computon-primitive-deadlock} and Remark \ref{rem:computon-primitive-deadlock}). In this paper, we consider three classes of primitive computons, namely \emph{fork computons}, \emph{join computons} and \emph{functional computons}.

\begin{definition} \label{def:computon-primitive}
A primitive computon is a computon whose diagram in \textbf{Set} has the form:\footnote{We use $1$ and $\rightarrowtail$ to denote a singleton set and an injective function, respectively. The symbol $\triangle$ is the operator for symmetric difference given by ${A \triangle B=(A \setminus B)\cup(B \setminus A)}$ for sets $A$ and $B$.}
\[
\begin{tikzcd}
 & O \arrow[dl, twoheadrightarrow, "\sigma"'] \arrow[dr, tail, "t"] & & & \\
1 & & P \arrow[r, twoheadrightarrow, "c"] & \Sigma & \text{with } P = Im(s) \triangle Im(t) \\
 & I \arrow[ul, twoheadrightarrow, "\tau"]\arrow[ur, tail, "s"'] & & &
\end{tikzcd}
\]
\end{definition}

\begin{remark}
A glance at Definition \ref{def:computon-primitive} reveals that a primitive computon has no i-ports because ${Im(s)\cap Im(t)=\emptyset}$ follows from ${P = Im(s) \triangle Im(t)}$. Consequently, ${Im(s)\setminus Im(t)=Im(s)}$ and ${Im(t)\setminus Im(s)=Im(t)}$ also hold.
\end{remark}

\begin{proposition} \label{prop:computon-primitive-insouts}
If $\lambda$ is a primitive computon, then ${P=P^+ \triangle P^-}$. 
\end{proposition}
\begin{proof}
If $\lambda$ is a primitive computon, then ${p \in P \iff p \in Im(s) \triangle Im(t)}$ (see Definition \ref{def:computon-primitive}). 
\begin{itemize}
\item If $p \in Im(s) \setminus Im(t)$, then $p \in P^+ \setminus P^-$ because ${(\exists i \in I)[s(i)=p]}$ and ${(\nexists o \in O)[t(o)=p]}$.
\item If $p \in Im(t) \setminus Im(s)$, then $p \in P^- \setminus P^+$ because ${(\exists o \in O)[t(o)=p]}$ and ${(\nexists i \in I)[s(i)=p]}$.
\end{itemize}
By the above cases and by the definition of symmetric difference, we have that ${p \in P} \iff {p \in Im(s) \triangle Im(t)} \iff {p \in (P^+ \setminus P^-) \cup (P^- \setminus P^+)} \iff p \in P^+ \triangle P^-$. Hence, we conclude ${P = Im(s) \triangle Im(t) = (P^+ \setminus P^-) \cup (P^- \setminus P^+) = P^+ \triangle P^-}$, as required.
\end{proof}

\begin{proposition} \label{prop:computon-primitive-connected}
Every primitive computon is a connected computon.
\end{proposition}
\begin{proof}
Assuming $\lambda$ is a primitive computon with ${p \in Im(s)\cup P^+}$, we first show ${Im(s)=P^+}$ as follows: ${p \in Im(s)}\iff{p \notin P^-}$ by Definition \ref{def:computon-interface} ${\iff p \in P^+}$ by Proposition \ref{prop:computon-primitive-insouts}. Having ${Im(s)=P^+}$ implies we just have to prove for either ${p \in Im(s)}$ or ${p \in P^+}$.
 
If ${p \in Im(s)}$, we know there must be some ${i \in I}$ for which ${s(i)=p}$. If $u$ is the only computation unit in $U$ (see Definition \ref{def:computon-primitive}), ${\tau(i)=u}$ because $\tau$ is total and surjective (see Definition \ref{def:computon}). As $\sigma$ is also surjective, there exists some ${o \in O}$ where ${\sigma(o)=u}$. By the totality of $t$, there is some ${q \in P}$ where ${t(o)=q}$. Having ${q \in Im(t)}$ and ${P=Im(s)\triangle Im(t)}$ allows us to deduce ${q \notin Im(s)}$ so that ${q \in P^-}$ by Definition \ref{def:computon-interface}. As ${p \xrightarrow{\exists} q}$ holds for ${p \in Im(s) \cup P^+}$ and ${q \in P^-}$, we conclude that $\lambda$ adheres to Definition \ref{def:computon-connected} and it is therefore a connected computon.
\end{proof}

\begin{proposition} \label{prop:computon-primitive-deadlock}
If $\lambda$ is a primitive computon, the nets $\mathcal{N}(\lambda)$ and $\mathcal{C}(\mathfrak{E}(\lambda))$ are deadlock-free.
\end{proposition}
\begin{proof}
The proof is obvious since every primitive computon's net has only one transition to which all input places (i.e., e-inports) and output places (i.e., e-outports) are attached (see Definition \ref{def:computon-primitive} and Proposition \ref{prop:computon-primitive-insouts}).
\end{proof}

\begin{remark}\label{rem:computon-primitive-deadlock}
The only scenario in which $\mathcal{D}(\lambda)$ satisfies Remark \ref{rem:deadlock-freedom} (i.e., initial and final states exist) is when $\lambda$ has both ed-inports and ed-outports. When this occurs, the proof that $\mathcal{D}(\lambda)$ is deadlock-free is analogous to that of Proposition \ref{prop:computon-primitive-deadlock}.
\end{remark}

\subsection{Fork Computons} \label{sec:primitive-computons-fork}

A fork computon has exactly one ec-inport and two ec-outports, i.e., it has no data ports at all (see Definition \ref{def:computon-fork}). Intuitively, it just duplicates the control received in its unique ec-inport into all its ec-outports.

\begin{definition} \label{def:computon-fork}
A fork computon $\lambda$ is a primitive computon with $|O|=2$ and $|I|=|\Sigma|=1$.
\end{definition}

From Definition \ref{def:computon-fork}, we can deduce a fork computon $\lambda$ has exactly three ports because $s$ and $t$ are injective, ${P=Im(s) \triangle Im(t)}$, ${|O|=2}$ and ${|I|=1}$. Specifically, ${|P^-|=2}$ and ${|P^+|=1}$ because ${P=P^+ \triangle P^-}$ (see Definition \ref{def:computon-interface} and Proposition \ref{prop:computon-primitive-insouts}). As ${|\Sigma|=1}$ and $c$ is total and ${C^+ \neq \emptyset \neq C^-}$ (see Definition \ref{def:computon}), $\Sigma=\{0\}$ so ${(\forall p \in P)[c(p)=0]}$. The general structure of a fork computon, together with its underlying Petri nets, is depicted in Figure \ref{fig:computon-fork}.

\vspace{5pt}

\begin{figure}[!h]
\centering
\subcaptionbox{Fork computon $\lambda$}
{
\begin{tikzpicture}
\forkplain{0}{0};
\flow{{0,0.25}}{{1,0.25}}{dashed}{};
\flow{{1.25,0.4}}{{1.25+0.95,0.4}}{dashed}{};
\flow{{1.25,0.1}}{{1.25+0.95,0.1}}{dashed}{};
\node[draw=black,fill=white,inner sep=0pt,minimum size=3pt,label=left:\scriptsize $c(p_1)$] (p) at (0,0.25) {};
\node[draw=black,fill=black,inner sep=0pt,minimum size=3pt,label=right:\scriptsize $c(p_2)$] (q) at (2.25,0.4) {};
\node[draw=black,fill=black,inner sep=0pt,minimum size=3pt,label=right:\scriptsize $c(p_3)$] (r) at (2.25,0.1) {};
\node(del) at (2.5,0){};
\end{tikzpicture}
}
\hspace{2cm}
\subcaptionbox{$\mathcal{N}(\lambda)$ or $\mathcal{C}(\mathfrak{E}(\lambda))$}
{
\begin{tikzpicture}
\node[transition,fill=black,minimum width=0.5mm,minimum height=10mm] (t1) at (1,0) {};

\node[place,label={180:\scriptsize $p_1$},minimum size=3mm] (q0) at (0,0) {};
\node[place,label={0:\scriptsize $p_2$},minimum size=3mm] (q1) at (2,0.3) {};
\node[place,label={0:\scriptsize $p_3$},minimum size=3mm] (q2) at (2,-0.3) {};

\draw[-latex,thick] (q0) -- (t1);
\draw[-latex,thick] (t1) -- (q1);
\draw[-latex,thick] (t1) -- (q2);
\end{tikzpicture}
}
\hspace{2cm}
\subcaptionbox{$\mathcal{D}(\lambda)$}
{
\begin{tikzpicture}
\node at (0,0){};
\node[transition,fill=black,minimum width=0.5mm,minimum height=10mm] (t1) at (0.5,0) {};
\node at (1,0){};
\end{tikzpicture}
}
\hspace{1.5cm}
{
\begin{tikzpicture}
\matrix [below, ampersand replacement=\&] at (current bounding box.south) {
\draw[fill=gray!30] (1,0.12) node[anchor=north]{}
  -- (1.12,0.25) node[anchor=north]{}
  -- (1.12,0) node[anchor=south]{}
  -- cycle; \& \node[inner sep=0pt,minimum size=0pt,label=right:{\scriptsize Computation unit}]{}; \& \node{}; \& \&   
 \draw[dashed] (2.3,0) to node[pos=0.5,yshift=4]{\arrowflow} (2.8,0); \& \node[inner sep=0pt,minimum size=3pt,label=right:{\scriptsize Control flow edge}]{}; \& \node{}; \& \&
 \node[draw=black,fill=white,inner sep=0pt,minimum size=3pt,label=right:{\scriptsize Ec-inport}] {}; \& \node{};
 \& \node[fill=black,inner sep=0pt,minimum size=3pt,label={right:{\scriptsize Ec-outport}}] {}; \& \node{};\\
};
\end{tikzpicture}
}
\caption{A fork computon $\lambda$ and its corresponding Petri nets. As $\lambda$ has control ports only, the Petri nets $\mathcal{N}(\lambda)$ and $\mathcal{C}(\mathfrak{E}(\lambda))$ are isomorphic. By the same reason, $\mathcal{D}(\lambda)$ has no places at all but just a single transition. Labels on places are just for reference purposes since we deal with non-labelled Petri nets, as discussed in Section \ref{sec:operational-semantics}.}
\label{fig:computon-fork}
\end{figure}

\subsection{Join Computons} \label{sec:primitive-computons-join}

A join computon is the dual of a fork computon since it has exactly two ec-inports and one ec-outport (see Definition \ref{def:computon-join-1}). Intuitively, it merges the control received in its ec-inports into its unique ec-outport. 

\begin{definition}\label{def:computon-join-1}
A join computon $\lambda$ is a primitive computon with $|O|=|\Sigma|=1$ and $|I|=2$. 
\end{definition}

The properties of a join computon $\lambda$ are almost identical to that of a fork computon so it is true that $c(p)=0$ for all $p \in P$. The only difference is in terms of the number of ec-inports and ec-outports. The general structure of a join computon, together with its underlying Petri nets, is depicted in Figure \ref{fig:computon-join}.

\begin{figure}[!h]
\centering
\subcaptionbox{Join computon $\lambda$}
{
\begin{tikzpicture}
\joinplain{0}{0};
\flow{{1.25,0.25}}{{1.25+0.95,0.25}}{dashed}{};
\flow{{0,0.4}}{{1,0.4}}{dashed}{};
\flow{{0,0.1}}{{1,0.1}}{dashed}{};
\node[draw=black,fill=black,inner sep=0pt,minimum size=3pt,label=right:\scriptsize $c(p_3)$] at (2.25,0.25) {};
\node[draw=black,fill=white,inner sep=0pt,minimum size=3pt,label=left:\scriptsize $c(p_1)$] at (0,0.4) {};
\node[draw=black,fill=white,inner sep=0pt,minimum size=3pt,label=left:\scriptsize $c(p_2)$] at (0,0.1) {};

\node(del) at (2.5,0){};
\end{tikzpicture}
}
\hspace{2cm}
\subcaptionbox{$\mathcal{N}(\lambda)$ or $\mathcal{C}(\mathfrak{E}(\lambda))$}
{
\begin{tikzpicture}
\node[transition,fill=black,minimum width=0.5mm,minimum height=10mm] (t1) at (1,0) {};

\node[place,label={180:\scriptsize $p_1$},minimum size=3mm] (q0) at (0,0.3) {};
\node[place,label={180:\scriptsize $p_2$},minimum size=3mm] (q1) at (0,-0.3) {};
\node[place,label={0:\scriptsize $p_3$},minimum size=3mm] (q2) at (2,0) {};

\draw[-latex,thick] (q0) -- (t1);
\draw[-latex,thick] (q1) -- (t1);
\draw[-latex,thick] (t1) -- (q2);
\end{tikzpicture}
}
\hspace{2cm}
\subcaptionbox{$\mathcal{D}(\lambda)$}
{
\begin{tikzpicture}
\node at (0,0){};
\node[transition,fill=black,minimum width=0.5mm,minimum height=10mm] (t1) at (0.5,0) {};
\node at (1,0){};
\end{tikzpicture}
}
\hspace{0.8cm}
{
\begin{tikzpicture}
\matrix [below, ampersand replacement=\&] at (current bounding box.south) {
\draw[fill=gray!30] (1.12,0.12) node[anchor=north]{}
  -- (1,0.25) node[anchor=north]{}
  -- (1,0) node[anchor=south]{}
  -- cycle; \& \node[inner sep=0pt,minimum size=0pt,label=right:{\scriptsize Computation unit}]{}; \& \node{}; \& \& 
 \draw[dashed] (2.3,0) to node[pos=0.5,yshift=4]{\arrowflow} (2.8,0); \& \node[inner sep=0pt,minimum size=3pt,label=right:{\scriptsize Control flow edge}]{}; \& \node{}; \& \&
 \node[draw=black,fill=white,inner sep=0pt,minimum size=3pt,label=right:{\scriptsize Ec-inport}] {}; \& \node{};
 \& \node[fill=black,inner sep=0pt,minimum size=3pt,label={right:{\scriptsize Ec-outport}}] {}; \& \node{};\\
};
\end{tikzpicture}
}
\caption{A join computon $\lambda$ and its corresponding Petri nets. As $\lambda$ has control ports only, the Petri nets $\mathcal{N}(\lambda)$ and $\mathcal{C}(\mathfrak{E}(\lambda))$ are isomorphic. By the same reason, $\mathcal{D}(\lambda)$ has no places at all but just a single transition. Labels on places are just for reference purposes since we deal with non-labelled Petri nets, as discussed in Section \ref{sec:operational-semantics}.}
\label{fig:computon-join}
\end{figure}

\subsection{Functional Computons} \label{sec:primitive-computons-functional}

A functional computon has exactly one ec-inport, one ec-outport, any number of ed-inports and any number of ed-outports, as formalised in Definition \ref{def:computon-functional}. Intuitively, the unique computation unit is a high-level representation of a (potentially halting) computation, triggered after receiving a control signal and a number of input data values. The successful termination of such a computation results in a single control signal and a number of output data values.

\begin{definition} \label{def:computon-functional}
A functional computon $\lambda$ is a primitive computon where ${(\exists !p \in P^+)[c(p)=0]}$ and ${(\exists !q \in P^-)[c(q)=0]}$.
\end{definition}

The general structure of a functional computon with $j$ ed-inports and $k$ ed-outports is illustrated in Figure \ref{fig:computon-functional}, together with its corresponding Petri nets.\footnote{We acknowledge that the definition of computons does not include labels for computation units. However, for increased clarity, we took the liberty of using the symbol of a functional computon for labelling its unique computation unit.}  

\vspace{5pt}

\begin{figure}[!h]
\centering
\subcaptionbox{Functional computon $\lambda$}
{
\begin{tikzpicture}
\computonPrimitive{1}{0}{1}{1.9}{$\lambda$}

\flow{{0.2,1.7}}{$(0.2,1.7)+(0.8,0)$}{dashed}{};
\node[draw=black,fill=white,inner sep=0pt,minimum size=3pt,label=left:\scriptsize $c(p_0)$] (1q0) at (0.2,1.7) {};
\din{1i1}{0.2}{1.1}{$c(p_1)$};
\node[label={$\vdots$}] (1i2) at (0.5,0.25) {};
\din{1in}{0.2}{0.2}{$c(p_j)$};

\flow{{2,1.7}}{$(2,1.7)+(0.8,0)$}{dashed}{};
\node[fill=black,inner sep=0pt,minimum size=3pt,label=right:\scriptsize $c(q_0)$] (1q1) at (2.8,1.7) {};
\dout{1o1}{2}{1.1}{$c(q_1)$};
\node[label={$\vdots$}] (1o2) at (2.5,0.25) {};
\dout{1om}{2}{0.2}{$c(q_k)$};
\end{tikzpicture}
}
\hspace{-0.4cm}
\subcaptionbox{$\mathcal{N}(\lambda)$}
{
\begin{tikzpicture}
\node[transition,fill=black,minimum width=0.5mm,minimum height=10mm] (t1) at (1,0.7) {};

\node[place,label={180:\scriptsize $p_0$},minimum size=3mm] (q0) at (0,1.7) {};
\node[place,label={180:\scriptsize $p_1$},minimum size=3mm] (i1) at (0,1.1) {};
\node[label={$\vdots$}] (i2) at (0,0.1) {};
\node[place,label={180:\scriptsize $p_j$},minimum size=3mm] (in) at (0,0) {};

\node[place,label={0:\scriptsize $q_0$},minimum size=3mm] (q1) at (2,1.7) {};
\node[place,label={0:\scriptsize $q_1$},minimum size=3mm] (o1) at (2,1.1) {};
\node[label={$\vdots$}] (o2) at (2,0.1) {};
\node[place,label={0:\scriptsize $q_k$},minimum size=3mm] (om) at (2,0) {};

\draw[-latex,thick] (q0) -- (t1);
\draw[-latex,thick] (i1) -- (t1);
\draw[-latex,thick] (in) -- (t1);
\draw[-latex,thick] (t1) -- (q1);
\draw[-latex,thick] (t1) -- (o1);
\draw[-latex,thick] (t1) -- (om);
\end{tikzpicture}
}
\hspace{-0.4cm}
\subcaptionbox{$\mathcal{C}(\mathfrak{E}(\lambda))$}
{
\begin{tikzpicture}
\node at (0,0){};
\node[transition,fill=black,minimum width=0.5mm,minimum height=10mm] (t1) at (1,0.7) {};
\node[place,label={180:\scriptsize $p_0$},minimum size=3mm] (q0) at (0,0.7) {};
\node[place,label={0:\scriptsize $q_0$},minimum size=3mm] (q1) at (2,0.7) {};

\draw[-latex,thick] (q0) -- (t1);
\draw[-latex,thick] (t1) -- (q1);
\end{tikzpicture}
}
\hspace{-0.4cm}
\subcaptionbox{$\mathcal{D}(\lambda)$}
{
\begin{tikzpicture}
\node at (0,0){};
\node[transition,fill=black,minimum width=0.5mm,minimum height=10mm] (t1) at (1,0.5) {};

\node[place,label={180:\scriptsize $p_1$},minimum size=3mm] (i1) at (0,1.1) {};
\node[label={$\vdots$}] (i2) at (0,0.1) {};
\node[place,label={180:\scriptsize $p_j$},minimum size=3mm] (in) at (0,0) {};

\node[place,label={0:\scriptsize $q_1$},minimum size=3mm] (o1) at (2,1.1) {};
\node[label={$\vdots$}] (o2) at (2,0.1) {};
\node[place,label={0:\scriptsize $q_k$},minimum size=3mm] (om) at (2,0) {};

\draw[-latex,thick] (i1) -- (t1);
\draw[-latex,thick] (in) -- (t1);
\draw[-latex,thick] (t1) -- (o1);
\draw[-latex,thick] (t1) -- (om);
\end{tikzpicture}
}
\hspace{-1cm}
{
\begin{tikzpicture}
\matrix [below, ampersand replacement=\&] at (current bounding box.south) {
 \draw[draw=black,fill=white] (0,0) rectangle ++(0.12,0.25);\& \node[inner sep=0pt,minimum size=0pt,label=right:{\scriptsize Computation unit}]{}; \& \node{}; \& \&
 \draw[dashed] (2.3,0) to node[pos=0.5,yshift=4]{\arrowflow} (2.8,0); \& \node[inner sep=0pt,minimum size=3pt,label=right:{\scriptsize Control flow edge}]{}; \& \node{}; \& \&
 \node[draw=black,fill=white,inner sep=0pt,minimum size=3pt,label=right:{\scriptsize Ec-inport}] {}; \& \node{};
 \& \node[fill=black,inner sep=0pt,minimum size=3pt,label={right:{\scriptsize Ec-outport}}] {}; \\
};
\matrix [below, ampersand replacement=\&] at (current bounding box.south) {
 \draw (0,0) to node[pos=0.5,yshift=4]{\arrowflow} (0.5,0); \& \node[inner sep=0pt,minimum size=3pt,label=right:{\scriptsize Data flow edge}]{}; \& \node{}; \& \&
 \node[circle,draw=black,fill=white,inner sep=0pt,minimum size=3pt,label={right:{\scriptsize Ed-inport}}] {}; \& \node{};
 \& \node[circle,fill=black,inner sep=0pt,minimum size=3pt,label={right:{\scriptsize Ed-outport}}] {}; \\
};
\end{tikzpicture}
}
\vspace{-10pt}
\caption{A functional computon $\lambda$ and its corresponding Petri nets. Evidently, there always is a place to buffer incoming control and another one to buffer outgoing control, in order to respect the structural constraints given by Definition \ref{def:computon-functional}. Particularly, the net $\mathcal{N}(\lambda)$ embodies the whole structure of $\lambda$, whereas $\mathcal{C}(\mathfrak{E}(\lambda))$ and $\mathcal{D}(\lambda)$ have places for control and data only, respectively. As places to store data are optional as per Definition \ref{def:computon}, we have ${j,k\geq 0}$. Labels on places are just for reference purposes since we deal with non-labelled Petri nets, as discussed in Section \ref{sec:operational-semantics}.}
\label{fig:computon-functional}
\end{figure}

Up to isomorphism, there exists a particular functional computon which possesses only one ec-inport and only one ec-outport. This sort of computon, referred to as \emph{the glue computon}, is explicitly defined when $|I|=|O|=1$. Intuitively, it just echoes the control signal received in its unique ec-inport into its only ec-outport.

\section{Composite Computons} \label{sec:composite-computons}

A composite computon is algebraically formed via a composition operator which defines an explicit control flow structure for the execution of computons in some order. More formally, a composite computon is the finite colimit of some diagram in $\textbf{Set}^\textbf{Comp}$, whose construction is given by a composition operation characterised as a colimit computation. In this section, we describe separate operations to form sequential, parallel, branching and iterative computons. Table \ref{tab:composites-index} summarises the colimit constructions each composition operation builds upon as well as resulting composites. This table also shows the subsection each operation is described in. 

\begin{table}[!h]
  \begin{center}
    \caption{Relationship between our composition operations and colimits in $\textbf{Set}^\textbf{Comp}$.}
    \label{tab:composites-index}
    \resizebox{15.5cm}{!}{
\begin{tabular}{ |c|c|c|c| } 
 \hline
 \emph{Composition operation} & \emph{Built upon} & \emph{Result} & \emph{Subsection} \\ 
 \hline
 Total sequencing & Pushout & Total sequential computon & \ref{sec:sequential-computons} \\ 
 \hline
 Partial sequencing & Pushout & Partial sequential computon & \ref{sec:sequential-computons} \\
 \hline
 Asynchronous parallel & Coproduct & P-async computon & \ref{sec:parallel-computons} \\
 \hline
 Synchronous parallel & Pushout and coproduct & P-sync computon & \ref{sec:parallel-computons} \\
 \hline
 Branching & Pushout and coproduct & Branching computon & \ref{sec:branching-computons} \\
 \hline
 Head-iteration & Pushout and coproduct & Head-iterative computon & \ref{sec:iterative-computons} \\
 \hline
 Tail-iteration & Pushout and coproduct & Tail-iterative computon & \ref{sec:iterative-computons} \\
 \hline
\end{tabular}
}
\end{center}
\end{table}

For each operation, we describe their category-theoretic foundations, operational semantics (in \textbf{Petri}) and encapsulation. The latter is a property that allows hiding the internals of composites so as to treat them as self-contained, black-box units. In this section, we will focus on encapsulation of control and data flow and we we will show that, while control flow is explicitly defined, data flow is implicitly sewn through our colimit-based composition operations. Contrary to what might seem obvious, data does not always follow control, especially in the cases of partial sequencing and parallel composition.

Compositionality is a key term in this section, which refers to the property of constructing complex computons from simpler ones through the proposed operators, while ensuring closure and structure preservation of the simpler entities in the complex ones. By closure, we mean an operation from Table \ref{tab:composites-index} remains within $\textbf{Set}^\textbf{Comp}$, i.e., the result of composing computons, each adhering to Definition \ref{def:computon}, is another computon which also conforms to Definition \ref{def:computon}. By preservation, we mean ensuring that the computation units, ports, colouring and flow adjacency of the composed entities are retained in the corresponding composite.

To realise preservation, our operators build upon pushout and coproduct which, in turn, rely on computon morphisms to (intuitively) ``insert'' a computon into another (see Definition \ref{def:computon-morphism}). That is, these two constructions do not introduce new structures but just embed computons in different ways, as discussed in Section \ref{sec:computons}. In the case of pushouts, two computons are merged into another via an apex object which represents the common part between the entities being merged. When it comes to coproduct, two computons are put in a side-by-side structure without identifying any elements. 

As $\textbf{Set}^\textbf{Comp}$ has all coproducts (see Proposition \ref{prop:computon-coproduct}), it is true that any two computons can always be combined under this colimit. Pushouts only exist for pushable spans (see Proposition \ref{prop:pushout-pushable}) so computons $\lambda_1$ and $\lambda_2$ can only be combined under this construction when there is a computon $\lambda_0$ that allows the formation of a pushable span ${\lambda_1 \xleftarrow{} \lambda_0 \xrightarrow{} \lambda_2}$. 

In this section we prove that, although $\textbf{Set}^\textbf{Comp}$ does not have all pushouts, such a category is closed under the operators we propose since the colimits they define always exist, i.e., they build upon coproduct, pushouts over pushable spans or any combination thereof. The respective proofs of existence are given for each operator in the corresponding subsection.

\subsection{Sequential Computons} \label{sec:sequential-computons}

Sequential composition is an operation we characterise as a particular pushout in $\textbf{Set}^\textbf{Comp}$. It is particular because, intuitively, the common object needs to be a trivial computon $\lambda_0$ that can be ``embedded'' into some or all the e-outports of a computon $\lambda_1$ and into some or all the e-inports of a computon $\lambda_2$. Particularly, every port $p_0 \in P_0$ that is embedded into an e-outport $p_1 \in P_1^-$ needs to be embedded into an e-inport $p_2 \in P_2^+$ and vice versa (see Definition \ref{def:span-sequentiable}). This restriction, given by so-called \emph{sequentiable spans}, enables a strict sequence in which $\lambda_1$ is computed before $\lambda_2$. A converse computation is possible and requires a different embedding since sequencing is a non-commutative operation in which order matters. The notion of a sequential computon is formalised in Definition \ref{def:computon-sequential} and its proof of existence is directly derivable from Lemma \ref{lem:span-sequentiable-pushable}.

\begin{definition} [Sequentiable Spans] \label{def:span-sequentiable}
A span ${\lambda_1 \xleftarrow{\alpha_1} \lambda_0 \xrightarrow{\alpha_2} \lambda_2}$ of computon morphisms is sequentiable if (i) $\lambda_0$ is a trivial computon with ${P_0=\vec{i}(\alpha_1) \cap \vec{o}(\alpha_2)}$, (ii) $\lambda_1$ and $\lambda_2$ are connected computons, (iii) ${\alpha_1(\vec{o}(\alpha_2)) \subseteq P_1^-}$ and (iv) ${\alpha_2(\vec{i}(\alpha_1)) \subseteq P_2^+}$. Particularly, if ${\alpha_1(\vec{o}(\alpha_2)) = P_1^-}$ and ${\alpha_2(\vec{i}(\alpha_1)) = P_2^+}$, then the span is totally sequentiable. Otherwise, if ${\alpha_1(\vec{o}(\alpha_2)) \subset P_1^-}$ or ${\alpha_2(\vec{i}(\alpha_1)) \subset P_2^+}$, the span is partially sequentiable.
\end{definition}

\begin{remark}
If a span $\rho$ is totally sequentiable, then neither ${\alpha_1(\vec{o}(\alpha_2)) \subset P_1^-}$ nor ${\alpha_2(\vec{i}(\alpha_1)) \subset P_2^+}$ can hold because equality forbids proper inclusion; conversely, if $\rho$ is partially sequentiable then at least one of the equalities ${\alpha_1(\vec{o}(\alpha_2)) = P_1^-}$ or ${\alpha_2(\vec{i}(\alpha_1)) = P_2^+}$ fails. Therefore, totally sequentiable spans are not partially sequentiable (and vice versa), i.e., these two notions are mutually exclusive.
\end{remark}

\begin{lemma} \label{lem:span-sequentiable-pushable}
Every (partially or totally) sequentiable span of computon morphisms is pushable.
\end{lemma}
\begin{proof}
Let ${\lambda_1 \xleftarrow{\alpha_1} \lambda_0 \xrightarrow{\alpha_2} \lambda_2}$ be a sequentiable span of computon morphisms. By Definition \ref{def:span-sequentiable}, ${p_0 \in P_0 \iff p_0 \in \vec{i}(\alpha_1) \cap \vec{o}(\alpha_2)}$ so ${\alpha_1(p_0) \in P_1^-}$ and ${\alpha_2(p_0) \in P_2^+}$. That is, the equation ${\alpha_1(p_0)\bullet=\emptyset=\bullet\alpha_2(p_0)}$ holds. As ${p_0\bullet=\emptyset=\bullet p_0}$ (because ${\lambda_0}$ is a trivial computon), it follows that ${\alpha_1(p_0)\bullet\setminus\alpha_1(p_0\bullet)=\emptyset=\bullet\alpha_2(p_0)\setminus\alpha_2(\bullet p_0)}$, meaning ${p_0 \notin \vec{o}(\alpha_1)\cup\vec{i}(\alpha_2)}$ and, hence, ${\vec{o}(\alpha_1)\cup\vec{i}(\alpha_2)=\emptyset}$. For ${j\in\{1,2\}}$, we observe that ${P_j^+ \neq \emptyset \neq P_j^-}$ (by Definition \ref{def:computon}) and that ${\lambda_j}$ is not a trivial computon because it is a connected computon (see Proposition \ref{prop:computon-connected-alwaysunits}). Thus, ${\alpha_1(\vec{i}(\alpha_2)) \cup \alpha_1(\vec{o}(\alpha_2)) \subseteq P_1^- \subset P_1^+ \cup P_1^-}$ and ${\alpha_2(\vec{i}(\alpha_1)) \cup \alpha_2(\vec{o}(\alpha_1)) \subseteq P_2^+ \subset P_2^+ \cup P_2^-}$.
\end{proof}

\begin{corollary} \label{cor:sequential}
For a sequentiable span ${\lambda_1 \xleftarrow{\alpha_1} \lambda_0 \xrightarrow{\alpha_2} \lambda_2}$ of computon morphisms, we have ${\vec{o}(\alpha_1)\cup\vec{i}(\alpha_2)=\emptyset}$. 
\end{corollary}
\begin{proof}
See the proof of Lemma \ref{lem:span-sequentiable-pushable}.
\end{proof}

\begin{definition} [Sequential Computon] \label{def:computon-sequential}
Let $\rho$ be a sequentiable span ${\lambda_1 \xleftarrow{\alpha_1} \lambda_0 \xrightarrow{\alpha_2} \lambda_2}$ of computon morphisms. If $\rho$ is partially sequentiable, then its pushout is called a partial sequential computon, written $\lambda_1 \rhd_{\rho} \lambda_2$. Otherwise, if $\rho$ is totally sequentiable, then its pushout is called a total sequential computon, written $\lambda_1 \unrhd_{\rho} \lambda_2$. In any case, $\lambda_0$ is called the apex computon, $\lambda_1$ the left operand and $\lambda_2$ the right operand.
\end{definition}

When an apex computon is embedded into all the e-outports of the left operand and into all the e-inports of the right one, we say that the respective span is totally sequentiable; otherwise, we say it is partially sequentiable (see Definition \ref{def:span-sequentiable}). Basically, a sequentiable span chooses the e-outports of the left operand and the e-inports of the right one that will become i-ports in the corresponding sequential composite. Port renomination, which is realised by the induced pushout morphisms, follows Proposition \ref{prop:computon-morphism-eports-become-iports}. By Proposition \ref{prop:sequential-totalispartial}, total sequentiality implies partial sequentiality in the general case. The only exception occurs when the left and right operands have exactly one e-outport and one e-inport, respectively, in which case total sequentiality is the only alternative. 

\begin{proposition} \label{prop:sequential-totalispartial}
If $\rho_1$ is a totally sequentiable span ${\lambda_1 \xleftarrow{\alpha_1} \lambda_0 \xrightarrow{\alpha_2} \lambda_2}$ of computon morphisms with ${|P_1^-|>1 \lor |P_2^+|>1}$, there exists a partial sequential computon ${\lambda_1 \rhd_{\rho_2} \lambda_2}$ where ${\rho_2}$ is the partially sequentiable span ${\lambda_1 \xleftarrow{\alpha_1 \circ \alpha_0} \Lambda \xrightarrow{\alpha_2 \circ \alpha_0} \lambda_2}$ with ${\alpha_0: \Lambda \rightarrow \lambda_0}$.
\end{proposition}
\begin{proof}
Let ${\lambda_1 \unrhd_{\rho_1} \lambda_2}$ be the pushout of a totally sequentiable span ${\lambda_1 \xleftarrow{\alpha_1} \lambda_0 \xrightarrow{\alpha_2} \lambda_2}$ of computon morphisms. Since $\lambda_0$ is required to have at least one ec-inoutport (see Definition \ref{def:computon}), we know there exists at least one computon morphism ${\alpha_0: \Lambda \rightarrow \lambda_0}$ such that ${\alpha_0(p) \in P_0}\iff{\alpha_0(p) \in \vec{i}(\alpha_1) \cap \vec{o}(\alpha_2)}$ (see Definition \ref{def:span-sequentiable}). Consequently, the span ${\rho_2:=\lambda_1 \xleftarrow{\alpha_1 \circ \alpha_0} \Lambda \xrightarrow{\alpha_2 \circ \alpha_0} \lambda_2}$ exists.

As ${\alpha_0(p) \in \vec{o}(\alpha_2)}$ and ${\alpha_1(\vec{o}(\alpha_2))=P_1^-}$ (by Definition \ref{def:span-sequentiable} of totally sequentiable spans), we have ${\alpha_1(\alpha_0(p)) \in P_1^-}$, i.e., ${\alpha_1(\alpha_0(p))\bullet=\emptyset}$. The fact that ${\lambda_1}$ is a connected computon implies there exists some computation unit ${u_1 \in \bullet \alpha_1(\alpha_0(p))}$ so that ${\bullet \alpha_1(\alpha_0(p)) \neq \emptyset}$ (see Proposition \ref{prop:computon-connected-alwaysunits}). As ${\bullet p = \emptyset}$ because $\Lambda$ is a trivial computon, ${\bullet \alpha_1(\alpha_0(p)) \setminus \alpha_1(\alpha_0(\bullet p)) \neq \emptyset}$ holds, meaning ${p \in \vec{i}(\alpha_1 \circ \alpha_0)}$ (see Definition \ref{def:computon-morphism}). A similar reasoning can be used to deduce ${\alpha_2(\alpha_0(p)) \in P_2^+}$ and ${p \in \vec{o}(\alpha_2 \circ \alpha_0)}$, i.e., ${\bullet\alpha_2(\alpha_0(p))=\emptyset}$. Having ${P=\{p\}}$ together with ${\bullet\alpha_2(\alpha_0(p))=\alpha_1(\alpha_0(p))\bullet=\bullet p=p \bullet=\emptyset}$ allows us to deduce ${\vec{i}(\alpha_2 \circ \alpha_0)=\emptyset=\vec{o}(\alpha_1 \circ \alpha_0)}$. That is, ${P=\vec{i}(\alpha_1 \circ \alpha_0)\cap\vec{o}(\alpha_2 \circ \alpha_0)}$.

The facts ${\alpha_1(\alpha_0(p)) \in P_1^-}$ and ${p \in \vec{o}(\alpha_2 \circ \alpha_0)}$ together imply ${\alpha_1(\alpha_0(\vec{o}(\alpha_2 \circ \alpha_0))) \subseteq P_1^-}$. Similarly, ${\alpha_2(\alpha_0(p)) \in P_2^+}$ and ${p \in \vec{i}(\alpha_1 \circ \alpha_0)}$ imply ${\alpha_2(\alpha_0(\vec{i}(\alpha_1 \circ \alpha_0))) \subseteq P_2^+}$. Hence, by Definition \ref{def:span-sequentiable}, ${\rho_2}$ is a sequentiable span. 

The morphisms of such a span are particularly injective because ${|P|=1}$. So, if ${|P_1^-|>1}$ or ${|P_2^+|>1}$, ${\alpha_1(\alpha_0(\vec{o}(\alpha_2 \circ \alpha_0))) \subset P_1^-}$ or ${\alpha_2(\alpha_0(\vec{i}(\alpha_1 \circ \alpha_0))) \subset P_2^+}$. That is, by Lemma \ref{lem:span-sequentiable-pushable}, ${\rho_2}$ is a partially sequentiable span whose pushout forms a partial sequential computon ${\lambda_1 \rhd_{\rho_2} \lambda_2}$.
\end{proof}

Any computon operand can be put in any order within a sequential computon, since ec-inports and ec-outports always possess the same colour (i.e., zero). This property, combined with the fact that a computon always has at least one ec-inport and at least one ec-outport, allows us to compose any two connected computons sequentially regardless of the data they require or produce (see Theorem \ref{th:always-sequentiable}). Composing two connected computons sequentially always results in another connected computon (see Proposition \ref{prop:computon-sequential-connected}).  

\begin{theorem} \label{th:always-sequentiable}
Let $\lambda_1$ and $\lambda_2$ be two computons. Then, there is a span $\rho$ whose pushout is ${\lambda_1 \rhd_\rho \lambda_2}$ or ${\lambda_1 \unrhd_\rho \lambda_2}$ $\iff$ $\lambda_1$ and $\lambda_2$ are connected computons.
\end{theorem}
\begin{proof}
$(\implies)$ This part of the proof follows directly from Definition \ref{def:span-sequentiable}.

$(\impliedby)$ Assuming $\lambda_1$ and $\lambda_2$ are connected computons, we first prove there exists a sequentiable span ${\rho:=\lambda_1 \xleftarrow{\alpha_1} \lambda_0 \xrightarrow{\alpha_2} \lambda_2}$ of computon morphisms. For this, we choose $\lambda_0$ to be $\Lambda$ which is a trivial computon with a unique port $p \in P^+ \cap P^-$ and $c(p)=0$ (see Definition \ref{def:computon-unit}). Below we construct computon morphisms $\alpha_1: \Lambda \rightarrow \lambda_1$ and $\alpha_2: \Lambda \rightarrow \lambda_2$ by only considering port mapping because ${I=O=U=\emptyset}$ and ${\Sigma=\{0\}}$.

Since any computon has at least one ec-inport and at least one ec-outport, ${C_1^- \neq \emptyset \neq C_2^+}$. If we trivially define ${\alpha_1(p)=p_1 \in C_1^-}$ and ${\alpha_2(p)=p_2 \in C_2^+}$, then ${p \in \vec{i}(\alpha_1) \cap \vec{o}(\alpha_2)}$ because ${\lambda_1}$ and ${\lambda_2}$ are connected computons (see Definition \ref{def:computon-connected} and Proposition \ref{prop:computon-connected-alwaysunits}). As ${P=\{p\}}$, we have ${P=\vec{i}(\alpha_1) \cap \vec{o}(\alpha_2)}$ so ${\alpha_1(\vec{o}(\alpha_2)) \subseteq C_1^- \subseteq P_1^-}$ and ${\alpha_2(\vec{i}(\alpha_1)) \subseteq C_2^+ \subseteq P_2^+}$. That is, $\rho$ is a sequentiable span of computon morphisms (by Definition \ref{def:span-sequentiable}), whose pushout is a sequential computon (by Lemma \ref{lem:span-sequentiable-pushable}).

If ${\alpha_1(\vec{o}(\alpha_2)) \subset C_1^-}$ then ${\alpha_1(\vec{o}(\alpha_2)) \subset P_1^-}$ because ${C_1^- \subseteq P_1^-}$. Similarly, ${\alpha_2(\vec{i}(\alpha_1)) \subset C_2^+}$ implies ${\alpha_2(\vec{i}(\alpha_1)) \subset P_2^+}$ because ${C_2^+ \subseteq P_2^+}$. Thus, ${\lambda_1 \rhd_{\rho} \lambda_2}$ would be the pushout of $\rho$ in both cases. Now, when $\alpha_1(\vec{o}(\alpha_2)) = C_1^-$, we have two possibilities: ${C_1^- \subset P_1^-}$ or ${C_1^- = P_1^-}$. If ${C_1^- \subset P_1^-}$, then ${\alpha_1(\vec{o}(\alpha_2)) \subset P_1^-}$ so that ${\lambda_1 \rhd_{\rho} \lambda_2}$ would also be the pushout of $\rho$. A similar approach can be used to prove that the partial sequential computon ${\lambda_1 \rhd_{\rho} \lambda_2}$ is the pushout of $\rho$ when both ${\alpha_2(\vec{i}(\alpha_1)) = C_2^+}$ and ${C_2^+ \subset P_2^+}$ hold.

The only scenario in which ${\lambda_1 \unrhd_{\rho} \lambda_2}$ is the pushout of $\rho$ is when $\lambda_1$ posesses only one ec-outport with no ed-outports while $\lambda_2$ has only one ec-inport with no ed-inports. More precisely, ${\lambda_1 \unrhd_{\rho} \lambda_2}$ can be formed exactly when $\alpha_1(\vec{o}(\alpha_2)) = C_1^- = P_1^-$ and $\alpha_2(\vec{i}(\alpha_1)) = C_2^+ = P_1^+$. 

Therefore, we conclude that for every pair $(\lambda_1,\lambda_2)$ of connected computons, either ${\lambda_1 \rhd_{\rho} \lambda_2}$ or ${\lambda_1 \unrhd_{\rho} \lambda_2}$ exists.
\end{proof}

\begin{proposition} \label{prop:computon-sequential-connected}
A sequential computon is a connected computon.
\end{proposition}
\begin{proof}
We know that a sequential computon is the pushout of a sequentiable span $\rho$ of computon morphisms (as per Definition \ref{def:computon-sequential}) and that the legs of $\rho$ are connected computons (as per Definition \ref{def:span-sequentiable}). By Proposition \ref{prop:pushout-connected}, the pushout of $\rho$ must be a connected computon.
\end{proof}

Theorem \ref{th:always-sequentiable} is important for our theory since it entails that any two connected computons can always be composed sequentially. Although an apex computon always exists, it is important to note that it does not need to correspond to the entire common part between the e-outports of the left operand and the e-inports of the right operand. By common, we mean ports sharing the same colour. Figure \ref{fig:computon-sequential-example} depicts a scenario of this sort. 

\vspace{1pt}

\begin{figure}[!h]
\centering
\subcaptionbox{}
{
\begin{tikzpicture}
\qmatch{0q}{3.5}{6.1}{};
\dmatch{0d}{3.5}{5.35}{$3$}{above};
\draw[->] (3,5.85) to[bend right=25] node[above]{\scriptsize$\alpha_1$}  (1.5,5.1);
\draw[->] (4,5.85) to[bend left=25] node[above]{\scriptsize$\alpha_2$}  (5.5,5.1);

\computonPrimitive{1}{3}{1}{1.9}{$\lambda_1$};
\qin{1q0}{0.2}{4.7}{};
\din{1i1}{0.2}{4.1}{$1$};
\din{1i2}{0.2}{3.5}{$2$};
\qout{1q1}{2}{4.7}{};
\dout{1o1}{2}{4.1}{$3$};
\dout{1o2}{2}{3.5}{$4$};
\draw[->] (1.5,2.8) to[bend right=25] node[left]{\scriptsize$\beta_1$}  (1.80,1.8);
\draw[->] (5.5,2.8) to[bend left=25] node[right]{\scriptsize$\beta_2$}  (5.2,1.8);

\computonPrimitive{5}{3}{1}{1.9}{$\lambda_2$};
\qin{2q0}{4.2}{4.7}{};
\din{2i1}{4.2}{4.1}{$3$};
\din{2i2}{4.2}{3.5}{$4$};
\qout{2q1}{6}{4.7}{};
\dout{2o1}{6}{4.1}{$5$};

\computonComposite{0.95}{-1}{5.1}{2.6};
\qmatch{3q}{3.5}{1.2}{};\flow{{2.5,1.2}}{{3.5,1.2}}{dashed}{};\flow{{3.6,1.2}}{{4.6,1.2}}{dashed}{};
\dmatch{3d}{3.5}{0.6}{$3$}{above};\flow{{2.5,0.6}}{{3.5,0.6}}{}{};\flow{{3.56,0.6}}{{4.6,0.6}}{}{};
\computonPrimitive{1.5}{-0.5}{1}{1.9}{$\lambda_1$};
\qin{3q0}{0.7}{1.2}{};
\din{3i1}{0.7}{0.6}{$1$};
\din{3i2}{0.7}{0}{$2$};
\node[circle,fill=black,inner sep=0pt,minimum size=3pt,label=right:\scriptsize $4$] (3o2) at (6.5,-0.6) {};\flow{{2.5,0}}{3o2}{}{bend right=20};
\computonPrimitive{4.5}{-0.5}{1}{1.9}{$\lambda_2$}
\node[circle,draw=black,fill=white,inner sep=0pt,minimum size=3pt,label=left:\scriptsize $4$] (3i3) at (0.5,-0.6) {};\flow{3i3}{{4.5,0}}{}{bend right=20};
\qout{3q1}{5.5}{1.2}{};
\dout{3o1}{5.5}{0.6}{$5$};
\end{tikzpicture}
}\hspace{9cm}
{
\begin{tikzpicture}
\begin{scope}
\draw[draw=black,fill=white] (0,0.4) rectangle ++(0.12,0.25);
\node[inner sep=0pt,minimum size=0pt,label=right:{\scriptsize Computation unit}] at (0.2,0.5) {};
\draw[dashed] (3.2,0.5) to node[pos=0.5]{\arrowflow} (3.8,0.5);
\node[inner sep=0pt,minimum size=3pt,label=right:{\scriptsize Control flow edge}] at (3.8,0.5) {};
\node[draw=black,fill=white,inner sep=0pt,minimum size=3pt,label=right:{\scriptsize Ec-inport}] at (6.8,0.5) {};
\node[fill=black,inner sep=0pt,minimum size=3pt,label={right:{\scriptsize Ec-outport}}] at (8.8,0.5) {};
\node[draw,fill=white,fill fraction={black}{0.5},inner sep=0pt,minimum size=3pt,label=right:{\scriptsize Ic-port}] at (11,0.5) {};
\end{scope}

\begin{scope}[xshift=1.5cm]
\draw (0,0) to node[pos=0.5]{\arrowflow} (0.5,0);
\node[inner sep=0pt,minimum size=3pt,label=right:{\scriptsize Data flow edge}] at (0.5,0){};
\node[circle,draw=black,fill=white,inner sep=0pt,minimum size=3pt,label={right:{\scriptsize Ed-inport}}] at (3.3,0) {};
\node[circle,fill=black,inner sep=0pt,minimum size=3pt,label={right:{\scriptsize Ed-outport}}] at (5.3,0) {};
\node[circle,draw,fill=white,fill fraction={black}{0.5},inner sep=0pt,minimum size=3pt,label=right:{\scriptsize Id-port}] at (7.5,0) {};
\end{scope}
\end{tikzpicture}
}
\caption{Constructing a partial sequential computon $\lambda_1 \rhd_{\rho} \lambda_2$ from a partially sequentiable span ${\rho:=\lambda_1 \xleftarrow{\alpha_1} \lambda_0 \xrightarrow{\alpha_2} \lambda_2}$ of computon morphisms.}
\label{fig:computon-sequential-example}
\end{figure}

Figure \ref{fig:computon-sequential-example} shows that, when a partial sequential computon $\lambda_1 \rhd_{\rho} \lambda_2$ is constructed from a partially sequentiable span ${\rho:=\lambda_1 \xleftarrow{\alpha_1} \lambda_0 \xrightarrow{\alpha_2} \lambda_2}$, there is an implicit effect in which all the e-outports of $\lambda_1$ that are not in the image of $\alpha_1$ become e-outports in $\lambda_1 \rhd_{\rho} \lambda_2$. Similarly, all the e-inports of $\lambda_2$ that are not in the image of $\alpha_2$ become e-inports in $\lambda_1 \rhd_{\rho} \lambda_2$. Naturally, this generative effect does not occur in the case of total sequential composition since the images of the computon morphisms involved would cover every e-outport of the left operand and every e-inport of the right one. Instead, each $p_1 \in P_1^-$ and each $p_2 \in P_2^+$ would be mapped to an i-port of $\lambda_1 \unrhd_{\rho} \lambda_2$. No matter whether a partial or a total sequential computon is formed, it is true that every e-inport of the left operand and every e-outport of the right one are preserved in the resulting sequential computon (see Propositions \ref{prop:computon-sequential-inports-outports-1} and \ref{prop:computon-sequential-inports-outports-2}).

\begin{proposition} \label{prop:computon-sequential-inports-outports-1}
If ${(\beta_1:\lambda_1 \rightarrow \lambda_3, \lambda_3, \beta_2:\lambda_2 \rightarrow \lambda_3)}$ is the pushout of a sequentiable span ${\lambda_1 \xleftarrow{\alpha_1} \lambda_0 \xrightarrow{\alpha_2} \lambda_2}$ of computon morphisms, then ${\beta_1(P_1^+) \subseteq P_3^+}$ and ${\beta_2(P_2^-) \subseteq P_3^-}$.
\end{proposition}
\begin{proof}
By letting $(\beta_1:\lambda_1 \rightarrow \lambda_3, \lambda_3, \beta_2:\lambda_2 \rightarrow \lambda_3)$ be the pushout of a sequentiable span ${\lambda_1 \xleftarrow{\alpha_1} \lambda_0 \xrightarrow{\alpha_2} \lambda_2}$ of computon morphisms, we just show ${\beta_1(P_1^+) \subseteq P_3^+}$ since the proof of ${\beta_2(P_2^-) \subseteq P_3^-}$ is completely analogous. 

Assume for contrapositive that ${p_3 \notin P_3^+}$ so there is some ${o_3 \in O_3}$ where ${t_3(o_3)=p_3}$. The fact ${O_3=O_1 +_{O_0} O_2}$ implies three possibilities: (i) there exclusively is some ${o_1 \in O_1}$ where ${\beta_1(o_1)=o_3}$, (ii) there exclusively is some ${o_2 \in O_2}$ where ${\beta_2(o_2)=o_3}$ or (iii) there are ${o_4 \in O_1}$ and ${o_5 \in O_2}$ such that ${\beta_1(o_4)=o_3=\beta_2(o_5)}$. The third scenario never holds since ${\lambda_0}$ is a trivial computon by Definition \ref{def:span-sequentiable}. So, we just prove for (i) and (ii).

For (i), ${\beta_1(t_1(o_1))=t_3(\beta_1(o_1))=t_3(o_3)=p_3}$ implies there is some ${p_1 \in P_1}$ with ${t_1(o_1)=p_1}$ and ${\beta_1(p_1)=p_3}$. Consequently, by Definition \ref{def:computon-interface}, ${p_1 \notin P_1^+}$ so ${p_3 \notin \beta_1(P_1^+)}$. If (ii) holds, ${\beta_2(t_2(o_2))=t_3(\beta_2(o_2))=t_3(o_3)=p_3}$ implies there is some ${p_2 \in P_2}$ for which ${t_2(o_2)=p_2}$ and ${\beta_2(p_2)=p_3}$. If there is some ${p_1 \in P_1}$ where ${\beta_1(p_1)=p_3=\beta_2(p_2)}$, there is some ${p_0 \in P_0}$ where ${\alpha_1(p_0)=p_1}$ and ${\alpha_2(p_0)=p_2}$. As ${(\nexists o_1 \in O_1)[\beta_1(o_1)=o_3=\beta_2(o_2)]}$ and ${\sigma_2}$ is surjective, ${p_0 \in \vec{i}(\alpha_2)}$ which contradicts the fact ${\vec{i}(\alpha_2)=\emptyset}$ (see Corollary \ref{cor:sequential}). So, there is no ${p_1 \in P_1}$ where ${\beta_1(p_1)=p_3=\beta_2(p_2)}$. That is, ${p_3 \notin \beta_1(P_1^+)}$. 

Proving ${p_3 \notin P_3^+ \implies p_3 \notin \beta_1(P_1^+)}$ in the above cases entails ${\beta_1(P_1^+) \subseteq P_3^+}$, as required.  
\end{proof}

\begin{proposition} \label{prop:computon-sequential-inports-outports-2}
If ${(\beta_1:\lambda_1 \rightarrow \lambda_3, \lambda_3, \beta_2:\lambda_2 \rightarrow \lambda_3)}$ is the pushout of a totally sequentiable span ${\lambda_1 \xleftarrow{\alpha_1} \lambda_0 \xrightarrow{\alpha_2} \lambda_2}$ of computon morphisms, then ${\beta_1(P_1^+)=P_3^+}$ and ${\beta_2(P_2^-)=P_3^-}$.
\end{proposition}
\begin{proof}
Assuming $(\beta_1:\lambda_1 \rightarrow \lambda_3, \lambda_3, \beta_2:\lambda_2 \rightarrow \lambda_3)$ is the pushout of a totally sequentiable span ${\rho:=\lambda_1 \xleftarrow{\alpha_1} \lambda_0 \xrightarrow{\alpha_2} \lambda_2}$ of computon morphisms, below we just show ${\beta_1(P_1^+)=P_3^+}$ since the proof of ${\beta_2(P_2^-)=P_3^-}$ is completely analogous.

As Proposition \ref{prop:computon-sequential-inports-outports-1} says ${\beta_1(P_1^+) \subseteq P_3^+}$, we just have to prove ${P_3^+ \subseteq \beta_1(P_1^+)}$. If we let ${p_3 \in P_3^+}$, by Proposition \ref{prop:computon-morphism-inports-outports} and by the fact ${P_3=P_1 +_{P_0} P_2}$, we have three options: (i) there exclusively is some ${p_1 \in P_1^+}$ such that ${\beta_1(p_1)=p_3}$, (ii) there exclusively is some ${p_2 \in P_2^+}$ such that ${\beta_2(p_2)=p_3}$ or (iii) there are ${p_4 \in P_1^+}$ and ${p_5 \in P_2^+}$ such that ${\beta_1(p_4)=p_3=\beta_2(p_5)}$. If (i) is true, then ${p_3 \in \beta_1(P_1^+)}$ follows directly. We now show that (ii) and (iii) cannot hold. 

Supposing (ii) is true, we have ${p_2 \in \alpha_2(\vec{i}(\alpha_1))}$ because ${\alpha_2(\vec{i}(\alpha_1))=P_2^+}$ (by the fact that $\rho$ is totally sequentiable). Therefore, ${(\exists p_0 \in P_0 \cap \vec{i}(\alpha_1))(\exists p_1 \in P_1)[\alpha_1(p_0)=p_1~\land~\alpha_2(p_0)=p_2]}$. As commutativity contradicts (ii), there is no ${p_2 \in P_2^+}$ such that ${\beta_2(p_2)=p_3 \in P_3^+}$. That is, ${p_3 \notin \beta_2(P_2^+)}$. To disprove (iii), we deduce by commutativity the existence of some ${p \in P_0}$ where ${\alpha_1(p)=p_4}$ and ${\alpha_2(p)=p_5}$. Since ${p_4 \in P_1^+}$ and ${p_5 \in P_2^+}$, Proposition \ref{prop:computon-morphism-inports-outports} says ${p \in P_0^+}$. The fact that ${\lambda_2}$ is a connected computon and that ${\alpha_2(p)=p_5 \in P_2^+}$ entail ${p \in \vec{o}(\alpha_2)}$. Because ${\rho}$ is totally sequentiable, it is true that ${\alpha_1(\vec{o}(\alpha_2))=P_1^-}$ and, consequently, ${\alpha_1(p)=p_4 \in P_1^-}$. But ${\lambda_1}$ is also a connected computon, so ${p_4 \in P_1^+ \cap P_1^-}$ cannot hold because that would contradict Proposition \ref{prop:computon-connected-isolated-port}. This contradiction implies there is no ${p_5 \in P_2^+}$ where ${\beta_2(p_5)=p_3 \in P_3^+}$, i.e., ${p_3 \notin \beta_2(P_2^+)}$. 

Proving (i) and disproving (ii) and (iii) entails ${p_3 \in P_3^+ \implies [p_3 \in \beta_1(P_1^+) \setminus \beta_2(P_2^+)]}$, i.e., ${P_3^+ \subseteq \beta_1(P_1^+)}$, as required.
\end{proof}

Although Figure \ref{fig:computon-sequential-example} shows an example of partial sequential composition, the same computon operands can be used to perform total sequential composition. This is because, in this case, there exists an apex computon that can be inserted into all the e-outports of $\lambda_1$ and into all the e-inports of $\lambda_2$. Such an apex does not always exist so partiality does not imply totality and, thus, the reverse of Proposition \ref{prop:sequential-totalispartial} does not hold. Proposition \ref{prop:sequential-totalispartial} combined with Theorem \ref{th:always-sequentiable} states that if any two connected computons can be composed into a total sequential computon, they can also be composed into a partial sequential computon, only if the left operand has more than one e-outport or if the right operand has at least two e-inports; otherwise, such computons can only form a total sequential computon. While total sequential composition is an associative operation (see Proposition \ref{prop:computon-sequential-total-associative}), partial sequential composition is not (see Proposition \ref{prop:computon-sequential-partial-associative}). In both cases, commutativity does not hold in the sense that the order of the operands matters (see Proposition \ref{prop:sequencing-commutative}).

\begin{proposition}[Total sequential composition is associative]\label{prop:computon-sequential-total-associative}
There is an isomorphism between ${\lambda_1 \unrhd_{\rho_3} (\lambda_2 \unrhd_{\rho_2} \lambda_3)}$ and ${(\lambda_1 \unrhd_{\rho_1} \lambda_2) \unrhd_{\rho_4} \lambda_3}$ for any choice of total sequential computons ${\lambda_1 \unrhd_{\rho_1} \lambda_2}$, ${\lambda_2 \unrhd_{\rho_2} \lambda_3}$, ${\lambda_1 \unrhd_{\rho_3} (\lambda_2 \unrhd_{\rho_2} \lambda_3)}$ and ${(\lambda_1 \unrhd_{\rho_1} \lambda_2) \unrhd_{\rho_4} \lambda_3}$.
\end{proposition}
\begin{proof}
Let ${\rho_1:=\lambda_1 \xleftarrow{\alpha_1} \lambda_0 \xrightarrow{\alpha_2} \lambda_2}$ and ${\rho_2:=\lambda_2 \xleftarrow{\alpha_3} \lambda_4 \xrightarrow{\alpha_4} \lambda_3}$ be two totally sequentiable spans of computon morphisms. By Definition \ref{def:computon-sequential} and Lemma \ref{lem:span-sequentiable-pushable}, we know that the pushouts of $\rho_1$ and $\rho_2$ are the total sequential computons ${\lambda_1 \unrhd_{\rho_1} \lambda_2}$ and ${\lambda_2 \unrhd_{\rho_2} \lambda_3}$, respectively. Accordingly, consider the following commutative diagram:

\[
\begin{tikzcd}
 & \lambda_4 \arrow[d, "\alpha_3"']\arrow[r, "\alpha_4"] & \lambda_3 \arrow[d, "\beta_4"] \\
\lambda_0 \arrow[d, "\alpha_1"']\arrow[r, "\alpha_2"] & \lambda_2 \arrow[d, "\beta_2"]\arrow[r, "\beta_3"] & \lambda_2 \unrhd_{\rho_2} \lambda_3 \arrow[d, "\beta_6"] \\
\lambda_1 \arrow[r, "\beta_1"'] & \lambda_1 \unrhd_{\rho_1} \lambda_2 \arrow[r, "\beta_5"'] & \lambda_5 
\end{tikzcd}
\]

where ${\lambda_1 \unrhd_{\rho_1} \lambda_2 \xleftarrow{\beta_2} \lambda_2 \xrightarrow{\beta_3} \lambda_2 \unrhd_{\rho_2} \lambda_3}$ is a pushout-induced span of computon morphisms, which evidently is not sequentiable by the fact that $\lambda_2$ is not a trivial computon. As it is routine to show it is pushable, we have that its pushout ${\lambda_5}$ can be constructed. Consequently, each square of the above diagram is a pushout. Using categorical algebra to horizontally and vertically compose morphisms, we obtain the following diagrams:

\[
\begin{tikzcd}
\lambda_0 \arrow[d, "\alpha_1"']\arrow[r, "\beta_3\circ\alpha_2"] & \lambda_2 \unrhd_{\rho_2} \lambda_3 \arrow[d, "\beta_6"] & & \lambda_4 \arrow[d, "\beta_2\circ\alpha_3"']\arrow[r, "\alpha_4"] & \lambda_3 \arrow[d, "\beta_6\circ\beta_4"] \\
\lambda_1 \arrow[r, "\beta_5\circ\beta_1"'] & \lambda_5 & & \lambda_1 \unrhd_{\rho_1} \lambda_2 \arrow[r, "\beta_5"'] & \lambda_5
\end{tikzcd}
\]

By letting $\rho_3$ be the (above) span ${\lambda_1 \xleftarrow{\alpha_1} \lambda_0 \xrightarrow{\beta_3\circ\alpha_2} \lambda_2 \unrhd_{\rho_2} \lambda_3}$, we now show its pushout $\lambda_5$ is the total sequential computon ${\lambda_1 \unrhd_{\rho_3} (\lambda_2 \unrhd_{\rho_2} \lambda_3)}$. For this, we first prove that $\rho_3$ is totally sequentiable: $p_0 \in P_0 \iff p_0 \in \vec{i}(\alpha_1) \cap \vec{o}(\alpha_2)$ (because $\rho_1$ is sequentiable) $\iff$ $p_0 \in \vec{i}(\alpha_1)$ and $\alpha_2(p_0)\bullet\setminus\alpha_2(p_0 \bullet)\neq\emptyset$ (by Definition \ref{def:computon-morphism}) $\iff$ $p_0 \in \vec{i}(\alpha_1)$ and $\beta_3(\alpha_2(p_0))\bullet\setminus\beta_3(\alpha_2(p_0 \bullet))\neq\emptyset$ (by the preservation of computation units) $\iff$ $p_0 \in \vec{i}(\alpha_1)\cap\vec{o}(\beta_3\circ\alpha_2)$ (by Definition \ref{def:computon-morphism}). Therefore, $P_0=\vec{i}(\alpha_1)\cap\vec{o}(\beta_3\circ\alpha_2)$, i.e., Condition (i) of Definition \ref{def:span-sequentiable} is met by $\rho_3$.

As $\lambda_1$ and $\lambda_2 \unrhd_{\rho_2} \lambda_3$ are connected computons (by Definition \ref{def:span-sequentiable} and Proposition \ref{prop:computon-sequential-connected}), it follows that Condition (ii) of Definition \ref{def:span-sequentiable} is also met by $\rho_3$. To prove $\alpha_1(\vec{o}(\beta_3\circ\alpha_2))=P_1^-$, consider the following chain of double implications: $p \in \alpha_1(\vec{o}(\beta_3\circ\alpha_2)) \iff p \in \alpha_1(P_0)$ (because $P_0=\vec{o}(\beta_3\circ\alpha_2)$) $\iff p \in P_1^-$ (because $\alpha_1(P_0)=P_1^-$ by the fact that $\rho_1$ is totally sequentiable). To prove the last condition of totally sequentiable spans, we proceed as follows: $q \in \beta_3(\alpha_2(\vec{i}(\alpha_1))) \iff q \in \beta_3(P_2^+)$ (because $\alpha_2(\vec{i}(\alpha_1))=P_2^+$ by the fact that $\rho_1$ is totally sequentiable) $\iff q$ is an e-inport of $\lambda_2 \unrhd_{\rho_2} \lambda_3$ (by Proposition \ref{prop:computon-sequential-inports-outports-2}).

Above we proved that $\rho_3$ is a totally sequentiable span of computon morphisms where $\lambda_0$ is the apex computon, $\lambda_1$ the left operand and $\lambda_2 \unrhd_{\rho_2} \lambda_3$ the right operand. Using Lemma \ref{lem:span-sequentiable-pushable} and Definition \ref{def:computon-sequential}, we have that the pushout of $\rho_3$ is the total sequential computon $\lambda_1 \unrhd_{\rho_3}(\lambda_2 \unrhd_{\rho_2} \lambda_3)$. Deducing that ${(\lambda_1 \unrhd_{\rho_1} \lambda_2) \unrhd_{\rho_4} \lambda_3}$ is the pushout of the span ${\rho_4:=\lambda_1 \unrhd_{\rho_1} \lambda_2 \xleftarrow{\beta_2\circ\alpha_3} \lambda_4 \xrightarrow{\alpha_4} \lambda_3}$ can be done analogously. 

As $\lambda_1 \unrhd_{\rho_3}(\lambda_2 \unrhd_{\rho_2} \lambda_3)$, $\lambda_5$ and  $(\lambda_1 \unrhd_{\rho_1} \lambda_2) \unrhd_{\rho_4} \lambda_3$ are evidently isomorphic, we conclude that total sequential composition is an associative operation up to isomorphism.
\end{proof}

\vspace{0.5pt}

\begin{proposition}\label{prop:computon-sequential-partial-associative}
Partial sequential composition is not associative.
\end{proposition}
\begin{proof}
We show there is no isomorphism between ${(\lambda_1 \rhd_{\rho_1} \lambda_2) \rhd_{\rho_2} \lambda_3}$ and ${\lambda_1 \rhd_{\rho_4} (\lambda_2 \rhd_{\rho_3} \lambda_3)}$ for some choice of partial sequential computons ${\lambda_1 \rhd_{\rho_1} \lambda_2}$, ${(\lambda_1 \rhd_{\rho_1} \lambda_2) \rhd_{\rho_2} \lambda_3}$, ${\lambda_2 \rhd_{\rho_3} \lambda_3}$ and ${\lambda_1 \rhd_{\rho_4} (\lambda_2 \rhd_{\rho_3} \lambda_3)}$.

For this, let us consider Figure \ref{fig:computon-sequential-partial-example}(a) in which there are two partially sequentiable spans of computon morphisms: ${\rho_1:=\lambda_1 \xleftarrow{\alpha_1} \lambda_0 \xrightarrow{\alpha_2} \lambda_2}$ and ${\rho_2:=(\lambda_1 \rhd_{\rho_1} \lambda_2) \xleftarrow{\alpha_3} \lambda_4 \xrightarrow{\alpha_4} \lambda_3}$.

\begin{figure}[!h]
\centering
\subcaptionbox{Constructing the partial sequential computon $(\lambda_1 \rhd_{\rho_1} \lambda_2) \rhd_{\rho_2} \lambda_3$.}
{
\begin{tikzpicture}[scale=0.92]
\qmatch{0q}{3.5}{10.1}{};
\dmatch{0d}{3.5}{9.35}{$1$}{above};
\draw[->] (3,9.85) to[bend right=25] node[above]{\scriptsize$\alpha_1$}  (1.5,9.1);
\draw[->] (4,9.85) to[bend left=25] node[above]{\scriptsize$\alpha_2$}  (5.5,9.1);

\computonPrimitive{1}{7}{1}{1.9}{$\lambda_1$};
\qin{1q0}{0.2}{8.7}{};
\qout{1q1}{2}{8.7}{};
\dout{1o1}{2}{8.1}{$1$};
\dout{1o2}{2}{7.5}{$2$};
\draw[->] (1.5,6.8) to[bend right=25] node[left]{\scriptsize$\beta_1$}  (1.80,5.8);
\draw[->] (5.5,6.8) to[bend left=25] node[right]{\scriptsize$\beta_2$}  (5.2,5.8);

\computonPrimitive{5}{7}{1}{1.9}{$\lambda_2$};
\qin{2q0}{4.2}{8.7}{};
\din{2i1}{4.2}{8.1}{$1$};
\qout{2q1}{6}{8.7}{};
\dout{2o1}{6}{8.1}{$3$};

\computonComposite{0.95}{3}{5.1}{2.6};
\qmatch{3q}{3.5}{5.2}{};\flow{{2.5,5.2}}{{3.5,5.2}}{dashed}{};\flow{{3.6,5.2}}{{4.6,5.2}}{dashed}{};
\dmatch{3d}{3.5}{4.6}{$1$}{below};\flow{{2.5,4.6}}{{3.5,4.6}}{}{};\flow{{3.56,4.6}}{{4.6,4.6}}{}{};
\node[circle,fill=black,inner sep=0pt,minimum size=3pt,label=right:\scriptsize $2$] (3o) at (6.5,3.4) {};\flow{{2.5,4}}{3o}{}{bend right=20};
\computonPrimitive{1.5}{3.5}{1}{1.9}{$\lambda_1$};
\qin{3q0}{0.7}{5.2}{};
\computonPrimitive{4.5}{3.5}{1}{1.9}{$\lambda_2$}
\qout{3q1}{5.5}{5.2}{};
\dout{3q1}{5.5}{4.6}{$3$};

\qmatch{0q}{8.3}{6.8}{};
\dmatch{0d}{8.3}{6.05}{$2$}{above};
\draw[->] (7.8,6.55) to[bend right=25] node[above]{\scriptsize$\alpha_3$}  (6.3,5.8);
\draw[->] (8.8,6.55) to[bend left=25] node[above]{\scriptsize$\alpha_4$}  (10.3,5.8);

\computonPrimitive{10}{3.5}{1}{1.9}{$\lambda_3$};
\qin{4q0}{9.2}{5.2}{};
\din{4i1}{9.2}{4.6}{$2$};
\qout{4q1}{11}{5.2}{};
\draw[->] (6.5,2.8) to[bend right=25] node[left]{\scriptsize$\beta_3$}  (6.8,1.8);
\draw[->] (10.5,2.8) to[bend left=25] node[right]{\scriptsize$\beta_4$}  (10.2,1.8);

\computonComposite{4.4}{-1.1}{8.2}{2.8};
\computonComposite{4.5}{-1}{4.7}{2.6};
\qmatch{5q}{7}{1.2}{};\flow{{6,1.2}}{{7,1.2}}{dashed}{};\flow{{7.1,1.2}}{{8.1,1.2}}{dashed}{};
\dmatch{5d}{7}{0.6}{$1$}{below};\flow{{6,0.6}}{{7,0.6}}{}{};\flow{{7.06,0.6}}{{8.1,0.6}}{}{};
\node[circle,fill=black,inner sep=0pt,minimum size=3pt,label=right:\scriptsize $3$] (5o) at (13,-0.6) {};\flow{{9,0}}{5o}{}{bend right=20};
\computonPrimitive{5}{-0.5}{1}{1.9}{$\lambda_1$};
\qin{5q0}{4.2}{1.2}{};
\computonPrimitive{8}{-0.5}{1}{1.9}{$\lambda_2$}
\qmatch{5qx}{10}{1.2}{};\flow{{9,1.2}}{{10,1.2}}{dashed}{};\flow{{10.1,1.2}}{{11.1,1.2}}{dashed}{};
\dmatch{5dx}{10}{0.6}{$2$}{below};\flow{{6,0}}{5dx}{}{bend right=70};\flow{{10.06,0.6}}{{11.1,0.6}}{}{};
\computonPrimitive{11}{-0.5}{1}{1.9}{$\lambda_3$}
\qout{5q1}{12}{1.2}{};
\end{tikzpicture}
}\vspace{12pt}
\subcaptionbox{Partial sequential computon $\lambda_2 \rhd_{\rho_3} \lambda_3$.}
{
\begin{tikzpicture}[scale=0.92]
\computonComposite{0.95}{3}{5.1}{2.6};
\qmatch{3q}{3.5}{5.2}{};\flow{{2.5,5.2}}{{3.5,5.2}}{dashed}{};\flow{{3.6,5.2}}{{4.6,5.2}}{dashed}{};
\dinplain{3i}{0.7}{3.3}{$2$}{left};\flow{3i}{{4.5,4.6}}{}{bend right=30};

\node[circle,fill=black,inner sep=0pt,minimum size=3pt,label=right:\scriptsize $3$] (3o) at (6.5,3.4) {};\flow{{2.5,4}}{3o}{}{bend right=20};
\computonPrimitive{1.5}{3.5}{1}{1.9}{$\lambda_2$};
\qin{3q0}{0.7}{5.2}{};
\din{3d0}{0.7}{4.6}{$1$};
\computonPrimitive{4.5}{3.5}{1}{1.9}{$\lambda_3$}
\qout{3q1}{5.5}{5.2}{};
\end{tikzpicture}
}
{
\begin{tikzpicture}
\begin{scope}
\draw[draw=black,fill=white] (0,0.4) rectangle ++(0.12,0.25);
\node[inner sep=0pt,minimum size=0pt,label=right:{\scriptsize Computation unit}] at (0.2,0.5) {};
\draw[dashed] (3.2,0.5) to node[pos=0.5]{\arrowflow} (3.8,0.5);
\node[inner sep=0pt,minimum size=3pt,label=right:{\scriptsize Control flow edge}] at (3.8,0.5) {};
\node[draw=black,fill=white,inner sep=0pt,minimum size=3pt,label=right:{\scriptsize Ec-inport}] at (6.8,0.5) {};
\node[fill=black,inner sep=0pt,minimum size=3pt,label={right:{\scriptsize Ec-outport}}] at (8.8,0.5) {};
\node[draw,fill=white,fill fraction={black}{0.5},inner sep=0pt,minimum size=3pt,label=right:{\scriptsize Ic-port}] at (11,0.5) {};
\end{scope}

\begin{scope}[xshift=1.5cm]
\draw (0,0) to node[pos=0.5]{\arrowflow} (0.5,0);
\node[inner sep=0pt,minimum size=3pt,label=right:{\scriptsize Data flow edge}] at (0.5,0){};
\node[circle,draw=black,fill=white,inner sep=0pt,minimum size=3pt,label={right:{\scriptsize Ed-inport}}] at (3.3,0) {};
\node[circle,fill=black,inner sep=0pt,minimum size=3pt,label={right:{\scriptsize Ed-outport}}] at (5.3,0) {};
\node[circle,draw,fill=white,fill fraction={black}{0.5},inner sep=0pt,minimum size=3pt,label=right:{\scriptsize Id-port}] at (7.5,0) {};
\end{scope}
\end{tikzpicture}
}
\caption{Counterexample that disproves the associativity of partial sequential composition.}
\label{fig:computon-sequential-partial-example}
\end{figure}
\newpage
Carefully observing Figure \ref{fig:computon-sequential-partial-example}(a) reveals that, up to isomorphism, there is only one partially sequentiable span ${\rho_3:=\lambda_2 \xleftarrow{\alpha_5} \lambda_5 \xrightarrow{\alpha_6} \lambda_3}$ where $\lambda_5$ is a unit computon which can only be injected into the unique ec-outport of $\lambda_2$ and into the unique ec-inport of $\lambda_3$ to yield the partial sequential computon ${\lambda_2 \rhd_{\rho_3} \lambda_3}$ depicted in Figure \ref{fig:computon-sequential-partial-example}(b). Now, to construct a partial sequential computon $\lambda_1 \rhd_{\rho_4} (\lambda_2 \rhd_{\rho_3} \lambda_3)$ isomorphic to $(\lambda_1 \rhd_{\rho_1} \lambda_2) \rhd_{\rho_2} \lambda_3$, $\rho_4$ must necessarily be totally sequentiable. So, the proposition being proved is true. 
\end{proof}

\begin{proposition}\label{prop:sequencing-commutative}
Sequential composition is not commutative.
\end{proposition}
\begin{proof}
For this proposition, we show:
\begin{enumerate}
\item There is no isomorphism between ${\lambda_1 \rhd_{\rho_1} \lambda_2}$ and ${\lambda_2 \rhd_{\rho_2} \lambda_1}$ for some choice of partial sequential computons $\lambda_1 \rhd_{\rho_1} \lambda_2$ and ${\lambda_2 \rhd_{\rho_2} \lambda_1}$.\label{prop:sequencing-commutative-1}
\item There is no isomorphism between ${\lambda_3 \unrhd_{\rho_3} \lambda_4}$ and ${\lambda_4 \unrhd_{\rho_4} \lambda_3}$ for some choice of total sequential computons $\lambda_3 \unrhd_{\rho_3} \lambda_4$ and ${\lambda_4 \unrhd_{\rho_4} \lambda_3}$.\label{prop:sequencing-commutative-2}
\end{enumerate}

For \ref{prop:sequencing-commutative-1}, let us consider the partially sequentiable span ${\rho_1:=\lambda_1 \xleftarrow{\alpha_1} \lambda_0 \xrightarrow{\alpha_2} \lambda_2}$ described in Figure \ref{fig:computon-sequential-example} and the obvious partially sequentiable span ${\rho_2:=\lambda_2 \xleftarrow{\gamma_2} \Lambda \xrightarrow{\gamma_1} \lambda_1}$ given by $\gamma_2(p)=p_2 \in C_2^-$ and $\gamma_1(p)=p_1 \in C_1^+$. As it is trivial to check that $\lambda_1 \rhd_{\rho_1} \lambda_2$ and $\lambda_2 \rhd_{\rho_2} \lambda_1$ are not isomorphic, we have just constructed an example that disproves the commutativity of partial sequential composition. For \ref{prop:sequencing-commutative-2}, consider the example depicted in Figure \ref{fig:computon-sequential-total-example} which evidently shows that the total sequential computons (a) and (b) are not isomorphic. Hence, we conclude that total sequential composition is not a commutative operation either. 
\begin{figure}[!h]
\centering
\subcaptionbox{}
{
\begin{tikzpicture}
\qmatch{0q}{3.5}{6.8}{};
\dmatch{0d}{3.5}{6.1}{$3$}{above};
\dmatch{0d2}{3.5}{5.35}{$4$}{above};
\draw[->] (3,5.85) to[bend right=25] node[above]{\scriptsize$\alpha_3$}  (1.5,5.1);
\draw[->] (4,5.85) to[bend left=25] node[above]{\scriptsize$\alpha_4$}  (5.5,5.1);

\computonPrimitive{1}{3}{1}{1.9}{$\lambda_3$};
\qin{1q0}{0.2}{4.7}{};
\din{1i1}{0.2}{4.1}{$1$};
\qout{1q1}{2}{4.7}{};
\dout{1o1}{2}{4.1}{$3$};
\dout{1o2}{2}{3.5}{$4$};
\draw[->] (1.5,2.8) to[bend right=25] node[left]{\scriptsize$\beta_3$}  (1.80,1.8);
\draw[->] (5.5,2.8) to[bend left=25] node[right]{\scriptsize$\beta_4$}  (5.2,1.8);

\computonPrimitive{5}{3}{1}{1.9}{$\lambda_4$};
\qin{2q0}{4.2}{4.7}{};
\din{2i1}{4.2}{4.1}{$3$};
\din{2i1}{4.2}{3.5}{$4$};
\qout{2q1}{6}{4.7}{};
\dout{2o1}{6}{4.1}{$1$};

\computonComposite{0.95}{-0.7}{5.1}{2.3};
\qmatch{3q}{3.5}{1.2}{};\flow{{2.5,1.2}}{{3.5,1.2}}{dashed}{};\flow{{3.6,1.2}}{{4.6,1.2}}{dashed}{};
\dmatch{3d}{3.5}{0.6}{$3$}{above};\flow{{2.5,0.6}}{{3.5,0.6}}{}{};\flow{{3.56,0.6}}{{4.6,0.6}}{}{};
\dmatch{3d2}{3.5}{0}{$4$}{above};\flow{{2.5,0}}{{3.5,0}}{}{};\flow{{3.56,0}}{{4.6,0}}{}{};
\computonPrimitive{1.5}{-0.5}{1}{1.9}{$\lambda_3$};
\qin{3q0}{0.7}{1.2}{};
\din{3i1}{0.7}{0.6}{$1$};
\computonPrimitive{4.5}{-0.5}{1}{1.9}{$\lambda_4$}
\qout{3q1}{5.5}{1.2}{};
\dout{3o1}{5.5}{0.6}{$1$};
\end{tikzpicture}
}
\hspace{-0.4cm}
\subcaptionbox{}
{
\begin{tikzpicture}
\qmatch{0q}{3.5}{6.1}{};
\dmatch{0d}{3.5}{5.35}{$1$}{above};
\draw[->] (3,5.85) to[bend right=25] node[above]{\scriptsize$\gamma_3$}  (1.5,5.1);
\draw[->] (4,5.85) to[bend left=25] node[above]{\scriptsize$\gamma_4$}  (5.5,5.1);

\computonPrimitive{1}{3}{1}{1.9}{$\lambda_4$};
\qin{1q0}{0.2}{4.7}{};
\din{1i1}{0.2}{4.1}{$3$};
\din{1i2}{0.2}{3.5}{$4$};
\qout{1q1}{2}{4.7}{};
\dout{1o1}{2}{4.1}{$1$};
\draw[->] (1.5,2.8) to[bend right=25] node[left]{\scriptsize$\beta_5$}  (1.80,1.8);
\draw[->] (5.5,2.8) to[bend left=25] node[right]{\scriptsize$\beta_6$}  (5.2,1.8);

\computonPrimitive{5}{3}{1}{1.9}{$\lambda_3$};
\qin{2q0}{4.2}{4.7}{};
\din{2i1}{4.2}{4.1}{$1$};
\qout{2q1}{6}{4.7}{};
\dout{2o1}{6}{4.1}{$3$};
\dout{2o2}{6}{3.5}{$4$};

\computonComposite{0.95}{-0.7}{5.1}{2.3};
\qmatch{3q}{3.5}{1.2}{};\flow{{2.5,1.2}}{{3.5,1.2}}{dashed}{};\flow{{3.6,1.2}}{{4.6,1.2}}{dashed}{};
\dmatch{3d}{3.5}{0.6}{$1$}{above};\flow{{2.5,0.6}}{{3.5,0.6}}{}{};\flow{{3.56,0.6}}{{4.6,0.6}}{}{};
\computonPrimitive{1.5}{-0.5}{1}{1.9}{$\lambda_4$};
\qin{3q0}{0.7}{1.2}{};
\din{3i1}{0.7}{0.6}{$3$};
\din{3i2}{0.7}{0}{$4$};
\computonPrimitive{4.5}{-0.5}{1}{1.9}{$\lambda_3$}
\qout{3q1}{5.5}{1.2}{};
\dout{3o1}{5.5}{0.6}{$3$};
\dout{3o2}{5.5}{0}{$4$};
\end{tikzpicture}
}
{
\begin{tikzpicture}
\begin{scope}
\draw[draw=black,fill=white] (0,0.4) rectangle ++(0.12,0.25);
\node[inner sep=0pt,minimum size=0pt,label=right:{\scriptsize Computation unit}] at (0.2,0.5) {};
\draw[dashed] (3.2,0.5) to node[pos=0.5]{\arrowflow} (3.8,0.5);
\node[inner sep=0pt,minimum size=3pt,label=right:{\scriptsize Control flow edge}] at (3.8,0.5) {};
\node[draw=black,fill=white,inner sep=0pt,minimum size=3pt,label=right:{\scriptsize Ec-inport}] at (6.8,0.5) {};
\node[fill=black,inner sep=0pt,minimum size=3pt,label={right:{\scriptsize Ec-outport}}] at (8.8,0.5) {};
\node[draw,fill=white,fill fraction={black}{0.5},inner sep=0pt,minimum size=3pt,label=right:{\scriptsize Ic-port}] at (11,0.5) {};
\end{scope}

\begin{scope}[xshift=1.5cm]
\draw (0,0) to node[pos=0.5]{\arrowflow} (0.5,0);
\node[inner sep=0pt,minimum size=3pt,label=right:{\scriptsize Data flow edge}] at (0.5,0){};
\node[circle,draw=black,fill=white,inner sep=0pt,minimum size=3pt,label={right:{\scriptsize Ed-inport}}] at (3.3,0) {};
\node[circle,fill=black,inner sep=0pt,minimum size=3pt,label={right:{\scriptsize Ed-outport}}] at (5.3,0) {};
\node[circle,draw,fill=white,fill fraction={black}{0.5},inner sep=0pt,minimum size=3pt,label=right:{\scriptsize Id-port}] at (7.5,0) {};
\end{scope}
\end{tikzpicture}
}
\caption{Counterexample that disproves the commutativity of total sequential composition.}
\label{fig:computon-sequential-total-example}
\end{figure}

\end{proof}

\subsubsection{Operational semantics for sequential computons (in the theory of Petri nets)}

No matter whether we use any of the three functors presented in Section \ref{sec:operational-semantics}, the Petri net of a sequential computon does not introduce any additional places or transitions, as a result of using a pushout operation on a sequentiable span of computon morphisms. In the case of a total sequential computon, the corresponding net takes all the e-inports from the left operand as input places and all the e-outports of the right operand as output places (see Figure \ref{fig:computon-sequential-net}(a) and Proposition \ref{prop:computon-sequential-inports-outports-2}). The net of a partial sequential computon has a similar structure, with the addition it has the unmatched e-inports of the right operand as input places and the unmatched e-ouports of the left operand as output places (see Figure \ref{fig:computon-sequential-net}(b)). 

\begin{figure}[!h]
\centering
\subcaptionbox{Net of a total sequential computon $\lambda_1\unrhd_\rho\lambda_2$ constructed from a connected computon $\lambda_1$ that has $n$ e-inports and $j$ e-outports, and a connected computon $\lambda_2$ that has $j$ e-inports and $k$ e-outports.}
{
\begin{tikzpicture}
\node[place,label={180:\scriptsize $p_n$},minimum size=3mm] (pn) at (0,2) {};
\node at (0,2.6){$\vdots$};
\node[place,label={180:\scriptsize $p_1$},minimum size=3mm] (p1) at (0,3) {};
\draw[dotted] (0.7,1.7) rectangle (2.2,3.2);\node at (1.4,2.4){\scriptsize $\lambda_1$-net};
\node[place,label={30:\scriptsize $r_1$},minimum size=3mm] (r1) at (2.9,3) {};
\node at (2.9,2.6){$\vdots$};
\node[place,label={30:\scriptsize $r_j$},minimum size=3mm] (rj) at (2.9,2) {};

\draw[dotted] (3.6,1.7) rectangle (5.1,3.2);\node at (4.3,2.4){\scriptsize $\lambda_2$-net};
\node[place,label={8:\scriptsize $s_1$},minimum size=3mm] (s1) at (5.8,3) {};
\node at (5.8,2.6){$\vdots$};
\node[place,label={8:\scriptsize $s_k$},minimum size=3mm] (sk) at (5.8,2) {};

\draw[-latex,thick] (p1) -- ($(p1)+(0.7,-0.2)$);
\draw[-latex,thick] (pn) -- ($(pn)+(0.7,0.2)$);
\draw[-latex,thick] ($(r1)+(-0.7,-0.2)$) -- (r1);
\draw[-latex,thick] ($(rj)+(-0.7,0.2)$) -- (rj);
\draw[-latex,thick] (r1) -- ($(r1)+(0.7,-0.2)$);
\draw[-latex,thick] (rj) -- ($(rj)+(0.7,0.2)$);
\draw[-latex,thick] ($(s1)+(-0.7,-0.2)$) -- (s1);
\draw[-latex,thick] ($(sk)+(-0.7,0.2)$) -- (sk);
\end{tikzpicture}
}
\subcaptionbox{Net of a partial sequential computon $\lambda_1\rhd_\rho\lambda_2$ constructed from a connected computon $\lambda_1$ that has $n$ e-inports and ${j+l}$ e-outports, and a connected computon $\lambda_2$ that has ${j+m}$ e-inports and $k$ e-outports.} 
{
\begin{tikzpicture}
\node[place,label={180:\scriptsize $q_m$},minimum size=3mm] (qm) at (0,0) {};
\node at (0,0.6){$\vdots$};
\node[place,label={180:\scriptsize $q_1$},minimum size=3mm] (q1) at (0,1) {};
\node[place,label={180:\scriptsize $p_n$},minimum size=3mm] (pn) at (0,2) {};
\node at (0,2.6){$\vdots$};
\node[place,label={180:\scriptsize $p_1$},minimum size=3mm] (p1) at (0,3) {};
\draw[dotted] (0.7,1.7) rectangle (2.2,3.2);\node at (1.4,2.4){\scriptsize $\lambda_1$-net};
\node[place,label={30:\scriptsize $r_1$},minimum size=3mm] (r1) at (2.9,3) {};
\node at (2.9,2.6){$\vdots$};
\node[place,label={30:\scriptsize $r_j$},minimum size=3mm] (rj) at (2.9,2) {};

\draw[dotted] (3.6,1.7) rectangle (5.1,3.2);\node at (4.3,2.4){\scriptsize $\lambda_2$-net};
\node[place,label={8:\scriptsize $s_1$},minimum size=3mm] (s1) at (5.8,3) {};
\node at (5.8,2.6){$\vdots$};
\node[place,label={8:\scriptsize $s_k$},minimum size=3mm] (sk) at (5.8,2) {};
\node[place,label={right:\scriptsize $v_1$},minimum size=3mm] (v1) at (5.8,1) {};
\node at (5.8,0.6){$\vdots$};
\node[place,label={right:\scriptsize $v_l$},minimum size=3mm] (vl) at (5.8,0) {};

\draw[-latex,thick] (p1) -- ($(p1)+(0.7,-0.2)$);
\draw[-latex,thick] (pn) -- ($(pn)+(0.7,0.2)$);
\draw[-latex,thick] ($(r1)+(-0.7,-0.2)$) -- (r1);
\draw[-latex,thick] ($(rj)+(-0.7,0.2)$) -- (rj);
\draw[-latex,thick] (r1) -- ($(r1)+(0.7,-0.2)$);
\draw[-latex,thick] (rj) -- ($(rj)+(0.7,0.2)$);
\draw[-latex,thick] (q1) to[bend right=15] ($(rj)+(0.7,0)$);
\draw[-latex,thick] (qm) to[bend right=15] ($(rj)+(0.7,-0.2)$);
\draw[-latex,thick] ($(s1)+(-0.7,-0.2)$) -- (s1);
\draw[-latex,thick] ($(sk)+(-0.7,0.2)$) -- (sk);
\draw[-latex,thick] ($(rj)+(-0.7,0)$) to[bend right=10] (v1);
\draw[-latex,thick] ($(rj)+(-0.7,-0.20)$) to[bend right=10] (vl);
\end{tikzpicture}
}
\caption{General structure of sequential computon nets. Both figures are applicable to all the functorial constructions from Section \ref{sec:operational-semantics}, namely $\mathcal{N}$, $\mathcal{C}\circ\mathfrak{E}$ and $\mathcal{D}$. In the case of ${\mathcal{D}}$, ${j,k,l,m,n \geq 0}$ and for the others, ${j,k,l,m,n > 0}$ according to Definition \ref{def:computon}.}
\label{fig:computon-sequential-net}
\end{figure}

Although there are no new places or transitions that could cause deadlocks in the net of a (partial or total) sequential computon, there is no guarantee such a net is deadlock-free, even though the nets of the composed computons are. To fully ensure deadlock-freedom, we have to synchronise the e-outports of the left operand with the e-inports of the right one, in order to prevent the net of the right operand from being executed before reaching its initial state. This can structurally be done by introducing a primitive computon that acts as a synchronisation point between the left and right operands. The formal notion of such a computon is given in Definition \ref{def:sync}. 

\begin{definition}[In- and Out-Sync Computons]\label{def:sync}
A primitive computon $\lambda$ is an in-sync of a connected computon $\lambda_1$ if the domain of $\lambda_1^+$ is the domain of both $\lambda^+$ and $\lambda^-$. It is an out-sync of $\lambda_1$ if the domain of $\lambda_1^-$ is the domain of both $\lambda^+$ and $\lambda^-$.
\end{definition}

Evidently, for any connected computon there always are in- and out-sync computons by the fact marker morphisms always exist (see Proposition \ref{prop:markers-always}). Propositions \ref{prop:sync-in} and \ref{prop:sync-out} show that the process of adapting an in- or an out-sync to a connected computon precisely corresponds to a total sequencing operation. 

\begin{proposition}\label{prop:sync-in}
If $\lambda$ is the in-sync of a connected computon $\lambda_1$, there is a span $\rho$ whose pushout is the total sequential computon $\lambda\unrhd_{\rho}\lambda_1$.
\end{proposition}
\begin{proof}
Assuming $\lambda$ is the in-sync of a connected computon $\lambda_1$, Definition \ref{def:sync} says there is a trivial computon $\lambda_2$ that gives rise to the span $\rho:=\lambda\xleftarrow{\lambda^-}\lambda_2\xrightarrow{\lambda_1^+}\lambda_1$. Verifying $\rho$ is totally sequentiable follows directly from Definition \ref{def:computon-morphism-markers} and Proposition \ref{prop:computon-primitive-connected}. So, by Definition \ref{def:computon-sequential}, the pushout of $\rho$ is the total sequential computon $\lambda\unrhd_{\rho}\lambda_1$.
\end{proof}

\begin{proposition}\label{prop:sync-out}
If $\lambda$ is the out-sync of a connected computon $\lambda_1$, there is a span $\rho$ whose pushout is the total sequential computon $\lambda_1\unrhd_{\rho}\lambda$.
\end{proposition}
\begin{proof}
The proof is analogous to that of Proposition \ref{prop:sync-in}.
\end{proof}

\begin{proposition}\label{prop:deadlock-total-right}
Let $\lambda_1\unrhd_\rho\lambda_2$ be a total sequential computon and $\lambda$ the in-sync of $\lambda_2$. If $\mathcal{N}(\lambda_1)$ and $\mathcal{N}(\lambda_2)$ are deadlock-free, there are spans $\rho_1$ and $\rho_2$ such that $\mathcal{N}(\lambda_1\unrhd_{\rho_2}\lambda\unrhd_{\rho_1}\lambda_2)$ is deadlock-free. 
\end{proposition}
\begin{proof}
Let $\lambda_1\unrhd_\rho\lambda_2$ be a total sequential computon constructed from the totally sequentiable span $\rho:=\lambda_1\xleftarrow{\alpha_1}\lambda_0\xrightarrow{\alpha_2}\lambda_2$ of computon morphisms. If $\lambda$ is the in-sync of $\lambda_2$, Proposition \ref{prop:sync-in} says there is a trivial computon $\lambda_3$ such that the pushout of $\rho_1:=\lambda\xleftarrow{\lambda^-}\lambda_3\xrightarrow{\lambda_2^+}\lambda_2$ is the total sequential computon $\lambda\unrhd_{\rho_1}\lambda_2$. If ${\beta:\lambda_2\rightarrow\lambda\unrhd_{\rho_1}\lambda_2}$ is one of the morphisms induced by such a pushout, we deduce the existence of ${\rho_2:=\lambda_1\xleftarrow{\alpha_1}\lambda_0\xrightarrow{\beta\circ\alpha_2}\lambda\unrhd_{\rho_1}\lambda_2}$ which trivially is totally sequentiable by the fact $\rho_1$ and $\rho$ are. Consequently, by Definition \ref{def:computon-sequential}, there is a total sequential computon $\lambda_1\unrhd_{\rho_2}\lambda\unrhd_{\rho_1}\lambda_2$ whose underlying net $\mathcal{N}(\lambda_1\unrhd_{\rho_2}\lambda\unrhd_{\rho_1}\lambda_2)$ has the following form when applying the functorial construction from Definition \ref{def:functor-computon-to-petri}:

\begin{center}
\begin{tikzpicture}
\node[place,label={180:\scriptsize $p_n$},minimum size=3mm] (pn) at (0,2) {};
\node at (0,2.6){$\vdots$};
\node[place,label={180:\scriptsize $p_1$},minimum size=3mm] (p1) at (0,3) {};
\draw[dotted] (0.7,1.7) rectangle (2.2,3.2);\node at (1.4,2.4){\scriptsize $\mathcal{N}(\lambda_1)$};
\node[place,label={80:\scriptsize $q_1$},minimum size=3mm] (q1) at (2.9,3) {};
\node at (2.9,2.6){$\vdots$};
\node[place,label={80:\scriptsize $q_j$},minimum size=3mm] (qj) at (2.9,2) {};

\node[transition,fill=black,minimum width=0.1mm,minimum height=10mm,label=\scriptsize$\mathcal{N}(\lambda)$] (t1) at (4.3,2.4) {};
\node[place,label={80:\scriptsize $r_1$},minimum size=3mm] (r1) at (5.8,3) {};
\node at (5.8,2.6){$\vdots$};
\node[place,label={80:\scriptsize $r_k$},minimum size=3mm] (rk) at (5.8,2) {};

\draw[dotted] (6.5,1.7) rectangle (8,3.2);\node at (7.2,2.4){\scriptsize $\mathcal{N}(\lambda_2)$};
\node[place,label={8:\scriptsize $s_1$},minimum size=3mm] (s1) at (8.7,3) {};
\node at (8.7,2.6){$\vdots$};
\node[place,label={8:\scriptsize $s_l$},minimum size=3mm] (sl) at (8.7,2) {};

\draw[-latex,thick] (p1) -- ($(p1)+(0.7,-0.2)$);
\draw[-latex,thick] (pn) -- ($(pn)+(0.7,0.2)$);
\draw[-latex,thick] ($(q1)+(-0.7,-0.2)$) -- (q1);
\draw[-latex,thick] ($(qj)+(-0.7,0.2)$) -- (qj);
\draw[-latex,thick] (q1) -- (t1);
\draw[-latex,thick] (qj) -- (t1);
\draw[-latex,thick] (t1) -- (r1);
\draw[-latex,thick] (t1) -- (rk);
\draw[-latex,thick] (r1) -- ($(r1)+(0.7,-0.2)$);
\draw[-latex,thick] (rk) -- ($(rk)+(0.7,0.2)$);
\draw[-latex,thick] ($(s1)+(-0.7,-0.2)$) -- (s1);
\draw[-latex,thick] ($(sl)+(-0.7,0.2)$) -- (sl);
\end{tikzpicture}
\end{center}

By Definition \ref{def:marking}, the initial state ${M_i}$ of the above net is a marking function where ${M_i(p)>0}$ for all ${p\in\{p_1,\ldots,p_n\}}$ and no tokens for all the other places, including those inside ${\mathcal{N}(\lambda_1)}$ and ${\mathcal{N}(\lambda_2)}$. As this marking corresponds to the initial marking of ${\mathcal{N}(\lambda_1)}$, only states from ${\mathcal{N}(\lambda_1)}$ are reachable from ${M_i}$ in the next time step. Assuming ${\mathcal{N}(\lambda_1)}$ and ${\mathcal{N}(\lambda_2)}$ are deadlock-free, we now have the following cases:
\begin{itemize}
\item If no state of $\mathcal{N}(\lambda_1)$ puts tokens in all the input places of $\mathcal{N}(\lambda)$, no state of $\mathcal{N}(\lambda)$ or $\mathcal{N}(\lambda_2)$ will ever be reached. Even though $\mathcal{N}(\lambda_1)$ will not terminate successfully, there is a guarantee $\mathcal{N}(\lambda_1\unrhd_{\rho_2}\lambda\unrhd_{\rho_1}\lambda_2)$ will never be stuck because $\mathcal{N}(\lambda_1)$ is deadlock-free. 
\item If a state of ${\mathcal{N}(\lambda_1)}$ puts tokens in all the places in ${\{q_1,\ldots,q_j\}}$, the only transition of ${\mathcal{N}(\lambda)}$ will be fired to reach a state that marks all the places in ${\{r_1,\ldots,r_k\}}$, which evidently corresponds to the initial marking of ${\mathcal{N}(\lambda_2)}$. As ${\mathcal{N}(\lambda_2)}$ is deadlock-free, ${\mathcal{N}(\lambda_1\unrhd_{\rho_2}\lambda\unrhd_{\rho_1}\lambda_2)}$ will not be stuck.
\end{itemize}
Hence, we conclude ${\mathcal{N}(\lambda_1\unrhd_{\rho_2}\lambda\unrhd_{\rho_1}\lambda_2)}$ is deadlock-free, as required.
\end{proof}

\begin{proposition}\label{prop:deadlock-total-left}
Let $\lambda_1\unrhd_\rho\lambda_2$ be a total sequential computon and $\lambda$ the out-sync of $\lambda_1$. If $\mathcal{N}(\lambda_1)$ and $\mathcal{N}(\lambda_2)$ are deadlock-free, there are spans $\rho_1$ and $\rho_2$ such that $\mathcal{N}(\lambda_1\unrhd_{\rho_1}\lambda\unrhd_{\rho_2}\lambda_2)$ is deadlock-free. 
\end{proposition}
\begin{proof}
Let $\lambda_1\unrhd_\rho\lambda_2$ be a total sequential computon constructed from the totally sequentiable span ${\rho:=\lambda_1\xleftarrow{\alpha_1}\lambda_0\xrightarrow{\alpha_2}\lambda_2}$ of computon morphisms. If $\lambda$ is the out-sync of $\lambda_1$, Proposition \ref{prop:sync-out} says there is a trivial computon $\lambda_3$ such that the pushout of ${\rho_1:=\lambda_1\xleftarrow{\lambda_1^-}\lambda_3\xrightarrow{\lambda^+}\lambda}$ is the total sequential computon $\lambda_1\unrhd_{\rho_1}\lambda$. If ${\beta:\lambda_1\rightarrow\lambda_1\unrhd_{\rho_1}\lambda}$ is one of the morphisms induced by such a pushout, we deduce the existence of ${\rho_2:=\lambda_1\unrhd_{\rho_1}\lambda\xleftarrow{\beta\circ\alpha_1}\lambda_0\xrightarrow{\alpha_2}\lambda_2}$ which trivially is totally sequentiable by the fact $\rho_1$ and $\rho$ are. Consequently, by Definition \ref{def:computon-sequential}, there is a total sequential computon ${\lambda_1\unrhd_{\rho_1}\lambda\unrhd_{\rho_2}\lambda_2}$. 

Applying the functorial construction $\mathcal{N}$ from Definition \ref{def:functor-computon-to-petri} on ${\lambda_1\unrhd_{\rho_1}\lambda\unrhd_{\rho_2}\lambda_2}$ results in a Petri net with the form shown in the proof of Proposition \ref{prop:deadlock-total-right}. Therefore, by that proposition, ${\mathcal{N}(\lambda_1\unrhd_{\rho_1}\lambda\unrhd_{\rho_2}\lambda_2)}$ is deadlock-free.
\end{proof}

\begin{corollary}\label{cor:deadlock-total}
Let $\lambda_1\unrhd_\rho\lambda_2$ be a total sequential computon, $\lambda_3$ the in-sync of $\lambda_2$ and $\lambda_4$ the out-sync of $\lambda_1$. Then, there are spans $\rho_1$, $\rho_2$, $\rho_3$ and $\rho_4$ such that ${\mathcal{N}(\lambda_1\unrhd_{\rho_1}\lambda_3\unrhd_{\rho_2}\lambda_2)}\cong{\mathcal{N}(\lambda_1\unrhd_{\rho_3}\lambda_4\unrhd_{\rho_4}\lambda_2)}$.
\end{corollary}
\begin{proof}
The proof follows directly from Propositions \ref{prop:deadlock-total-right} and \ref{prop:deadlock-total-left}, with the observation that $\lambda_3$ must necessarily be isomorphic to $\lambda_4$.
\end{proof}

Propositions \ref{prop:deadlock-total-right} and \ref{prop:deadlock-total-left} together entail it is always possible to construct a deadlock-free net from any total sequential computon. To do so, it suffices to place a sync computon between the left and right operands. More precisely, we can attach either an in-sync computon to the right operand or an out-sync computon to the left one (see Corollary \ref{cor:deadlock-total}) via a total sequencing operation (see Propositions \ref{prop:sync-in} and \ref{prop:sync-out}). The resulting composite can then be composed into a new total sequential computon that respects the mapping given by the original sequentiable span of computon morphisms. No matter whether we adapt the right or the left operand, the net of the resulting composite is deadlock-free by Propositions \ref{prop:deadlock-total-right} and \ref{prop:deadlock-total-left}. For partial sequential computons, a similar approach enables deadlock-freedom, as described by Propositions \ref{prop:deadlock-partial-right} and \ref{prop:deadlock-partial-left}.

\begin{proposition}\label{prop:deadlock-partial-right}
Let $\lambda_1\rhd_\rho\lambda_2$ be a partial sequential computon and $\lambda$ the in-sync of $\lambda_2$. If $\mathcal{N}(\lambda_1)$ and $\mathcal{N}(\lambda_2)$ are deadlock-free, there are spans $\rho_1$ and $\rho_2$ such that $\mathcal{N}(\lambda_1\rhd_{\rho_2}(\lambda\unrhd_{\rho_1}\lambda_2))$ is deadlock-free. 
\end{proposition}
\begin{proof}
Let $\lambda_1\rhd_\rho\lambda_2$ be a partial sequential computon constructed from the partially sequentiable span ${\rho:=\lambda_1\xleftarrow{\alpha_1}\lambda_0\xrightarrow{\alpha_2}\lambda_2}$ of computon morphisms. If $\lambda$ is the in-sync of $\lambda_2$, Proposition \ref{prop:sync-in} says there is a trivial computon $\lambda_3$ such that the pushout of ${\rho_1:=\lambda\xleftarrow{\lambda^-}\lambda_3\xrightarrow{\lambda_2^+}\lambda_2}$ is the total sequential computon $\lambda\unrhd_{\rho_1}\lambda_2$. If ${\beta:\lambda_2\rightarrow\lambda\unrhd_{\rho_1}\lambda_2}$ is one of the morphisms induced by such a pushout, we deduce the existence of ${\rho_2:=\lambda_1\xleftarrow{\alpha_1}\lambda_0\xrightarrow{\beta\circ\alpha_2}\lambda\unrhd_{\rho_1}\lambda_2}$ which trivially is partially sequentiable by the fact $\rho$ is. Consequently, by Definition \ref{def:computon-sequential}, there is a partial sequential computon ${\lambda_1\rhd_{\rho_2}(\lambda\unrhd_{\rho_1}\lambda_2)}$ whose underlying net ${\mathcal{N}(\lambda_1\rhd_{\rho_2}(\lambda\unrhd_{\rho_1}\lambda_2))}$ has the following form when applying the functorial construction from Definition \ref{def:functor-computon-to-petri}:

\begin{center}
\begin{tikzpicture}
\node[place,label={180:\scriptsize $q_m$},minimum size=3mm] (qm) at (0,0) {};
\node at (0,0.6){$\vdots$};
\node[place,label={180:\scriptsize $q_1$},minimum size=3mm] (q1) at (0,1) {};
\node[place,label={180:\scriptsize $p_n$},minimum size=3mm] (pn) at (0,2) {};
\node at (0,2.6){$\vdots$};
\node[place,label={180:\scriptsize $p_1$},minimum size=3mm] (p1) at (0,3) {};
\draw[dotted] (0.7,1.7) rectangle (2.2,3.2);\node at (1.4,2.4){\scriptsize $\mathcal{N}(\lambda_1)$};
\node[place,label={[yshift=-0.1cm,xshift=0.3cm]:\scriptsize $r_1$},minimum size=3mm] (r1) at (2.9,3) {};
\node at (2.9,2.6){$\vdots$};
\node[place,label={[yshift=-0.1cm,xshift=0.3cm]:\scriptsize $r_j$},minimum size=3mm] (rj) at (2.9,2) {};

\node[transition,fill=black,minimum width=0.1mm,minimum height=10mm,label=\scriptsize$\mathcal{N}(\lambda)$] (t1) at (4.3,2.4) {};
\node[place,label={[yshift=-0.2cm,xshift=0.3cm]:\scriptsize $s_1$},minimum size=3mm] (s1) at (5.8,3.8) {};
\node at (5.8,3.4){$\vdots$};
\node[place,label={[yshift=-0.2cm,xshift=0.3cm]:\scriptsize $s_j$},minimum size=3mm] (sk) at (5.8,2.8) {};
\node[place,label={[yshift=-0.2cm,xshift=0.3cm]:\scriptsize $w_1$},minimum size=3mm] (w1) at (5.8,2.3) {};
\node at (5.8,1.9){$\vdots$};
\node[place,label={right:\scriptsize $w_m$},minimum size=3mm] (wx) at (5.8,1.3) {};

\draw[dotted] (6.5,1.7) rectangle (8,3.2);\node at (7.2,2.4){\scriptsize $\mathcal{N}(\lambda_2)$};
\node[place,label={right:\scriptsize $v_1$},minimum size=3mm] (v1) at (8.7,3) {};
\node at (8.7,2.6){$\vdots$};
\node[place,label={right:\scriptsize $v_l$},minimum size=3mm] (vl) at (8.7,2) {};
\node[place,label={right:\scriptsize $z_1$},minimum size=3mm] (z1) at (8.7,1) {};
\node at (8.7,0.6){$\vdots$};
\node[place,label={right:\scriptsize $z_x$},minimum size=3mm] (zy) at (8.7,0) {};

\draw[-latex,thick] (p1) -- ($(p1)+(0.7,-0.2)$);
\draw[-latex,thick] (pn) -- ($(pn)+(0.7,0.2)$);
\draw[-latex,thick] ($(r1)+(-0.7,-0.2)$) -- (r1);
\draw[-latex,thick] ($(rj)+(-0.7,0.2)$) -- (rj);
\draw[-latex,thick] (r1) -- (t1);
\draw[-latex,thick] (rj) -- (t1);
\draw[-latex,thick] (t1) -- (s1);
\draw[-latex,thick] (t1) -- (sk);
\draw[-latex,thick] (t1) -- (w1);
\draw[-latex,thick] (t1) -- (wx);
\draw[-latex,thick] (s1) -- ($(s1)+(0.7,-0.8)$);
\draw[-latex,thick] (sk) -- ($(sk)+(0.7,0)$);
\draw[-latex,thick] (w1) -- ($(w1)+(0.7,0)$);
\draw[-latex,thick] (wx) -- ($(wx)+(0.7,0.5)$);
\draw[-latex,thick] ($(v1)+(-0.7,-0.2)$) -- (v1);
\draw[-latex,thick] ($(vl)+(-0.7,0.2)$) -- (vl);
\draw[-latex,thick] (q1) to[bend right=15] (t1);
\draw[-latex,thick] (qm) to[bend right=15] (t1);
\draw[-latex,thick] ($(rj)+(-0.7,0)$) to[bend right=10] (z1);
\draw[-latex,thick] ($(rj)+(-0.7,-0.20)$) to[bend right=10] (zy);
\end{tikzpicture}
\end{center}

By Definition \ref{def:marking}, the initial state ${M_i}$ of the above net is a marking function where ${M_i(p)>0}$ for all ${p\in\{p_1,\ldots,p_n,q_1,\ldots,q_m\}}$ and no tokens for all the other places, including those inside ${\mathcal{N}(\lambda_1)}$ and ${\mathcal{N}(\lambda_2)}$. This initial marking can only reach states from ${\mathcal{N}(\lambda_1)}$ in the next time step, since the only transition of ${\mathcal{N}(\lambda)}$ is not yet enabled (${j>0}$ by Definition \ref{def:computon}). Assuming ${\mathcal{N}(\lambda_1)}$ and ${\mathcal{N}(\lambda_2)}$ are deadlock-free, we now have the following cases:
\begin{itemize}
\item If no state of ${\mathcal{N}(\lambda_1)}$ puts tokens in all the places in ${\{r_1,\ldots,r_j\}}$, the only transition of ${\mathcal{N}(\lambda)}$ will not be enabled. Consequently, no state of ${\mathcal{N}(\lambda)}$ or ${\mathcal{N}(\lambda_2)}$ will ever be reached. Even though ${\mathcal{N}(\lambda_1)}$ will not terminate successfully, there is a guarantee ${\mathcal{N}(\lambda_1\rhd_{\rho_2}(\lambda\unrhd_{\rho_1}\lambda_2))}$ will never be stuck because ${\mathcal{N}(\lambda_1)}$ is deadlock-free.
\item If a state of ${\mathcal{N}(\lambda_1)}$ puts tokens in all the places in ${\{r_1,\ldots,r_j\}}$, the unique transition of ${\mathcal{N}(\lambda)}$ will be enabled since there are also tokens in ${\{q_1,\ldots,q_m\}}$ previously placed by ${M_i}$. As firing such a transition will evidently reach the initial marking of ${\mathcal{N}(\lambda_2)}$ and ${\mathcal{N}(\lambda_2)}$ is deadlock-free, ${\mathcal{N}(\lambda_1\rhd_{\rho_2}(\lambda\unrhd_{\rho_1}\lambda_2))}$ will not be stuck.
\end{itemize}
Hence, ${\mathcal{N}(\lambda_1\rhd_{\rho_2}(\lambda\unrhd_{\rho_1}\lambda_2))}$ is deadlock-free, as required.
\end{proof}

\begin{remark}
Due to the non-associativity property of partial sequential composition (see Proposition \ref{prop:computon-sequential-partial-associative}), the total sequential computon $\lambda\unrhd_{\rho_1}\lambda_2$ must be defined before constructing the corresponding partial sequential computon. This is the reason we use parentheses for the syntactic expression $\lambda_1\rhd_{\rho_2}(\lambda\unrhd_{\rho_1}\lambda_2)$.
\end{remark}

\begin{remark}
Having the same subindex for $r$- and $s$-places (and for $q$- and $w$-places) is not a coincidence. We did this to reflect the fact there is a one-to-one correspondence between the input and outputs places of the net of a sync computon. As we are adapting $\lambda_2$ rather than $\lambda_1$ and $\rho_2$ is partially sequentiable, there are some input places of $\mathcal{N}(\lambda)$ for which there is no match (i.e., $q_1,\ldots,q_m$) and some output places of $\mathcal{N}(\lambda_1)$ for which there is no match (i.e., $z_1,\ldots,z_x$). This evidently is a structural consequence of partial composition semantics.
\end{remark}

\begin{proposition}\label{prop:deadlock-partial-left}
Let $\lambda_1\rhd_\rho\lambda_2$ be a partial sequential computon and $\lambda$ the out-sync of $\lambda_1$. If ${\mathcal{N}(\lambda_1)}$ and ${\mathcal{N}(\lambda_2)}$ are deadlock-free, there are spans $\rho_1$ and $\rho_2$ such that ${\mathcal{N}((\lambda_1\unrhd_{\rho_1}\lambda)\rhd_{\rho_2}\lambda_2)}$ is deadlock-free. 
\end{proposition}
\begin{proof}
The proof similar to that of Proposition \ref{prop:deadlock-partial-right}.
\end{proof}

\begin{remark}\label{rem:sequential-deadlock}
Although Propositions \ref{prop:deadlock-total-right}--\ref{prop:deadlock-partial-left} and Corollary \ref{cor:deadlock-total} are statements about the functor $\mathcal{N}$, they are applicable to the functors $\mathcal{C}\circ\mathfrak{E}$ and $\mathcal{D}$ presented in Section \ref{sec:operational-semantics}. The proofs are valid for $\mathcal{C}\circ\mathfrak{E}$ since Proposition \ref{prop:functor-control-petri} says ${\mathcal{C}}$ is just a restriction of $\mathcal{N}$ to $\mathfrak{E}(\textbf{Set}^{\textbf{Comp}})$. 

By Remark \ref{rem:deadlock-freedom}, we only need to check deadlock-freedom for $\mathcal{D}$-nets that have initial and final states. In the case of Proposition \ref{prop:deadlock-total-right}, this occurs when $n,l>0$. As total sequencing connects all the ed-outports of the left operand with all the ed-inports of the right one, we only have to check the cases when $j=0=k$ and $j,k>0$. In the first one, only states of $\mathcal{D}(\lambda_1)$ are reachable from the initial marking. If such a net is deadlock-free, $\mathcal{D}(\lambda_1\unrhd_{\rho_2}\lambda\unrhd_{\rho_1}\lambda_2)$ is deadlock-free too. When $j,k>0$, the proof is identical to that of Proposition \ref{prop:deadlock-total-right}. Proposition \ref{prop:deadlock-total-left} and Corollary \ref{cor:deadlock-total} follow analogously for $\mathcal{D}$.

For Proposition \ref{prop:deadlock-partial-right}, we observe $\mathcal{D}(\lambda_1\rhd_{\rho_2}(\lambda\unrhd_{\rho_1}\lambda_2))$ has the following possibilities for $j$ and $m$ (considering $n,l>0$ by Remark \ref{rem:deadlock-freedom}): 
\begin{enumerate}
\item If $j=0$ and $m=0$, $\mathcal{D}(\lambda_1)$ and $\mathcal{D}(\lambda_2)$ do not share any data. In this case, $M_i$ marks the places $p_1,\ldots,p_n$ so only states of $\mathcal{D}(\lambda_1)$ are reached. Since such a net is deadlock-free, $\mathcal{D}(\lambda_1\rhd_{\rho_2}(\lambda\unrhd_{\rho_1}\lambda_2))$ is deadlock-free too. 
\item If $j=0$ and $m>0$, $\mathcal{D}(\lambda_1)$ and $\mathcal{D}(\lambda_2)$ do not share any data. In this case, $M_i$ marks the places $p_1,\ldots,p_n,q_1,\ldots,q_m$ to reach states from $\mathcal{D}(\lambda_1)$ and $\mathcal{D}(\lambda)$ simultaneously. Since $j=0$, the only transition of $\mathcal{D}(\lambda)$ is enabled by $M_i$ and firing it results in tokens in $w_1,\ldots,w_m$ thereby reaching the initial state of $\mathcal{D}(\lambda_2)$. As $\mathcal{D}(\lambda_1)$ and $\mathcal{D}(\lambda_2)$ are both deadlock-free, $\mathcal{D}(\lambda_1\rhd_{\rho_2}(\lambda\unrhd_{\rho_1}\lambda_2))$ is deadlock-free too.
\item If $j>0$ and $m=0$, $\mathcal{D}(\lambda_1)$ and $\mathcal{D}(\lambda_2)$ do share data. In this case, $M_i$ marks the places $p_1,\ldots,p_n$ to solely activate $\mathcal{D}(\lambda_1)$. If no state of $\mathcal{D}(\lambda_1)$ ever puts tokens in all the places in $\{r_1,\ldots,r_j\}$, no state of $\mathcal{D}(\lambda)$ or $\mathcal{D}(\lambda_2)$ will ever be reached. As $\mathcal{D}(\lambda_1)$ is deadlock-free, $\mathcal{D}(\lambda_1\rhd_{\rho_2}(\lambda\unrhd_{\rho_1}\lambda_2))$ will not get stuck. If a state of $\mathcal{D}(\lambda_1)$ puts tokens in all the places in $\{r_1,\ldots,r_j\}$, then states of $\mathcal{D}(\lambda_2)$ will be reached. Since $\mathcal{D}(\lambda_2)$ is deadlock-free too, $\mathcal{D}(\lambda_1\rhd_{\rho_2}(\lambda\unrhd_{\rho_1}\lambda_2))$ will not get stuck.
\item If $j>0$ and $m>0$, $\mathcal{D}(\lambda_1)$ and $\mathcal{D}(\lambda_2)$ do share data. In this case, the proof is identical to that of Proposition \ref{prop:deadlock-partial-right}.
\end{enumerate}
A similar reasoning applies to $\mathcal{D}((\lambda_1\unrhd_{\rho_3}\lambda)\rhd_{\rho_4}\lambda_2)$ with respect to Proposition \ref{prop:deadlock-partial-left}.
\end{remark}

\subsubsection{Encapsulation of control flow and data flow in sequential computons}

A total sequential computon encapsulates sequential control flow and up to sequential data flow, as a result of matching all the e-outports of the left operand with all the e-inports of the right one via a pushout operation on a totally sequentiable span of computon morphisms. Sequential control flow allows the strict execution of one computon after another, whilst sequential data flow serves to transfer all the computation results of the left operand to the right one. Figure \ref{fig:encapsulation-sequential-total} illustrates how the total sequential computon from Figure \ref{fig:computon-sequential-total-example}(a) encapsulates both sequential control flow for the invocation of $\lambda_3$ and $\lambda_4$ (in that order) and sequential data flow for passing two data items between them. 

\begin{figure}[!h]
\centering
\begin{tikzpicture}
\begin{scope}
\computonComposite{0.95}{-0.7}{5.1}{2.3};
\qmatch{3q}{3.5}{1.2}{};\flow{{2.5,1.2}}{{3.5,1.2}}{dashed}{};\flow{{3.6,1.2}}{{4.6,1.2}}{dashed}{};
\dmatch{3d}{3.5}{0.6}{$3$}{above};\flow{{2.5,0.6}}{{3.5,0.6}}{}{};\flow{{3.56,0.6}}{{4.6,0.6}}{}{};
\dmatch{3d2}{3.5}{0}{$4$}{above};\flow{{2.5,0}}{{3.5,0}}{}{};\flow{{3.56,0}}{{4.6,0}}{}{};
\computonPrimitive{1.5}{-0.5}{1}{1.9}{$\lambda_3$};
\qin{3q0}{0.7}{1.2}{};
\din{3i1}{0.7}{0.6}{$1$};
\computonPrimitive{4.5}{-0.5}{1}{1.9}{$\lambda_4$}
\qout{3q1}{5.5}{1.2}{};
\dout{3o1}{5.5}{0.6}{$1$};
\end{scope}
\draw[opacity=0.3,line width=1.5pt, ->, -Latex] (7,0.5) to node[pos=0.4,yshift=7,xshift=-6]{\scriptsize $\mathcal{C}\circ \mathfrak{E}$} (8.5,1.5);
\begin{scope}[xshift=8cm,yshift=2cm]
\node[opacity=0.2] at (2,1){\scriptsize Control Flow Net};
\node[place,label={80:},minimum size=3mm] (4i) at (0,0) {};
\node[transition,fill=black,minimum width=0.1mm,minimum height=10mm] (4) at (1,0) {};
\node[place,label={80:},minimum size=3mm] (4o) at (2,0) {};
\node[transition,fill=black,minimum width=0.1mm,minimum height=10mm] (3) at (3,0) {};
\node[place,label={80:},minimum size=3mm] (3o) at (4,0) {};
\draw[-latex,thick] (4i)--(4);\draw[-latex,thick](4)--(4o);\draw[-latex,thick](4o)--(3);\draw[-latex,thick](3)--(3o);
\end{scope}
\draw[opacity=0.3,line width=1.5pt, ->, -Latex] (7,0.2) to node[pos=0.4,yshift=-4,xshift=-5]{\scriptsize $\mathcal{D}$} (8.5,-0.5);
\begin{scope}[xshift=8cm,yshift=-1cm]
\node[opacity=0.2] at (2,1){\scriptsize Data Flow Net};
\node[place,label={80:},minimum size=3mm] (4i) at (0,0) {\scriptsize $1$};
\node[transition,fill=black,minimum width=0.1mm,minimum height=10mm] (4) at (1,0) {};
\node[place,label={80:},minimum size=3mm] (4o1) at (2,0.5) {\scriptsize $3$};
\node[place,label={80:},minimum size=3mm] (4o2) at (2,-0.5) {\scriptsize $4$};
\node[transition,fill=black,minimum width=0.1mm,minimum height=10mm] (3) at (3,0) {};
\node[place,label={80:},minimum size=3mm] (3o) at (4,0) {\scriptsize $1$};
\draw[-latex,thick] (4i)--(4);\draw[-latex,thick](4)--(4o1);\draw[-latex,thick](4o1)--(3);\draw[-latex,thick](3)--(3o);
\draw[-latex,thick] (4)--(4o2);\draw[-latex,thick](4o2)--(3);
\end{scope}
\draw[opacity=0.3,line width=1.5pt, ->, -Latex] (3.5,-0.9) to node[pos=0.4,yshift=-4,xshift=-10]{\scriptsize $\mathcal{N}$} (3.5,-1.8);
\begin{scope}[xshift=1.5cm,yshift=-2.6cm]
\node[opacity=0.2] at (-1.5,0.3){\scriptsize Control and};\node[opacity=0.2] at (-1.5,0){\scriptsize Data Flow};\node[opacity=0.2] at (-1.5,-0.3){\scriptsize Net};
\node[place,label={80:},minimum size=3mm] (4ic) at (0,0.5) {};
\node[place,label={80:},minimum size=3mm] (4i) at (0,-0.5) {\scriptsize $1$};
\node[transition,fill=black,minimum width=0.1mm,minimum height=10mm] (4) at (1,0) {};
\node[place,label={80:},minimum size=3mm] (4o1) at (2,0.5) {};
\node[place,label={80:},minimum size=3mm] (4o1c) at (2,0) {\scriptsize $3$};
\node[place,label={80:},minimum size=3mm] (4o2) at (2,-0.5) {\scriptsize $4$};
\node[transition,fill=black,minimum width=0.1mm,minimum height=10mm] (3) at (3,0) {};
\node[place,label={80:},minimum size=3mm] (3o) at (4,0.5) {};
\node[place,label={80:},minimum size=3mm] (3oc) at (4,-0.5) {\scriptsize $1$};
\draw[-latex,thick] (4i)--(4);\draw[-latex,thick] (4ic)--(4);\draw[-latex,thick](4)--(4o1);\draw[-latex,thick](4)--(4o1c);\draw[-latex,thick](4o1)--(3);\draw[-latex,thick](3)--(3o);\draw[-latex,thick](3)--(3oc);
\draw[-latex,thick] (4)--(4o2);\draw[-latex,thick](4o2)--(3);\draw[-latex,thick](4o1c)--(3);
\end{scope}
\end{tikzpicture}
\caption{Sequential control flow and sequential data flow encapsulated by the total sequential computon from Figure \ref{fig:computon-sequential-total-example}(a). We label some places for mapping purposes even though Petri nets are not labelled (see Section \ref{sec:operational-semantics}).}
\label{fig:encapsulation-sequential-total}
\end{figure}

We know a partial sequential computon $\lambda_1\rhd_\rho\lambda_2$ is the pushout of a partially sequentiable span $\rho:=\lambda_1\xleftarrow{\alpha_1}\lambda_0\xrightarrow{\alpha_2}\lambda_2$ of computon morphisms. By Definition \ref{def:span-sequentiable}, some of the e-outports of $\lambda_1$ are matched with some of the e-inports of $\lambda_2$, with the characteristic there always is at least one ec-outport in $C_1^-$ identified with some ec-inport in $C_2^+$ (because the apex $\lambda_0$ is necessarily a trivial computon with at least one ec-inoutport --- see Definition \ref{def:computon}). For this reason, $\lambda_1\rhd_\rho\lambda_2$ also encapsulates sequential control flow and up to partial sequential data flow. By partial, we mean some data elements are passed from the left operand to the right one. To give a concrete example, Figure \ref{fig:encapsulation-sequential-partial} shows the encapsulation given by the partial sequential computon that results from the pushout construction depicted in Figure \ref{fig:computon-sequential-example}.

\begin{figure}[!h]
\centering
\begin{tikzpicture}
\begin{scope}
\computonComposite{0.95}{-1}{5.1}{2.6};
\qmatch{3q}{3.5}{1.2}{};\flow{{2.5,1.2}}{{3.5,1.2}}{dashed}{};\flow{{3.6,1.2}}{{4.6,1.2}}{dashed}{};
\dmatch{3d}{3.5}{0.6}{$3$}{above};\flow{{2.5,0.6}}{{3.5,0.6}}{}{};\flow{{3.56,0.6}}{{4.6,0.6}}{}{};
\computonPrimitive{1.5}{-0.5}{1}{1.9}{$\lambda_1$};
\qin{3q0}{0.7}{1.2}{};
\din{3i1}{0.7}{0.6}{$1$};
\din{3i2}{0.7}{0}{$2$};
\node[circle,fill=black,inner sep=0pt,minimum size=3pt,label=right:\scriptsize $4$] (3o2) at (6.5,-0.6) {};\flow{{2.5,0}}{3o2}{}{bend right=20};
\computonPrimitive{4.5}{-0.5}{1}{1.9}{$\lambda_2$}
\node[circle,draw=black,fill=white,inner sep=0pt,minimum size=3pt,label=left:\scriptsize $4$] (3i3) at (0.5,-0.6) {};\flow{3i3}{{4.5,0}}{}{bend right=20};
\qout{3q1}{5.5}{1.2}{};
\dout{3o1}{5.5}{0.6}{$5$};
\end{scope}
\draw[opacity=0.3,line width=1.5pt, ->, -Latex] (7,0.5) to node[pos=0.4,yshift=7,xshift=-6]{\scriptsize $\mathcal{C}\circ \mathfrak{E}$} (8.5,1.5);
\begin{scope}[xshift=8cm,yshift=2cm]
\node[opacity=0.2] at (2,1){\scriptsize Control Flow Net};
\node[place,label={80:},minimum size=3mm] (4i) at (0,0) {};
\node[transition,fill=black,minimum width=0.1mm,minimum height=10mm] (4) at (1,0) {};
\node[place,label={80:},minimum size=3mm] (4o) at (2,0) {};
\node[transition,fill=black,minimum width=0.1mm,minimum height=10mm] (3) at (3,0) {};
\node[place,label={80:},minimum size=3mm] (3o) at (4,0) {};
\draw[-latex,thick] (4i)--(4);\draw[-latex,thick](4)--(4o);\draw[-latex,thick](4o)--(3);\draw[-latex,thick](3)--(3o);
\end{scope}
\draw[opacity=0.3,line width=1.5pt, ->, -Latex] (7,0.2) to node[pos=0.4,yshift=-4,xshift=-5]{\scriptsize $\mathcal{D}$} (8.5,-0.5);
\begin{scope}[xshift=8cm,yshift=-1cm]
\node[opacity=0.2] at (2,1){\scriptsize Data Flow Net};
\node[place,label={80:},minimum size=3mm] (4i1) at (0,0.3) {\scriptsize $1$};
\node[place,label={80:},minimum size=3mm] (4i2) at (0,-0.3) {\scriptsize $2$};
\node[place,label={80:},minimum size=3mm] (4i4) at (0,-0.8) {\scriptsize $4$};
\node[transition,fill=black,minimum width=0.1mm,minimum height=10mm] (4) at (1,0) {};
\node[place,label={80:},minimum size=3mm] (4o1) at (2,0) {\scriptsize $3$};
\node[transition,fill=black,minimum width=0.1mm,minimum height=10mm] (3) at (3,0) {};
\node[place,label={80:},minimum size=3mm] (3o5) at (4,0) {\scriptsize $5$};
\node[place,label={80:},minimum size=3mm] (3o4) at (4,-0.8) {\scriptsize $4$};
\draw[-latex,thick] (4i1)--(4);\draw[-latex,thick] (4i2)--(4);\draw[-latex,thick](4)--(4o1);\draw[-latex,thick](4o1)--(3);\draw[-latex,thick](3)--(3o5);
\draw[-latex,thick] (4i4) to[bend right=20] (3);\draw[-latex,thick](4) to[bend right=20] (3o4);
\end{scope}
\draw[opacity=0.3,line width=1.5pt, ->, -Latex] (3.5,-1.2) to node[pos=0.4,yshift=-4,xshift=-10]{\scriptsize $\mathcal{N}$} (3.5,-2.1);
\begin{scope}[xshift=1.5cm,yshift=-2.8cm]
\node[opacity=0.2] at (-1.5,0.3){\scriptsize Control and};\node[opacity=0.2] at (-1.5,0){\scriptsize Data Flow};\node[opacity=0.2] at (-1.5,-0.3){\scriptsize Net};
\node[place,label={80:},minimum size=3mm] (4i1c) at (0,0.5) {};
\node[place,label={80:},minimum size=3mm] (4i1) at (0,0) {\scriptsize $1$};
\node[place,label={80:},minimum size=3mm] (4i2) at (0,-0.5) {\scriptsize $2$};
\node[place,label={80:},minimum size=3mm] (4i4) at (0,-1) {\scriptsize $4$};
\node[transition,fill=black,minimum width=0.1mm,minimum height=10mm] (4) at (1,0) {};
\node[place,label={80:},minimum size=3mm] (4o1c) at (2,0.3) {};
\node[place,label={80:},minimum size=3mm] (4o1) at (2,-0.3) {\scriptsize $3$};
\node[transition,fill=black,minimum width=0.1mm,minimum height=10mm] (3) at (3,0) {};
\node[place,label={80:},minimum size=3mm] (3o5c) at (4,0.5) {};
\node[place,label={80:},minimum size=3mm] (3o5) at (4,-0.5) {\scriptsize $5$};
\node[place,label={80:},minimum size=3mm] (3o4) at (4,-1) {\scriptsize $4$};
\draw[-latex,thick] (4i1c)--(4);\draw[-latex,thick] (4i1)--(4);\draw[-latex,thick] (4i2)--(4);\draw[-latex,thick](4)--(4o1c);\draw[-latex,thick](4)--(4o1);\draw[-latex,thick](4o1c)--(3);\draw[-latex,thick](4o1)--(3);\draw[-latex,thick](3)--(3o5);
\draw[-latex,thick] (4i4) to[bend right=20] (3);\draw[-latex,thick](3)--(3o5c);\draw[-latex,thick](4) to[bend right=20] (3o4);
\end{scope}
\end{tikzpicture}
\caption{Sequential control flow and partial sequential data flow encapsulated by the partial sequential computon from Figure \ref{fig:computon-sequential-example}. We label some places for mapping purposes even though Petri nets are not labelled (see Section \ref{sec:operational-semantics}).}
\label{fig:encapsulation-sequential-partial}
\end{figure}

In Figure \ref{fig:encapsulation-sequential-partial}, it is easy to observe that not all data from $\lambda_1$ is passed to $\lambda_2$, even though control flow is necessarily sequential. In other words, data does not always follow control within a partial sequential computon, as a result of identifying some ed-outports of the left operand with some ed-inports of the right one. For example, in our scenario, only a $3$-coloured data item is passed from $\lambda_1$ to $\lambda_2$, whilst $\lambda_2$/$\lambda_1$ consumes/produces a $4$-coloured data item from/for the external environment. 

\subsection{Parallel Computons} \label{sec:parallel-computons}

In this subsection, we present two major classes of parallel computons, \emph{p-async} and \emph{p-sync}, which allow the asynchronous and synchronous execution of connected computons, respectively. These two classes serve to capture parallel processes that do not interfere with one another.

\subsubsection{Asynchronous Parallel Computons}

A \emph{p-async computon} intutitively permits the independent, simultaneous execution of two connected computons, without the need of forking or synchronizing control. Its formal notion is given in Definition \ref{def:computon-parallel-async}.

\begin{definition}\label{def:computon-parallel-async}
A p-async computon is the coproduct of two connected computons.
\end{definition}

Definition \ref{def:computon-parallel-async} implies that a p-async composite puts two connected computons side by side by offering multiple ec-inports to trigger some control-driven computation concurrently. Figure \ref{fig:computon-parallel-async-example} depicts a self-descriptive example for the construction of a p-async computon in which the connected computons being put in parallel are the same we used in Figure \ref{fig:computon-sequential-example}. This example demonstrates a particular feature of our theory, which is to allow the composition of the same connected computons into sequential or p-async composite structures, no matter the data such computons require or produce. Theorem \ref{th:computon-parallel-async-always} generalises this assertion by stating that two arbitrary connected computons are sufficient and necessary to form a p-async computon which, by Proposition \ref{prop:computon-parallel-async-connected}, is always connected. As per Propositions \ref{prop:computon-parallel-async-commutative} and \ref{prop:computon-parallel-async-associative}, the operation for forming a p-async computon is both commutative and associative. 

\begin{figure}[!h]
\centering
{
\begin{tikzpicture}
\computonPrimitive{1}{8.2}{1}{1.9}{$\lambda_1$};
\qin{1q0}{0.2}{9.9}{};
\din{1i1}{0.2}{9.3}{$1$};
\din{1i2}{0.2}{8.7}{$2$};
\qout{1q1}{2}{9.9}{};
\dout{1o1}{2}{9.3}{$3$};
\dout{1o2}{2}{8.7}{$4$};

\computonPrimitive{5}{8.2}{1}{1.9}{$\lambda_2$};
\qin{2q0}{4.2}{9.9}{};
\din{2i1}{4.2}{9.3}{$3$};
\din{2i2}{4.2}{8.7}{$4$};
\qout{2q1}{6}{9.9}{};
\dout{2o1}{6}{9.3}{$5$};

\draw[->] (1.5,8) to[bend right=25] node[left]{\scriptsize$\beta_1$}  (1.80,7);
\draw[->] (5.5,8) to[bend left=25] node[right]{\scriptsize$\beta_2$}  (5.2,7);

\computonComposite{2.45}{2.9}{2.2}{4.5};
\computonPrimitive{3}{5.3}{1}{1.9}{$\lambda_1$};
\qin{1q0}{2.2}{7}{};
\din{1i1}{2.2}{6.4}{$1$};
\din{1i2}{2.2}{5.8}{$2$};
\qout{1q1}{4}{7}{};
\dout{1o1}{4}{6.4}{$3$};
\dout{1o2}{4}{5.8}{$4$};
\computonPrimitive{3}{3}{1}{1.9}{$\lambda_2$};
\qin{1q0}{2.2}{4.7}{};
\din{1i1}{2.2}{4.1}{$3$};
\din{1i2}{2.2}{3.5}{$4$};
\qout{1q1}{4}{4.7}{};
\dout{1o1}{4}{4.1}{$5$};
\end{tikzpicture}
}
{
\begin{tikzpicture}
\begin{scope}
\draw[draw=black,fill=white] (0,0.4) rectangle ++(0.12,0.25);
\node[inner sep=0pt,minimum size=0pt,label=right:{\scriptsize Computation unit}] at (0.2,0.5) {};
\draw[dashed] (3.2,0.5) to node[pos=0.5]{\arrowflow} (3.8,0.5);
\node[inner sep=0pt,minimum size=3pt,label=right:{\scriptsize Control flow edge}] at (3.8,0.5) {};
\node[draw=black,fill=white,inner sep=0pt,minimum size=3pt,label=right:{\scriptsize Ec-inport}] at (6.8,0.5) {};
\node[fill=black,inner sep=0pt,minimum size=3pt,label={right:{\scriptsize Ec-outport}}] at (8.8,0.5) {};
\end{scope}

\begin{scope}[xshift=1.5cm]
\draw (0,0) to node[pos=0.5]{\arrowflow} (0.5,0);
\node[inner sep=0pt,minimum size=3pt,label=right:{\scriptsize Data flow edge}] at (0.5,0){};
\node[circle,draw=black,fill=white,inner sep=0pt,minimum size=3pt,label={right:{\scriptsize Ed-inport}}] at (3.3,0) {};
\node[circle,fill=black,inner sep=0pt,minimum size=3pt,label={right:{\scriptsize Ed-outport}}] at (5.3,0) {};
\end{scope}
\end{tikzpicture}
}
\caption{Constructing a p-async computon $\lambda_1+\lambda_2$ where $\lambda_1$ and $\lambda_2$ are isomorphic to the respective left and right operands presented in Figure \ref{fig:computon-sequential-example}.}
\label{fig:computon-parallel-async-example}
\end{figure}

\begin{theorem}\label{th:computon-parallel-async-always}
$\lambda_1$ and $\lambda_2$ are connected computons $\iff$ the p-async computon $\lambda_1+\lambda_2$ exists.
\end{theorem}
\begin{proof}
The implication and reverse implication follow from Proposition \ref{prop:computon-coproduct} and Definition \ref{def:computon-parallel-async}, respectively.
\end{proof}

\begin{proposition}\label{prop:computon-parallel-async-connected}
A p-async computon is a connected computon. 
\end{proposition}
\begin{proof}
The proof trivially follows from Definition \ref{def:computon-parallel-async} and Proposition \ref{prop:computon-coproduct-connected}.
\end{proof}

\begin{proposition}[Asynchronous parallel composition is commutative]\label{prop:computon-parallel-async-commutative}
There is an isomorphism between $\lambda_1+\lambda_2$ and $\lambda_2+\lambda_1$ for any (connected) computons $\lambda_1$ and $\lambda_2$.
\end{proposition}
\begin{proof}
The proof follows directly from the well-known fact that categorical coproduct is commutative up to unique isomorphism.
\end{proof}

\begin{proposition}[Asynchronous parallel composition is associative]\label{prop:computon-parallel-async-associative}
There is an isomorphism between $(\lambda_1+\lambda_2)+\lambda_3$ and $\lambda_1+(\lambda_2+\lambda_3)$ for any (connected) computons $\lambda_1$, $\lambda_2$ and $\lambda_3$.
\end{proposition}
\begin{proof}
The proof follows directly from the well-known fact that categorical coproduct is associative up to unique isomorphism.
\end{proof}

\subsubsection{Operational semantics for p-async computons (in the theory of Petri nets)}

No matter whether we use any of the three functors from Section \ref{sec:operational-semantics}, the Petri net of a p-async computon does not introduce any additional places or transitions and the nets of the operands do not interact in any way. This structural organisation, depicted in Figure \ref{fig:computon-parallel-async-net}, results from defining a p-async computon in the form of a coproduct construction. By Proposition \ref{prop:computon-parallel-async-deadlock} and Remark \ref{rem:parallel-asynchronous-deadlock}, any p-async's net is deadlock free when the nets of the composed computons are too. 

\begin{figure}[!h]
\centering
{
\begin{tikzpicture}
\node[place,label={left:\scriptsize $p_n$},minimum size=3mm] (pn) at (0,2) {};
\node at (0,2.6){$\vdots$};
\node[place,label={left:\scriptsize $p_1$},minimum size=3mm] (p1) at (0,3) {};
\draw[dotted] (0.7,1.7) rectangle (2.2,3.2);\node at (1.4,2.4){\scriptsize $\lambda_1$-net};
\node[place,label={right:\scriptsize $q_1$},minimum size=3mm] (q1) at (2.9,3) {};
\node at (2.9,2.6){$\vdots$};
\node[place,label={right:\scriptsize $q_j$},minimum size=3mm] (qj) at (2.9,2) {};

\node[place,label={left:\scriptsize $r_m$},minimum size=3mm] (rm) at (0,0) {};
\node at (0,0.6){$\vdots$};
\node[place,label={left:\scriptsize $r_1$},minimum size=3mm] (r1) at (0,1) {};
\draw[dotted] (0.7,-0.3) rectangle (2.2,1.2);\node at (1.4,0.4){\scriptsize $\lambda_2$-net};
\node[place,label={right:\scriptsize $s_1$},minimum size=3mm] (s1) at (2.9,1) {};
\node at (2.9,0.6){$\vdots$};
\node[place,label={right:\scriptsize $s_k$},minimum size=3mm] (sk) at (2.9,0) {};

\draw[-latex,thick] (p1) -- ($(p1)+(0.7,-0.2)$);
\draw[-latex,thick] (pn) -- ($(pn)+(0.7,0.2)$);
\draw[-latex,thick] ($(q1)+(-0.7,-0.2)$) -- (q1);
\draw[-latex,thick] ($(qj)+(-0.7,0.2)$) -- (qj);
\draw[-latex,thick] (r1) -- ($(r1)+(0.7,-0.2)$);
\draw[-latex,thick] (rm) -- ($(rm)+(0.7,0.2)$);
\draw[-latex,thick] ($(s1)+(-0.7,-0.2)$) -- (s1);
\draw[-latex,thick] ($(sk)+(-0.7,0.2)$) -- (sk);
\end{tikzpicture}
}
\caption{General structure of the Petri net of a p-async computon ${\lambda_1+\lambda_2}$ constructed from a connected computon $\lambda_1$ with $n$ e-inports and $j$ e-outports, and a connected computon $\lambda_2$ with $m$ e-inports and $k$ e-outports. This structure is applicable to all the functorial constructions presented in Section \ref{sec:operational-semantics}, namely $\mathcal{N}$, $\mathcal{C}\circ\mathfrak{E}$ and $\mathcal{D}$.}
\label{fig:computon-parallel-async-net}
\end{figure}

\begin{proposition}\label{prop:computon-parallel-async-deadlock}
If $\mathcal{N}(\lambda_1)$ and $\mathcal{N}(\lambda_2)$ are deadlock-free, $\mathcal{N}(\lambda_1+\lambda_2)$ is deadlock-free.
\end{proposition}
\begin{proof}
If the net from Figure \ref{fig:computon-parallel-async-net} corresponds to ${\mathcal{N}(\lambda_1+\lambda_2)}$, Definition \ref{def:marking} says the initial state ${M_i}$ of such a net is a marking function where ${M_i(p)>0}$ for all ${p\in\{p_1,\ldots,p_n,r_1,\ldots,r_m\}}$ and no tokens for all the other places, including those inside ${\mathcal{N}(\lambda_1)}$ and ${\mathcal{N}(\lambda_2)}$. As this marking evidently reaches states from ${\mathcal{N}(\lambda_1)}$ and ${\mathcal{N}(\lambda_2)}$ in parallel, we have that ${\mathcal{N}(\lambda_1+\lambda_2)}$ is deadlock-free whenever ${\mathcal{N}(\lambda_1)}$ and ${\mathcal{N}(\lambda_2)}$ are.
\end{proof}

\begin{remark}\label{rem:parallel-asynchronous-deadlock}
Although it is a statement about $\mathcal{N}$, Proposition \ref{prop:computon-parallel-async-deadlock} is applicable to the functors $\mathcal{C}\circ\mathfrak{E}$ and $\mathcal{D}$ presented in Section \ref{sec:operational-semantics}. The proof is valid for $\mathcal{C}\circ\mathfrak{E}$ since Proposition \ref{prop:functor-control-petri} says ${\mathcal{C}}$ is just a restriction of $\mathcal{N}$ to $\mathfrak{E}(\textbf{Set}^{\textbf{Comp}})$. 

For $\mathcal{D}$, we are only interested in nets with initial and final states (see Remark \ref{rem:deadlock-freedom}). Thus, considering the form depicted in Figure \ref{fig:computon-parallel-async-net}, we have the following cases for a net $\mathcal{D}(\lambda_1+\lambda_2)$:
\begin{enumerate}
\item If $m=0$ and $n>0$, or $m>0$ and $n=0$, only states of $\mathcal{D}(\lambda_1)$ or $\mathcal{D}(\lambda_2)$ are reached from $M_i$. Since both nets are deadlock-free, $\mathcal{D}(\lambda_1+\lambda_2)$ is deadlock-free. 
\item If $m>0$ and $n>0$, states of $\mathcal{D}(\lambda_1)$ and $\mathcal{D}(\lambda_2)$ are simultaneously reached from $M_i$. As both nets are deadlock-free, $\mathcal{D}(\lambda_1+\lambda_2)$ is deadlock-free.
\end{enumerate}
\end{remark}

\subsubsection{Encapsulation of control flow and data flow in p-async computons}

By Definition \ref{def:computon-parallel-async}, we know a p-async computon results from a coproduct construction built upon disjoint union. Consequently, there are no ports of one computon identified with ports of the other, meaning there is no way of structurally exchanging either data or control. For this reason, a p-async computon encapsulates asynchronous parallel control flow and up to asynchronous parallel data flow. To give a concrete example, Figure \ref{fig:encapsulation-pasync} shows the encapsulation given by the p-async computon from Figure \ref{fig:computon-parallel-async-example}.

\begin{figure}[!h]
\centering
\begin{tikzpicture}
\begin{scope}[xshift=0cm,yshift=-4cm]
\computonComposite{2.45}{2.9}{2.2}{4.5};
\computonPrimitive{3}{5.3}{1}{1.9}{$\lambda_1$};
\qin{1q0}{2.2}{7}{};
\din{1i1}{2.2}{6.4}{$1$};
\din{1i2}{2.2}{5.8}{$2$};
\qout{1q1}{4}{7}{};
\dout{1o1}{4}{6.4}{$3$};
\dout{1o2}{4}{5.8}{$4$};
\computonPrimitive{3}{3}{1}{1.9}{$\lambda_2$};
\qin{1q0}{2.2}{4.7}{};
\din{1i1}{2.2}{4.1}{$3$};
\din{1i2}{2.2}{3.5}{$4$};
\qout{1q1}{4}{4.7}{};
\dout{1o1}{4}{4.1}{$5$};
\end{scope}
\draw[opacity=0.3,line width=1.5pt, ->, -Latex] (5.5,1.5) to node[pos=0.4,yshift=7,xshift=-8]{\scriptsize $\mathcal{C}\circ \mathfrak{E}$} (6.7,3.5);
\begin{scope}[xshift=7.2cm,yshift=3cm]
\node[opacity=0.2] at (3.4,0.9){\scriptsize Control Flow};\node[opacity=0.2] at (3.4,0.6){\scriptsize Net};
\node[place,label={80:},minimum size=3mm] (4i) at (0,1.3) {};
\node[transition,fill=black,minimum width=0.1mm,minimum height=10mm] (4) at (1,1.3) {};
\node[place,label={80:},minimum size=3mm] (4o) at (2,1.3) {};

\node[place,label={80:},minimum size=3mm] (4ix) at (0,0) {};
\node[transition,fill=black,minimum width=0.1mm,minimum height=10mm] (4x) at (1,0) {};
\node[place,label={80:},minimum size=3mm] (4ox) at (2,0) {};

\draw[-latex,thick] (4i)--(4);\draw[-latex,thick](4)--(4o);
\draw[-latex,thick] (4ix)--(4x);\draw[-latex,thick](4x)--(4ox);
\end{scope}
\draw[opacity=0.3,line width=1.5pt, ->, -Latex] (5.5,1.2) to node[pos=0.4,yshift=-4,xshift=-5]{\scriptsize $\mathcal{D}$} (6.7,0);
\begin{scope}[xshift=7.2cm,yshift=-0.7cm]
\node[opacity=0.2] at (3.4,0.9){\scriptsize Data Flow};\node[opacity=0.2] at (3.4,0.6){\scriptsize Net};
\node[place,label={80:},minimum size=3mm] (4i1) at (0,2) {\scriptsize $1$};
\node[place,label={80:},minimum size=3mm] (4i2) at (0,1) {\scriptsize $2$};
\node[transition,fill=black,minimum width=0.1mm,minimum height=10mm] (4) at (1,1.5) {};
\node[place,label={80:},minimum size=3mm] (4o3) at (2,2) {\scriptsize $3$};
\node[place,label={80:},minimum size=3mm] (4o4) at (2,1) {\scriptsize $4$};

\node[place,label={80:},minimum size=3mm] (4ix3) at (0,0.5) {\scriptsize $3$};
\node[place,label={80:},minimum size=3mm] (4ix4) at (0,-0.5) {\scriptsize $4$};
\node[transition,fill=black,minimum width=0.1mm,minimum height=10mm] (4x) at (1,0) {};
\node[place,label={80:},minimum size=3mm] (4ox) at (2,0) {\scriptsize $5$};

\draw[-latex,thick] (4i1)--(4);\draw[-latex,thick] (4i2)--(4);\draw[-latex,thick](4)--(4o3);\draw[-latex,thick](4)--(4o4);
\draw[-latex,thick] (4ix3)--(4x);\draw[-latex,thick] (4ix4)--(4x);\draw[-latex,thick](4x)--(4ox);
\end{scope}
\draw[opacity=0.3,line width=1.5pt, ->, -Latex] (2,1.2) to node[pos=0.4,yshift=6]{\scriptsize $\mathcal{N}$} (0.8,1.2);
\begin{scope}[xshift=-1.5cm,yshift=0.5cm]
\node[opacity=0.2] at (-1.7,1){\scriptsize Control and};\node[opacity=0.2] at (-1.7,0.7){\scriptsize Data Flow};\node[opacity=0.2] at (-1.7,0.4){\scriptsize Net};
\node[place,label={80:},minimum size=3mm] (4i1) at (0,2) {};
\node[place,label={80:},minimum size=3mm] (c) at (0,0) {\scriptsize $3$};
\node[place,label={80:},minimum size=3mm] (4i2) at (0,1) {\scriptsize $2$};
\node[transition,fill=black,minimum width=0.1mm,minimum height=10mm] (4) at (1,1.5) {};
\node[place,label={80:},minimum size=3mm] (d) at (2,0.3) {};
\node[place,label={80:},minimum size=3mm] (4o3) at (2,2) {};
\node[place,label={80:},minimum size=3mm] (4o4) at (2,1) {\scriptsize $4$};

\node[place,label={80:},minimum size=3mm] (e) at (0,1.5) {\scriptsize $1$};
\node[place,label={80:},minimum size=3mm] (4ix3) at (0,0.5) {};
\node[place,label={80:},minimum size=3mm] (4ix4) at (0,-0.5) {\scriptsize $4$};
\node[transition,fill=black,minimum width=0.1mm,minimum height=10mm] (4x) at (1,0) {};
\node[place,label={80:},minimum size=3mm] (f) at (2,1.5) {\scriptsize $3$};
\node[place,label={80:},minimum size=3mm] (4ox) at (2,-0.3) {\scriptsize $5$};

\draw[-latex,thick] (4i1)--(4);\draw[-latex,thick] (4i2)--(4);\draw[-latex,thick](4)--(4o3);\draw[-latex,thick](4)--(4o4);
\draw[-latex,thick] (4ix3)--(4x);\draw[-latex,thick] (4ix4)--(4x);\draw[-latex,thick](4x)--(4ox);
\draw[-latex,thick] (e)--(4);\draw[-latex,thick] (c)--(4x);\draw[-latex,thick] (4x)--(d);\draw[-latex,thick] (4)--(f);
\end{scope}
\end{tikzpicture}
\caption{Asynchronous parallel control flow and asynchronous parallel data flow encapsulated by the p-async computon from Figure \ref{fig:computon-parallel-async-example}. We label some places for mapping purposes even though Petri nets are not labelled (see Section \ref{sec:operational-semantics}).}
\label{fig:encapsulation-pasync}
\end{figure}

\subsubsection{Synchronous Parallel Computons}\label{sec:p-sync-sub}

Structurally, a \emph{p-sync computon} consists of a fork computon, two arbitrary connected computons and a join computon. The roles of the fork and join are to split into and synchronise control from the connected computons, respectively. Formally, a p-sync computon is constructed from a so-called \emph{p-diagram} which satisfies the requirements imposed by Definition \ref{def:diagram-p}. Such a composite computon is defined as a colimit in $\textbf{Set}^\textbf{Comp}$ which, by Lemma \ref{lem:computon-parallel-sync-exists}, can always be computed via coproduct and pushout constructions.

\begin{definition}[P-Diagram]\label{def:diagram-p}
A p-diagram $\rho$ is a diagram with the following shape in $\textbf{Set}^\textbf{Comp}$:
\[
\begin{tikzcd}
 & & \lambda_4 & & \\
 & \lambda_0 \arrow[ur, "\alpha_0"]\arrow[dl, "\alpha_2"'] & & \lambda_1 \arrow[ul, "\alpha_1"']\arrow[dr, "\alpha_3"] & \\ 
\lambda_6 & & & & \lambda_7 \\
 & \lambda_2 \arrow[ul, "\alpha_4"]\arrow[dr, "\alpha_6"'] & & \lambda_3 \arrow[ur, "\alpha_5"']\arrow[dl, "\alpha_7"] & \\ 
 & & \lambda_5 & & 
\end{tikzcd}
\]
where:
\begin{enumerate}
\item $\lambda_0$, $\lambda_1$, $\lambda_2$ and $\lambda_3$ are unit computons (see Definition \ref{def:computon-unit}), \label{def:computon-parallel-sync-diagram-1}
\item $\lambda_4$ and $\lambda_5$ are connected computons, \label{def:computon-parallel-sync-diagram-2}
\item $\lambda_6$ is a fork computon, \label{def:computon-parallel-sync-diagram-3}
\item $\lambda_7$ is a join computon, \label{def:computon-parallel-sync-diagram-4}
\item ${\lambda_6 \xleftarrow{\alpha_2} \lambda_0 \xrightarrow{\alpha_0} \lambda_4}$, ${\lambda_4 \xleftarrow{\alpha_1} \lambda_1 \xrightarrow{\alpha_3} \lambda_7}$, ${\lambda_5 \xleftarrow{\alpha_7} \lambda_3 \xrightarrow{\alpha_5} \lambda_7}$ and ${\lambda_6 \xleftarrow{\alpha_4} \lambda_2 \xrightarrow{\alpha_6} \lambda_5}$ are partially sequentiable spans of computon morphisms, \label{def:computon-parallel-sync-diagram-5}  
\item $\alpha_2(P_0) \neq \alpha_4(P_2)$ and \label{def:computon-parallel-sync-diagram-6}
\item $\alpha_3(P_1) \neq \alpha_5(P_3)$. \label{def:computon-parallel-sync-diagram-7}
\end{enumerate}
\end{definition}

\begin{definition}[P-Sync Computon]\label{def:computon-parallel-sync}
A p-sync computon is the colimit of a p-diagram.
\end{definition}

\begin{notation} \label{notation:parallel-sync}
For convenience, we use a pipe to reflect the fact that two connected computons are being put into a synchronous parallel structure. For example, we write ${\lambda_{4} \mid_{\rho} \lambda_{5}}$ for the p-sync computon obtained by computing the colimit of the p-diagram $\rho$ shown in Definition \ref{def:diagram-p}. 
\end{notation}

\begin{lemma}\label{lem:computon-parallel-sync-exists}
A p-sync computon can always be constructed in $\textbf{Set}^\textbf{Comp}$.
\end{lemma}
\begin{proof}
Considering the p-diagram shown in Definition \ref{def:diagram-p}, let ${(\beta_1:\lambda_6 \rightarrow \lambda_8, \lambda_8, \beta_2:\lambda_4 \rightarrow \lambda_8)}$, ${(\beta_3:\lambda_4 \rightarrow \lambda_9, \lambda_9, \beta_4:\lambda_7 \rightarrow \lambda_9)}$, ${(\beta_5:\lambda_6 \rightarrow \lambda_{10}, \lambda_{10}, \beta_6:\lambda_5 \rightarrow \lambda_{10})}$ and ${(\beta_7:\lambda_5 \rightarrow \lambda_{11}, \lambda_{11},} \beta_8:\lambda_7 \rightarrow \lambda_{11})$ be the respective pushouts of the partially sequentiable spans ${\lambda_6 \xleftarrow{\alpha_2} \lambda_0 \xrightarrow{\alpha_0} \lambda_4}$, ${\lambda_4 \xleftarrow{\alpha_1} \lambda_1 \xrightarrow{\alpha_3} \lambda_7}$, ${\lambda_6 \xleftarrow{\alpha_4} \lambda_2 \xrightarrow{\alpha_6} \lambda_5}$ and ${\lambda_5 \xleftarrow{\alpha_7} \lambda_3 \xrightarrow{\alpha_5} \lambda_7}$. By Definition \ref{def:computon-sequential} and Lemma \ref{lem:span-sequentiable-pushable}, we know such pushouts can be constructed to yield partial sequential computons. For example, $\lambda_8$ corresponds to the partial sequential computon $\lambda_6\rhd_{\rho_1}\lambda_4$ constructed from the partially sequentiable span ${\rho_1:=\lambda_6 \xleftarrow{\alpha_2} \lambda_0 \xrightarrow{\alpha_0} \lambda_4}$.

We now show that the induced span ${\lambda_8 \xleftarrow{\beta_2} \lambda_4 \xrightarrow{\beta_3} \lambda_9}$ is pushable. For this, let $p_8 \in \beta_2(\vec{o}(\beta_3))$ so there exists some ${p_4 \in \vec{o}(\beta_3)}$ where ${\beta_2(p_4)=p_8}$. Since $p_4 \in \vec{o}(\beta_3)$ and $\beta_3$ is induced by the pushout of the partially sequentiable span ${\lambda_4 \xleftarrow{\alpha_1} \lambda_1 \xrightarrow{\alpha_3} \lambda_7}$, we have $\alpha_1(p_1)=p_4 \in P_4^-$ for the unique port $p_1 \in P_1$ (recall $\lambda_1$ is a unit computon and $\lambda_4$ is an arbitary connected computon --- see Conditions \ref{def:computon-parallel-sync-diagram-1} and \ref{def:computon-parallel-sync-diagram-2} of Definition \ref{def:diagram-p}). By Proposition \ref{prop:computon-sequential-inports-outports-1}, ${\beta_2(P_4^-) \subseteq P_8^-}$ because $\lambda_8$ is the partial sequential computon $\lambda_6\rhd_{\rho_1}\lambda_4$. So, $\beta_2(p_4)=p_8 \in P_8^-$. 

As the other conditions of Definition \ref{def:computon-morphisms-pushable} follow analogously, we have that the pushout ${(\beta_9:\lambda_8 \rightarrow \lambda_{12}, \lambda_{12}, \beta_{10}:\lambda_9 \rightarrow \lambda_{12})}$ of ${\lambda_8 \xleftarrow{\beta_2} \lambda_4 \xrightarrow{\beta_3} \lambda_9}$ can be constructed. A similar approach can be used to prove the existence of the pushout ${(\beta_{11}:\lambda_{10} \rightarrow \lambda_{13}, \lambda_{13}, \beta_{12}:\lambda_{11} \rightarrow \lambda_{13})}$ of the induced span ${\lambda_{10} \xleftarrow{\beta_6} \lambda_5 \xrightarrow{\beta_7} \lambda_{11}}$. 

Now, Proposition \ref{prop:computon-coproduct} says that $\lambda_6+\lambda_7$ can be formed. By the universal property of coproducts, we deduce there are unique computon morphisms ${(\beta_{9}\circ\beta_1,\beta_{10}\circ\beta_4):\lambda_6+\lambda_7 \rightarrow \lambda_{12}}$ and ${(\beta_{11}\circ\beta_5,\beta_{12}\circ\beta_8):\lambda_6+\lambda_7 \rightarrow \lambda_{13}}$. Considering Conditions \ref{def:computon-parallel-sync-diagram-6} and \ref{def:computon-parallel-sync-diagram-7} of Definition \ref{def:diagram-p}, it is routine to check that the pushout $\lambda_{14}$ of ${\lambda_{13} \xleftarrow{(\beta_{11}\circ\beta_5,\beta_{12}\circ\beta_8)} \lambda_6+\lambda_7 \xrightarrow{(\beta_{9}\circ\beta_1,\beta_{10}\circ\beta_4)} \lambda_{12}}$ can be constructed. 

To check that $\lambda_{14}$ is indeed the colimit of the p-diagram shown in Definition \ref{def:diagram-p}, consider the following cone:
\vspace{-7pt}
\[
\begin{tikzcd}[row sep=scriptsize, column sep=scriptsize]
 & & \lambda_4 \arrow[ddddr, opacity=0.4] & & & \lambda_1 \arrow[lll, crossing over, "\alpha_1"']\arrow[rrr, crossing over, "\alpha_3"]\arrow[ddddll, opacity=0.4] & & & \lambda_7 \arrow[ddddlllll, bend right=6, opacity=0.4] \\
 & \lambda_0 \arrow[ur, crossing over, "\alpha_0"]\arrow[dl, crossing over, "\alpha_2"']\arrow[dddrr, opacity=0.4] & & & & & & \lambda_3 \arrow[ur, crossing over, "\alpha_5"']\arrow[dl, crossing over, "\alpha_7"]\arrow[dddllll, bend left=2, opacity=0.4] & \\
\lambda_6 \arrow[ddrrr, opacity=0.4] & & & \lambda_2 \arrow[lll, crossing over, "\alpha_4"]\arrow[rrr, crossing over, "\alpha_6"']\arrow[dd, opacity=0.4] & & & \lambda_5 \arrow[ddlll, bend left=10, opacity=0.4] & & \\
 & & & & & & & & \\
 & & & \lambda_{15} & & & & &
\end{tikzcd}
\]
Since ${\lambda_{8}}$, ${\lambda_{9}}$, ${\lambda_{10}}$ and ${\lambda_{11}}$ are pushouts of the spans of the original p-diagram, by the universal property of pushouts, it is true that there are unique computon morphisms ${\lambda_{8} \rightarrow \lambda_{15}}$, ${\lambda_{9} \rightarrow \lambda_{15}}$, ${\lambda_{10} \rightarrow \lambda_{15}}$ and ${\lambda_{11} \rightarrow \lambda_{15}}$ that make the corresponding diagram commute. Likewise, as such morphisms exist, there also are unique computon morphisms from the respective pushouts of ${\lambda_8 \xleftarrow{\beta_2} \lambda_4 \xrightarrow{\beta_3} \lambda_9}$ and ${\lambda_{10} \xleftarrow{\beta_6} \lambda_5 \xrightarrow{\beta_7} \lambda_{11}}$. That is, $\lambda_{12} \rightarrow \lambda_{15}$ and $\lambda_{13} \rightarrow \lambda_{15}$ exist. 

In the above cone, it is clear there are morphisms ${\lambda_6 \rightarrow \lambda_{15}}$ and ${\lambda_7 \rightarrow \lambda_{15}}$. Using the universal property of coproducts, we deduce the existence of a unique computon morphism ${\lambda_6+\lambda_7 \rightarrow\lambda_{15}}$. Finally, as ${\lambda_{12} \rightarrow \lambda_{15}}$ and ${\lambda_{13} \rightarrow \lambda_{15}}$ exist, we use the universal property of pushouts again to deduce there is a unique computon morphism ${\lambda_{14} \rightarrow \lambda_{15}}$. As ${\lambda_{14} \rightarrow \lambda_{15}}$ makes everything commute in our construction, it is true that ${\lambda_{14}}$ is the colimit of the original p-diagram.
\end{proof}

\begin{corollary} \label{cor:computon-parallel-sync-connected}
A p-sync computon is a connected computon.
\end{corollary}
\begin{proof}
Consider the construction presented in the proof of Lemma \ref{lem:computon-parallel-sync-exists}. By Proposition \ref{prop:pushout-connected}, we have that $\lambda_8$, $\lambda_9$, $\lambda_{10}$ and $\lambda_{11}$ are connected computons because $\lambda_4$, $\lambda_5$, $\lambda_6$ and $\lambda_7$ also are (recall forks and joins are primitive computons which, by Proposition \ref{prop:computon-primitive-connected}, adhere to Definition \ref{def:computon-connected}). Consequently, the pushouts $\lambda_{12}$ and $\lambda_{13}$ (of the induced spans ${\lambda_8 \xleftarrow{\beta_2} \lambda_4 \xrightarrow{\beta_3} \lambda_9}$ and ${\lambda_{10} \xleftarrow{\beta_6} \lambda_5 \xrightarrow{\beta_7} \lambda_{11}}$, respectively) are connected computons too. Using Proposition \ref{prop:pushout-connected} again, we deduce that the pushout $\lambda_{14}$ of the unique span ${\lambda_{13} \xleftarrow{(\beta_{11}\circ\beta_5,\beta_{12}\circ\beta_8)} \lambda_6+\lambda_7 \xrightarrow{(\beta_{9}\circ\beta_1,\beta_{10}\circ\beta_4)} \lambda_{12}}$ is a connected computon. As $\lambda_{14}$ is the colimit of the original p-diagram (shown in Definition \ref{def:diagram-p}), we conclude that every p-sync computon is a connected computon.
\end{proof}

To elucidate the proof of Lemma \ref{lem:computon-parallel-sync-exists}, Figure \ref{fig:computon-parallel-sync-example} provides a complete, self-descriptive example for constructing a p-sync computon from the connected computons used as operands in one of our examples of partial sequential composition (see Figure \ref{fig:computon-sequential-example}). A glance at Figure \ref{fig:computon-parallel-sync-example} reveals that the initial p-diagram $\rho$ (displayed in the middle) specifies the basic building blocks for constructing a p-sync computon, namely four unit computons, a fork computon, a join computon and two connected computons (i.e., the computons being put in parallel). The construction starts by computing four pushout operations that produce a partial sequential computon each as per Definitions \ref{def:computon-sequential} and \ref{def:diagram-p} (see the squares marked with $S$). The induced morphisms of such pushouts form two spans whose respective pushouts freely behave as in Definition \ref{def:pushout-computation}, i.e., they are pushouts that just ``merge'' computons via some common object (see the squares marked with $M$). In this case, such common objects are $\lambda_1$ and $\lambda_2$, respectively.

\begin{figure*}[!h]
\centering
\begin{tikzpicture}
\begin{scope}[scale=0.84]
\begin{scope}[xshift=0cm,yshift=11cm]\fork{0}{1.6}{}{}{};\join{0}{0.8}{}{}{};\end{scope}
\begin{scope}[xshift=2.6cm,yshift=11.3cm]\draw[->,opacity=0.4] (1.25,0) to node[right]{} (0,1);\end{scope} 
\begin{scope}[xshift=2.6cm,yshift=11cm]\draw[->,opacity=0.4] (5.5,0) to node[right]{} (0,1.7);\end{scope} 
\begin{scope}
\draw[->,dashed,opacity=0.4] (0.4,11.3) to node[right,pos=0.2]{\scriptsize $(\beta_{11}\circ\beta_5,\beta_{12}\circ\beta_8)$} (0.4,0.9) -- (3.1,0.9);
\end{scope}
\begin{scope}
\draw[->,dashed,opacity=0.4] (0.4,13.4) to node[right,pos=0.2]{\scriptsize $(\beta_{9}\circ\beta_{1},\beta_{10}\circ\beta_4)$} (0.4,21.7) -- (3.1,21.7);
\end{scope}
\begin{scope}
\draw[->,opacity=0.4] (11.1,21.7) -- (13.8,21.7) to node[right]{\scriptsize $\beta_{13}$} (13.8,13.6);
\end{scope}
\begin{scope}
\draw[->,opacity=0.4] (11.1,0.9) -- (13.8,0.9) to node[right]{\scriptsize $\beta_{14}$} (13.8,8.8);
\end{scope}

\begin{scope}[xshift=11cm,yshift=8.8cm]
  \computonComposite{0.5}{0.3}{6.4}{4.2}
  
	\forkplain{0}{2.5};
	\qin{fq0}{0.2}{2.75}{};
	\qmatch{fq1}{2}{3.2}{};\flow{fq1}{$(fq1)+(1,0)$}{dashed}{};\flowdiag{{1.3,2.9}}{fq1}{dashed}{}{pos=0.5,rotate=27};	
	\qmatch{fq2}{2}{1.9}{};\flow{fq2}{$(fq2)+(1,0)$}{dashed}{};\flowdiag{{1.3,2.6}}{fq2}{dashed}{}{pos=0.52,rotate=311};
	
	\joinplain{5}{2.5};	
	\qmatch{jq0}{5}{3.2}{};\flow{{4,3.2}}{jq0}{dashed}{};\flowdiag{jq0}{{6,2.9}}{dashed}{}{pos=0.52,rotate=350};
	\qmatch{jq1}{5}{1.9}{};\flow{{4,1.9}}{jq1}{dashed}{};\flowdiag{jq1}{{6,2.6}}{dashed}{}{pos=0.5,rotate=33};
	\qout{jq2}{6.3}{2.75}{};

	\computonPrimitive{3}{2.5}{1}{1.9}{$\lambda_1$};
	\dinplain{1i1}{0.2}{4.2}{$1$};\flow{1i1}{$(1i1)+(2.8,0)$}{}{};
	\dinplain{1i2}{0.2}{3.6}{$2$};\flow{1i2}{$(1i2)+(2.8,0)$}{}{};
	\doutplain{1o1}{6.3}{4.2}{$3$};\flow{{4,4.2}}{1o1}{}{};
	\doutplain{1o2}{6.3}{3.6}{$4$};\flow{{4,3.6}}{1o2}{}{};			
  
	\computonPrimitive{3}{0.4}{1}{1.9}{$\lambda_2$};
	\dinplain{2i1}{0.2}{1.5}{$3$};\flow{2i1}{$(2i1)+(2.8,0)$}{}{};
	\dinplain{2i2}{0.2}{0.9}{$4$};\flow{2i2}{$(2i2)+(2.8,0)$}{}{};
	\doutplain{2o1}{6.3}{1.5}{$5$};\flow{{4,1.5}}{2o1}{}{};
\end{scope}

\begin{scope}[xshift=2.4cm,yshift=19cm]\draw[->,opacity=0.4] (0,1) to node[left]{\scriptsize $\beta_{9}$} (1,2);\end{scope} 
\begin{scope}[xshift=10.8cm,yshift=19cm]\draw[->,opacity=0.4] (1,1) to node[right]{\scriptsize $\beta_{10}$} (0,2);\end{scope}

\begin{scope}[xshift=3cm,yshift=15cm]\draw[->,opacity=0.4] (0.9,-3.5) to[bend left=20] node[left]{\scriptsize $\beta_{1}$} (0.4,2);\end{scope} 
\begin{scope}[xshift=10.1cm,yshift=15cm]\draw[->,opacity=0.4] (-0.1,-3.5) to[bend right=20] node[right]{\scriptsize $\beta_{4}$} (0.4,2);\end{scope} 
\begin{scope}[xshift=4.4cm,yshift=15cm]\draw[->,opacity=0.4] (1,1) to node[left]{\scriptsize $\beta_{2}$} (0,2);\end{scope} 
\begin{scope}[xshift=8.3cm,yshift=15cm]\draw[->,opacity=0.4] (0,1) to node[right]{\scriptsize $\beta_{3}$} (1,2);\end{scope}
\begin{scope}[xshift=4.7cm,yshift=12cm]\draw[->] (1,1) to node[left]{\scriptsize $\alpha_0$} (1,2);\end{scope}
\begin{scope}[xshift=8.3cm,yshift=12cm]\draw[->] (0,1) to node[right]{\scriptsize $\alpha_1$} (0,2);\end{scope}
\begin{scope}[xshift=4.7cm,yshift=10.4cm]\draw[->] (1,2) to node[left]{\scriptsize $\alpha_2$} (1,1);\end{scope} 
\begin{scope}[xshift=8.3cm,yshift=10.4cm]\draw[->] (0,2) to node[right]{\scriptsize $\alpha_3$} (0,1);\end{scope} 

\node[opacity=0.4] at (4.1,14.8) {$S$};\node[opacity=0.4] at (9.7,14.8) {$S$};\node[opacity=0.4] at (4.1,8) {$S$};\node[opacity=0.4] at (9.7,8) {$S$};
\node[opacity=0.4] at (7,16.55) {$M$};\node[opacity=0.4] at (7,5.5) {$M$};

\begin{scope}[xshift=3.5cm,yshift=20.1cm]
\computonComposite{0.5}{0}{6.4}{2.3}

	\forkplain{0}{0.3};
	\qin{fq0}{0.2}{0.55}{};
	\qmatch{fq2}{2}{0.7}{};\flow{fq2}{$(fq2)+(1,0)$}{dashed}{};\flow{{1.3,0.7}}{fq2}{dashed}{};
	\qoutplain{fq1}{6.3}{0.1}{};\draw[dashed] (1.3,0.4) to [bend right=5] node [pos=0.935] {\arrowflow} (fq1);	
	
	\computonPrimitive{3}{0.2}{1}{1.9}{$\lambda_1$};
	\dinplain{2i1}{0.2}{1.9}{$1$};\flow{2i1}{$(2i1)+(2.8,0)$}{}{};
	\dinplain{2i2}{0.2}{1.3}{$2$};\flow{2i2}{$(2i2)+(2.8,0)$}{}{};
	\doutplain{2o1}{6.3}{1.9}{$3$};\flow{{4,1.9}}{2o1}{}{};	
	\doutplain{2o2}{6.3}{1.3}{$4$};\flow{{4,1.3}}{2o2}{}{};	
	
	\joinplain{5}{0.3};		
	\qmatch{jq1}{5}{0.7}{};\flow{{4,0.7}}{jq1}{dashed}{};\flow{jq1}{$(jq1)+(1,0)$}{dashed}{};	
	\qinplain{jq0}{0.2}{0.1}{};\draw[dashed] (jq0) to [bend right=5] node [pos=0.065] {\arrowflow} (6,0.4);
	\qout{jq2}{6.3}{0.55}{};
\end{scope}

\begin{scope}[xshift=1.5cm,yshift=17.3cm]
	\computonComposite{0.5}{0}{4.1}{2.3};
  
  \forkplain{0}{0.3};
	\qin{fq0}{0.2}{0.55}{};
	\qoutplain{fq1}{4}{0.1}{};\draw[dashed] (1.3,0.4) to [bend right=5] node [pos=0.9] {\arrowflow} (fq1);
	\qmatch{fq2}{2}{0.7}{};\flow{fq2}{$(fq2)+(1,0)$}{dashed}{};\flow{{1.3,0.7}}{fq2}{dashed}{};
	
	\computonPrimitive{3}{0.2}{1}{1.9}{$\lambda_1$};
	\dinplain{2i1}{0.2}{1.9}{$1$};\flow{2i1}{$(2i1)+(2.8,0)$}{}{};
	\dinplain{2i2}{0.2}{1.3}{$2$};\flow{2i2}{$(2i2)+(2.8,0)$}{}{};
	\dout{2o1}{4}{1.9}{$3$};
	\dout{2o2}{4}{1.3}{$4$};
	\qout{2q1}{4}{0.7}{};	
\end{scope}
\begin{scope}[xshift=7.4cm,yshift=17.3cm]
	\computonComposite{0.4}{0}{4.5}{2.3};	
	
	\computonPrimitive{1}{0.2}{1}{1.9}{$\lambda_1$};	
	\din{2i1}{0.2}{1.9}{$1$};
	\din{2i2}{0.2}{1.3}{$2$};
	\qin{2q0}{0.2}{0.7}{};
	\doutplain{2o1}{4.3}{1.9}{$3$};\draw (2,1.9) to node [pos=0.86] {\arrowflow} (2o1);
	\doutplain{2o2}{4.3}{1.3}{$4$};\draw (2,1.3) to node [pos=0.86] {\arrowflow} (2o2);
	
	\joinplain{3}{0.3};	
	\qmatch{fq2}{3}{0.7}{};\flow{{2,0.7}}{fq2}{dashed}{};\flow{fq2}{$(fq2)+(1,0)$}{dashed}{};
	\qinplain{jq1}{0.2}{0.1}{};\draw[dashed] (jq1) to [bend right=5] node [pos=0.1] {\arrowflow} (4,0.4);
	\qout{jq2}{4.3}{0.55}{};
\end{scope}

\begin{scope}[xshift=5.5cm,yshift=14cm]
	\computonPrimitive{1}{0}{1}{1.9}{$\lambda_1$}  
  \din{1i1}{0.2}{1.7}{$1$};
  \din{1i2}{0.2}{1.1}{$2$};
  \qin{1q0}{0.2}{0.5}{};  
  \dout{1o1}{2}{1.7}{$3$};
  \dout{1o2}{2}{1.1}{$4$};
  \qout{1q1}{2}{0.5}{}
\end{scope}

\begin{scope}[xshift=5.7cm,yshift=12.7cm]\qmatch{0q}{0}{0}{};\end{scope}
\begin{scope}[xshift=8.3cm,yshift=12.7cm]\qmatch{0q}{0}{0}{};\end{scope}

\begin{scope}[xshift=4.7cm,yshift=8.5cm]\draw[->] (1,1) to node[left]{\scriptsize $\alpha_4$} (1,2);\end{scope} 
\begin{scope}[xshift=8.3cm,yshift=8.5cm]\draw[->] (0,1) to node[right]{\scriptsize $\alpha_5$} (0,2);\end{scope} 
\begin{scope}[xshift=4.7cm,yshift=7cm]\draw[->] (1,2) to node[left]{\scriptsize $\alpha_6$} (1,1);\end{scope} 
\begin{scope}[xshift=8.3cm,yshift=7cm]\draw[->] (0,2) to node[right]{\scriptsize $\alpha_7$} (0,1);\end{scope} 
\begin{scope}[xshift=3cm,yshift=8.5cm]\draw[->,opacity=0.4] (0.9,2) to[bend right=20] node[left]{\scriptsize $\beta_{5}$} (0.4,-3.3);\end{scope} 
\begin{scope}[xshift=10.1cm,yshift=8.5cm]\draw[->,opacity=0.4] (-0.1,2) to[bend left=20] node[right]{\scriptsize $\beta_{8}$} (0.4,-3.3);\end{scope} 
\begin{scope}[xshift=4.2cm,yshift=4.2cm]\draw[->,opacity=0.4] (1,2) to node[left]{\scriptsize $\beta_{6}$} (0,1);\end{scope} 
\begin{scope}[xshift=8.5cm,yshift=4.2cm]\draw[->,opacity=0.4] (0,2) to node[right]{\scriptsize $\beta_{7}$} (1,1);\end{scope}
\begin{scope}[xshift=2.4cm,yshift=0.5cm]\draw[->,opacity=0.4] (0,2) to node[left]{\scriptsize $\beta_{11}$} (1,1);\end{scope} 
\begin{scope}[xshift=10.8cm,yshift=0.5cm]\draw[->,opacity=0.4] (1,2) to node[right]{\scriptsize $\beta_{12}$} (0,1);\end{scope} 

\begin{scope}[xshift=3.45cm,yshift=10.7cm]\fork{0}{0}{}{}{};\end{scope}
\begin{scope}[xshift=8.3cm,yshift=10.7cm]\join{0}{0}{}{}{};\end{scope}

\begin{scope}[xshift=5.7cm,yshift=9.25cm]\qmatch{0q}{0}{0}{};\end{scope}
\begin{scope}[xshift=8.3cm,yshift=9.25cm]\qmatch{0q}{0}{0}{};\end{scope}

\begin{scope}[xshift=5.5cm,yshift=6.1cm]
  \computonPrimitive{1}{0}{1}{1.9}{$\lambda_2$}
  \qin{2q0}{0.2}{1.7}{}
  \din{2i1}{0.2}{1.1}{$3$};
  \din{2i2}{0.2}{0.5}{$4$};
  \qout{2q1}{2}{1.7}{}
  \dout{2o1}{2}{1.1}{$5$};  
\end{scope}

\begin{scope}[xshift=1.5cm,yshift=2.8cm]
  \computonComposite{0.5}{0}{4.1}{2.3};
  
  \forkplain{0}{1.55};
	\qin{fq0}{0.2}{1.8}{};
	\qoutplain{fq1}{4}{2.1}{};\draw[dashed] (1.3,1.9) to [bend left=5] node [pos=0.9] {\arrowflow} (fq1);
	\qmatch{fq2}{2}{1.7}{};\flow{fq2}{$(fq2)+(1,0)$}{dashed}{};\flow{{1.3,1.7}}{fq2}{dashed}{};
	
	\computonPrimitive{3}{0.1}{1}{1.9}{$\lambda_2$};
	\dinplain{2i1}{0.2}{1.2}{$3$};\flow{2i1}{$(2i1)+(2.8,0)$}{}{};
	\dinplain{2i2}{0.2}{0.6}{$4$};\flow{2i2}{$(2i2)+(2.8,0)$}{}{};
	\qout{2q1}{4}{1.7}{};
	\dout{2o1}{4}{1.2}{$5$};
\end{scope}
\begin{scope}[xshift=7.4cm,yshift=2.8cm]
  \computonComposite{0.4}{0}{4.5}{2.3};
	\qmatch{fq2}{3}{1.7}{};\flow{{2,1.7}}{fq2}{dashed}{};\flow{fq2}{$(fq2)+(1,0)$}{dashed}{};
	
	\computonPrimitive{1}{0.1}{1}{1.9}{$\lambda_2$};
	\qin{2q0}{0.2}{1.8}{};
	\din{2i1}{0.2}{1.2}{$3$};
	\din{2i2}{0.2}{0.6}{$4$};
	\doutplain{2o1}{4.3}{1.2}{$5$};\draw (2,1.2) to node [pos=0.86] {\arrowflow} (2o1);
	
	\joinplain{3}{1.55};	
	\qinplain{jq1}{0.2}{2.1}{};\draw[dashed] (jq1) to [bend left=5] node [pos=0.1] {\arrowflow} (4,1.9);
	\qout{jq2}{4.3}{1.8}{};
\end{scope}

\begin{scope}[xshift=3.5cm]
	\computonComposite{0.5}{0}{6.4}{2.3}

	\forkplain{0}{1.55};
	\qin{fq0}{0.2}{1.8}{};
	\qoutplain{fq1}{6.3}{2.1}{};\draw[dashed] (1.3,1.9) to [bend left=5] node [pos=0.935] {\arrowflow} (fq1);
	\qmatch{fq2}{2}{1.7}{};\flow{fq2}{$(fq2)+(1,0)$}{dashed}{};\flow{{1.3,1.7}}{fq2}{dashed}{};
	
	\computonPrimitive{3}{0.1}{1}{1.9}{$\lambda_2$};
	\dinplain{2i1}{0.2}{1.2}{$3$};\flow{2i1}{$(2i1)+(2.8,0)$}{}{};
	\dinplain{2i2}{0.2}{0.6}{$4$};\flow{2i2}{$(2i2)+(2.8,0)$}{}{};
	\doutplain{2o1}{6.3}{1.2}{$5$};\flow{{4,1.2}}{2o1}{}{};	
	
	\joinplain{5}{1.55};	
	\qinplain{jq0}{0.2}{2.1}{};\draw[dashed] (jq0) to [bend left=5] node [pos=0.065] {\arrowflow} (6,1.9);
	\qmatch{jq1}{5}{1.7}{};\flow{{4,1.7}}{jq1}{dashed}{};\flow{jq1}{$(jq1)+(1,0)$}{dashed}{};	
	\qout{jq2}{6.3}{1.8}{};
\end{scope}
\end{scope}
\end{tikzpicture}
\caption{Constructing a p-sync computon $\lambda_1 |_{\rho} \lambda_2$ where $\lambda_1$ and $\lambda_2$ are isomorphic to the operands presented in Figure \ref{fig:computon-sequential-example} and $\rho$ is the p-diagram in the middle (whose morphisms are displayed as black arrows). Intuitively, the pushouts marked with `S' are first computed to form four partial sequential computons. Then, the pushouts marked with `M' are computed to merge such partial sequential computons into the top- and bottom-level composites, respectively. Finally, the pushout of the span formed by $(\beta_{11}\circ\beta_5,\beta_{12}\circ\beta_8)$ and $(\beta_{9}\circ\beta_1,\beta_{10}\circ\beta_4)$ is computed to yield the p-sync computon $\lambda_{1} \mid_{\rho} \lambda_{2}$ (displayed on the rightmost part of this figure).}
\label{fig:computon-parallel-sync-example}
\vspace{-2pt}
\end{figure*}

Our construction finalises by computing the pushout of the unique computon morphisms deduced from the universal property of coproducts. The coproduct, in this case, is the juxtaposition of a fork computon and a join computon (in fact a p-async computon), which serves as a common object for the pushout of the unique (induced) span of $(\beta_{11}\circ\beta_5,\beta_{12}\circ\beta_8)$ and $(\beta_{9}\circ\beta_1,\beta_{10}\circ\beta_4)$, i.e., for constructing the p-sync computon $\lambda_{1} \mid_{\rho} \lambda_{2}$ (see Notation \ref{notation:parallel-sync}).

As $\lambda_1 \mid_{\rho} \lambda_2$ is constructed from pushouts that rely on unit computons as apices, only an ec-inport ${p_1 \in C_1^+}$, an ec-outport ${q_1 \in C_1^-}$, an ec-inport $p_2 \in C_2^+$ and an ec-outport $q_2 \in C_2^-$ become i-ports in $\lambda_1 \mid_{\rho} \lambda_2$. The rest of e-inports and e-outports of the arbitrary connected computons become e-inports and e-outports in $\lambda_1 \mid_{\rho} \lambda_2$, respectively. This structural implication is derived from the fact that fork and join computons have control ports only; so, unlike sequential composition and like p-async computons, $\lambda_1$ and $\lambda_2$ do not have any structural means to exchange data when composed into a synchronous parallel structure. To ensure a consistent construction of the p-sync computon $\lambda_1 \mid_{\rho} \lambda_2$, Condition \ref{def:computon-parallel-sync-diagram-6} of Definition \ref{def:diagram-p} intuitively says that $p_1$ and $p_2$ cannot be mapped to the same ec-outport of the fork computon. A similar constraint is imposed by Condition \ref{def:computon-parallel-sync-diagram-7} which states that $q_1$ and $q_2$ cannot be mapped to the same ec-inport of the join computon.

Another difference with respect to sequential composition is that the order of the computons being parallelised does not matter. So, even if $\lambda_1$ and $\lambda_2$ are interchanged in the p-diagram from Figure \ref{fig:computon-parallel-sync-example}, we will always have the same colimit result, i.e., synchronous parallel composition is a commutative operation (see Proposition \ref{prop:computon-parallel-sync-commutative}). Proposition \ref{prop:computon-parallel-sync-associative} shows that, unlike total sequencing, synchronous parallel composition is not associative so that grouping matters. Although such an algebraic property is not satisfied, the result of synchronous parallel composition is always a connected computon (see Corollary \ref{cor:computon-parallel-sync-connected}). Also, any two connected computons can be put into a synchronous parallel structure regardless of the data they require or produce (see Theorem \ref{th:always-parallelisable-sync}). 

\begin{proposition}[Synchronous parallel composition is commutative] \label{prop:computon-parallel-sync-commutative}
There is an isomorphism between ${\lambda_1\mid_{\rho_1}\lambda_2}$ and ${\lambda_2\mid_{\rho_2}\lambda_1}$ for any p-sync computons ${\lambda_1\mid_{\rho_1}\lambda_2}$ and ${\lambda_2\mid_{\rho_2}\lambda_1}$.
\end{proposition}
\begin{proof}
The proof is obvious. It follows from the fact that fork computons are trivially isomorphic, with the same being true for join and unit computons. 
\end{proof}

\begin{proposition}[Synchronous parallel composition is not associative] \label{prop:computon-parallel-sync-associative}
There is no isomorphism between ${(\lambda_1\mid_{\rho_1}\lambda_2)\mid_{\rho_2}\lambda_3}$ and ${\lambda_1\mid_{\rho_4}(\lambda_2\mid_{\rho_3}\lambda_3)}$ for some choice of p-sync computons $\lambda_1\mid_{\rho_1}\lambda_2$, ${(\lambda_1\mid_{\rho_1}\lambda_2)\mid_{\rho_2}\lambda_3}$, $\lambda_2\mid_{\rho_3}\lambda_3$ and ${\lambda_1\mid_{\rho_4}(\lambda_2\mid_{\rho_3}\lambda_3)}$.
\end{proposition}
\begin{proof}
Suppose $\rho_1$, $\rho_2$, $\rho_3$ and $\rho_4$ are p-diagrams. Considering Figure \ref{fig:computon-parallel-sync-associativity}, we let (a), (b), (c) and (d) be the colimits of $\rho_1$, $\rho_2$, $\rho_3$ and $\rho_4$, respectively. As it is clear there is no isomorphism from the p-sync computon (b) to the p-sync computon (d), we conclude that the proposition being proved is true.
\vspace{-18pt}
\begin{figure}[!h]
\centering
\subcaptionbox{Colimit $\lambda_1\mid_{\rho_1}\lambda_2$ of p-diagram $\rho_1$.}
{
\begin{tikzpicture}[scale=0.65]
  \computonComposite{0.5}{0.3}{6.4}{4.2}
  
	\forkplain{0}{2.5};
	\qin{fq0}{0.2}{2.75}{};
	\qmatch{fq1}{2}{3.2}{};\flow{fq1}{$(fq1)+(1,0)$}{dashed}{};\flowdiag{{1.3,2.9}}{fq1}{dashed}{}{pos=0.5,rotate=27};	
	\qmatch{fq2}{2}{1.9}{};\flow{fq2}{$(fq2)+(1,0)$}{dashed}{};\flowdiag{{1.3,2.6}}{fq2}{dashed}{}{pos=0.52,rotate=311};
	
	\joinplain{5}{2.5};	
	\qmatch{jq0}{5}{3.2}{};\flow{{4,3.2}}{jq0}{dashed}{};\flowdiag{jq0}{{6,2.9}}{dashed}{}{pos=0.52,rotate=350};
	\qmatch{jq1}{5}{1.9}{};\flow{{4,1.9}}{jq1}{dashed}{};\flowdiag{jq1}{{6,2.6}}{dashed}{}{pos=0.5,rotate=33};
	\qout{jq2}{6.3}{2.75}{};

	\computonPrimitive{3}{2.5}{1}{1.9}{$\lambda_1$};
	\dinplain{1i1}{0.2}{4.2}{$1$};\flow{1i1}{$(1i1)+(2.8,0)$}{}{};
  
	\computonPrimitive{3}{0.4}{1}{1.9}{$\lambda_2$};
	\dinplain{2i1}{0.2}{1}{$2$};\flow{2i1}{$(2i1)+(2.8,0)$}{}{};
\end{tikzpicture}
}
\subcaptionbox{Colimit ${(\lambda_1\mid_{\rho_1}\lambda_2)\mid_{\rho_2}\lambda_3}$ of p-diagram $\rho_2$.}
{
\begin{tikzpicture}[scale=0.65]
  \computonComposite{0.2}{0.3}{10.7}{6.4}
  \computonComposite{2.3}{2.4}{6.4}{4.2}
  
  \forkplain{-0.2}{3.2};
  \qin{f1q0}{0}{3.45}{};
  \qmatch{f1q1}{2}{4.85}{};\flow{f1q1}{$(f1q1)+(1,0)$}{dashed}{};\flowdiag{{1.1,3.6}}{f1q1}{dashed}{}{pos=0.5,rotate=47};	
  \qmatch{f1q2}{2}{1.95}{};\flow{f1q2}{$(f1q2)+(2.8,0)$}{dashed}{};\flowdiag{{1.1,3.3}}{f1q2}{dashed}{}{pos=0.5,rotate=311};	
    
	\forkplain{1.8}{4.6};
	\qmatch{fq1}{3.8}{5.3}{};\flow{fq1}{$(fq1)+(1,0)$}{dashed}{};\flowdiag{{3.1,5}}{fq1}{dashed}{}{pos=0.5,rotate=27};	
	\qmatch{fq2}{3.8}{4}{};\flow{fq2}{$(fq2)+(1,0)$}{dashed}{};\flowdiag{{3.1,4.7}}{fq2}{dashed}{}{pos=0.52,rotate=311};
	
	\joinplain{6.8}{4.6};	
	\qmatch{jq0}{6.8}{5.3}{};\flow{{5.8,5.3}}{jq0}{dashed}{};\flowdiag{jq0}{{7.8,5}}{dashed}{}{pos=0.52,rotate=350};
	\qmatch{jq1}{6.8}{4}{};\flow{{5.8,4}}{jq1}{dashed}{};\flowdiag{jq1}{{7.8,4.7}}{dashed}{}{pos=0.5,rotate=33};

	\computonPrimitive{4.8}{4.6}{1}{1.9}{$\lambda_1$};
	\dinplain{1i1}{0}{6.3}{$1$};\flowdiag{1i1}{$(1i1)+(4.8,0)$}{}{}{pos=0.1};
  
	\computonPrimitive{4.8}{2.5}{1}{1.9}{$\lambda_2$};
	\dinplain{2i1}{0}{5.1}{$2$};\flowdiag{2i1}{$(2i1)+(4.8,-2)$}{}{bend right=20}{pos=0.1, rotate=311};
		
	\computonPrimitive{4.8}{0.4}{1}{1.9}{$\lambda_3$};
	
	\joinplain{9}{3.2};	
	\qmatch{jq2}{8.9}{4.85}{};\flow{{8.1,4.85}}{jq2}{dashed}{};\flowdiag{jq2}{$(jq2)+(1.1,-1.4)$}{dashed}{}{pos=0.5,rotate=311};	
	\qmatch{jq3}{8.9}{1.95}{};\flow{{5.8,1.95}}{jq3}{dashed}{};\flowdiag{jq3}{$(jq3)+(1.1,1.4)$}{dashed}{}{pos=0.5,rotate=47};	
	\qout{jq4}{10.3}{3.45}{};
\end{tikzpicture}
}
\subcaptionbox{Colimit ${\lambda_2\mid_{\rho_3}\lambda_3}$ of p-diagram $\rho_3$.}
{
\begin{tikzpicture}[scale=0.65]
  \computonComposite{0.5}{0.3}{6.4}{4.2}
  
	\forkplain{0}{2.5};
	\qin{fq0}{0.2}{2.75}{};
	\qmatch{fq1}{2}{3.2}{};\flow{fq1}{$(fq1)+(1,0)$}{dashed}{};\flowdiag{{1.3,2.9}}{fq1}{dashed}{}{pos=0.5,rotate=27};	
	\qmatch{fq2}{2}{1.9}{};\flow{fq2}{$(fq2)+(1,0)$}{dashed}{};\flowdiag{{1.3,2.6}}{fq2}{dashed}{}{pos=0.52,rotate=311};
	
	\joinplain{5}{2.5};	
	\qmatch{jq0}{5}{3.2}{};\flow{{4,3.2}}{jq0}{dashed}{};\flowdiag{jq0}{{6,2.9}}{dashed}{}{pos=0.52,rotate=350};
	\qmatch{jq1}{5}{1.9}{};\flow{{4,1.9}}{jq1}{dashed}{};\flowdiag{jq1}{{6,2.6}}{dashed}{}{pos=0.5,rotate=33};
	\qout{jq2}{6.3}{2.75}{};

	\computonPrimitive{3}{2.5}{1}{1.9}{$\lambda_2$};  
	\dinplain{1i1}{0.2}{4.2}{$2$};\flow{1i1}{$(1i1)+(2.8,0)$}{}{};
	\computonPrimitive{3}{0.4}{1}{1.9}{$\lambda_3$};
\end{tikzpicture}
}
\subcaptionbox{Colimit ${\lambda_1\mid_{\rho_4}(\lambda_2\mid_{\rho_3}\lambda_3)}$ of p-diagram $\rho_4$.}
{
\begin{tikzpicture}[scale=0.65]
  \computonComposite{0.2}{0.2}{10.7}{6.4}
  \computonComposite{2.3}{0.3}{6.4}{4.2}
  
  \forkplain{-0.2}{4};
  \qin{f1q0}{0}{4.25}{};
  \qmatch{f1q1}{2}{5.65}{};\flow{f1q1}{$(f1q1)+(3.2,0)$}{dashed}{};\flowdiag{{1.1,4.4}}{f1q1}{dashed}{}{pos=0.5,rotate=47};	
  \qmatch{f1q2}{2}{2.75}{};\flow{f1q2}{$(f1q2)+(0.8,0)$}{dashed}{};\flowdiag{{1.1,4.1}}{f1q2}{dashed}{}{pos=0.5,rotate=311};	
    
	\forkplain{1.8}{2.5};
	\qmatch{fq1}{3.8}{3.2}{};\flow{fq1}{$(fq1)+(1,0)$}{dashed}{};\flowdiag{{3.1,2.8}}{fq1}{dashed}{}{pos=0.5,rotate=27};	
	\qmatch{fq2}{3.8}{1.9}{};\flow{fq2}{$(fq2)+(1,0)$}{dashed}{};\flowdiag{{3.1,2.6}}{fq2}{dashed}{}{pos=0.52,rotate=311};
	
	\joinplain{6.8}{2.5};	
	\qmatch{jq0}{6.8}{3.2}{};\flow{{5.8,3.2}}{jq0}{dashed}{};\flowdiag{jq0}{{7.8,2.9}}{dashed}{}{pos=0.52,rotate=350};
	\qmatch{jq1}{6.8}{1.9}{};\flow{{5.8,1.9}}{jq1}{dashed}{};\flowdiag{jq1}{{7.8,2.6}}{dashed}{}{pos=0.5,rotate=33};

	\computonPrimitive{4.8}{2.5}{1}{1.9}{$\lambda_2$};
	\dinplain{1i1}{0}{6.3}{$1$};\flowdiag{1i1}{$(1i1)+(4.8,0)$}{}{}{pos=0.1};
	\dinplain{2i1}{0}{3.1}{$2$};\flowdiag{2i1}{$(2i1)+(4.8,1)$}{}{bend left=20}{pos=0.1, rotate=37};
  
	\computonPrimitive{4.8}{0.4}{1}{1.9}{$\lambda_3$};
		
	\computonPrimitive{4.8}{4.6}{1}{1.9}{$\lambda_1$};
	
	\joinplain{9}{4};	
	\qmatch{jq2}{8.9}{5.65}{};\flow{{5.8,5.65}}{jq2}{dashed}{};\flowdiag{jq2}{$(jq2)+(1.1,-1.4)$}{dashed}{}{pos=0.5,rotate=311};	
	\qmatch{jq3}{8.9}{2.75}{};\flow{{8.1,2.75}}{jq3}{dashed}{};\flowdiag{jq3}{$(jq3)+(1.1,1.4)$}{dashed}{}{pos=0.5,rotate=47};	
	\qout{jq4}{10.3}{4.25}{};
\end{tikzpicture}
}
\caption{Counterexample that disproves the associativity property of synchronous parallel composition.}
\label{fig:computon-parallel-sync-associativity}
\end{figure}

\end{proof}

\begin{theorem} \label{th:always-parallelisable-sync}
For every pair $(\lambda_1,\lambda_2)$ of connected computons, there is a p-diagram $\rho$ such that $\lambda_1 \mid_{\rho} \lambda_2$ exists.
\end{theorem}
\begin{proof}
Let ${\lambda_4}$ and ${\lambda_{5}}$ be two arbitrary connected computons, ${\lambda_6}$ a fork computon, ${\lambda_7}$ a join computon, and $\lambda_j$ a unit computon for $j\in\{0,1,2,3\}$. We first construct the following spans of computon morphisms: ${\lambda_6 \xleftarrow{\alpha_2} \lambda_0 \xrightarrow{\alpha_0} \lambda_4}$, ${\lambda_4 \xleftarrow{\alpha_1} \lambda_1 \xrightarrow{\alpha_3} \lambda_7}$, ${\lambda_5 \xleftarrow{\alpha_7} \lambda_3 \xrightarrow{\alpha_5} \lambda_7}$ and ${\lambda_6 \xleftarrow{\alpha_4} \lambda_2 \xrightarrow{\alpha_6} \lambda_5}$. As the common domain of each span is a unit computon, each morphism is a diagram of the form:
\[
\begin{tikzcd}
1\arrow[d, hook] & 1 \arrow[l]\arrow[d] & \emptyset\arrow[l]\arrow[r]\arrow[d] & \emptyset\arrow[d] & \emptyset\arrow[l]\arrow[r]\arrow[d] & 1\arrow[r]\arrow[d] & 1\arrow[d, hook] \\
\Sigma & P\arrow[l] & O\arrow[l]\arrow[r] & U & I\arrow[l]\arrow[r] & P\arrow[r] & \Sigma
\end{tikzcd} 
\]
In the above diagram, it is evident that the only morphism components that are not empty functions are those mapping ports and colours, respectively. As the set of colours of a unit computon is always ${\{0\}}$ and $0$ is in the set of colours of every computon (by Definition \ref{def:computon}), the respective $\Sigma$-component of each morphism can be defined in the obvious way to yield an inclusion function. Now, if ${p_j \in P_j}$, the $P$-component of each morphism is given as follows: ${\alpha_{2}(p_0) \in C_{6}^-}$, ${\alpha_{0}(p_0) \in C_{4}^+}$, ${\alpha_{1}(p_1) \in C_{4}^-}$, ${\alpha_{3}(p_1) \in C_{7}^+}$, ${\alpha_{7}(p_3) \in C_{5}^-}$, ${\alpha_{5}(p_3) \in C_{7}^+}$, ${\alpha_{4}(p_2) \in C_{6}^-}$ and ${\alpha_{6}(p_2) \in C_{5}^+}$ such that ${\alpha_2(p_0) \neq \alpha_4(p_2)}$ and ${\alpha_3(p_1) \neq \alpha_5(p_3)}$.

Since ${\lambda_{4}}$ and ${\lambda_{6}}$ are connected computons and ${\lambda_0}$ is a trivial computon (see Definition \ref{def:computon-connected} and Proposition \ref{prop:computon-connected-alwaysunits}), ${p_0 \in \vec{i}(\alpha_{2}) \cap \vec{o}(\alpha_{0})}$ and, consequently, ${P_0=\vec{i}(\alpha_{2}) \cap \vec{o}(\alpha_{0})}$ because ${P_0=\{p_0\}}$. The facts ${p_0 \in \vec{i}(\alpha_{2}) \cap \vec{o}(\alpha_{0})}$, ${\alpha_{2}(p_0) \in C_{6}^-}$ and ${\alpha_{0}(p_0) \in C_{4}^+}$ allow us to further deduce ${\alpha_{2}(\vec{o}(\alpha_{0})) \subseteq C_{6}^- \subseteq P_{6}^-}$ and ${\alpha_{0}(\vec{i}(\alpha_{2})) \subseteq C_{4}^+ \subseteq P_{4}^+}$. In particular, ${\alpha_{2}(\vec{o}(\alpha_{0})) \subset P_{6}^-}$ because ${|P_0|=1}$ and ${|P_{6}^-|=|C_{6}^-|=2}$ (see the above diagram and Definition \ref{def:computon-fork}). This means that, by Definition \ref{def:span-sequentiable}, ${\lambda_6 \xleftarrow{\alpha_2} \lambda_0 \xrightarrow{\alpha_0} \lambda_4}$ is partially sequentiable. Proving that the other spans are also partially sequentiable is completely analogous.

As our construction corresponds to a p-diagram $\rho$ whose colimit can be computed by Lemma \ref{lem:computon-parallel-sync-exists}, it follows that the p-sync computon $\lambda_4 \mid_{\rho} \lambda_5$ exists in $\textbf{Set}^\textbf{Comp}$ (see Definition \ref{def:computon-parallel-sync} and Notation \ref{notation:parallel-sync}).
\end{proof}

\subsubsection{Operational semantics for p-sync computons (in the theory of Petri nets)}

No matter whether we use any of the three functors presented in Section \ref{sec:operational-semantics} (i.e., $\mathcal{N}$, ${\mathcal{C}\circ\mathfrak{E}}$ or $\mathcal{D}$), the Petri net of a p-sync computon does not introduce any additional places or transitions beyond those from the nets of the computons of the corresponding p-diagram. In the case of $\mathcal{N}$ and ${\mathcal{C}\circ\mathfrak{E}}$, the net of a p-sync has the form depicted in Figure \ref{fig:computon-parallel-sync-net}(a); whereas for $\mathcal{D}$, the corresponding net has the form depicted in Figure \ref{fig:computon-parallel-sync-net}(b). 

\begin{figure}[!h]
\centering
{
\subcaptionbox{For $\mathcal{N}\circ\mathfrak{E}$ or $\mathcal{C}$.}
{
\begin{tikzpicture}
\node[place,minimum size=3mm,label=left:\scriptsize $f$] (f) at (-2,1.5) {};
\node[transition,fill=black,minimum width=0.1mm,minimum height=10mm,label=\scriptsize $\lambda_3$-net] (t1) at (-1.2,1.5) {};

\node[place,label={left:\scriptsize $p_n$},minimum size=3mm] (pn) at (0,2) {};
\node at (0,2.6){$\vdots$};
\node[place,label={left:\scriptsize $p_1$},minimum size=3mm] (p1) at (0,3) {};
\draw[dotted] (0.7,1.7) rectangle (2.2,3.2);\node at (1.4,2.4){\scriptsize $\lambda_1$-net};
\node[place,label={right:\scriptsize $q_1$},minimum size=3mm] (q1) at (2.9,3) {};
\node at (2.9,2.6){$\vdots$};
\node[place,label={right:\scriptsize $q_j$},minimum size=3mm] (qj) at (2.9,2) {};

\node[place,label={left:\scriptsize $r_m$},minimum size=3mm] (rm) at (0,0) {};
\node at (0,0.6){$\vdots$};
\node[place,label={left:\scriptsize $r_1$},minimum size=3mm] (r1) at (0,1) {};
\draw[dotted] (0.7,-0.3) rectangle (2.2,1.2);\node at (1.4,0.4){\scriptsize $\lambda_2$-net};
\node[place,label={right:\scriptsize $s_1$},minimum size=3mm] (s1) at (2.9,1) {};
\node at (2.9,0.6){$\vdots$};
\node[place,label={right:\scriptsize $s_k$},minimum size=3mm] (sk) at (2.9,0) {};

\node[place,minimum size=3mm,label=right:\scriptsize $j$] (j) at (4.9,1.5) {};
\node[transition,fill=black,minimum width=0.1mm,minimum height=10mm,label=\scriptsize $\lambda_4$-net] (t2) at (4.1,1.5) {};

\draw[-latex,thick] (p1) -- ($(p1)+(0.7,-0.2)$);
\draw[-latex,thick] (pn) -- ($(pn)+(0.7,0.2)$);
\draw[-latex,thick] ($(q1)+(-0.7,-0.2)$) -- (q1);
\draw[-latex,thick] ($(qj)+(-0.7,0.2)$) -- (qj);
\draw[-latex,thick] (r1) -- ($(r1)+(0.7,-0.2)$);
\draw[-latex,thick] (rm) -- ($(rm)+(0.7,0.2)$);
\draw[-latex,thick] ($(s1)+(-0.7,-0.2)$) -- (s1);
\draw[-latex,thick] ($(sk)+(-0.7,0.2)$) -- (sk);
\draw[-latex,thick] (f) -- (t1);
\draw[-latex,thick] (t1) -- (pn);
\draw[-latex,thick] (t1) -- (r1);
\draw[-latex,thick] (qj) -- (t2);
\draw[-latex,thick] (s1) -- (t2);
\draw[-latex,thick] (t2) -- (j);
\end{tikzpicture}
}
\subcaptionbox{For $\mathcal{D}$, all places buffer data so the transitions for the fork and join computons do not have any input or output places (cf. Figures \ref{fig:computon-fork}(c) and \ref{fig:computon-join}(c)).}
{
\begin{tikzpicture}
\node[transition,fill=black,minimum width=0.1mm,minimum height=10mm,label=\scriptsize $\lambda_3$-net] (t1) at (-1.2,1.5) {};

\node[place,label={left:\scriptsize $p_n$},minimum size=3mm] (pn) at (0,2) {};
\node at (0,2.6){$\vdots$};
\node[place,label={left:\scriptsize $p_1$},minimum size=3mm] (p1) at (0,3) {};
\draw[dotted] (0.7,1.7) rectangle (2.2,3.2);\node at (1.4,2.4){\scriptsize $\lambda_1$-net};
\node[place,label={right:\scriptsize $q_1$},minimum size=3mm] (q1) at (2.9,3) {};
\node at (2.9,2.6){$\vdots$};
\node[place,label={right:\scriptsize $q_j$},minimum size=3mm] (qj) at (2.9,2) {};

\node[place,label={left:\scriptsize $r_m$},minimum size=3mm] (rm) at (0,0) {};
\node at (0,0.6){$\vdots$};
\node[place,label={left:\scriptsize $r_1$},minimum size=3mm] (r1) at (0,1) {};
\draw[dotted] (0.7,-0.3) rectangle (2.2,1.2);\node at (1.4,0.4){\scriptsize $\lambda_2$-net};
\node[place,label={right:\scriptsize $s_1$},minimum size=3mm] (s1) at (2.9,1) {};
\node at (2.9,0.6){$\vdots$};
\node[place,label={right:\scriptsize $s_k$},minimum size=3mm] (sk) at (2.9,0) {};

\node[transition,fill=black,minimum width=0.1mm,minimum height=10mm,label=\scriptsize $\lambda_4$-net] (t2) at (4.1,1.5) {};

\draw[-latex,thick] (p1) -- ($(p1)+(0.7,-0.2)$);
\draw[-latex,thick] (pn) -- ($(pn)+(0.7,0.2)$);
\draw[-latex,thick] ($(q1)+(-0.7,-0.2)$) -- (q1);
\draw[-latex,thick] ($(qj)+(-0.7,0.2)$) -- (qj);
\draw[-latex,thick] (r1) -- ($(r1)+(0.7,-0.2)$);
\draw[-latex,thick] (rm) -- ($(rm)+(0.7,0.2)$);
\draw[-latex,thick] ($(s1)+(-0.7,-0.2)$) -- (s1);
\draw[-latex,thick] ($(sk)+(-0.7,0.2)$) -- (sk);
\end{tikzpicture}
}
}
\caption{General structure of the Petri net of a p-sync computon $\lambda_1\mid_{\rho}\lambda_2$ constructed from a connected computon $\lambda_1$ with $n$ e-inports and $j$ e-outports, and a connected computon $\lambda_2$ with $m$ e-inports and $k$ e-outports. The places $f$ and $j$ correspond to the only ec-inport and the only ec-outport of the fork and join computons, respectively.}
\label{fig:computon-parallel-sync-net}
\end{figure}

By Proposition \ref{prop:computon-parallel-sync-deadlock} and Remark \ref{rem:parallel-synchronous-deadlock}, the underlying net of any p-sync computon is deadlock-free only if the nets of the composed computons are deadlock-free too.

\begin{proposition}\label{prop:computon-parallel-sync-deadlock}
A Petri net ${\mathcal{N}(\lambda_1\mid_{\rho}\lambda_2)}$ is deadlock-free if ${\mathcal{N}(\lambda_1)}$ and ${\mathcal{N}(\lambda_2)}$ are deadlock-free, and the fork transition of ${\mathcal{N}(\lambda_1\mid_{\rho}\lambda_2)}$ is the only transition enabled in the initial state of ${\mathcal{N}(\lambda_1\mid_{\rho}\lambda_2)}$.
\end{proposition}
\begin{proof}
If the net from Figure \ref{fig:computon-parallel-sync-net}(a) corresponds to $\mathcal{N}(\lambda_1\mid_{\rho}\lambda_2)$, we know by Definition \ref{def:marking} that its initial state $M_i$ is a marking function that puts tokens in $f$ together with the input places of $\mathcal{N}(\lambda_1)$ different than $p_n$ and the input places of $\mathcal{N}(\lambda_2)$ different than $r_1$, while keeping no tokens in all the other places, including those inside $\mathcal{N}(\lambda_1)$ and $\mathcal{N}(\lambda_2)$. If the unique transition of ${\mathcal{N}(\lambda_3)}$ is the only one enabled under $M_i$, firing it reaches a marking that corresponds to the initial states of $\mathcal{N}(\lambda_1)$ and $\mathcal{N}(\lambda_2)$, i.e., only the places ${p_1,\ldots,p_n}$ and ${r_1,\ldots,r_m}$ have tokens. Thus, if $\mathcal{N}(\lambda_1)$ and $\mathcal{N}(\lambda_2)$ are deadlock free, $\mathcal{N}(\lambda_1\mid_{\rho}\lambda_2)$ must be deadlock-free too, considering it is evident that the places and transitions of $\mathcal{N}(\lambda_4)$ do not introduce any deadlocks.
\end{proof}

\begin{remark}\label{rem:parallel-synchronous-deadlock}
Although it is a statement about the functor $\mathcal{N}$, Proposition \ref{prop:computon-parallel-sync-deadlock} is applicable to the functors $\mathcal{C}\circ\mathfrak{E}$ and $\mathcal{D}$ presented in Section \ref{sec:operational-semantics}. The proof is valid for $\mathcal{C}\circ\mathfrak{E}$ since Proposition \ref{prop:functor-control-petri} says ${\mathcal{C}}$ is just a restriction of $\mathcal{N}$ to $\mathfrak{E}(\textbf{Set}^{\textbf{Comp}})$. As $\mathcal{D}(\lambda_1\mid_{\rho}\lambda_2)$ has isolated transitions for fork and join computons (see Figure \ref{fig:computon-parallel-sync-net}(b)), deadlock-freedom follows directly from Remark \ref{rem:parallel-asynchronous-deadlock}.
\end{remark}

By restricting the fork transition of ${\mathcal{N}(\lambda_1\mid_{\rho}\lambda_2)}$ to be the only one enabled under $M_i$, Proposition \ref{prop:computon-parallel-sync-deadlock} disregards the possibility of triggering either the $\lambda_1$-net or the $\lambda_2$-net before/upon firing the unique transition of the $\lambda_3$-net. Should any of these two cases occur, there is a possibility of deadlock. 

\vspace{5pt}

\subsubsection{Encapsulation of control flow and data flow in p-sync computons}

By Definition \ref{def:computon-parallel-sync}, we know a p-sync computon $\lambda_1\mid_{\rho}\lambda_2$ is the colimit of a p-diagram $\rho$ which, by Definition \ref{def:diagram-p}, integrates a fork computon, a join computon, two connected computons and four unit computons. As a result of the colimit construction described in Lemma \ref{lem:computon-parallel-sync-exists}, a p-sync computon connects the ec-outports of the fork and the ec-inports of the join with ec-inports and ec-outports of the connected computons, respectively. Thus, forming a composite that encapsulates synchronous parallel control flow. It is synchronous in the sense the fork enables the parallel invocation of $\lambda_1$ and $\lambda_2$, while the join waits for their termination. As no data ports are linked (through the colimit of $\rho$), there is no data exchange within a p-sync computon so $\lambda_1\mid_{\rho}\lambda_2$ encapsulates up to asynchronous parallel data flow (just as p-async computons do). To give a concrete example, Figure \ref{fig:encapsulation-psync} shows the encapsulation given by the p-sync computon resulting from the colimit construction depicted in Figure \ref{fig:computon-parallel-sync-example}.

\vspace{1.1pt}

Having synchronous control and asynchronous data entails that data does not follow control within a p-sync computon. So, data items can be received before forking control or produced before joining control. Despite of this asynchronous behaviour, the connected computons being paralleised cannot consume data until receiving a control signal from the corresponding fork computon. This is enforced in nets under $\mathcal{N}$. For example, in the scenario depicted in Figure \ref{fig:encapsulation-psync}, $\mathcal{N}(\lambda_1)$ cannot perform any computation until receiving control from the fork as well as $1$- and $2$-coloured data items from the external environment. In other words, the transition representing $\lambda_1$ needs to be enabled by both the transition representing the fork computon and the external environment.

\vspace{1.1pt}

As computons interact via input/output ports, it might seem that the computon model is not yet ready to encode synchronised concurrent behaviours (cf. open automata \cite{katis_spangraph_1997}). We conjecture that an instance of such class of behaviours can be expressed as a p-sync computon to simultaneously trigger computons that encode the transitions being synchronised. Although constructing such composites might be infeasible in practice due to their potentially complex structure, proving this conjecture would support our thesis that control flow is ever present in any (concurrent or sequential) computation.

\begin{figure}[!h]
\centering
\begin{tikzpicture}
\begin{scope}[xshift=-4cm,yshift=0cm]
\computonComposite{0.5}{0.3}{6.4}{4.2}
  
	\forkplain{0}{2.5};
	\qin{fq0}{0.2}{2.75}{};
	\qmatch{fq1}{2}{3.2}{};\flow{fq1}{$(fq1)+(1,0)$}{dashed}{};\flowdiag{{1.3,2.9}}{fq1}{dashed}{}{pos=0.5,rotate=27};	
	\qmatch{fq2}{2}{1.9}{};\flow{fq2}{$(fq2)+(1,0)$}{dashed}{};\flowdiag{{1.3,2.6}}{fq2}{dashed}{}{pos=0.52,rotate=311};
	
	\joinplain{5}{2.5};	
	\qmatch{jq0}{5}{3.2}{};\flow{{4,3.2}}{jq0}{dashed}{};\flowdiag{jq0}{{6,2.9}}{dashed}{}{pos=0.52,rotate=350};
	\qmatch{jq1}{5}{1.9}{};\flow{{4,1.9}}{jq1}{dashed}{};\flowdiag{jq1}{{6,2.6}}{dashed}{}{pos=0.5,rotate=33};
	\qout{jq2}{6.3}{2.75}{};

	\computonPrimitive{3}{2.5}{1}{1.9}{$\lambda_1$};
	\dinplain{1i1}{0.2}{4.2}{$1$};\flow{1i1}{$(1i1)+(2.8,0)$}{}{};
	\dinplain{1i2}{0.2}{3.6}{$2$};\flow{1i2}{$(1i2)+(2.8,0)$}{}{};
	\doutplain{1o1}{6.3}{4.2}{$3$};\flow{{4,4.2}}{1o1}{}{};
	\doutplain{1o2}{6.3}{3.6}{$4$};\flow{{4,3.6}}{1o2}{}{};			
  
	\computonPrimitive{3}{0.4}{1}{1.9}{$\lambda_2$};
	\dinplain{2i1}{0.2}{1.5}{$3$};\flow{2i1}{$(2i1)+(2.8,0)$}{}{};
	\dinplain{2i2}{0.2}{0.9}{$4$};\flow{2i2}{$(2i2)+(2.8,0)$}{}{};
	\doutplain{2o1}{6.3}{1.5}{$5$};\flow{{4,1.5}}{2o1}{}{};
\end{scope}
\draw[opacity=0.3,line width=1.5pt, ->, -Latex] (3.5,2.5) to node[pos=0.4,yshift=7,xshift=-8]{\scriptsize $\mathcal{C}\circ \mathfrak{E}$} (5,3.7);
\begin{scope}[xshift=5cm,yshift=2.7cm]
\node[opacity=0.2] at (0.6,2.5){\scriptsize Control Flow Net};
\node[place,label={80:},minimum size=3mm] (f) at (0,1.3) {};
\node[transition,fill=black,minimum width=0.1mm,minimum height=10mm] (fork) at (1,1.3) {};

\node[place,label={80:},minimum size=3mm] (1i) at (2,2) {};
\node[transition,fill=black,minimum width=0.1mm,minimum height=10mm] (1) at (3,2) {};
\node[place,label={80:},minimum size=3mm] (1o) at (4,2) {};

\node[place,label={80:},minimum size=3mm] (2i) at (2,0.6) {};
\node[transition,fill=black,minimum width=0.1mm,minimum height=10mm] (2) at (3,0.6) {};
\node[place,label={80:},minimum size=3mm] (2o) at (4,0.6) {};

\node[transition,fill=black,minimum width=0.1mm,minimum height=10mm] (join) at (5,1.3) {};
\node[place,label={80:},minimum size=3mm] (j) at (6,1.3) {};

\draw[-latex,thick] (f)--(fork);\draw[-latex,thick] (fork)--(1i);\draw[-latex,thick] (fork)--(2i);
\draw[-latex,thick] (1i)--(1);\draw[-latex,thick] (1)--(1o);
\draw[-latex,thick] (2i)--(2);\draw[-latex,thick] (2)--(2o);
\draw[-latex,thick] (1o)--(join);\draw[-latex,thick] (2o)--(join);\draw[-latex,thick] (join)--(j);
\end{scope}
\draw[opacity=0.3,line width=1.5pt, ->, -Latex] (3.5,2.2) to node[pos=0.4,yshift=-4,xshift=-5]{\scriptsize $\mathcal{D}$} (5,1);
\begin{scope}[xshift=5.5cm,yshift=-0.7cm]
\node[opacity=0.2] at (3.5,0.5){\scriptsize Data Flow Net};
\node[place,label={80:},minimum size=3mm] (4i1) at (0,2) {\scriptsize $1$};
\node[place,label={80:},minimum size=3mm] (4i2) at (0,1) {\scriptsize $2$};
\node[transition,fill=black,minimum width=0.1mm,minimum height=10mm] (4) at (1,1.5) {};
\node[place,label={80:},minimum size=3mm] (4o3) at (2,2) {\scriptsize $3$};
\node[place,label={80:},minimum size=3mm] (4o4) at (2,1) {\scriptsize $4$};

\node[place,label={80:},minimum size=3mm] (4ix3) at (0,0.5) {\scriptsize $3$};
\node[place,label={80:},minimum size=3mm] (4ix4) at (0,-0.5) {\scriptsize $4$};
\node[transition,fill=black,minimum width=0.1mm,minimum height=10mm] (4x) at (1,0) {};
\node[place,label={80:},minimum size=3mm] (4ox) at (2,0) {\scriptsize $5$};

\draw[-latex,thick] (4i1)--(4);\draw[-latex,thick] (4i2)--(4);\draw[-latex,thick](4)--(4o3);\draw[-latex,thick](4)--(4o4);
\draw[-latex,thick] (4ix3)--(4x);\draw[-latex,thick] (4ix4)--(4x);\draw[-latex,thick](4x)--(4ox);
\end{scope}
\draw[opacity=0.3,line width=1.5pt, ->, -Latex] (0,0) to node[pos=0.4,yshift=7,xshift=-8]{\scriptsize $\mathcal{N}$} (0,-0.8);
\begin{scope}[xshift=-3cm,yshift=-3.5cm]
\node[opacity=0.2] at (0,2.8){\scriptsize Control and};\node[opacity=0.2] at (0,2.5){\scriptsize Data Flow Net};\node[opacity=0.2] at (0,2.2){\scriptsize Net};
\node[place,label={80:},minimum size=3mm] (i1) at (2,2.8) {\scriptsize $1$};
\node[place,label={80:},minimum size=3mm] (i2) at (2,2.3) {\scriptsize $2$};
\node[place,label={80:},minimum size=3mm] (i3) at (2,0.3) {\scriptsize $3$};
\node[place,label={80:},minimum size=3mm] (i4) at (2,-0.2) {\scriptsize $4$};

\node[place,label={80:},minimum size=3mm] (f) at (0,1.3) {};
\node[transition,fill=black,minimum width=0.1mm,minimum height=10mm] (fork) at (1,1.3) {};

\node[place,label={80:},minimum size=3mm] (1i) at (2,1.7) {};
\node[transition,fill=black,minimum width=0.1mm,minimum height=10mm] (1) at (3,2) {};
\node[place,label={80:},minimum size=3mm] (1o) at (4,1.7) {};

\node[place,label={80:},minimum size=3mm] (2i) at (2,0.9) {};
\node[transition,fill=black,minimum width=0.1mm,minimum height=10mm] (2) at (3,0.6) {};
\node[place,label={80:},minimum size=3mm] (2o) at (4,0.9) {};

\node[transition,fill=black,minimum width=0.1mm,minimum height=10mm] (join) at (5,1.3) {};
\node[place,label={80:},minimum size=3mm] (j) at (6,1.3) {};

\node[place,label={80:},minimum size=3mm] (o3) at (4,2.8) {\scriptsize $3$};
\node[place,label={80:},minimum size=3mm] (o4) at (4,2.3) {\scriptsize $4$};
\node[place,label={80:},minimum size=3mm] (o5) at (4,0.3) {\scriptsize $5$};

\draw[-latex,thick] (f)--(fork);\draw[-latex,thick] (fork)--(1i);\draw[-latex,thick] (fork)--(2i);
\draw[-latex,thick] (1i)--(1);\draw[-latex,thick] (1)--(1o);
\draw[-latex,thick] (2i)--(2);\draw[-latex,thick] (2)--(2o);
\draw[-latex,thick] (1o)--(join);\draw[-latex,thick] (2o)--(join);\draw[-latex,thick] (join)--(j);
\draw[-latex,thick](i1)--(1);\draw[-latex,thick](i2)--(1);\draw[-latex,thick](i3)--(2);\draw[-latex,thick](i4)--(2);
\draw[-latex,thick](1)--(o3);\draw[-latex,thick](1)--(o4);\draw[-latex,thick](2)--(o5);
\end{scope}
\end{tikzpicture}
\caption{Synchronous parallel control flow and asynchronous parallel data flow encapsulated by the p-sync computon from Figure \ref{fig:computon-parallel-sync-example}. We label some places for mapping purposes even though Petri nets are not labelled (see Section \ref{sec:operational-semantics}).}
\label{fig:encapsulation-psync}
\end{figure}

\subsection{Branching Computons}
\label{sec:branching-computons}

A \emph{branching computon} structurally consists of two connected computons whose e-inports and e-outports overlap, respectively. This overlapping restriction is captured by a so-called \emph{b-diagram} whose morphisms are all markers (see Definition \ref{def:computon-branching-diagram}). 

\begin{definition}[B-Diagram]\label{def:computon-branching-diagram}
A b-diagram is a diagram with the following shape in $\textbf{Set}^\textbf{Comp}$:
\[
\begin{tikzcd}
 & \lambda_0 \arrow[dl, "\lambda_2^+"']\arrow[dr, "\lambda_3^+"] & \\
\lambda_2 & & \lambda_3 \\
 & \lambda_1 \arrow[ul, "\lambda_2^-"]\arrow[ur, "\lambda_3^-"'] & 
\end{tikzcd}
\]
where:
\begin{enumerate}
\item $\lambda_2$ is a connected computon with an in-marker $\lambda_2^+$ and an out-marker $\lambda_2^-$, and
\item $\lambda_3$ is a connected computon with an in-marker $\lambda_3^+$ and an out-marker $\lambda_3^-$.
\end{enumerate}
Evidently, by Definition \ref{def:computon-morphism-markers}, $\lambda_0$ and $\lambda_1$ are trivial computons, serving as domains for the markers involved in the b-diagram.
\end{definition}

A branching computon operates in a non-deterministic manner by exclusively choosing a connected computon out of two possible ones. Its role is to capture the essential structure of decision-making by abstracting away from particular conditions or external environment influences which might be variable, unpredictable or too complex to enumerate exhaustively. Intuitively, a branching computon is like an \emph{if-else} programming construct but without explicit conditions deciding computational paths.\footnote{Another reason for decoupling control flow from external influences is to facilitate formal verification of crucial control flow properties such as deadlock absence. Without decoupling, a verification problem could become intractable since, in general, the external environment cannot be modelled completely. Even when assumptions are made, they might not capture the real-world accurately \cite{van_der_aalst_application_1998}.} 

To construct a branching composite, it suffices to compute the colimit of a b-diagram by performing a pushout operation along a coproduct construction. More specifically, if we consider the b-diagram from Definition \ref{def:computon-branching-diagram}, the operation is $\lambda_2 +_{\lambda_0+\lambda_1} \lambda_3$ such that $\lambda_2$ and $\lambda_3$ are the computons being branched (see Definition \ref{def:computon-branching}). As the apex of such a pushout is the coproduct of a trivial computon $\lambda_0$ (which can be injected into all the e-inports of both operands) and a trivial computon $\lambda_1$ (which can be identified with all the e-outports of both operands), a branching structure can be constructed only from computons with isomorphic interfaces. By Lemma \ref{lem:computon-branching-exists}, this pushout construction can always be computed in $\textbf{Set}^\textbf{Comp}$. 

\begin{definition}[Branching Computon]\label{def:computon-branching}
A branching computon is the colimit of a b-diagram.
\end{definition}

\begin{notation} \label{notation:branching}
For convenience, we use a question mark to reflect the fact that computons are chosen non-deterministically. For example, we write $\lambda_{2} ?_{\rho} \lambda_{3}$ for the colimit of the b-diagram $\rho$ shown in Definition \ref{def:computon-branching-diagram}.
\end{notation}

\begin{lemma}\label{lem:computon-branching-exists}
A branching computon can always be constructed in $\textbf{Set}^\textbf{Comp}$. 
\end{lemma}
\begin{proof}
Consider the b-diagram shown in Definition \ref{def:computon-branching-diagram}. By Proposition \ref{prop:computon-coproduct}, we know that the coproduct $\lambda_0+\lambda_1$ can be formed. As there are markers ${\lambda_j^+:\lambda_0 \rightarrow \lambda_j}$ and ${\lambda_j^-:\lambda_1 \rightarrow \lambda_j}$ for $j \in \{2,3\}$, we use the universal property of coproducts to deduce there also are unique computon morphisms ${(\lambda_2^+,\lambda_2^-):\lambda_0+\lambda_1 \rightarrow \lambda_2}$ and ${(\lambda_3^+,\lambda_3^-):\lambda_0+\lambda_1 \rightarrow \lambda_3}$. Assuming ${\beta_1: \lambda_0 \rightarrow \lambda_0+\lambda_1}$ and ${\beta_2: \lambda_1 \rightarrow \lambda_0+\lambda_1}$ are the canonical injections into ${\lambda_0+\lambda_1}$, we now prove that the induced span ${\lambda_2 \xleftarrow{(\lambda_2^+,\lambda_2^-)} \lambda_0+\lambda_1 \xrightarrow{(\lambda_3^+,\lambda_3^-)} \lambda_3}$ is pushable.

If $p_2 \in (\lambda_2^+,\lambda_2^-)(\vec{i}(\lambda_3^+,\lambda_3^-))$, there is some $q \in \vec{i}(\lambda_3^+,\lambda_3^-)$ for which $(\lambda_2^+,\lambda_2^-)(q)=p_2$, i.e., $\bullet (\lambda_3^+,\lambda_3^-)(q) \setminus (\lambda_3^+,\lambda_3^-)(\bullet q) \neq \emptyset$ which implies $(\lambda_3^+,\lambda_3^-)(q) \notin P_3^+$. Consequently, $(\lambda_3^+,\lambda_3^-)(q)$ is not in the image of the in-marker $\lambda_3^+$ so, by coproduct definition, there must be some ${p_1 \in \vec{i}(\lambda_3^-)}$ such that ${\lambda_3^-(p_1)=(\lambda_3^+,\lambda_3^-)(q) \in P_3^-}$. As ${\lambda_3^-=(\lambda_3^+,\lambda_3^-)\circ\beta_2}$ (by coproduct commutativity), we have ${\lambda_3^-(p_1)=(\lambda_3^+,\lambda_3^-)(q)=(\lambda_3^+,\lambda_3^-)(\beta_2(p_1))}$. That is, ${q=\beta_2(p_1)}$ given that the $P$-component of $(\lambda_3^+,\lambda_3^-)$ is injective. Injectivity follows from the fact that ${\lambda_3^+}$ and ${\lambda_3^-}$ are both injective (see Definition \ref{def:computon-morphism-markers}) and that ${P_3^+\cap P_3^-=\emptyset}$ because $\lambda_3$ is a connected computon (see Definition \ref{def:computon-branching-diagram} and Proposition \ref{prop:computon-connected-isolated-port}).

Now, observing ${\lambda_2^-=(\lambda_2^+,\lambda_2^-)\circ\beta_2}$ and considering that ${\lambda_2^-}$ is an out-marker, we obtain $\lambda_2^-(p_1)=(\lambda_2^+,\lambda_2^-)(\beta_2(p_1))=(\lambda_2^+,\lambda_2^-)(q)=p_2 \in P_2^-$. Hence, $(\lambda_2^+,\lambda_2^-)(\vec{i}(\lambda_3^+,\lambda_3^-)) \subseteq P_2^- \subseteq P_2^- \cup P_2^+$. As the other conditions of Definition \ref{def:computon-morphisms-pushable} can be proved analogously, the pushout ${(\beta_3:\lambda_2 \rightarrow \lambda_4, \lambda_4, \beta_4:\lambda_3 \rightarrow \lambda_4)}$ of ${\lambda_2 \xleftarrow{(\lambda_2^+,\lambda_2^-)} \lambda_0+\lambda_1 \xrightarrow{(\lambda_3^+,\lambda_3^-)} \lambda_3}$ can be constructed. To prove $\lambda_4$ is the colimit of the original b-diagram, suppose there is a cone:
\[
\begin{tikzcd}
 & \lambda_0 \arrow[ddl, crossing over, "\lambda_2^+"']\arrow[drr, crossing over, "\lambda_3^+"]\arrow[dddddr, opacity=0.4] & & \\
 & & & \lambda_3 \arrow[ddddl, opacity=0.4] \\
\lambda_2 \arrow[dddrr, opacity=0.4] & & & \\
 & & \lambda_1 \arrow[ull, crossing over, "\lambda_2^-"]\arrow[uur, crossing over, "\lambda_3^-"']\arrow[dd, opacity=0.4] & \\
 & & & \\
 & & \lambda_5 & 
\end{tikzcd}
\]
As there evidently are morphisms ${\lambda_0 \rightarrow \lambda_5}$ and ${\lambda_1 \rightarrow \lambda_5}$, we use the universal property of coproducts to deduce there is a unique computon morphism ${\lambda_0+\lambda_1 \rightarrow \lambda_5}$ such that the corresponding diagram commutes. As there also are computon morphisms ${\lambda_2 \rightarrow \lambda_5}$ and ${\lambda_3 \rightarrow \lambda_5}$, we use the universal property of pushouts to deduce there is a unique morphism ${\lambda_4 \rightarrow \lambda_5}$ that makes everything commute in our construction. Therefore, $\lambda_4$ is the colimit of the original b-diagram.
\end{proof}

\begin{corollary}\label{cor:computon-branching-connected}
A branching computon is a connected computon.
\end{corollary}
\begin{proof}
As a branching computon is the pushout of a pushable span whose legs are connected computons (see the proof of Lemma \ref{lem:computon-branching-exists} and Definition \ref{def:computon-branching-diagram}), we simply use Proposition \ref{prop:pushout-connected} to deduce that any branching computon is a connected computon.
\end{proof}

To clarify the construction presented in the proof of Lemma \ref{lem:computon-branching-exists}, Figure \ref{fig:computon-branching-example} presents a complete, self-descriptive example for the construction of a branching computon $\lambda_1 ?_\rho \lambda_2$ where $\rho$ is the b-diagram displayed at the top. Particularly, for $j \in \{1,2\}$, $\lambda_j$ is a computon with an in-marker $\lambda_j^+$ and an out-marker $\lambda_j^-$ and the computon morphisms $\beta_1$ and $\beta_2$ are canonical injections satisfying the universal property of coproducts. By this property, there are unique computon morphisms $(\lambda_1^+,\lambda_1^-)$ and $(\lambda_2^+,\lambda_2^-)$ whose pushout yields the branching computon $\lambda_1 ?_\rho \lambda_2$ together with induced morphisms $\gamma_1$ and $\gamma_2$. 

\begin{figure}[!h]
\centering
\begin{tikzpicture}
\begin{scope}
\begin{scope}\draw[->,opacity=0.4] (3.8,4.8) to node[left]{\scriptsize $\gamma_1$} (4.8,3.5);\end{scope}
\begin{scope}\draw[->,opacity=0.4] (8.8,4.8) to node[right]{\scriptsize $\gamma_2$} (8,3.5);\end{scope}
\begin{scope}\draw[->,dashed,opacity=0.4] (5.3,8.4) to node[left]{\scriptsize $(\lambda_1^+,\lambda_1^-)$} (3.8,7.1);\end{scope}
\begin{scope}\draw[->,dashed,opacity=0.4] (7.3,8.4) to node[right]{\scriptsize $(\lambda_2^+,\lambda_2^-)$} (8.8,7.1);\end{scope}
\begin{scope}\draw[->,opacity=0.4] (5.3,9.9) to node[left]{\scriptsize $\beta_1$} (6.1,8.9);\end{scope}
\begin{scope}\draw[->,opacity=0.4] (7.1,9.9) to node[right]{\scriptsize $\beta_2$} (6.3,8.9);\end{scope}
\begin{scope}\draw[<-] (2.7,7.1) to[bend left=25] node[left]{\scriptsize $\lambda_1^+$} (4.4,10);\end{scope}
\begin{scope}\draw[<-] (9.9,7.1) to[bend right=25] node[right]{\scriptsize $\lambda_2^-$} (8.2,10);\end{scope}
\begin{scope}\draw[<-] (4.6,7.1) to[bend left=25] node[left]{\scriptsize $\lambda_1^-$} (7.1,10.1);\end{scope}
\begin{scope}\draw[<-] (8.1,7.1) to[bend right=25] node[right]{\scriptsize $\lambda_2^+$} (5.3,10.1);\end{scope}

\begin{scope}[xshift=5cm,yshift=7.3cm]
  \qmatch{1i0}{0}{3}{};
  \dmatch{1i1}{0}{2.6}{$1$}{left};
  \dmatch{1i2}{0}{2.2}{$2$}{left};
\end{scope}
\begin{scope}[xshift=6.4cm,yshift=7.3cm]
  \qmatch{2i0}{1}{3}{};
  \dmatch{2i1}{1}{2.6}{$3$}{right};
  \dmatch{2i2}{1}{2.2}{$4$}{right};
\end{scope}

\begin{scope}[xshift=5.7cm,yshift=5.7cm]
  \qmatch{1i0}{0}{3}{};
  \dmatch{1i1}{0}{2.6}{$1$}{left};
  \dmatch{1i2}{0}{2.2}{$2$}{left};
  
  \qmatch{2i0}{1}{3}{};
  \dmatch{2i1}{1}{2.6}{$3$}{right};
  \dmatch{2i2}{1}{2.2}{$4$}{right};
\end{scope}

\begin{scope}[xshift=2.5cm,yshift=5cm]
  \computonPrimitive{0.8}{0}{1}{1.9}{$\lambda_1$}
  \qin{1q0}{0}{1.7}{}
  \din{1i1}{0}{1.1}{$1$};
  \din{1i2}{0}{0.5}{$2$};
  \qout{1q1}{1.8}{1.7}{}
  \dout{1o1}{1.8}{1.1}{$3$};
  \dout{1o2}{1.8}{0.5}{$4$};
\end{scope}
\begin{scope}[xshift=7.5cm,yshift=5cm]
  \computonPrimitive{0.8}{0}{1}{1.9}{$\lambda_2$}
  \qin{1q0}{0}{1.7}{}
  \din{1i1}{0}{1.1}{$1$};
  \din{1i2}{0}{0.5}{$2$};
  \qout{1q1}{1.8}{1.7}{}
  \dout{1o1}{1.8}{1.1}{$3$};
  \dout{1o2}{1.8}{0.5}{$4$};
\end{scope}

\begin{scope}[xshift=5cm]
	\computonComposite{0.3}{0}{2.4}{4.5};

	\computonPrimitive{1}{2.4}{1}{1.9}{$\lambda_1$};
	\computonPrimitive{1}{0.2}{1}{1.9}{$\lambda_2$};
	
	\qinplain{q0}{0}{2.7}{};
	\dinplain{j0}{0}{2.3}{$1$};
	\dinplain{k0}{0}{1.9}{$2$};
	\flow{q0}{{1,3.2}}{dashed}{};\flow{q0}{{1,1.3}}{dashed}{};
	\flow{j0}{{1,3}}{}{};\flow{j0}{{1,1.1}}{}{};
	\flow{k0}{{1,2.8}}{}{};\flow{k0}{{1,0.9}}{}{};

	\qoutplain{q1}{2.2}{2.7}{};
	\doutplain{j1}{2.2}{2.3}{$3$};
	\doutplain{k1}{2.2}{1.9}{$4$};
	\flow{{2,3.2}}{q1}{dashed}{};\flow{{2,1.3}}{q1}{dashed}{};
	\flow{{2,3}}{j1}{}{};\flow{{2,1.1}}{j1}{}{};
	\flow{{2,2.8}}{k1}{}{};\flow{{2,0.9}}{k1}{}{};
\end{scope}
\end{scope}
\end{tikzpicture}
\caption{Constructing a branching computon $\lambda_1 ?_\rho \lambda_2$ where $\rho$ is the b-diagram at the top (whose morphisms are displayed as black arrows). Here, $\lambda_1$ and $\lambda_2$ are both isomorphic to the left operand presented in the example of Figure \ref{fig:computon-sequential-example}, and the $\beta$-morphisms are the canonical injections into the coproduct of the trivial computons in $\rho$. The branching computon $\lambda_1 ?_\rho \lambda_2$ is simply the pushout of the unique morphisms induced from the universal property of coproducts, namely $(\lambda_1^+,\lambda_1^-)$ and $(\lambda_2^+,\lambda_2^-)$.}
\label{fig:computon-branching-example}
\end{figure}

Branching is an operation that enables the non-deterministic selection of a computon out of two possible ones. So, even if we interchange $\lambda_1$ and $\lambda_2$ in the construction depicted in Figure \ref{fig:computon-branching-example}, the colimit would be isomorphic, i.e., constructing a branching computon is a commutative operation (see Proposition \ref{prop:computon-branching-commutative}). As grouping does not alter the colimit result either, branching is associative in addition (see Proposition \ref{prop:computon-branching-associative}). 

\begin{proposition}[Branching composition is commutative]\label{prop:computon-branching-commutative}
There is an isomorphism between ${\lambda_1 ?_{\rho_1} \lambda_2}$ and ${\lambda_2 ?_{\rho_2} \lambda_1}$ for any branching computons ${\lambda_1 ?_{\rho_1} \lambda_2}$ and ${\lambda_2 ?_{\rho_2} \lambda_1}$.
\end{proposition}
\begin{proof}
The proof follows directly from the well-known fact that categorical pushout is commutative up to unique isomorphism. 
\end{proof}

\vspace{2pt}

\begin{proposition}[Branching composition is associative]\label{prop:computon-branching-associative}
There is an isomorphism between ${(\lambda_1 ?_{\rho_1} \lambda_2) ?_{\rho_2} \lambda_3}$ and ${\lambda_1 ?_{\rho_4} (\lambda_2 ?_{\rho_3} \lambda_3)}$ for any branching computons ${\lambda_1 ?_{\rho_1} \lambda_2}$, ${(\lambda_1 ?_{\rho_1} \lambda_2) ?_{\rho_2} \lambda_3}$, ${\lambda_2 ?_{\rho_3} \lambda_3}$ and ${\lambda_1 ?_{\rho_4} (\lambda_2 ?_{\rho_3} \lambda_3)}$.
\end{proposition}
\begin{proof}
The proof is similar to that of Proposition \ref{prop:computon-sequential-total-associative}.
\end{proof}

\vspace{2pt}

Unfortunately, not every pair of connected computons is a candidate to define a branching composite. This is because the e-inports of one computon must totally match the e-inports of the other, with the same being true for e-outports (hence the retrictions imposed by the morphisms of a b-diagram --- see Definition \ref{def:computon-branching-diagram}). Nevertheless, when a pair of computons meets such restrictions, Corollary \ref{cor:computon-branching-connected} states that their corresponding branching composite is always a connected computon.

\vspace{2pt}

\subsubsection{Operational semantics for branching computons (in the theory of Petri nets)}

No matter whether we use any of the three functorial constructions presented in Section \ref{sec:operational-semantics}, the underlying Petri net of a branching computon does not have any additional places or transitions beyond those from the composed computon nets. The general structure of a branching computon's net is depicted in Figure \ref{fig:computon-branching-net}. 

\begin{figure}[!h]
\centering
{
\begin{tikzpicture}
\node[place,label={left:\scriptsize $p_n$},minimum size=3mm] (pn) at (0,0.9) {};
\node at (0,1.5){$\vdots$};
\node[place,label={left:\scriptsize $p_1$},minimum size=3mm] (p1) at (0,1.9) {};

\draw[dotted] (0.7,1.7) rectangle (2.2,3.2);\node at (1.4,2.4){\scriptsize $\lambda_1$-net};
\draw[dotted] (0.7,-0.3) rectangle (2.2,1.2);\node at (1.4,0.4){\scriptsize $\lambda_2$-net};

\node[place,label={right:\scriptsize $q_j$},minimum size=3mm] (q1) at (2.9,0.9) {};
\node at (2.9,1.5){$\vdots$};
\node[place,label={right:\scriptsize $q_1$},minimum size=3mm] (qj) at (2.9,1.9) {};

\draw[-latex,thick] (p1) -- ($(p1)+(0.7,0.5)$);\draw[-latex,thick] (p1) -- ($(p1)+(0.7,-1)$);
\draw[-latex,thick] (pn) -- ($(pn)+(0.7,1)$);\draw[-latex,thick] (pn) -- ($(pn)+(0.7,-0.5)$);
\draw[-latex,thick] ($(q1)+(-0.7,1)$) -- (q1);\draw[-latex,thick] ($(q1)+(-0.7,-0.5)$) -- (q1);
\draw[-latex,thick] ($(qj)+(-0.7,0.5)$) -- (qj);\draw[-latex,thick] ($(qj)+(-0.7,-1)$) -- (qj);
\end{tikzpicture}
}
\caption{General structure of the Petri net of a branching computon $\lambda_1?_{\rho}\lambda_2$ constructed from connected computons $\lambda_1$ and $\lambda_2$ that have $n$ e-inports and $j$ e-outports each. This structure is applicable to all the functorial constructions from Section \ref{sec:operational-semantics}, namely $\mathcal{N}$, $\mathcal{C}\circ\mathfrak{E}$ and $\mathcal{D}$.}
\label{fig:computon-branching-net}
\end{figure} 

Unfortunately, a net ${\mathcal{N}(\lambda_1?_{\rho}\lambda_2)}$ is not deadlock-free in general, especially when $\lambda_1$ or $\lambda_2$ rely on partial sequencing. To fully ensure deadlock-freedom, the composed computons need to be ``adapted'' by attaching an in-sync to each of them; thereby, forming total sequential computons that synchronise e-inports. These sequential entities can then composed into a branching structure which is guaranteed to be deadlock-free, provided that ${\mathcal{N}(\lambda_1)}$ and ${\mathcal{N}(\lambda_2)}$ are deadlock-free too. Proposition \ref{prop:computon-branching-deadlock} uses this adaptation process to show that it is possible to convert the net of any branching computon into a deadlock-free one. 

\begin{proposition}\label{prop:computon-branching-deadlock}
Let ${\lambda_1?_{\rho}\lambda_2}$ be a branching computon and assume ${\lambda_3}$ and ${\lambda_4}$ are in-syncs of ${\lambda_1}$ and ${\lambda_2}$, respectively. If ${\mathcal{N}(\lambda_1)}$ and ${\mathcal{N}(\lambda_2)}$ are deadlock-free, there are spans ${\rho_1}$, ${\rho_2}$ and ${\rho_3}$ such that ${\mathcal{N}((\lambda_3\unrhd_{\rho_1}\lambda_1)?_{\rho_3}(\lambda_4\unrhd_{\rho_2}\lambda_2))}$ is deadlock-free.
\end{proposition}
\begin{proof}
If ${\lambda_1?_{\rho}\lambda_2}$ is a branching computon, we know by Definition \ref{def:computon-branching-diagram} that ${\lambda_1}$ and ${\lambda_2}$ are connected computons. If ${\lambda_3}$ and ${\lambda_4}$ are in-syncs of ${\lambda_1}$ and ${\lambda_2}$, respectively, Proposition \ref{prop:sync-in} says there must be spans ${\rho_1}$ and ${\rho_2}$ such that ${\lambda_3\unrhd_{\rho_1}\lambda_1}$ and ${\lambda_4\unrhd_{\rho_2}\lambda_2}$ exist. By Proposition \ref{prop:computon-sequential-inports-outports-2} and Definition \ref{def:sync}, we deduce that the in-markers of ${\lambda_3\unrhd_{\rho_1}\lambda_1}$ and ${\lambda_4\unrhd_{\rho_2}\lambda_2}$ share domain with ${\lambda_1^+}$ and ${\lambda_2^+}$, respectively. Similarly, using Proposition \ref{prop:computon-sequential-inports-outports-2}, we further derive that the out-markers of ${\lambda_3\unrhd_{\rho_1}\lambda_1}$ and ${\lambda_4\unrhd_{\rho_2}\lambda_2}$ share domain with ${\lambda_1^-}$ and ${\lambda_2^-}$, respectively. Considering that the domains of ${\lambda_1^+}$ and ${\lambda_2^+}$ coincide in the b-diagram ${\rho}$ (as per Definition \ref{def:computon-branching-diagram}) and that sequential computons are connected (see Proposition \ref{prop:computon-sequential-connected}), we have just constructed a new b-diagram ${\rho_3}$ out of ${\rho}$. Simply using Lemma \ref{lem:computon-branching-exists}, we deduce ${(\lambda_3\unrhd_{\rho_1}\lambda_1)?_{\rho_3}(\lambda_4\unrhd_{\rho_2}\lambda_2)}$ exists. According to the functorial construction presented in Definition \ref{def:functor-computon-to-petri}, we know the net ${\mathcal{N}((\lambda_3\unrhd_{\rho_1}\lambda_1)?_{\rho_3}(\lambda_4\unrhd_{\rho_2}\lambda_2))}$ must have the following form:

\begin{center}
\begin{tikzpicture}
\node[place,label={left:\scriptsize $p_n$},minimum size=3mm] (pn) at (0,0.9) {};
\node at (0,1.5){$\vdots$};
\node[place,label={left:\scriptsize $p_1$},minimum size=3mm] (p1) at (0,1.9) {};

\node[transition,fill=black,minimum width=0.1mm,minimum height=10mm,label=\scriptsize$\mathcal{N}(\lambda_3)$] (t1) at (0.7,2.4) {};
\node[transition,fill=black,minimum width=0.1mm,minimum height=10mm,label=below:\scriptsize$\mathcal{N}(\lambda_4)$] (t2) at (0.7,0.6) {};

\node[place,label={[yshift=-0.1cm,xshift=0.3cm]:\scriptsize $r_n$},minimum size=3mm] (r1) at (1.4,1.9) {};
\node at (1.4,2.5){$\vdots$};
\node[place,label={[yshift=-0.1cm,xshift=0.3cm]:\scriptsize $r_1$},minimum size=3mm] (rn) at (1.4,2.9) {};
\draw[dotted] (2.3,1.7) rectangle (3.8,3.2);\node at (3,2.4){\scriptsize $\mathcal{N}(\lambda_1)$};

\node[place,label={[yshift=-0.1cm,xshift=0.3cm]:\scriptsize $s_n$},minimum size=3mm] (s1) at (1.4,0) {};
\node at (1.4,0.6){$\vdots$};
\node[place,label={[yshift=-0.1cm,xshift=0.3cm]:\scriptsize $s_1$},minimum size=3mm] (sn) at (1.4,1) {};
\draw[dotted] (2.3,-0.3) rectangle (3.8,1.2);\node at (3,0.4){\scriptsize $\mathcal{N}(\lambda_2)$};

\node[place,label={right:\scriptsize $q_j$},minimum size=3mm] (q1) at (4.5,0.9) {};
\node at (4.5,1.5){$\vdots$};
\node[place,label={right:\scriptsize $q_1$},minimum size=3mm] (qj) at (4.5,1.9) {};

\draw[-latex,thick] (p1) -- ($(p1)+(0.7,0.8)$);\draw[-latex,thick] (p1) -- ($(p1)+(0.7,-1.2)$);
\draw[-latex,thick] (pn) -- ($(pn)+(0.7,1.3)$);\draw[-latex,thick] (pn) -- ($(pn)+(0.7,-0.6)$);
\draw[-latex,thick] ($(q1)+(-0.7,1)$) -- (q1);\draw[-latex,thick] ($(q1)+(-0.7,-0.5)$) -- (q1);
\draw[-latex,thick] ($(qj)+(-0.7,0.5)$) -- (qj);\draw[-latex,thick] ($(qj)+(-0.7,-1)$) -- (qj);

\draw[-latex,thick] (t1) -- (r1);\draw[-latex,thick] (r1) -- ($(r1)+(0.9,0)$);
\draw[-latex,thick] (t1) -- (rn);\draw[-latex,thick] (rn) -- ($(rn)+(0.9,0)$);
\draw[-latex,thick] (t2) -- (s1);\draw[-latex,thick] (s1) -- ($(s1)+(0.9,0)$);
\draw[-latex,thick] (t2) -- (sn);\draw[-latex,thick] (sn) -- ($(sn)+(0.9,0)$);
\end{tikzpicture}
\end{center}

Definition \ref{def:marking} says that the initial state ${M_i}$ of the above net is a marking function where ${M_i(p)>0}$ for all ${p\in\{p_1,\ldots,p_n\}}$ and no tokens in all the other places, including those inside ${\mathcal{N}(\lambda_1)}$ and ${\mathcal{N}(\lambda_2)}$. Although ${M_i}$ enables the only transitions of ${\mathcal{N}(\lambda_3)}$ and ${\mathcal{N}(\lambda_4)}$, it is clear that only one of them will be chosen for firing due to mutual exclusion. Consequently, we have two possible execution paths:
\begin{itemize}
\item ${\mathcal{N}(\lambda_3)}$ fires to put tokens in each place in ${\{r_1,\ldots,r_n\}}$; thereby, reaching a correspondence with the initial state of ${\mathcal{N}(\lambda_1)}$. 
\item ${\mathcal{N}(\lambda_4)}$ fires to put tokens in each place in ${\{s_1,\ldots,s_n\}}$, thereby, reaching a correspondence with the initial state of ${\mathcal{N}(\lambda_2)}$.
\end{itemize}
If we assume that ${\mathcal{N}(\lambda_1)}$ and ${\mathcal{N}(\lambda_2)}$ are both deadlock-free, it is evident that the net ${\mathcal{N}((\lambda_3\unrhd_{\rho_1}\lambda_1)?_{\rho_3}(\lambda_4\unrhd_{\rho_2}\lambda_2))}$ cannot get stuck in any of the two possible execution paths. Therefore, ${\mathcal{N}((\lambda_3\unrhd_{\rho_1}\lambda_1)?_{\rho_3}(\lambda_4\unrhd_{\rho_2}\lambda_2))}$ is deadlock-free, as required.
\end{proof}

\begin{remark}
As ${\lambda_1?_{\rho}\lambda_2}$ identifies all the e-inports of ${\lambda_1}$ with all e-inports of ${\lambda_2}$, it is evident that the in-syncs ${\lambda_3}$ and ${\lambda_4}$ must be isomorphic. Accordingly, we deliberately use the same subindices for $p$-, $r$- and $s$-places so as to show that there is a one-to-one correspondence between the input and output places of ${\mathcal{N}(\lambda_3)}$ and ${\mathcal{N}(\lambda_4)}$. 
\end{remark}

\begin{remark}\label{rem:branching-deadlock}
Although it is a statement about the functor $\mathcal{N}$, Proposition \ref{prop:computon-branching-deadlock} is applicable to the functors $\mathcal{C}\circ\mathfrak{E}$ and $\mathcal{D}$ presented in Section \ref{sec:operational-semantics}. The proof is valid for $\mathcal{C}\circ\mathfrak{E}$ since Proposition \ref{prop:functor-control-petri} says ${\mathcal{C}}$ is just a restriction of $\mathcal{N}$ to $\mathfrak{E}(\textbf{Set}^{\textbf{Comp}})$. 

As we are only interested in checking deadlock-freedom for $\mathcal{D}$-nets with initial and final states (see Remark \ref{rem:deadlock-freedom}), we only need to consider the form depicted in the proof of Proposition \ref{prop:computon-branching-deadlock} for $j,n>0$. As this is exactly what we have for any net under $\mathcal{N}$ (because any computon always has ec-inports and ec-outports), the proof that any $\mathcal{D}$-net (with initial and final states) is deadlock-free is analogous to that of Proposition \ref{prop:computon-branching-deadlock}.
\end{remark}

\vspace{0.5pt}

\subsubsection{Encapsulation of control flow and data flow in branching computons}

A branching computon encapsulates branching control flow and up to branching data flow, as a result of matching all the e-inports/e-outports of one computon with all the e-inports/e-outports of another. It is branching in the sense a corresponding net chooses an execution path out of two possible ones. For instance, Figure \ref{fig:encapsulation-branching} shows the encapsulation given by the branching computon resulting from the colimit construction depicted in Figure \ref{fig:computon-branching-example}.

\begin{figure}[!h]
\centering
\begin{tikzpicture}
\begin{scope}[xshift=-4cm,yshift=0cm]
	\computonComposite{0.3}{0}{2.4}{4.5};

	\computonPrimitive{1}{2.4}{1}{1.9}{$\lambda_1$};
	\computonPrimitive{1}{0.2}{1}{1.9}{$\lambda_2$};
	
	\qinplain{q0}{0}{2.7}{};
	\dinplain{j0}{0}{2.3}{$1$};
	\dinplain{k0}{0}{1.9}{$2$};
	\flow{q0}{{1,3.2}}{dashed}{};\flow{q0}{{1,1.3}}{dashed}{};
	\flow{j0}{{1,3}}{}{};\flow{j0}{{1,1.1}}{}{};
	\flow{k0}{{1,2.8}}{}{};\flow{k0}{{1,0.9}}{}{};

	\qoutplain{q1}{2.2}{2.7}{};
	\doutplain{j1}{2.2}{2.3}{$3$};
	\doutplain{k1}{2.2}{1.9}{$4$};
	\flow{{2,3.2}}{q1}{dashed}{};\flow{{2,1.3}}{q1}{dashed}{};
	\flow{{2,3}}{j1}{}{};\flow{{2,1.1}}{j1}{}{};
	\flow{{2,2.8}}{k1}{}{};\flow{{2,0.9}}{k1}{}{};
\end{scope}
\draw[opacity=0.3,line width=1.5pt, ->, -Latex] (0,2.5) to node[pos=0.4,yshift=7,xshift=-8]{\scriptsize $\mathcal{C}\circ \mathfrak{E}$} (1.5,3.7);
\begin{scope}[xshift=2.2cm,yshift=2.7cm]
\node[opacity=0.2] at (-0.5,1.8){\scriptsize Control Flow};\node[opacity=0.2] at (-0.5,1.5){\scriptsize Net};
\node[place,label={80:},minimum size=3mm] (i) at (0,0.8) {};
\node[transition,fill=black,minimum width=0.1mm,minimum height=10mm] (1) at (1,1.6) {};
\node[transition,fill=black,minimum width=0.1mm,minimum height=10mm] (2) at (1,0) {};
\node[place,label={80:},minimum size=3mm] (o) at (2,0.8) {};

\draw[-latex,thick] (i)--(1);\draw[-latex,thick] (i)--(2);
\draw[-latex,thick] (1)--(o);\draw[-latex,thick] (2)--(o);
\end{scope}
\draw[opacity=0.3,line width=1.5pt, ->, -Latex] (0,2.2) to node[pos=0.4,yshift=-4,xshift=-5]{\scriptsize $\mathcal{D}$} (1.5,1);
\begin{scope}[xshift=2.2cm,yshift=-0.7cm]
\node[opacity=0.2] at (-1,0.9){\scriptsize Data Flow};\node[opacity=0.2] at (-1,0.6){\scriptsize Net};
\node[place,label={80:},minimum size=3mm] (i1) at (0,1.2) {\scriptsize $1$};
\node[place,label={80:},minimum size=3mm] (i2) at (0,0.4) {\scriptsize $2$};
\node[transition,fill=black,minimum width=0.1mm,minimum height=10mm] (1) at (1,1.6) {};
\node[transition,fill=black,minimum width=0.1mm,minimum height=10mm] (2) at (1,0) {};
\node[place,label={80:},minimum size=3mm] (o3) at (2,1.2) {\scriptsize $3$};
\node[place,label={80:},minimum size=3mm] (o4) at (2,0.4) {\scriptsize $4$};

\draw[-latex,thick] (i1)--(1);\draw[-latex,thick] (i1)--(2);\draw[-latex,thick] (i2)--(1);\draw[-latex,thick] (i2)--(2);
\draw[-latex,thick] (1)--(o3);\draw[-latex,thick] (2)--(o3);\draw[-latex,thick] (1)--(o4);\draw[-latex,thick] (2)--(o4);
\end{scope}
\draw[opacity=0.3,line width=1.5pt, ->, -Latex] (-4.7,2.5) to node[pos=0.4,yshift=8]{\scriptsize $\mathcal{N}$} (-5.9,2.5);
\begin{scope}[xshift=-8.5cm,yshift=2cm]
\node[opacity=0.2] at (-1.6,1){\scriptsize Control and};\node[opacity=0.2] at (-1.7,0.7){\scriptsize Data Flow};\node[opacity=0.2] at (-1.6,0.4){\scriptsize Net};
\node[place,label={80:},minimum size=3mm] (ic) at (0,1.4) {};
\node[place,label={80:},minimum size=3mm] (i1) at (0,0.9) {\scriptsize $1$};
\node[place,label={80:},minimum size=3mm] (i2) at (0,0.4) {\scriptsize $2$};
\node[transition,fill=black,minimum width=0.1mm,minimum height=10mm] (1) at (1,1.6) {};
\node[transition,fill=black,minimum width=0.1mm,minimum height=10mm] (2) at (1,0) {};
\node[place,label={80:},minimum size=3mm] (oc) at (2,1.4) {};
\node[place,label={80:},minimum size=3mm] (o3) at (2,0.9) {\scriptsize $3$};
\node[place,label={80:},minimum size=3mm] (o4) at (2,0.4) {\scriptsize $4$};

\draw[-latex,thick] (i1)--(1);\draw[-latex,thick] (i1)--(2);\draw[-latex,thick] (i2)--(1);\draw[-latex,thick] (i2)--(2);
\draw[-latex,thick] (1)--(o3);\draw[-latex,thick] (2)--(o3);\draw[-latex,thick] (1)--(o4);\draw[-latex,thick] (2)--(o4);
\draw[-latex,thick] (ic)--(1);\draw[-latex,thick] (ic)--(2);\draw[-latex,thick] (1)--(oc);\draw[-latex,thick] (2)--(oc);
\end{scope}
\end{tikzpicture}
\caption{Branching control flow and branching data flow encapsulated by the branching computon from Figure \ref{fig:computon-branching-example}. We label some places for mapping purposes even though Petri nets are not labelled (see Section \ref{sec:operational-semantics}).}
\label{fig:encapsulation-branching}
\end{figure}

\subsection{Iterative Computons}
\label{sec:iterative-computons}

An \emph{iterative computon} is structurally made up of four connected computons, one of which corresponds to the computon $\lambda$ being iterated over while another allows the repeated invocation of $\lambda$. The other two connected computons serve as endpoints to respectively enter and exit the iterative computational structure being defined. In this subsection, we describe two classes of iterative computons: \emph{head-iterative} and \emph{tail-iterative}.

\subsubsection{Head-Iterative Computons}

A \emph{head-iterative computon} decides to either continue or terminate an iterative process just before executing an arbitrary connected computon. To construct it, it suffices to compute the colimit of a so-called \emph{h-diagram} which defines basic building blocks, namely four connected computons and two trivial computons, together with six marker morphisms (see Definitions \ref{def:diagram-h} and \ref{def:computon-iterative-head}). By Lemma \ref{lem:computon-iterative-head-exists}, such a colimit always exists in the category $\textbf{Set}^\textbf{Comp}$ so a head-iterative computon can always be constructed. By Corollary \ref{cor:computon-iterative-head-connected}, a head-iterative computon is a connected computon.

\begin{definition}[H-Diagram]\label{def:diagram-h}
An h-diagram is a diagram with the following shape in $\textbf{Set}^\textbf{Comp}$:
\[
\begin{tikzcd}
 & & \lambda_0 \arrow[dll, "\lambda_2^-"']\arrow[dl, "\lambda_4^+"]\arrow[dr, "\lambda^+"]\arrow[drrr, "\lambda_3^-"] & & & \\
\lambda_2 & \lambda_4 & & \lambda & \lambda_1 \arrow[l, "\lambda^-"]\arrow[r, "\lambda_3^+"'] & \lambda_3
\end{tikzcd}
\]
where:
\begin{itemize}
\item $\lambda$ is a connected computon with an in-marker $\lambda^+$ and an out-marker $\lambda^-$,
\item $\lambda_2$ is a connected computon with an out-marker $\lambda_2^-$,
\item $\lambda_3$ is a connected computon with an in-marker $\lambda_3^+$ and an out-marker $\lambda_3^-$, and
\item $\lambda_4$ is a connected computon with an in-marker $\lambda_4^+$.
\end{itemize}
Evidently, by Definition \ref{def:computon-morphism-markers}, $\lambda_0$ and $\lambda_1$ are trivial computons, serving as domains for the markers involved in the h-diagram.
\end{definition}

\begin{definition}[Head-Iterative Computon]\label{def:computon-iterative-head}
A head-iterative computon is the colimit of an h-diagram.
\end{definition}

\begin{notation}\label{notation:computon-iterative-head}
For convenience, we write a star symbol before a computon symbol $\lambda$ to indicate that the decision-making locus that terminates the iterative process is placed just before the computational structure of $\lambda$. For example, we write $*_{\rho}(\lambda)$ for the colimit of the h-diagram $\rho$ shown in Definition \ref{def:diagram-h}.
\end{notation}

\begin{lemma}\label{lem:computon-iterative-head-exists}
A head-iterative computon can always be constructed in $\textbf{Set}^\textbf{Comp}$.
\end{lemma}
\begin{proof}
Considering the h-diagram shown in Definition \ref{def:diagram-h}, by Propositions \ref{prop:pushout-pushable} and \ref{prop:markers-pushable}, we know that the pushouts of ${\lambda_2 \xleftarrow{\lambda_2^-} \lambda_0 \xrightarrow{\lambda_3^-} \lambda_3}$ and ${\lambda_4 \xleftarrow{\lambda_4^+} \lambda_0 \xrightarrow{\lambda^+} \lambda}$ can be constructed. Let us denote them $(\beta_1:\lambda_2 \rightarrow \lambda_5, \lambda_5,\beta_2:\lambda_3 \rightarrow \lambda_5)$ and $(\beta_3:\lambda_4 \rightarrow \lambda_6, \lambda_6,\beta_4:\lambda \rightarrow \lambda_6)$, respectively. By pushout commutativity, we deduce the existence of computon morphisms $f:\lambda_0 \rightarrow \lambda_5$ and $g:\lambda_0 \rightarrow \lambda_6$ such that $f=\beta_1\circ\lambda_2^-=\beta_2\circ\lambda_3^-$ and $g=\beta_3\circ\lambda_4^+=\beta_4\circ\lambda^+$.

Now, by Proposition \ref{prop:computon-coproduct}, we know that the coproduct $\lambda_0+\lambda_1$ can be formed. Using the universal property of coproducts, we also know there must be unique computon morphisms $(f,\beta_2\circ\lambda_3^+)$ and $(g,\beta_4\circ\lambda^-)$. To show that ${\lambda_5 \xleftarrow{(f,\beta_2\circ\lambda_3^+)} \lambda_0+\lambda_1 \xrightarrow{(g,\beta_4\circ\lambda^-)} \lambda_6}$ is pushable, assume $p_6 \in (g,\beta_4\circ\lambda^-)(\vec{o}(f,\beta_2\circ\lambda_3^+))$ so there is some $q \in \vec{o}(f,\beta_2\circ\lambda_3^+)$ such that $(g,\beta_4\circ\lambda^-)(q)=p_6$. As $(f,\beta_2\circ\lambda_3^+)(q)\bullet\setminus(f,\beta_2\circ\lambda_3^+)(q\bullet)\neq\emptyset$, it is true that $(f,\beta_2\circ\lambda_3^+)(q) \notin P_5^-$. By coproduct definition and considering that $\lambda_5$ is the pushout of the span ${\lambda_2 \xleftarrow{\lambda_2^-} \lambda_0 \xrightarrow{\lambda_3^-} \lambda_3}$ of out-marker morphisms, we use Proposition \ref{prop:markers-inout} to deduce there is some $p_1 \in P_1$ where $\beta_2(\lambda_3^+(p_1))=(f,\beta_2\circ\lambda_3^+)(q)$. Again, by coproduct definition, we get $(g,\beta_4\circ\lambda^-)(q)=p_6=\beta_4(\lambda^-(p_1)) \in P_6^-$ (because $\lambda^-$ is an out-marker morphism and $\beta_4$ is an induced morphism for the pushout of ${\lambda_4 \xleftarrow{\lambda_4^+} \lambda_0 \xrightarrow{\lambda^+} \lambda}$). As the other conditions of Definition \ref{def:computon-morphisms-pushable} are proved analogously, it is true that the pushout $\lambda_7$ of $(f,\beta_2\circ\lambda_3^+)$ and $(g,\beta_4\circ\lambda^-)$ can be constructed. We now show that such a pushout satisfies the universal property of the colimit of the original h-diagram by supposing there is a cone:

\[
\begin{tikzcd}
 & & \lambda_0 \arrow[dll, "\lambda_2^-"']\arrow[dl, "\lambda_4^+"]\arrow[dr, "\lambda^+"]\arrow[drrr, "\lambda_3^-"]\arrow[dd, opacity=0.4] & & & \\
\lambda_2 \arrow[drr, bend right=30, opacity=0.4] & \lambda_4 \arrow[dr, bend right=30, opacity=0.4] & & \lambda \arrow[dl, bend left=30, opacity=0.4] & \lambda_1 \arrow[dll, bend left=30, opacity=0.4] \arrow[l, "\lambda^-"]\arrow[r, "\lambda_3^+"'] & \lambda_3 \arrow[dlll, bend left=20, opacity=0.4] \\
 & & \lambda_8 & & &
\end{tikzcd}
\]

Since $\lambda_5$ and $\lambda_6$ are the respective pushouts of ${\lambda_2 \xleftarrow{\lambda_2^-} \lambda_0 \xrightarrow{\lambda_3^-} \lambda_3}$ and ${\lambda_4 \xleftarrow{\lambda_4^+} \lambda_0 \xrightarrow{\lambda^+} \lambda}$, by the universal property of pushouts, there are unique computon morphisms $\lambda_5 \rightarrow \lambda_8$ and $\lambda_6 \rightarrow \lambda_8$ that make the corresponding diagrams commute.

In the above cone, there are computon morphisms ${\lambda_0 \rightarrow \lambda_8}$ and ${\lambda_1 \rightarrow \lambda_8}$. Using the universal property of coproducts, we deduce there is a unique computon morphism ${\lambda_0+\lambda_1 \rightarrow \lambda_8}$. Finally, we use the existence of ${\lambda_5 \rightarrow \lambda_8}$ and ${\lambda_6 \rightarrow \lambda_8}$ and the universal property of pushouts to deduce there is a unique computon morphism ${\lambda_7 \rightarrow \lambda_8}$ that makes everything commute in our construction. Thus, proving that $\lambda_7$ is the colimit of the original h-diagram.
\end{proof}

\begin{corollary}\label{cor:computon-iterative-head-connected}
A head-iterative computon is a connected computon.
\end{corollary}
\begin{proof}
In the construction presented in the proof of Lemma \ref{lem:computon-iterative-head-exists}, the pushouts $\lambda_5$ and $\lambda_6$ are connected computons by the fact that $\lambda$, $\lambda_2$, $\lambda_3$ and $\lambda_4$ also are (see Proposition \ref{prop:pushout-connected}). Consequently, the pushout $\lambda_7$ of the pushable span ${\lambda_5 \xleftarrow{(f,\beta_2\circ\lambda_3^+)} \lambda_0+\lambda_1 \xrightarrow{(g,\beta_4\circ\lambda^-)} \lambda_6}$ is a connected computon. As $\lambda_7$ is the colimit of the (original) h-diagram shown in Definition \ref{def:diagram-h}, we conclude that every head-iterative computon is a connected computon.
\end{proof}

To clarify the construction presented in the proof of Lemma \ref{lem:computon-iterative-head-exists}, Figure \ref{fig:computon-iterative-head-example} illustrates a complete, self-descriptive example for the formation of a head-iterative computon $*_{\rho}(\lambda)$ over a functional computon $\lambda$ where $\rho$ is the h-diagram shown in the middle. 

\begin{figure*}[!h]
\centering
\begin{tikzpicture}[scale=0.77]
\begin{scope}[xshift=3cm,yshift=19.5cm]
  
\end{scope}
\begin{scope}[xshift=-0.5cm,yshift=8.7cm]
\qmatch{i0}{0}{1.7}{};
\dmatch{i1}{0}{1.3}{$1$}{left};
\dmatch{i2}{0}{0.9}{$2$}{left};
\qmatch{o0}{0.5}{1.7}{};
\dmatch{o1}{0.5}{1.3}{$3$}{right};
\end{scope}
\begin{scope}[xshift=10cm,yshift=9.7cm]
  \computonComposite{0.5}{0}{6}{4.9};
  
  \computonPrimitive{2.3}{0.2}{1}{1.9}{$\lambda_2$};
  \qinplain{2q0}{0.1}{1.8}{};\draw[dashed] (2q0) to node [pos=0.55] {\arrowflow} ({2.3,1.8});
  \dinplain{2i1}{0.1}{1.2}{$4$};\draw (2i1) to node [pos=0.55] {\arrowflow} ({2.3,1.2});
 	
  \computonPrimitive{2.3}{2.4}{1}{1.9}{$\lambda_3$};  
  \qmatch{xq1}{1.5}{4}{};
  \dmatch{xi1}{1.5}{3.4}{$3$}{above};
	\flow{xq1}{{2.3,4}}{dashed}{};
	\flow{xi1}{{2.3,3.4}}{}{};
	
	\qmatch{yq1}{4.1}{2.9}{};
  \dmatch{yo1}{4.1}{2.3}{$1$}{above};
  \dmatch{yo2}{4.1}{1.7}{$2$}{above};
  \flowdiag{{3.3,4}}{yq1}{dashed}{}{pos=0.5,rotate=308};
  \flowdiag{{3.3,1.8}}{yq1}{dashed}{}{pos=0.25,rotate=45};
	\flowdiag{{3.3,3.4}}{yo1}{}{}{pos=0.4,rotate=308};\flowdiag{{3.3,1.2}}{yo1}{}{}{pos=0.4,rotate=45};
	\flowdiag{{3.3,2.8}}{yo2}{}{}{pos=0.25,rotate=308};\flowdiag{{3.3,0.6}}{yo2}{}{}{pos=0.5,rotate=45};
	\flowdiag{yq1}{{4.9,1.8}}{dashed}{}{pos=0.75,rotate=308};	
	
	\computonPrimitive{4.9}{2.4}{1}{1.9}{$\lambda$};
	\flowdiag{yq1}{$(yq1)+(0.8,1.1)$}{dashed}{}{pos=0.5,rotate=45};
	\flowdiag{yo1}{$(yo1)+(0.8,1.1)$}{}{}{pos=0.5,rotate=45};
	\flowdiag{yo2}{$(yo2)+(0.8,1.1)$}{}{}{pos=0.75,rotate=45};
	
	\computonPrimitive{4.9}{0.2}{1}{1.9}{$\lambda_4$}; 
  \qoutplain{3q0}{6}{1.8}{};\flow{{5.9,1.8}}{3q0}{dashed}{};
  \doutplain{3i1}{6}{1.2}{$5$};\flow{{5.9,1.2}}{3i1}{}{}; 
  \flowdiag{yo1}{{4.9,1.2}}{}{}{pos=0.5,rotate=308};
	\flowdiag{yo2}{{4.9,0.6}}{}{}{pos=0.5,rotate=308};     
  
  \draw[dashed] ($(yq1)+(1.8,1.1)$) -- ($(yq1)+(2,1.1)$) -- ($(yq1)+(2,1.6)$) to node [pos=0.45] {\invertedarrowflow} (1.2,4.5) -- (1.2,4) -- (xq1);
  \draw ($(yo1)+(1.8,1.1)$) -- ($(yo1)+(2.2,1.1)$) -- ($(yo1)+(2.2,2.4)$) to node [pos=0.45] {\invertedarrowflow} (1,4.7) -- (1,3.4) -- (xi1);
\end{scope}

\begin{scope}
\begin{scope}\draw[->,opacity=0.4,dashed] (-0.35,9.3) to[bend right=30] node[left]{\scriptsize $(g,\beta_4\circ\lambda^-)$} (3,1);\end{scope}
\begin{scope}\draw[->,opacity=0.4,dashed] (-0.3,10.8) to[bend left=30] node[left]{\scriptsize $(f,\beta_2\circ\lambda_3^+)$} (2.4,17);\end{scope}
\begin{scope}\draw[->,opacity=0.4] (6.4,1) to[bend right=30] node[right]{\scriptsize $\beta_6$} (13.5,9.5);\end{scope}
\begin{scope}\draw[->,opacity=0.4] (6.3,17.3) to[bend left=30] node[right, yshift=1.5mm]{\scriptsize $\beta_5$} (13.5,14.8);\end{scope}

\begin{scope}[xshift=0.7cm,yshift=8.4cm]\draw[->,opacity=0.4] (3.2,1.1) to[bend left=25] node[left]{} (0,1.7);\end{scope}
\begin{scope}[xshift=0.7cm,yshift=8.4cm]\draw[->,opacity=0.4] (8.2,0.9) to[bend right=25] node[left]{} (0,2);\end{scope}

\begin{scope}[xshift=6.9cm,yshift=8.4cm]\draw[->] (2,1) to node[right]{\scriptsize $\lambda_3^+$} (0.7,2.8);\end{scope}
\begin{scope}[xshift=6.9cm,yshift=6.6cm]\draw[->] (2,2.6) to node[right]{\scriptsize $\lambda^-$} (0.7,1);\end{scope}

\begin{scope}[xshift=0.7cm,yshift=8.6cm]\draw[->] (2,1) to node[left]{\scriptsize $\lambda_2^-$} (0.7,2.6);\end{scope}
\begin{scope}[xshift=5cm,yshift=8.6cm]\draw[->] (0.7,1) to node[right]{\scriptsize $\lambda_3^-$} (2,2.6);\end{scope}
\begin{scope}[xshift=0.7cm,yshift=12.8cm]\draw[->,opacity=0.4] (0.7,1) to node[left]{\scriptsize $\beta_1$} (2,2.6);\end{scope} 
\begin{scope}[xshift=5.3cm,yshift=12.8cm]\draw[->,opacity=0.4] (2,1) to node[right]{\scriptsize $\beta_2$} (0.7,2.6);\end{scope}

\begin{scope}[xshift=0.7cm,yshift=6.6cm]\draw[->] (2,2.6) to node[left]{\scriptsize $\lambda_4^+$} (0.7,1);\end{scope}
\begin{scope}[xshift=5cm,yshift=6.6cm]\draw[->] (0.7,2.6) to node[right]{\scriptsize $\lambda^+$} (2,1);\end{scope}
\begin{scope}[xshift=0.7cm,yshift=2.4cm]\draw[->,opacity=0.4] (0.7,2.6) to node[left]{\scriptsize $\beta_3$} (2,1);\end{scope} 
\begin{scope}[xshift=5.3cm,yshift=2.4cm]\draw[->,opacity=0.4] (2,2.6) to node[right]{\scriptsize $\beta_4$} (0.7,1);\end{scope} 

\begin{scope}[xshift=3.2cm,yshift=14cm]
 	\computonComposite{0.3}{0}{2.1}{4.5};
  
  \computonPrimitive{0.8}{2.4}{1}{1.9}{$\lambda_3$}
  \qin{2q0}{0}{4}{}
  \din{2i1}{0}{3.4}{$3$};
  \computonPrimitive{0.8}{0.2}{1}{1.9}{$\lambda_2$}
  \qin{3q0}{0}{1.8}{}
  \din{3i1}{0}{1.2}{$4$};
  
  \qoutplain{q1}{1.8}{2.9}{}
  \doutplain{o1}{1.8}{2.3}{$1$};
  \doutplain{o2}{1.8}{1.7}{$2$};
	\flowdiag{{1.8,4}}{q1}{dashed}{}{pos=0.5,rotate=308};\flowdiag{{1.8,1.8}}{q1}{dashed}{}{pos=0.25,rotate=45};
	\flowdiag{{1.8,3.4}}{o1}{}{}{pos=0.35,rotate=308};\flowdiag{{1.8,1.2}}{o1}{}{}{pos=0.25,rotate=45};
	\flowdiag{{1.8,2.8}}{o2}{}{}{pos=0.25,rotate=308};\flowdiag{{1.8,0.6}}{o2}{}{}{pos=0.25,rotate=45};
\end{scope}

\begin{scope}[xshift=0.2cm,yshift=11.8cm]
  \computonPrimitive{0.8}{0}{1}{1.9}{$\lambda_2$}
  \qin{2q0}{0}{1.7}{}
  \din{2i1}{0}{1.1}{$4$};
  \qout{2q1}{1.8}{1.7}{}
  \dout{2o1}{1.8}{1.1}{$1$};
  \dout{2o2}{1.8}{0.5}{$2$};
\end{scope}
\begin{scope}[xshift=6cm,yshift=11.8cm]
  \computonPrimitive{0.8}{0}{1}{1.9}{$\lambda_3$}
  \qin{3q0}{0}{1.7}{}
  \din{3i1}{0}{1.1}{$3$};
  \qout{3q1}{1.8}{1.7}{}
  \dout{3o1}{1.8}{1.1}{$1$};
  \dout{3o1}{1.8}{0.5}{$2$};
\end{scope}

\begin{scope}[xshift=4.3cm,yshift=8.6cm]
\qmatch{i0}{0}{0.8}{};
\dmatch{i1}{0}{0.4}{$1$}{left};
\dmatch{i2}{0}{0}{$2$}{left};
\end{scope}
\begin{scope}[xshift=9.3cm,yshift=8.9cm]
\qmatch{o0}{0}{0.4}{};
\dmatch{o1}{0}{0}{$3$}{left};
\end{scope}

\begin{scope}[xshift=0.2cm,yshift=5.2cm]
\computonPrimitive{0.8}{0}{1}{1.9}{$\lambda_4$}
  \qin{4q0}{0}{1.7}{}
  \din{4i1}{0}{1.1}{$1$};
  \din{4i1}{0}{0.5}{$2$};
  \qout{q1}{1.8}{1.7}{}
  \dout{3o1}{1.8}{1.1}{$5$};
\end{scope}
\begin{scope}[xshift=6cm,yshift=5.2cm]
  \computonPrimitive{0.8}{0}{1}{1.9}{$\lambda$}
  \qin{q0}{0}{1.7}{}
  \din{i1}{0}{1.1}{$1$};
  \din{i2}{0}{0.5}{$2$};
  \qout{q1}{1.8}{1.7}{}
  \dout{o1}{1.8}{1.1}{$3$};
\end{scope}

\begin{scope}[xshift=3.2cm,yshift=0cm]
  \computonComposite{0.3}{0}{2.1}{4.5};
  
  \computonPrimitive{0.8}{2.4}{1}{1.9}{$\lambda$}
  \qout{2q0}{1.8}{4}{}
  \dout{2i1}{1.8}{3.4}{$3$};
  \computonPrimitive{0.8}{0.2}{1}{1.9}{$\lambda_4$}
  \qout{3q0}{1.8}{1.8}{}
  \dout{3i1}{1.8}{1.2}{$5$};
  
  \qinplain{q1}{0}{2.9}{}
  \dinplain{i1}{0}{2.3}{$1$};
  \dinplain{i2}{0}{1.7}{$2$};
	\flowdiag{q1}{{0.8,4}}{dashed}{}{pos=0.5,rotate=45};\flowdiag{q1}{{0.8,1.8}}{dashed}{}{pos=0.75,rotate=308};
	\flowdiag{i1}{{0.8,3.4}}{}{}{pos=0.5,rotate=45};\flowdiag{i1}{{0.8,1.2}}{}{}{pos=0.5,rotate=308};
	\flowdiag{i2}{{0.8,2.8}}{}{}{pos=0.75,rotate=45};\flowdiag{i2}{{0.8,0.6}}{}{}{pos=0.5,rotate=308};
\end{scope}
\end{scope}
\end{tikzpicture}
\caption{Constructing a head-iterative computon $*_{\rho}(\lambda)$ where $\rho$ is the h-diagram shown in the middle (whose morphisms are displayed as black arrows). Here, $f=\beta_1\circ\lambda_2^-=\beta_2\circ\lambda_3^-$ and $g=\beta_3\circ\lambda_4^+=\beta_4\circ\lambda^+$.}
\label{fig:computon-iterative-head-example}
\end{figure*} 

A glance at Figure \ref{fig:computon-iterative-head-example} reveals that $*_{\rho}(\lambda)$ is constructed from three additional connected computons (i.e., $\lambda_2$, $\lambda_3$ and $\lambda_4$) and two trivial computons. One of the trivial computons serves as the domain for the in-markers $\lambda_4^+$ and $\lambda^+$ so the e-inports of $\lambda_4$ and $\lambda$ match. This trivial computon also serves as the domain for the out-markers $\lambda_2^-$ and $\lambda_3^-$. The other trivial computon is the domain of the in-marker $\lambda_3^+$ and the out-marker $\lambda^-$, i.e., the e-inports of $\lambda_3$ are identified with the e-outports of $\lambda$.

The right-most composite in Figure \ref{fig:computon-iterative-head-example} shows that the connected computons $\lambda_2$ and $\lambda_4$ serve as the respective entry and exit points for the whole iterative structure of $*_{\rho}(\lambda)$. Particularly, $\lambda_2$ is needed because, without this, $*_{\rho}(\lambda)$ will enter into a closed loop with no entry points (i.e., no e-inports); thus, violating Definition \ref{def:computon}. Beyond this, there is no special requirement for the e-inports of $\lambda_2$ or the e-outports of $\lambda_4$, as evidenced by the h-diagram shown in Definition \ref{def:diagram-h}. Not enforcing specific requirements on this matter enables a high degree of modelling flexibility. For instance, it is possible to deem $\lambda_2$ as a computon that replicates data (when its e-inports and e-outports are isomorphic) or as a computon that receives data of a certain type, performs some processing on that data and returns data of a different type. In our particular example, as its e-inports and e-outports do not coincide, we can treat $\lambda_2$ as a computon that pre-processes (or filters) information before sending it into the iterative computation defined over $\lambda$. By Theorem \ref{th:computon-iterative-head-always}, a head-iterative computon can always be formed for any arbitrary connected computon, regardless of the data such an arbitrary computon requires or produces.

\begin{theorem}\label{th:computon-iterative-head-always}
$\lambda$ is a connected computon $\iff$ a head-iterative computon $*_\rho(\lambda)$ exists for some h-diagram $\rho$.
\end{theorem}
\begin{proof}
$(\implies)$ Let $\lambda$ be an arbitrary connected computon. By Proposition \ref{prop:markers-always}, we deduce the existence of an in-marker ${\lambda^+:\lambda_0 \rightarrow \lambda}$ and an out-marker ${\lambda^-:\lambda_1 \rightarrow \lambda}$. Now, if $\lambda_2$ and $\lambda_3$ are connected computons and duals of $\lambda$ (see Proposition \ref{prop:markers-connected-dual}), Definition \ref{def:computon-dual} says there is an in-marker ${\lambda_3^+:\lambda_1 \rightarrow \lambda_3}$ as well as out-markers ${\lambda_3^-:\lambda_0 \rightarrow \lambda_3}$ and ${\lambda_2^-:\lambda_0 \rightarrow \lambda_2}$. Finally, if $\lambda_4$ is a computon isomorphic to $\lambda$, $\lambda_4$ must be connected and there must evidently exists an in-marker ${\lambda_4^+:\lambda_0 \rightarrow \lambda_4}$. 

The above construction corresponds to that of an h-diagram $\rho$ so we simply apply Lemma \ref{lem:computon-iterative-head-exists} to deduce that the colimit of $\rho$ exists. Using Definition \ref{def:computon-iterative-head} and Notation \ref{notation:computon-iterative-head}, we conclude such a colimit is the head-iterative computon $*_\rho(\lambda)$.

$(\impliedby)$ This part of the proof follows directly from Definition \ref{def:diagram-h}.
\end{proof}

\subsubsection{Operational semantics for head-iterative computons (in the theory of Petri nets)}

No matter whether we use any of the three functorial constructions presented in Section \ref{sec:operational-semantics}, the Petri net of a head-iterative computon has no additional places or transitions beyond those from the nets of the computons of the corresponding h-diagram. The general structure of a net of this sort is depicted in Figure \ref{fig:computon-head-net}.

\begin{figure}[!h]
\centering
{
\begin{tikzpicture}
\node[place,label={left:\scriptsize $p_1$},minimum size=3mm] (p1) at (-0.8,1) {};
\node at (-0.8,0.6){$\vdots$};
\node[place,label={left:\scriptsize $p_m$},minimum size=3mm] (pm) at (-0.8,0) {};
\draw[dotted] (0.1,-0.3) rectangle (1.6,1.2);\node at (0.8,0.4){\scriptsize $\lambda_2$-net};
\node[place,label={above:\scriptsize $q_1$},minimum size=3mm] (q1) at (2.4,1.8) {};
\node at (2.4,1.4){$\vdots$};
\node[place,label={below:\scriptsize $q_n$},minimum size=3mm] (qn) at (2.4,0.8) {};

\draw[dotted] (3.3,-0.3) rectangle (4.8,1.2);\node at (4,0.4){\scriptsize $\lambda_4$-net};
\node[place,label={right:\scriptsize $r_1$},minimum size=3mm] (r1) at (5.6,1) {};
\node at (5.6,0.6){$\vdots$};
\node[place,label={right:\scriptsize $r_j$},minimum size=3mm] (rj) at (5.6,0) {};

\draw[dotted] (3.3,2) rectangle (4.8,3.5);\node at (4,2.7){\scriptsize $\lambda$-net};

\node[place,label={above:\scriptsize $s_1$},minimum size=3mm] (s1) at (-0.5,3.2) {};
\node at (-0.5,2.8){$\vdots$};
\node[place,label={below:\scriptsize $s_k$},minimum size=3mm] (sk) at (-0.5,2.2) {};
\draw[dotted] (0.1,2) rectangle (1.6,3.5);\node at (0.8,2.7){\scriptsize $\lambda_3$-net};

\draw[-latex,thick] (p1) -- ($(p1)+(0.9,0)$);\draw[-latex,thick] (pm) -- ($(pm)+(0.9,0)$);
\draw[-latex,thick] ($(qn)+(-0.8,0)$)--(q1);\draw[-latex,thick] ($(qn)+(-0.8,-1)$)--(qn);
\draw[-latex,thick] (q1) -- ($(qn)+(0.9,0)$);\draw[-latex,thick] (qn) -- ($(qn)+(0.9,-0.9)$);
\draw[-latex,thick] (q1) -- ($(q1)+(0.9,1.5)$);\draw[-latex,thick] (qn) -- ($(qn)+(0.9,1.5)$);
\draw[-latex,thick] ($(q1)+(-0.8,1.5)$)--(q1);\draw[-latex,thick] ($(q1)+(-0.8,0.6)$)--(qn);
\draw[-latex,thick] ($(r1)+(-0.8,0)$)--(r1);\draw[-latex,thick] ($(rj)+(-0.8,0)$)--(rj);
\draw[-latex,thick] (4.8,3.2) -- (5.1,3.2) -- (5.1,3.8) -- (-1,3.8) -- (-1,3.2) -- (s1);\draw[-latex,thick] (s1) -- ($(s1)+(0.6,0)$);
\draw[-latex,thick] (4.8,2.8) -- (5.3,2.8) -- (5.3,4) -- (-1.2,4) -- (-1.2,2.2) -- (sk);\draw[-latex,thick] (sk) -- ($(sk)+(0.6,0)$);
\end{tikzpicture}
}
\caption{General structure of the Petri net of a head-iterative computon, considering the h-diagram from Definition \ref{def:diagram-h}. This structure is applicable to the functorial constructions ${\mathcal{N}}$, ${\mathcal{C}\circ\mathfrak{E}}$ and ${\mathcal{D}}$ from Section \ref{sec:operational-semantics}.}
\label{fig:computon-head-net}
\end{figure} 

Unfortunately, there is no guarantee every head-iterative's net is deadlock-free even when the nets of the computons from the corresponding h-diagram are. Despite of this, it is still possible to enforce deadlock-freedom by using primitive computons as entry, exit and iteration points. Proposition \ref{prop:computon-primitive-deadlock} and Remark \ref{rem:computon-primitive-deadlock} together say every primitive computon's net is deadlock-free. So, as long as the net of the computon being iterated over never gets stuck, the corresponding head-iterative's net will be deadlock-free (see Proposition \ref{prop:computon-iterative-head-deadlock} and Remark \ref{rem:head-deadlock}).

\begin{proposition}\label{prop:computon-iterative-head-deadlock}
Consider the h-diagram $\rho$ depicted in Definition \ref{def:diagram-h} and assume ${M_i}$ and ${M_f}$ are the initial and final markings of the net ${\mathcal{N}(\lambda)}$, respectively. The net ${\mathcal{N}(*_{\rho}(\lambda))}$ is deadlock-free if:
\begin{enumerate}
\item ${\lambda_j}$ is a primitive computon for ${j=2,3,4}$,\label{prop:computon-iterative-head-deadlock-1}
\item ${\mathcal{N}(\lambda)}$ is deadlock-free, and \label{prop:computon-iterative-head-deadlock-2}
\item for every marking ${M}$ reachable from ${M_i}$, if ${M(s)>0}$ for each output place $s$ of ${\mathcal{N}(\lambda)}$ then ${M=M_f}$.\label{prop:computon-iterative-head-deadlock-3}
\end{enumerate}
\end{proposition}
\begin{proof}
Consider the h-diagram $\rho$ from Definition \ref{def:diagram-h} and assume $\lambda_j$ is a primitive computon for $j=2,3,4$. By Proposition \ref{prop:computon-primitive-connected}, $\rho$ is a well-defined h-diagram because each $\lambda_j$ is a connected computon. Using Lemma \ref{lem:computon-iterative-head-exists}, we deduce the existence of $*_{\rho}(\lambda)$ whose underlying net $\mathcal{N}(*_{\rho}(\lambda))$ has the following form according to the functorial construction presented in Definition \ref{def:functor-computon-to-petri}:
\begin{center}
\begin{tikzpicture}
\node[place,label={left:\scriptsize $p_1$},minimum size=3mm] (p1) at (-0.8,1) {};
\node at (-0.8,0.6){$\vdots$};
\node[place,label={left:\scriptsize $p_m$},minimum size=3mm] (pm) at (-0.8,0) {};
\node[transition,fill=black,minimum width=0.1mm,minimum height=10mm,label=\scriptsize $\mathcal{N}(\lambda_2)$] (l2) at (0.8,0.5) {};
\node[place,label={above:\scriptsize $q_1$},minimum size=3mm] (q1) at (2.4,1.8) {};
\node at (2.4,1.4){$\vdots$};
\node[place,label={below:\scriptsize $q_n$},minimum size=3mm] (qn) at (2.4,0.8) {};

\node[transition,fill=black,minimum width=0.1mm,minimum height=10mm,label=\scriptsize $\mathcal{N}(\lambda_4)$] (l4) at (4,0.5) {};
\node[place,label={right:\scriptsize $r_1$},minimum size=3mm] (r1) at (4.8,1) {};
\node at (4.8,0.6){$\vdots$};
\node[place,label={right:\scriptsize $r_j$},minimum size=3mm] (rj) at (4.8,0) {};

\draw[dotted] (3.3,2) rectangle (4.8,3.5);\node at (4,2.7){\scriptsize $\mathcal{N}(\lambda)$};

\node[place,label={above:\scriptsize $s_1$},minimum size=3mm] (s1) at (-0.5,3.2) {};
\node at (-0.5,2.8){$\vdots$};
\node[place,label={below:\scriptsize $s_k$},minimum size=3mm] (sk) at (-0.5,2.2) {};
\node[transition,fill=black,minimum width=0.1mm,minimum height=10mm,label=\scriptsize $\mathcal{N}(\lambda_3)$] (l3) at (0.8,2.7) {};

\draw[-latex,thick] (p1) -- (l2);\draw[-latex,thick] (pm) -- (l2);
\draw[-latex,thick] (l2) -- (q1);\draw[-latex,thick] (l2) -- (qn);
\draw[-latex,thick] (q1) -- (l4);\draw[-latex,thick] (qn) -- (l4);
\draw[-latex,thick] (q1) -- ($(q1)+(0.9,1)$);\draw[-latex,thick] (qn) -- ($(qn)+(0.9,1.5)$);
\draw[-latex,thick] (l3)--(q1);\draw[-latex,thick] (l3)--(qn);
\draw[-latex,thick] (l4) -- (r1);
\draw[-latex,thick] (l4) -- (rj);
\draw[-latex,thick] (4.8,3.2) -- (5.1,3.2) -- (5.1,3.8) -- (-1,3.8) -- (-1,3.2) -- (s1);\draw[-latex,thick] (s1) -- (l3);
\draw[-latex,thick] (4.8,2.8) -- (5.3,2.8) -- (5.3,4) -- (-1.2,4) -- (-1.2,2.2) -- (sk);\draw[-latex,thick] (sk) -- (l3);
\end{tikzpicture}
\end{center}
The above net evidently has the form depicted in Figure \ref{fig:computon-head-net}. The only difference is that, rather than black-boxing ${\mathcal{N}(\lambda_j)}$, we display its internals which consist of only one transition (because ${\lambda_j}$ is primitive). By Definition \ref{def:marking}, we know the initial state $M_i$ of ${\mathcal{N}(*_{\rho}(\lambda))}$ is a marking function where ${M_i(p)>0}$ for all ${p\in\{p_1,\ldots,p_m\}}$ and no tokens for all the other places, including those inside ${\mathcal{N}(\lambda)}$. This marking evidently enables the only transition of ${\mathcal{N}(\lambda_2)}$ and nothing else. Consequently, firing ${\mathcal{N}(\lambda_2)}$ reaches a state that marks each place in ${\{q_1,\ldots,q_n\}}$. Assuming ${\mathcal{N}(\lambda)}$ is deadlock-free and that only its final state can put tokens in each element from ${\{s_1,\ldots,s_k\}}$ (see Conditions \ref{prop:computon-iterative-head-deadlock-2} and \ref{prop:computon-iterative-head-deadlock-3}), we now have two possible execution paths (as per mutual exclusion):
\begin{enumerate}
\item If ${\mathcal{N}(\lambda_4)}$ is executed, the final state of ${\mathcal{N}(*_{\rho}(\lambda))}$ will immediately be reached with tokens in ${r_1,\ldots,r_j}$. Therefore, ${\mathcal{N}(*_{\rho}(\lambda))}$ will not get stuck.\label{head-deadlock-1}
\item If ${\mathcal{N}(\lambda)}$ is executed, we have two options: \label{head-deadlock-2}
\begin{enumerate}
\item No state of ${\mathcal{N}(\lambda)}$ ever puts tokens in all the places in ${\{s_1,\ldots,s_k\}}$. In this case, even though ${\mathcal{N}(\lambda)}$ never terminates successfully, there is a guarantee ${\mathcal{N}(*_{\rho}(\lambda))}$ will not get stuck because ${\mathcal{N}(\lambda)}$ is deadlock-free.
\item A state of ${\mathcal{N}(\lambda)}$ puts tokens in all the places in ${\{s_1,\ldots,s_k\}}$. If so, the final marking of ${\mathcal{N}(\lambda_3)}$ will be reached, which simply puts a token in each place in ${\{q_1,\ldots,q_n\}}$. As we have the same two execution alternatives again, we simply repeat \ref{head-deadlock-1} or \ref{head-deadlock-2} whichever applies.
\end{enumerate}
\end{enumerate}
By the above, it is evident that all the execution paths lead to a deadlock-free execution. Therefore, we conclude ${\mathcal{N}(*_{\rho}(\lambda))}$ is deadlock-free, as required.
\end{proof}

\begin{remark}\label{rem:head-deadlock}
Although it is a statement about the functor $\mathcal{N}$, Proposition \ref{prop:computon-iterative-head-deadlock} is applicable to the functors $\mathcal{C}\circ\mathfrak{E}$ and $\mathcal{D}$ presented in Section \ref{sec:operational-semantics}. The proof is valid for $\mathcal{C}\circ\mathfrak{E}$ since Proposition \ref{prop:functor-control-petri} says ${\mathcal{C}}$ is just a restriction of $\mathcal{N}$ to $\mathfrak{E}(\textbf{Set}^{\textbf{Comp}})$. 

Remark \ref{rem:deadlock-freedom} says we are only interested in checking deadlock-freedom for $\mathcal{D}$-nets that have initial and final states. A glance at the figure depicted in the proof of Proposition \ref{prop:computon-iterative-head-deadlock} reveals this is satisfied when $j,m>0$. Starting with the initial state $M_i$ that puts tokens in $p_1,\ldots,p_m$, we have the following cases:
\begin{itemize}
\item If $n=0$, $M_i$ enables the only transition of $\mathcal{D}(\lambda_2)$ which, upon firing, makes $\mathcal{D}(*_{\rho}(\lambda))$ enter into a deadlock state. 
\item If $k=0$ and $n>0$, $\mathcal{D}(\lambda)$ will never reach its final state. Despite of this, $\mathcal{D}(*_{\rho}(\lambda))$ is guaranteed to be deadlock-free when $\mathcal{D}(\lambda)$ also is.
\item If $k>0$ and $n>0$, the proof of deadlock-freedom for $\mathcal{D}(*_{\rho}(\lambda))$ is analogous to that of Proposition \ref{prop:computon-iterative-head-deadlock}.
\end{itemize}
Therefore, to guarantee $\mathcal{D}(*_{\rho}(\lambda))$ is deadlock-free, we must consider an h-diagram $\rho$ where the entry and iteration computons have ed-outports, apart from ensuring that $\mathcal{D}(\lambda)$ is deadlock-free and that satisfies Condition \ref{prop:computon-iterative-head-deadlock-3} of Proposition \ref{prop:computon-iterative-head-deadlock}.
\end{remark}

\subsubsection{Encapsulation of control flow and data flow in head-iterative computons}

\vspace{0.2cm}

By Definition \ref{def:computon-iterative-head}, we know a head-iterative computon $*_{\rho}(\lambda)$ is the colimit of an h-diagram $\rho$ which, by Definition \ref{def:diagram-h}, is formed by four connected computons and two trivial computons. One of the connected computons is $\lambda$ (i.e., the computon being iterated over) whereas the others serve as entry, exit and iteration points. Thus, $*_{\rho}(\lambda)$ encapsulates cyclic control flow and up to cyclic data flow. By cyclic, we mean $\lambda$ and the iteration entity are executed repeatedly. In a head-iterative computon, the decision whether to repeat $\lambda$ is made before executing it. To give a concrete example, Figure \ref{fig:encapsulation-head} illustrates the encapsulation given by the head-iterative computon resulting from the colimit construction depicted in Figure \ref{fig:computon-iterative-head-example}.

\begin{figure}[!h]
\centering
\begin{tikzpicture}[scale=0.9]
\begin{scope}[xshift=-8cm,yshift=0cm]
  \computonComposite{0.5}{0}{6}{4.9};
  
  \computonPrimitive{2.3}{0.2}{1}{1.9}{$\lambda_2$};
  \qinplain{2q0}{0.1}{1.8}{};\draw[dashed] (2q0) to node [pos=0.55] {\arrowflow} ({2.3,1.8});
  \dinplain{2i1}{0.1}{1.2}{$4$};\draw (2i1) to node [pos=0.55] {\arrowflow} ({2.3,1.2});
 	
  \computonPrimitive{2.3}{2.4}{1}{1.9}{$\lambda_3$};  
  \qmatch{xq1}{1.5}{4}{};
  \dmatch{xi1}{1.5}{3.4}{$3$}{above};
	\flow{xq1}{{2.3,4}}{dashed}{};
	\flow{xi1}{{2.3,3.4}}{}{};
	
	\qmatch{yq1}{4.1}{2.9}{};
  \dmatch{yo1}{4.1}{2.3}{$1$}{above};
  \dmatch{yo2}{4.1}{1.7}{$2$}{above};
  \flowdiag{{3.3,4}}{yq1}{dashed}{}{pos=0.5,rotate=308};
  \flowdiag{{3.3,1.8}}{yq1}{dashed}{}{pos=0.25,rotate=45};
	\flowdiag{{3.3,3.4}}{yo1}{}{}{pos=0.4,rotate=308};\flowdiag{{3.3,1.2}}{yo1}{}{}{pos=0.4,rotate=45};
	\flowdiag{{3.3,2.8}}{yo2}{}{}{pos=0.25,rotate=308};\flowdiag{{3.3,0.6}}{yo2}{}{}{pos=0.5,rotate=45};
	\flowdiag{yq1}{{4.9,1.8}}{dashed}{}{pos=0.75,rotate=308};	
	
	\computonPrimitive{4.9}{2.4}{1}{1.9}{$\lambda$};
	\flowdiag{yq1}{$(yq1)+(0.8,1.1)$}{dashed}{}{pos=0.5,rotate=45};
	\flowdiag{yo1}{$(yo1)+(0.8,1.1)$}{}{}{pos=0.5,rotate=45};
	\flowdiag{yo2}{$(yo2)+(0.8,1.1)$}{}{}{pos=0.75,rotate=45};
	
	\computonPrimitive{4.9}{0.2}{1}{1.9}{$\lambda_4$}; 
  \qoutplain{3q0}{6}{1.8}{};\flow{{5.9,1.8}}{3q0}{dashed}{};
  \doutplain{3i1}{6}{1.2}{$5$};\flow{{5.9,1.2}}{3i1}{}{}; 
  \flowdiag{yo1}{{4.9,1.2}}{}{}{pos=0.5,rotate=308};
	\flowdiag{yo2}{{4.9,0.6}}{}{}{pos=0.5,rotate=308};     
  
  \draw[dashed] ($(yq1)+(1.8,1.1)$) -- ($(yq1)+(2,1.1)$) -- ($(yq1)+(2,1.6)$) to node [pos=0.45] {\invertedarrowflow} (1.2,4.5) -- (1.2,4) -- (xq1);
  \draw ($(yo1)+(1.8,1.1)$) -- ($(yo1)+(2.2,1.1)$) -- ($(yo1)+(2.2,2.4)$) to node [pos=0.45] {\invertedarrowflow} (1,4.7) -- (1,3.4) -- (xi1);
\end{scope}
\draw[opacity=0.3,line width=1.5pt, ->, -Latex] (-0.3,2.5) to node[pos=0.4,yshift=7,xshift=-8]{\scriptsize $\mathcal{C}\circ \mathfrak{E}$} (1.2,3.7);
\begin{scope}[xshift=2cm,yshift=2.7cm]
\node[opacity=0.2] at (-1.7,2){\scriptsize Control Flow};\node[opacity=0.2] at (-1.7,1.7){\scriptsize Net};
\node[place,label={80:},minimum size=3mm] (i2) at (0,0) {};
\node[transition,fill=black,minimum width=0.1mm,minimum height=10mm] (2) at (1,0) {};
\node[place,label={80:},minimum size=3mm] (o2) at (2,0.8) {};
\node[transition,fill=black,minimum width=0.1mm,minimum height=10mm] (4) at (3,0) {};
\node[place,label={80:},minimum size=3mm] (o4) at (4,0) {};

\node[place,label={80:},minimum size=3mm] (i3) at (0.5,1.5) {};
\node[transition,fill=black,minimum width=0.1mm,minimum height=10mm] (3) at (1,1.5) {};

\node[transition,fill=black,minimum width=0.1mm,minimum height=10mm] (0) at (3,1.5) {};

\draw[-latex,thick] (i2)--(2);\draw[-latex,thick] (2)--(o2);\draw[-latex,thick] (o2)--(4);\draw[-latex,thick] (o2)--(0);\draw[-latex,thick] (4)--(o4);
\draw[-latex,thick] (0)--($(0)+(0.5,0)$)--($(0)+(0.5,0.8)$)--($(0)+(-3.5,0.8)$)--($(0)+(-3.5,0)$)--(i3);
\draw[-latex,thick] (i3)--(3);
\draw[-latex,thick] (3)--(o2);
\end{scope}
\draw[opacity=0.3,line width=1.5pt, ->, -Latex] (-0.3,2.2) to node[pos=0.4,yshift=-4,xshift=-5]{\scriptsize $\mathcal{D}$} (1.2,1);
\begin{scope}[xshift=2cm,yshift=-0.7cm]
\node[opacity=0.2] at (-0.8,0.9){\scriptsize Data Flow};\node[opacity=0.2] at (-0.8,0.6){\scriptsize Net};
\node[place,label={80:},minimum size=3mm] (i2) at (0,0) {\scriptsize $4$};
\node[transition,fill=black,minimum width=0.1mm,minimum height=10mm] (2) at (1,0) {};
\node[place,label={80:},minimum size=3mm] (o2) at (2,1) {\scriptsize $1$};
\node[place,label={80:},minimum size=3mm] (ox) at (2,0.6) {\scriptsize $2$};
\node[transition,fill=black,minimum width=0.1mm,minimum height=10mm] (4) at (3,0) {};
\node[place,label={80:},minimum size=3mm] (o4) at (4,0) {\scriptsize $5$};

\node[place,label={80:},minimum size=3mm] (i3) at (0.5,1.5) {\scriptsize $3$};
\node[transition,fill=black,minimum width=0.1mm,minimum height=10mm] (3) at (1,1.5) {};

\node[transition,fill=black,minimum width=0.1mm,minimum height=10mm] (0) at (3,1.5) {};

\draw[-latex,thick] (i2)--(2);\draw[-latex,thick] (2)--(o2);\draw[-latex,thick] (2)--(ox);\draw[-latex,thick] (o2)--(4);\draw[-latex,thick] (ox)--(4);
\draw[-latex,thick] (o2)--(0);\draw[-latex,thick] (ox)--(0);\draw[-latex,thick] (4)--(o4);
\draw[-latex,thick] (0)--($(0)+(0.5,0)$)--($(0)+(0.5,0.8)$)--($(0)+(-3.5,0.8)$)--($(0)+(-3.5,0)$)--(i3);
\draw[-latex,thick] (i3)--(3);
\draw[-latex,thick] (3)--(o2);\draw[-latex,thick] (3)--(ox);
\end{scope}
\draw[opacity=0.3,line width=1.5pt, ->, -Latex] (-4,-0.2) to node[pos=0.4,xshift=10]{\scriptsize $\mathcal{N}$} (-4,-1);
\begin{scope}[xshift=-6cm,yshift=-3.6cm]
\node[opacity=0.2] at (-1.2,0.8){\scriptsize Control and};\node[opacity=0.2] at (-1.2,0.5){\scriptsize Data Flow};\node[opacity=0.2] at (-1.2,0.2){\scriptsize Net};
\node[place,label={80:},minimum size=3mm] (ic2) at (0,0.5) {};
\node[place,label={80:},minimum size=3mm] (i2) at (0,-0.5) {\scriptsize $4$};
\node[transition,fill=black,minimum width=0.1mm,minimum height=10mm] (2) at (1,0) {};
\node[place,label={80:},minimum size=3mm] (oc2) at (2,1) {};
\node[place,label={80:},minimum size=3mm] (o2) at (2,0.6) {\scriptsize $1$};
\node[place,label={80:},minimum size=3mm] (ox) at (2,0.2) {\scriptsize $2$};
\node[transition,fill=black,minimum width=0.1mm,minimum height=10mm] (4) at (3,0) {};
\node[place,label={80:},minimum size=3mm] (oc4) at (4,0.5) {};
\node[place,label={80:},minimum size=3mm] (o4) at (4,-0.5) {\scriptsize $5$};

\node[place,label={80:},minimum size=3mm] (oc) at (0.5,1.7) {};
\node[place,label={80:},minimum size=3mm] (i3) at (0.5,1.3) {\scriptsize $3$};
\node[transition,fill=black,minimum width=0.1mm,minimum height=10mm] (3) at (1,1.5) {};

\node[transition,fill=black,minimum width=0.1mm,minimum height=10mm] (0) at (3,1.5) {};

\draw[-latex,thick] (i2)--(2);\draw[-latex,thick] (2)--(o2);\draw[-latex,thick] (2)--(ox);\draw[-latex,thick] (o2)--(4);\draw[-latex,thick] (ox)--(4);
\draw[-latex,thick] (o2)--(0);\draw[-latex,thick] (ox)--(0);\draw[-latex,thick] (4)--(o4);
\draw[-latex,thick] ($(0)+(0,-0.3)$)--($(0)+(0.7,-0.3)$)--($(0)+(0.7,1)$)--($(0)+(-3.8,1)$)--($(0)+(-3.8,-0.2)$)--(i3);
\draw[-latex,thick] (i3)--(3);
\draw[-latex,thick] (3)--(o2);\draw[-latex,thick] (3)--(ox);
\draw[-latex,thick] (ic2)--(2);\draw[-latex,thick] (2)--(oc2);\draw[-latex,thick] (oc2)--(4);\draw[-latex,thick] (oc2);
\draw[-latex,thick] (oc2)--(0);
\draw[-latex,thick] (4)--(oc4);
\draw[-latex,thick] (0)--($(0)+(0.5,0)$)--($(0)+(0.5,0.8)$)--($(0)+(-3.5,0.8)$)--($(0)+(-3.5,0.2)$)--(oc);
\draw[-latex,thick] (oc)--(3);
\end{scope}
\end{tikzpicture}
\caption{Cyclic control flow and cyclic data flow encapsulated by the head-iterative computon from Figure \ref{fig:computon-iterative-head-example}. We label some places for mapping purposes even though Petri nets are not labelled (see Section \ref{sec:operational-semantics}).}
\label{fig:encapsulation-head}
\end{figure}

\subsubsection{Tail-Iterative Computons}

\vspace{0.2cm}

A \emph{tail-iterative computon} is structurally similar to a head-iterative one in the sense it is formed from the same basic building blocks, namely four connected computons and two trivial computons, as specified by the notion of a \emph{t-diagram} (see Definition \ref{def:diagram-t}). The difference lies in the position of the structure that non-deterministically chooses continuation or termination of the iterative computation. While a head-iterative computon defines such a structure just before the computon being iterated over, a tail-iterative one specifies it immediately after. Analogically, in the realm of imperative programming languages, a head-iterative computon corresponds to a \emph{while} construct, whereas a tail-iterative is akin to a \emph{do-while} statement. Like head-iterative computons, tail-iteratives are connected computons which can always be constructed in $\textbf{Set}^\textbf{Comp}$ (see Definition \ref{def:computon-iterative-tail}, Lemma \ref{lem:computon-iterative-tail-exists} and Corollary \ref{cor:computon-iterative-tail-connected}). 

\begin{definition}[T-Diagram]\label{def:diagram-t}
A t-diagram is a diagram with the following shape in $\textbf{Set}^\textbf{Comp}$:
\[
\begin{tikzcd}
 & \lambda_0 \arrow[dl, "\lambda_2^-"']\arrow[dr, "\lambda_3^-"]\arrow[ddr, "\lambda^+"'] & & \lambda_1 \arrow[dl, "\lambda_3^+"']\arrow[dr, "\lambda_4^+"]\arrow[ddl, "\lambda^-"] \\
\lambda_2 & & \lambda_3 & & \lambda_4 \\
 & & \lambda & & 
\end{tikzcd}
\]

\vspace{10pt}
\noindent where:
\begin{itemize}
\item $\lambda$ is a connected computon with an in-marker $\lambda^+$ and an out-marker $\lambda^-$,
\item $\lambda_2$ is a connected computon with an out-marker $\lambda_2^-$,
\item $\lambda_3$ is a connected computon with an in-marker $\lambda_3^+$ and an out-marker $\lambda_3^-$, and
\item $\lambda_4$ is a connected computon with an in-marker $\lambda_4^+$.
\end{itemize}
Evidently, by Definition \ref{def:computon-morphism-markers}, $\lambda_0$ and $\lambda_1$ are trivial computons serving as domains for the markers involved in the t-diagram.
\end{definition}

\begin{definition}[Tail-Iterative Computon]\label{def:computon-iterative-tail}
A tail-iterative computon is the colimit of a t-diagram.
\end{definition}

\begin{notation}\label{notation:computon-iterative-tail}
For convenience, we write a star symbol after a computon symbol $\lambda$ to indicate that the decision-making locus that terminates the iterative process is placed immediately after the computational structure of $\lambda$. For example, we write $(\lambda)*_{\rho}$ for the colimit of the t-diagram $\rho$ shown in Definition \ref{def:diagram-t}.
\end{notation}

\begin{lemma}\label{lem:computon-iterative-tail-exists}
A tail-iterative computon can always be constructed in $\textbf{Set}^\textbf{Comp}$.
\end{lemma}
\begin{proof}
Considering the t-diagram shown in Definition \ref{def:diagram-t}, by Propositions \ref{prop:pushout-pushable} and \ref{prop:markers-pushable}, we know that the pushouts of ${\lambda_2 \xleftarrow{\lambda_2^-} \lambda_0 \xrightarrow{\lambda_3^-} \lambda_3}$ and ${\lambda_3 \xleftarrow{\lambda_3^+} \lambda_1 \xrightarrow{\lambda_4^+} \lambda_4}$ can be constructed. Let us denote them $(\beta_1:\lambda_2 \rightarrow \lambda_5, \lambda_5,\beta_2:\lambda_3 \rightarrow \lambda_5)$ and $(\beta_3:\lambda_3 \rightarrow \lambda_6, \lambda_6,\beta_4:\lambda_4 \rightarrow \lambda_6)$, respectively.

To show that the induced span ${\lambda_5 \xleftarrow{\beta_2} \lambda_3 \xrightarrow{\beta_3} \lambda_6}$ is pushable, we just prove $\beta_2(\vec{o}(\beta_3)) \subseteq P_5^+ \cup P_5^-$ since the other conditions of Definition \ref{def:computon-morphisms-pushable} follow analogously. For this, assume ${p_5 \in \beta_2(\vec{o}(\beta_3))}$ so there is some ${p_3 \in \vec{o}(\beta_3)}$ where ${\beta_2(p_3)=p_5}$. As $\beta_3$ is a computon morphism, Definition \ref{def:computon-morphism} says ${\vec{i}(\beta_3)\cup\vec{o}(\beta_3)\subseteq P_3^+\cup P_3^-}$. Supposing for contradiction ${\beta_2(p_3)\notin P_5^+\cup P_5^-}$, we have ${t_5(o_5)=\beta_2(p_3)=s_5(i_5)}$ for some ${o_5\in O_5}$ and some ${i_5\in I_5}$. Accordingly, we consider the following cases:
\begin{enumerate}
\item When ${p_3\in P_3^+}$, we use ${O_5=O_2+_{O_0}O_3}$ to determine the preimage of $o_5$. If ${o_5=\beta_1(o_2)}$ for some ${o_2\in O_2}$, ${\beta_1\circ t_2=t_5\circ\beta_1}$ entails ${\beta_1(t_2(o_2))=t_5(\beta_1(o_2))=t_5(o_5)=\beta_2(p_3)}$. As ${\beta_1}$ and ${\beta_2}$ are morphisms induced by the pushout of a span of out-markers, $\beta_1(t_2(o_2))=\beta_2(p_3)$ if and only if ${t_2(o_2)\in P_2^-}$ and ${p_3\in P_3^-}$. But ${\lambda_3}$ is a connected computon as per Definition \ref{def:diagram-t} so ${p_3\in P_3^-\cap P_3^+}$ cannot hold as per Proposition \ref{prop:computon-connected-isolated-port}. If ${o_5=\beta_2(o_3)}$ for some ${o_3\in O_3}$, ${\beta_2\circ t_3=t_5\circ\beta_2}$ entails ${\beta_2(t_3(o_3))=t_5(\beta_2(o_3))=t_5(o_5)=\beta_2(p_3)}$. As ${t_3(o_3)=p_3 \in P_3^+}$ directly contradicts Definition \ref{def:computon-interface}, we focus on ${t_3(o_3)\neq p_3}$. By considering that the $P$-component of $\beta_2$ can only be non-injective on the image of the $P$-component of $\lambda_3^-$, we deduce ${\beta_2(t_3(o_3))=\beta_2(p_3)}$ with ${t_3(o_3)\neq p_3\iff t_3(o_3),p_3\in P_3^-}$. Again, ${p_3\in P_3^-\cap P_3^+}$ contradicts Proposition \ref{prop:computon-connected-isolated-port} because $\lambda_3$ is a connected computon.
\item When ${p_3\in P_3^-}$, we use ${I_5=I_2+_{I_0}I_3}$ to determine the preimage of $i_5$. If ${i_5=\beta_1(i_2)}$ for some ${i_2\in I_2}$, ${\beta_1\circ s_2=s_5\circ\beta_1}$ entails ${\beta_1(s_2(i_2))=s_5(\beta_1(i_2))=s_5(i_5)=\beta_2(p_3)}$. As ${\beta_1}$ and ${\beta_2}$ are morphisms induced by the pushout of a span of out-markers, $\beta_1(s_2(i_2))=\beta_2(p_3)$ if and only if ${s_2(i_2)\in P_2^-}$ and ${p_3\in P_3^-}$. But ${s_2(i_2)\in P_2^-}$ cannot be true because that would contradict Definition \ref{def:computon-interface}. If ${i_5=\beta_2(i_3)}$ for some ${i_3\in I_3}$, ${\beta_2\circ s_3=s_5\circ\beta_2}$ entails ${\beta_2(s_3(i_3))=s_5(\beta_2(i_3))=s_5(i_5)=\beta_2(p_3)}$. As ${s_3(i_3)=p_3 \in P_3^-}$ directly contradicts Definition \ref{def:computon-interface}, we focus on ${s_3(i_3)\neq p_3}$. By considering that the $P$-component of $\beta_2$ can only be non-injective on the image of the $P$-component of $\lambda_3^-$, we deduce ${\beta_2(s_3(i_3))=\beta_2(p_3)}$ with ${s_3(i_3)\neq p_3}$ if and only if ${s_3(i_3),p_3\in P_3^-}$. Clearly, ${s_3(i_3)\in P_3^-}$ contradicts Definition \ref{def:computon-interface}.
\end{enumerate}

Disproving the above scenarios implies that our initial assumption must be false. So, ${\beta_2(p_3)\in P_5^+\cup P_5^-}$. Showing ${\lambda_5 \xleftarrow{\beta_2} \lambda_3 \xrightarrow{\beta_3} \lambda_6}$ is pushable allows us to use Proposition \ref{prop:pushout-pushable} to construct its pushout, denoted $(\beta_5:\lambda_5 \rightarrow \lambda_7, \lambda_7,\beta_6:\lambda_6 \rightarrow \lambda_7)$. By pushout commutativity, we deduce the existence of computon morphisms $f:\lambda_0 \rightarrow \lambda_7$ and $g:\lambda_1 \rightarrow \lambda_7$ where $f=\beta_5\circ\beta_1\circ\lambda_2^-=\beta_5\circ\beta_2\circ\lambda_3^-=\beta_6\circ\beta_3\circ\lambda_3^-$ and $g=\beta_5\circ\beta_2\circ\lambda_3^+=\beta_6\circ\beta_3\circ\lambda_3^+=\beta_6\circ\beta_4\circ\lambda_4^+$.

Now, by Proposition \ref{prop:computon-coproduct}, we know that the coproduct ${\lambda_0+\lambda_1}$ can be formed. Using the universal property of coproducts, we also know there must be unique computon morphisms ${(\lambda^+,\lambda^-):\lambda_0+\lambda_1 \rightarrow \lambda}$ and ${(f,g):\lambda_0+\lambda_1 \rightarrow \lambda_7}$. To show that ${\lambda_7 \xleftarrow{(f,g)} \lambda_0+\lambda_1 \xrightarrow{(\lambda^+,\lambda^-)} \lambda}$ is pushable, assume ${p_7 \in (f,g)(\vec{i}(\lambda^+,\lambda^-))}$ so there is some ${q \in \vec{i}(\lambda^+,\lambda^-)}$ where ${(f,g)(q)=p_7}$. As ${\bullet(\lambda^+,\lambda^-)(q)\setminus(\lambda^+,\lambda^-)(\bullet q)\neq\emptyset}$, we have ${(\lambda^+,\lambda^-)(q) \notin P^+}$. Consequently, by coproduct definition, ${(\lambda^+,\lambda^-)(q)}$ must be in the image of ${\lambda^-}$. That is, there must be some ${p_1 \in P_1}$ for which ${\lambda^-(p_1)=(\lambda^+,\lambda^-)(q) \in P^-}$. As ${\lambda_3^+}$ and ${\lambda_4^+}$ are in-markers, we can apply Proposition \ref{prop:markers-inout}, while recognising that $\lambda_5$ does not identify $P_3^+$-ports, to further deduce ${g(p_1) \in P_7^+}$. By coproduct definition, we have ${(f,g)(q)=p_7=g(p_1) \in P_7^+}$ and, therefore, ${(f,g)(\vec{i}(\lambda^+,\lambda^-)) \subseteq P_7^+ \cup P_7^-}$. As proving ${(f,g)(\vec{o}(\lambda^+,\lambda^-)) \subseteq P_7^+ \cup P_7^-}$ and $(\lambda^+,\lambda^-)(\vec{i}(f,g))\cup(\lambda^+,\lambda^-)(\vec{o}(f,g)) \subseteq P^+ \cup P^-$ can be done analogously, the pushout ${(\beta_7:\lambda \rightarrow \lambda_8, \lambda_8, \beta_8:\lambda_7 \rightarrow \lambda_8)}$ of ${(\lambda^+,\lambda^-)}$ and ${(f,g)}$ can be constructed. To show such a pushout satisfies the universal property of the colimit of the original t-diagram, suppose there is a cone:
\[
\begin{tikzcd}
 & \lambda_0 \arrow[dl, "\lambda_2^-"']\arrow[dr, "\lambda_3^-"]\arrow[ddr, "\lambda^+"']\arrow[dddr, bend right=30, opacity=0.4] & & \lambda_1 \arrow[dl, "\lambda_3^+"']\arrow[dr, "\lambda_4^+"]\arrow[ddl, "\lambda^-"]\arrow[dddl, bend left=30, opacity=0.4] \\
\lambda_2 \arrow[ddrr, bend right=30, opacity=0.4] & & \lambda_3 \arrow[dd, bend right=30, opacity=0.4] & & \lambda_4 \arrow[ddll, bend left=30, opacity=0.4] \\
 & & \lambda \arrow[d, opacity=0.4] & & \\
 & & \lambda_9 & &
\end{tikzcd}
\]
Since $\lambda_5$ and $\lambda_6$ are the pushouts of ${\lambda_2 \xleftarrow{\lambda_2^-} \lambda_0 \xrightarrow{\lambda_3^-} \lambda_3}$ and ${\lambda_3 \xleftarrow{\lambda_3^+} \lambda_1 \xrightarrow{\lambda_4^+} \lambda_4}$, respectively, we know there are unique computon morphisms $\lambda_5 \rightarrow \lambda_9$ and $\lambda_6 \rightarrow \lambda_9$ that make the corresponding triangles commute. Taking into account these morphisms and considering that $\lambda_7$ is the pushout of the induced span ${\lambda_5 \xleftarrow{\beta_2} \lambda_3 \xrightarrow{\beta_3} \lambda_6}$, we use the universal property of pushouts to deduce there is a unique computon morphism $\lambda_7 \rightarrow \lambda_9$ that also makes the corresponding diagram commute. 

In the above cone, there are computon morphisms ${\lambda_0 \rightarrow \lambda_9}$ and ${\lambda_1 \rightarrow \lambda_9}$. Using the universal property of coproducts, we deduce there is a unique computon morphism ${\lambda_0+\lambda_1 \rightarrow \lambda_9}$. As in the cone there also is ${\lambda \rightarrow \lambda_9}$, we use the universal property of pushouts to deduce the existence of a unique computon morphism ${\lambda_8 \rightarrow \lambda_9}$. Thus, proving that ${\lambda_8}$ is the colimit of the original t-diagram.
\end{proof}

\begin{corollary}\label{cor:computon-iterative-tail-connected}
A tail-iterative computon is a connected computon.
\end{corollary}
\begin{proof}
Considering the construction presented in the proof of Lemma \ref{lem:computon-iterative-tail-exists}, we know that the pushouts $\lambda_5$ and $\lambda_6$ are connected computons because $\lambda_2$, $\lambda_3$ and $\lambda_4$ also are (see Proposition \ref{prop:pushout-connected}). Consequently, the pushout $\lambda_7$ of the induced span ${\lambda_5 \xleftarrow{\beta_2} \lambda_3 \xrightarrow{\beta_3} \lambda_6}$ is also a connected computon.

As $\lambda_7$ and $\lambda$ are connected computons, we use again Proposition \ref{prop:pushout-connected} to deduce that the pushout $\lambda_8$ of the unique (pushable) span ${\lambda_7 \xleftarrow{(f,g)} \lambda_0+\lambda_1 \xrightarrow{(\lambda^+,\lambda^-)} \lambda}$ is a connected computon. As $\lambda_8$ is the colimit of the original t-diagram (shown in Definition \ref{def:diagram-t}), we conclude that every tail-iterative computon is a connected computon.
\end{proof}

\begin{figure*}[!h]
\centering
\begin{tikzpicture}[scale=0.77]
\begin{scope}[xshift=3cm,yshift=19.5cm]
  \computonPrimitive{0.8}{0}{1}{1.9}{$\lambda$}
  \qin{q0}{0}{1.7}{}
  \din{i1}{0}{1.1}{$1$};
  \din{i2}{0}{0.5}{$2$};
  \qout{q1}{1.8}{1.7}{}
  \dout{o1}{1.8}{1.1}{$3$};
\end{scope}
\begin{scope}[xshift=-0.5cm,yshift=14.4cm]
\qmatch{i0}{0}{1.7}{};
\dmatch{i1}{0}{1.3}{$1$}{left};
\dmatch{i2}{0}{0.9}{$2$}{left};
\qmatch{o0}{0.5}{1.7}{};
\dmatch{o1}{0.5}{1.3}{$3$}{right};
\end{scope}
\begin{scope}[xshift=8.5cm,yshift=15.4cm]
  \computonComposite{1.7}{0.3}{7.4}{4.6};
  
  \computonPrimitive{2.3}{1.6}{1}{1.9}{$\lambda_2$}; 
  \qin{2q0}{1.5}{3.2}{};
  \din{2i1}{1.5}{2.6}{$1$};
  \din{2i2}{1.5}{2}{$2$};
 
  \qmatch{yq1}{4.1}{3.2}{};
  \dmatch{yo1}{4.1}{2.6}{$1$}{above};
  \dmatch{yo2}{4.1}{2}{$2$}{above};  
  \flow{{3.3,3.2}}{yq1}{dashed}{};
  \flow{{3.3,2.6}}{yo1}{}{};
  \flow{{3.3,2}}{yo2}{}{};
	
  \computonPrimitive{4.9}{0.5}{1}{1.9}{$\lambda$};
  \qmatch{3q0}{6.8}{2.8}{};
  \dmatch{3i1}{6.8}{2.2}{$3$}{above};
  \flowdiag{{4.9,1.9}}{yq1}{dashed}{}{pos=0.25,rotate=308};
  \flowdiag{{4.9,1.3}}{yo1}{}{}{pos=0.4,rotate=308};
  \flowdiag{{4.9,0.7}}{yo2}{}{}{pos=0.4,rotate=308};
  \flowdiag{{5.9,1.9}}{3q0}{dashed}{}{pos=0.5,rotate=45};
  \flowdiag{{5.9,1.3}}{3i1}{}{}{pos=0.5,rotate=45};
  
  \computonPrimitive{4.9}{2.7}{1}{1.9}{$\lambda_3$};  
  \flowdiag{3q0}{{5.9,4}}{dashed}{}{pos=0.5,rotate=135};
  \flowdiag{3i1}{{5.9,3.4}}{}{}{pos=0.5,rotate=135};
  \flowdiag{{4.9,4.3}}{yq1}{dashed}{}{pos=0.3,rotate=225};
  \flowdiag{{4.9,3.7}}{yo1}{}{}{pos=0.3,rotate=225};
  \flowdiag{{4.9,3.1}}{yo2}{}{}{pos=0.2,rotate=225};
  
  \computonPrimitive{7.6}{1.6}{1}{1.9}{$\lambda_4$};
  \qout{4q0}{8.6}{3.2}{};
  \dout{4o0}{8.6}{2.6}{$3$};
  \flow{3q0}{{7.6,2.8}}{dashed}{};
  \flow{3i1}{{7.6,2.2}}{}{};
    
\end{scope}

\begin{scope}
\begin{scope}\draw[->,opacity=0.4,dashed] (-0.35,15) to[bend right=30] node[left]{\scriptsize $(f,g)$} (6,1);\end{scope}
\begin{scope}\draw[->,opacity=0.4,dashed] (-0.3,16.5) to[bend left=30] node[left]{\scriptsize $(\lambda^+,\lambda^-)$} (2.4,20.5);\end{scope}
\begin{scope}\draw[->,opacity=0.4] (9,1) to[bend right=30] node[right]{\scriptsize $\beta_8$} (15,15.4);\end{scope}
\begin{scope}\draw[->,opacity=0.4] (6.3,20.5) to[bend left=30] node[right, yshift=1.4mm]{\scriptsize $\beta_7$} (10,19.3);\end{scope}

\begin{scope}[xshift=3.9cm,yshift=14.2cm]\draw[->] (0.4,1.3) to node[left]{\scriptsize $\lambda^+$} (0.4,5.1);\end{scope}
\begin{scope}[xshift=6.4cm,yshift=14.2cm]\draw[->] (3.3,1) to node[right]{\scriptsize $\lambda^-$} (-1.8,5.1);\end{scope}

\begin{scope}[xshift=0.7cm,yshift=14.2cm]\draw[->,opacity=0.4] (3.2,1.1) to[bend right=25] node[left]{} (0,1.7);\end{scope}
\begin{scope}[xshift=0.7cm,yshift=14.2cm]\draw[->,opacity=0.4] (9,0.9) to[bend right=25] node[left]{} (0,2);\end{scope}

\begin{scope}[xshift=0.7cm,yshift=12.2cm]\draw[->] (2,2.6) to node[left]{\scriptsize $\lambda_2^-$} (0.7,1);\end{scope}
\begin{scope}[xshift=5cm,yshift=12.2cm]\draw[->] (0.7,2.6) to node[right]{\scriptsize $\lambda_3^-$} (2,1);\end{scope}
\begin{scope}[xshift=6.8cm,yshift=12.2cm]\draw[->] (2,2.6) to node[left]{\scriptsize $\lambda_3^+$} (0.7,1);\end{scope}
\begin{scope}[xshift=11cm,yshift=12.2cm]\draw[->] (0.7,2.6) to node[right]{\scriptsize $\lambda_4^+$} (2,1);\end{scope}
\begin{scope}[xshift=0.7cm,yshift=8.2cm]\draw[->,opacity=0.4] (0.7,2.6) to node[left]{\scriptsize $\beta_1$} (2,1);\end{scope} 
\begin{scope}[xshift=5.3cm,yshift=8.2cm]\draw[->,opacity=0.4] (2,2.6) to node[left]{\scriptsize $\beta_2$} (0.7,1);\end{scope} 
\begin{scope}[xshift=6.6cm,yshift=8.2cm]\draw[->,opacity=0.4] (0.7,2.6) to node[right]{\scriptsize $\beta_3$} (2,1);\end{scope} 
\begin{scope}[xshift=11.1cm,yshift=8.2cm]\draw[->,opacity=0.4] (2,2.6) to node[right]{\scriptsize $\beta_4$} (0.7,1);\end{scope} 
\begin{scope}[xshift=3.7cm,yshift=3cm]\draw[->,opacity=0.4] (0.7,2.6) to node[left]{\scriptsize $\beta_5$} (2,1);\end{scope} 
\begin{scope}[xshift=8.7cm,yshift=3cm]\draw[->,opacity=0.4] (2,2.6) to node[right]{\scriptsize $\beta_6$} (0.7,1);\end{scope} 

\begin{scope}[xshift=4.3cm,yshift=14.4cm]
\qmatch{i0}{0}{0.8}{};
\dmatch{i1}{0}{0.4}{$1$}{left};
\dmatch{i2}{0}{0}{$2$}{left};
\end{scope}
\begin{scope}[xshift=10.1cm,yshift=14.7cm]
\qmatch{o0}{0}{0.4}{};
\dmatch{o1}{0}{0}{$3$}{left};
\end{scope}

\begin{scope}[xshift=0.2cm,yshift=11cm]
  \computonPrimitive{0.8}{0}{1}{1.9}{$\lambda_2$}
  \qin{2q0}{0}{1.7}{}
  \din{2i1}{0}{1.1}{$1$};
  \din{2i2}{0}{0.5}{$2$};
  \qout{2q1}{1.8}{1.7}{}
  \dout{2o1}{1.8}{1.1}{$1$};
  \dout{2o2}{1.8}{0.5}{$2$};
\end{scope}
\begin{scope}[xshift=6cm,yshift=11cm]
  \computonPrimitive{0.8}{0}{1}{1.9}{$\lambda_3$}
  \qin{3q0}{0}{1.7}{}
  \din{3i1}{0}{1.1}{$3$};
  \qout{3q1}{1.8}{1.7}{}
  \dout{3o1}{1.8}{1.1}{$1$};
  \dout{3o1}{1.8}{0.5}{$2$};
\end{scope}
\begin{scope}[xshift=11.7cm,yshift=11cm]
\computonPrimitive{0.8}{0}{1}{1.9}{$\lambda_4$}
  \qin{3q0}{0}{1.7}{}
  \din{3i1}{0}{1.1}{$3$};
  \qout{3q1}{1.8}{1.7}{}
  \dout{3o1}{1.8}{1.1}{$3$};
\end{scope}

\begin{scope}[xshift=3.2cm,yshift=5.8cm]
  \computonComposite{0.3}{0}{2.1}{4.5};
  
  \computonPrimitive{0.8}{2.4}{1}{1.9}{$\lambda_2$}
  \qin{2q0}{0}{4}{}
  \din{2i1}{0}{3.4}{$1$};
  \din{2i2}{0}{2.8}{$2$};  
  \computonPrimitive{0.8}{0.2}{1}{1.9}{$\lambda_3$}
  \qin{3q0}{0}{1.8}{}
  \din{3i1}{0}{1.2}{$3$};
  
  \qoutplain{q1}{1.8}{2.9}{}
  \doutplain{o1}{1.8}{2.3}{$1$};
  \doutplain{o2}{1.8}{1.7}{$2$};
	\flowdiag{{1.8,4}}{q1}{dashed}{}{pos=0.5,rotate=308};\flowdiag{{1.8,1.8}}{q1}{dashed}{}{pos=0.25,rotate=45};
	\flowdiag{{1.8,3.4}}{o1}{}{}{pos=0.35,rotate=308};\flowdiag{{1.8,1.2}}{o1}{}{}{pos=0.25,rotate=45};
	\flowdiag{{1.8,2.8}}{o2}{}{}{pos=0.25,rotate=308};\flowdiag{{1.8,0.6}}{o2}{}{}{pos=0.25,rotate=45};
\end{scope}
\begin{scope}[xshift=8.6cm,yshift=5.8cm]
  \computonComposite{0.3}{0}{2.1}{4.5};
  
  \computonPrimitive{0.8}{2.4}{1}{1.9}{$\lambda_3$}
  \qout{2q0}{1.8}{4}{}
  \dout{2i1}{1.8}{3.4}{$1$};
  \dout{2i2}{1.8}{2.8}{$2$};  
  \computonPrimitive{0.8}{0.2}{1}{1.9}{$\lambda_4$}
  \qout{3q0}{1.8}{1.8}{}
  \dout{3i1}{1.8}{1.2}{$3$};
  
  \qinplain{q1}{0}{2.9}{}
  \dinplain{i1}{0}{2.3}{$3$};
	\flowdiag{q1}{{0.8,4}}{dashed}{}{pos=0.5,rotate=45};\flowdiag{q1}{{0.8,1.8}}{dashed}{}{pos=0.5,rotate=308};
	\flowdiag{i1}{{0.8,3.4}}{}{}{pos=0.5,rotate=45};\flowdiag{i1}{{0.8,1.2}}{}{}{pos=0.5,rotate=308};
\end{scope}

\begin{scope}[xshift=6cm]
	\computonComposite{0.3}{0}{2.1}{6.7};
  
  \computonPrimitive{0.8}{4.6}{1}{1.9}{$\lambda_2$}
  \qin{2q0}{0}{6.2}{}
  \din{2i1}{0}{5.6}{$1$};
  \din{2i2}{0}{5}{$2$};  
 	
  \computonPrimitive{0.8}{2.4}{1}{1.9}{$\lambda_3$}  
  
  \computonPrimitive{0.8}{0.2}{1}{1.9}{$\lambda_4$}
  \qout{3q0}{1.8}{1.8}{}
  \dout{3i1}{1.8}{1.2}{$3$};   
  
  \qinplain{q1}{0}{2.9}{}
  \dinplain{i1}{0}{2.3}{$3$};
	\flowdiag{q1}{{0.8,4}}{dashed}{}{pos=0.5,rotate=45};\flowdiag{q1}{{0.8,1.8}}{dashed}{}{pos=0.5,rotate=308};
	\flowdiag{i1}{{0.8,3.4}}{}{}{pos=0.5,rotate=45};\flowdiag{i1}{{0.8,1.2}}{}{}{pos=0.5,rotate=308};
	
	\qoutplain{q1}{1.8}{5.1}{}
  \doutplain{o1}{1.8}{4.5}{$1$};
  \doutplain{o2}{1.8}{3.9}{$2$};
	\flowdiag{{1.8,6.2}}{q1}{dashed}{}{pos=0.5,rotate=308};\flowdiag{{1.8,4}}{q1}{dashed}{}{pos=0.25,rotate=45};
	\flowdiag{{1.8,5.6}}{o1}{}{}{pos=0.35,rotate=308};\flowdiag{{1.8,3.4}}{o1}{}{}{pos=0.35,rotate=45};
	\flowdiag{{1.8,5}}{o2}{}{}{pos=0.25,rotate=308};\flowdiag{{1.8,2.8}}{o2}{}{}{pos=0.5,rotate=45};
\end{scope}
\end{scope}
\end{tikzpicture}
\caption{Constructing a tail-iterative computon $(\lambda)*_{\rho}$ where $\rho$ is the t-diagram whose morphisms are displayed as black arrows. Here, $f=\beta_5\circ\beta_1\circ\lambda_2^-=\beta_5\circ\beta_2\circ\lambda_3^-=\beta_6\circ\beta_3\circ\lambda_3^-$ and $g=\beta_5\circ\beta_2\circ\lambda_3^+=\beta_6\circ\beta_3\circ\lambda_3^+=\beta_6\circ\beta_4\circ\lambda_4^+$.}
\label{fig:computon-iterative-tail-example}
\end{figure*}

Building upon the proof of Lemma \ref{lem:computon-iterative-tail-exists}, Figure \ref{fig:computon-iterative-tail-example} shows a detailed, self-descriptive example for the construction of a tail-iterative computon $(\lambda)*_{\rho}$ where $\lambda$ is the same connected computon we use in Figure \ref{fig:computon-iterative-head-example} and $\rho$ is the t-diagram whose morphisms are displayed as black arrows. 

Figure \ref{fig:computon-iterative-tail-example} shows that, like head-iterative computons, the connected computons $\lambda_2$ and $\lambda_4$ serve as entry and exit points for the iterative computation, respectively, while the connected computon $\lambda_3$ enables the repeated invocation of $\lambda$. Although in this example the e-inports and e-outports of $\lambda_2$ are isomorphic (the same for $\lambda_4$), there no strict requirement for enforcing this as there is no in-marker $\lambda_2^+$ and no out-marker $\lambda_4^-$ in the corresponding t-diagram $\rho$. Not enforcing this structural feature enables a certain degree of flexibility in the sense the endpoints of a tail-iterative computon can or cannot expose the interface of the computon being iterated over (i.e., $\lambda$). Again, like head-iterative computons, it is possible to operationally implement $\lambda_2$ and $\lambda_4$ in different manners. For instance, in our particular scenario, $\lambda_2$ can be treated as a computon that either duplicates information or transforms data of the same type. As we are dealing with high-level computations, the internals of such functional computons are irrelevant. We just focus on structure from a ``birds-eye viewpoint''. By Theorem \ref{th:computon-iterative-tail-always}, a tail-iterative composite can always be constructed for any arbitrary connected computon, regardless of the data such an arbitrary computon requires or produces.

\begin{theorem}\label{th:computon-iterative-tail-always}
$\lambda$ is a connected computon $\iff$ a tail-iterative computon $(\lambda)*_{\rho}$ exists for some t-diagram $\rho$.
\end{theorem}
\begin{proof}
$(\implies)$ Let $\lambda$ be an arbitrary connected computon. By Proposition \ref{prop:markers-always}, we deduce there is an in-marker ${\lambda^+:\lambda_0 \rightarrow \lambda}$ and an out-marker ${\lambda^-:\lambda_1 \rightarrow \lambda}$. Now, if $\lambda_2$, $\lambda_3$ and $\lambda_4$ are connected computons and duals of $\lambda$ (see Proposition \ref{prop:markers-connected-dual}), Definition \ref{def:computon-dual} says there are in-markers ${\lambda_3^+:\lambda_1 \rightarrow \lambda_3}$ and ${\lambda_4^+:\lambda_1 \rightarrow \lambda_4}$ as well as out-markers ${\lambda_2^-:\lambda_0 \rightarrow \lambda_2}$ and ${\lambda_3^-:\lambda_0 \rightarrow \lambda_3}$. 

As the above construction corresponds to that of a t-diagram $\rho$, by Lemma \ref{lem:computon-iterative-tail-exists}, we have that the colimit of $\rho$ exists. Using Definition \ref{def:computon-iterative-tail} and Notation \ref{notation:computon-iterative-tail}, we conclude such a colimit is the tail-iterative computon $(\lambda)*_\rho$.

$(\impliedby)$ This part of the proof follows directly from Definition \ref{def:diagram-t}. 
\end{proof}

\subsubsection{Operational semantics for tail-iterative computons (in the theory of Petri nets)}

No matter whether we use any of the three functorial constructions from Section \ref{sec:operational-semantics}, the Petri net of a tail-iterative computon has no additional places or transitions beyond those from the nets of the computons of the corresponding t-diagram. The general structure of a net of this sort is depicted in Figure \ref{fig:computon-tail-net}.

\begin{figure}[!h]
\centering
{
\begin{tikzpicture}
\node[place,label={left:\scriptsize $p_1$},minimum size=3mm] (p1) at (-1.7,2.4) {};
\node at (-1.7,2){$\vdots$};
\node[place,label={left:\scriptsize $p_m$},minimum size=3mm] (pm) at (-1.7,1.4) {};
\draw[dotted] (-1.2,1.1) rectangle (0.3,2.6);\node at (-0.5,1.8){\scriptsize $\mathcal{N}(\lambda_2)$};
\node[place,label={above:\scriptsize $q_1$},minimum size=3mm] (q1) at (0.8,2.4) {};
\node at (0.8,2){$\vdots$};
\node[place,label={below:\scriptsize $q_n$},minimum size=3mm] (qn) at (0.8,1.4) {};

\draw[dotted] (1.6,0) rectangle (3.1,1.5);\node at (2.3,0.7){\scriptsize $\mathcal{N}(\lambda)$};
\draw[dotted] (1.6,2.3) rectangle (3.1,3.8);\node at (2.3,3){\scriptsize $\mathcal{N}(\lambda_3)$};

\node[place,label={right:\scriptsize $s_1$},minimum size=3mm] (s1) at (6.5,2.4) {};
\node at (6.5,2){$\vdots$};
\node[place,label={right:\scriptsize $s_k$},minimum size=3mm] (sk) at (6.5,1.4) {};

\draw[dotted] (4.5,1.1) rectangle (6,2.6);\node at (5.2,1.8){\scriptsize $\mathcal{N}(\lambda_4)$};
\node[place,label={above:\scriptsize $r_1$},minimum size=3mm] (r1) at (4,2.4) {};
\node at (4,2){$\vdots$};
\node[place,label={below:\scriptsize $r_j$},minimum size=3mm] (rj) at (4,1.4) {};

\draw[-latex,thick] (p1) -- ($(p1)+(0.5,0)$);\draw[-latex,thick] (pm) -- ($(pm)+(0.5,0)$);
\draw[-latex,thick] ($(q1)+(-0.5,0)$)--(q1);\draw[-latex,thick] ($(qn)+(-0.5,0)$)--(qn);
\draw[-latex,thick] ($(q1)+(0.8,1.2)$)--(q1);\draw[-latex,thick] ($(q1)+(0.8,0.3)$)--(qn);
\draw[-latex,thick] (q1) -- ($(q1)+(0.8,-1.3)$);\draw[-latex,thick] (qn) -- ($(qn)+(0.8,-1)$);
\draw[-latex,thick] (r1) -- ($(r1)+(0.5,0)$);\draw[-latex,thick] (rj) -- ($(rj)+(0.5,0)$);
\draw[-latex,thick] (r1) -- ($(r1)+(-0.9,1.2)$);\draw[-latex,thick] (rj) -- ($(r1)+(-0.9,0.3)$);
\draw[-latex,thick] ($(s1)+(-0.5,0)$)--(s1);\draw[-latex,thick] ($(sk)+(-0.5,0)$)--(sk);
\draw[-latex,thick] ($(r1)+(-0.9,-1.3)$) -- (r1);\draw[-latex,thick] ($(rj)+(-0.9,-1)$) -- (rj);
\end{tikzpicture}
}
\caption{General structure of the Petri net of a tail-iterative computon, considering the t-diagram from Definition \ref{def:diagram-t}. This structure is applicable to all the functorial constructions from Section \ref{sec:operational-semantics}: $\mathcal{N}$, $\mathcal{C}\circ\mathfrak{E}$ and $\mathcal{D}$.}
\label{fig:computon-tail-net}
\end{figure} 

Unfortunately, there is no guarantee every tail-iterative's net is deadlock-free even when the nets of the computons from the corresponding t-diagram are. Despite of this, it is still possible to enforce deadlock-freedom by employing primitive computons as entry, exit and iteration points. Proposition \ref{prop:computon-primitive-deadlock} and Remark \ref{rem:computon-primitive-deadlock} together entail every primitive computon's net is deadlock-free. So, as long as the net of the computon being iterated over never gets stuck, the net of the corresponding tail-iterative computon will be deadlock-free too (see Proposition \ref{prop:computon-iterative-tail-deadlock} and Remark \ref{rem:tail-deadlock}).

\begin{proposition}\label{prop:computon-iterative-tail-deadlock}
Consider the t-diagram $\rho$ depicted in Definition \ref{def:diagram-t} and assume ${M_i}$ and ${M_f}$ are the initial and final markings of the net ${\mathcal{N}(\lambda)}$, respectively. The net ${\mathcal{N}((\lambda)*_{\rho})}$ is deadlock-free if: 
\begin{enumerate}
\item ${\lambda_j}$ is a primitive computon for $j=2,3,4$,\label{prop:computon-iterative-tail-deadlock-1}
\item ${\mathcal{N}(\lambda)}$ is deadlock-free, and \label{prop:computon-iterative-tail-deadlock-2}
\item for every marking ${M}$ reachable from ${M_i}$, if ${M(r)>0}$ for each output place $r$ of ${\mathcal{N}(\lambda)}$ then ${M=M_f}$.\label{prop:computon-iterative-tail-deadlock-3}
\end{enumerate}

\end{proposition}
\begin{proof}
Consider the t-diagram $\rho$ from Definition \ref{def:diagram-t} and assume $\lambda_j$ is a primitive computon for $j=2,3,4$. By Proposition \ref{prop:computon-primitive-connected}, we know $\rho$ is a well-defined t-diagram because each $\lambda_j$ is a connected computon. Then, using Lemma \ref{lem:computon-iterative-tail-exists}, we deduce the existence of $(\lambda)*_{\rho}$ whose underlying net $\mathcal{N}((\lambda)*_{\rho})$ has the following form according to the functorial construction presented in Definition \ref{def:functor-computon-to-petri}:
\begin{center}
\begin{tikzpicture}
\node[place,label={left:\scriptsize $p_1$},minimum size=3mm] (p1) at (-0.8,2.4) {};
\node at (-0.8,2){$\vdots$};
\node[place,label={left:\scriptsize $p_m$},minimum size=3mm] (pm) at (-0.8,1.4) {};
\node[transition,fill=black,minimum width=0.1mm,minimum height=10mm,label=\scriptsize $\mathcal{N}(\lambda_2)$] (l2) at (0,1.9) {};
\node[place,label={above:\scriptsize $q_1$},minimum size=3mm] (q1) at (0.8,2.4) {};
\node at (0.8,2){$\vdots$};
\node[place,label={below:\scriptsize $q_n$},minimum size=3mm] (qn) at (0.8,1.4) {};

\draw[dotted] (1.6,0) rectangle (3.1,1.5);\node at (2.3,0.7){\scriptsize $\mathcal{N}(\lambda)$};
\node[transition,fill=black,minimum width=0.1mm,minimum height=10mm,label=\scriptsize $\mathcal{N}(\lambda_3)$] (l3) at (2.3,2.8) {};

\node[place,label={right:\scriptsize $s_1$},minimum size=3mm] (s1) at (5.6,2.4) {};
\node at (5.6,2){$\vdots$};
\node[place,label={right:\scriptsize $s_k$},minimum size=3mm] (sk) at (5.6,1.4) {};

\node[transition,fill=black,minimum width=0.1mm,minimum height=10mm,label=\scriptsize $\mathcal{N}(\lambda_4)$] (l4) at (4.8,1.9) {};
\node[place,label={above:\scriptsize $r_1$},minimum size=3mm] (r1) at (4,2.4) {};
\node at (4,2){$\vdots$};
\node[place,label={below:\scriptsize $r_j$},minimum size=3mm] (rj) at (4,1.4) {};

\draw[-latex,thick] (p1) -- (l2);\draw[-latex,thick] (pm) -- (l2);
\draw[-latex,thick] (l2) -- (q1);\draw[-latex,thick] (l2) -- (qn);
\draw[-latex,thick] (l3) -- (q1);\draw[-latex,thick] (l3) -- (qn);
\draw[-latex,thick] (q1) -- ($(q1)+(0.8,-1.3)$);\draw[-latex,thick] (qn) -- ($(qn)+(0.8,-1)$);
\draw[-latex,thick] (r1) -- (l3);\draw[-latex,thick] (rj) -- (l3);
\draw[-latex,thick] (r1) -- (l4);\draw[-latex,thick] (rj) -- (l4);
\draw[-latex,thick] (l4) -- (s1);\draw[-latex,thick] (l4) -- (sk);
\draw[-latex,thick] ($(r1)+(-0.9,-1.3)$) -- (r1);\draw[-latex,thick] ($(rj)+(-0.9,-1)$) -- (rj);
\end{tikzpicture}
\end{center}
The above net evidently has the form depicted in Figure \ref{fig:computon-tail-net}. The only difference is that, rather than black-boxing $\mathcal{N}(\lambda_j)$, we display its internals which consist of only one transition (since $\lambda_j$ is primitive). By Definition \ref{def:marking}, we know the initial state $M_i$ of $\mathcal{N}((\lambda)*_{\rho})$ is a marking function where $M_i(p)>0$ for all $p\in\{p_1,\ldots,p_m\}$ and no tokens for all the other places, including those inside $\mathcal{N}(\lambda)$. This marking evidently enables the only transition of $\mathcal{N}(\lambda_2)$ and nothing else. Consequently, firing $\mathcal{N}(\lambda_2)$ reaches a state that marks each place in $\{q_1,\ldots,q_n\}$. This new marking evidently corresponds to the initial state of $\mathcal{N}(\lambda)$. Assuming $\mathcal{N}(\lambda)$ is deadlock-free and that only its final state can put tokens in each element from $\{r_1,\ldots,r_j\}$ (see Conditions \ref{prop:computon-iterative-tail-deadlock-2} and \ref{prop:computon-iterative-tail-deadlock-3}), we now have the following cases:
\begin{enumerate}
\item If no state of $\mathcal{N}(\lambda)$ ever puts tokens in all the places in $\{r_1,\ldots,r_j\}$, $\mathcal{N}(\lambda)$ will never terminate successfully. Despite of this, there is a guarantee $\mathcal{N}((\lambda)*_{\rho})$ will not get stuck because $\mathcal{N}(\lambda)$ is deadlock-free.\label{path-tail-1}
\item If a state of $\mathcal{N}(\lambda)$ puts tokens in all the places in $\{r_1,\ldots,r_j\}$, there are two possible execution paths because $\mathcal{N}(\lambda_3)$ and $\mathcal{N}(\lambda_4)$ are both enabled (due to mutual exclusion):\label{path-tail-2}
\begin{enumerate}
\item If $\mathcal{N}(\lambda_4)$ is triggered, the final state of $\mathcal{N}((\lambda)*_{\rho})$ will be reached with tokens in $s_1,\ldots,s_k$. So, $\mathcal{N}((\lambda)*_{\rho})$ is deadlock-free.\label{path-tail-2a}
\item If $\mathcal{N}(\lambda_3)$ is triggered, only the places in $\{q_1,\ldots,q_n\}$ will be marked in the next state. As the initial state of $\mathcal{N}(\lambda)$ is reached once again, we simply repeat \ref{path-tail-1} or \ref{path-tail-2} whichever applies.\label{path-tail-2b}
\end{enumerate}
\end{enumerate}
By the above, it is evident that all the execution paths lead to a deadlock-free execution. Therefore, we conclude $\mathcal{N}((\lambda)*_{\rho})$ is deadlock-free, as required.
\end{proof}

\begin{remark}\label{rem:tail-deadlock}
Although it is a statement about the functor $\mathcal{N}$, Proposition \ref{prop:computon-iterative-tail-deadlock} is applicable to the functors $\mathcal{C}\circ\mathfrak{E}$ and $\mathcal{D}$ presented in Section \ref{sec:operational-semantics}. The proof is valid for $\mathcal{C}\circ\mathfrak{E}$ since Proposition \ref{prop:functor-control-petri} says ${\mathcal{C}}$ is just a restriction of $\mathcal{N}$ to $\mathfrak{E}(\textbf{Set}^{\textbf{Comp}})$. 

Remark \ref{rem:deadlock-freedom} says we are only interested in checking deadlock-freedom for $\mathcal{D}$-nets that have initial and final states. A glance at the figure depicted in the proof of Proposition \ref{prop:computon-iterative-tail-deadlock} reveals this is satisfied when $k,m>0$. Starting with the initial state $M_i$ that puts tokens in $p_1,\ldots,p_m$, we have following cases:
\begin{itemize}
\item If $n=0$, $M_i$ enables the only transition of $\mathcal{D}(\lambda_2)$ which, upon firing, makes $\mathcal{N}((\lambda)*_{\rho})$ enter into a deadlock state. 
\item If $j=0$ and $n>0$, the net $\mathcal{D}(\lambda)$ will never reach its final state. Despite of this, $\mathcal{N}((\lambda)*_{\rho})$ is guaranteed to be deadlock-free when $\mathcal{D}(\lambda)$ also is.
\item If $j>0$ and $n>0$, the proof of deadlock-freedom for $\mathcal{N}((\lambda)*_{\rho})$ is analogous to that of Proposition \ref{prop:computon-iterative-tail-deadlock}.
\end{itemize}
Hence, to guarantee $\mathcal{D}((\lambda)*_{\rho})$ is deadlock-free, we must consider a t-diagram $\rho$ where the entry and iteration computons have both ed-outports, apart from ensuring that $\mathcal{D}(\lambda)$ is deadlock-free and that satisfies Condition \ref{prop:computon-iterative-tail-deadlock-3} of Proposition \ref{prop:computon-iterative-tail-deadlock}.
\end{remark}

\subsubsection{Encapsulation of control flow and data flow in tail-iterative computons}

By Definition \ref{def:computon-iterative-tail}, we know a tail-iterative computon $(\lambda)*_{\rho}$ is the colimit of a t-diagram $\rho$ which, by Definition \ref{def:diagram-t}, consists of four connected computons and two trivial computons. One of the connected computons is $\lambda$ (i.e., the computon being iterated over) whereas the others serve as entry, exit and iteration points. Thus, like a head-iterative, $(\lambda)*_{\rho}$ encapsulates cyclic control flow and up to cyclic data flow. By cyclic, we mean $\lambda$ and the iteration entity are executed repeatedly. In a tail-iterative computon, the decision whether to repeat $\lambda$ is made after executing it. To give a concrete example, Figure \ref{fig:encapsulation-tail} illustrates the encapsulation given by the tail-iterative computon resulting from the colimit construction depicted in Figure \ref{fig:computon-iterative-tail-example}.

\begin{figure}[!h]
\centering
\begin{tikzpicture}[scale=0.89]
\begin{scope}[xshift=-9cm,yshift=0cm]
  \computonComposite{1.7}{0.3}{7.4}{4.6};
  
  \computonPrimitive{2.3}{1.6}{1}{1.9}{$\lambda_2$}; 
  \qin{2q0}{1.5}{3.2}{};
  \din{2i1}{1.5}{2.6}{$1$};
  \din{2i2}{1.5}{2}{$2$};
 
  \qmatch{yq1}{4.1}{3.2}{};
  \dmatch{yo1}{4.1}{2.6}{$1$}{above};
  \dmatch{yo2}{4.1}{2}{$2$}{above};  
  \flow{{3.3,3.2}}{yq1}{dashed}{};
  \flow{{3.3,2.6}}{yo1}{}{};
  \flow{{3.3,2}}{yo2}{}{};
	
  \computonPrimitive{4.9}{0.5}{1}{1.9}{$\lambda$};
  \qmatch{3q0}{6.8}{2.8}{};
  \dmatch{3i1}{6.8}{2.2}{$3$}{above};
  \flowdiag{{4.9,1.9}}{yq1}{dashed}{}{pos=0.25,rotate=308};
  \flowdiag{{4.9,1.3}}{yo1}{}{}{pos=0.4,rotate=308};
  \flowdiag{{4.9,0.7}}{yo2}{}{}{pos=0.4,rotate=308};
  \flowdiag{{5.9,1.9}}{3q0}{dashed}{}{pos=0.5,rotate=45};
  \flowdiag{{5.9,1.3}}{3i1}{}{}{pos=0.5,rotate=45};
  
  \computonPrimitive{4.9}{2.7}{1}{1.9}{$\lambda_3$};  
  \flowdiag{3q0}{{5.9,4}}{dashed}{}{pos=0.5,rotate=135};
  \flowdiag{3i1}{{5.9,3.4}}{}{}{pos=0.5,rotate=135};
  \flowdiag{{4.9,4.3}}{yq1}{dashed}{}{pos=0.3,rotate=225};
  \flowdiag{{4.9,3.7}}{yo1}{}{}{pos=0.3,rotate=225};
  \flowdiag{{4.9,3.1}}{yo2}{}{}{pos=0.2,rotate=225};
  
  \computonPrimitive{7.6}{1.6}{1}{1.9}{$\lambda_4$};
  \qout{4q0}{8.6}{3.2}{};
  \dout{4o0}{8.6}{2.6}{$3$};
  \flow{3q0}{{7.6,2.8}}{dashed}{};
  \flow{3i1}{{7.6,2.2}}{}{};
\end{scope}
\draw[opacity=0.3,line width=1.5pt, ->, -Latex] (1.1,2.5) to node[pos=0.4,yshift=7,xshift=-8]{\scriptsize $\mathcal{C}\circ \mathfrak{E}$} (2,3.7);
\begin{scope}[xshift=0.8cm,yshift=3.3cm]
\node[opacity=0.2] at (4.8,1.8){\scriptsize Control Flow Net};
\node[place,label={80:},minimum size=3mm] (i2) at (0,0.8) {};
\node[transition,fill=black,minimum width=0.1mm,minimum height=10mm] (2) at (1,0.8) {};
\node[place,label={80:},minimum size=3mm] (o2) at (2,0.8) {};

\node[transition,fill=black,minimum width=0.1mm,minimum height=10mm] (3) at (3,1.6) {};
\node[transition,fill=black,minimum width=0.1mm,minimum height=10mm] (0) at (3,0) {};

\node[place,label={80:},minimum size=3mm] (o) at (4,0.8) {};

\node[transition,fill=black,minimum width=0.1mm,minimum height=10mm] (4) at (5,0.8) {};
\node[place,label={80:},minimum size=3mm] (o4) at (6,0.8) {};

\draw[-latex,thick] (i2)--(2);\draw[-latex,thick] (2)--(o2);\draw[-latex,thick] (o2)--(0);\draw[-latex,thick] (3)--(o2);
\draw[-latex,thick] (o)--(3);\draw[-latex,thick] (0)--(o);
\draw[-latex,thick] (o)--(4);\draw[-latex,thick] (4)--(o4);
\end{scope}
\draw[opacity=0.3,line width=1.5pt, ->, -Latex] (1.1,2.2) to node[pos=0.4,yshift=-4,xshift=-5]{\scriptsize $\mathcal{D}$} (2,1);
\begin{scope}[xshift=0.8cm,yshift=-0.7cm]
\node[opacity=0.2] at (4.8,1.8){\scriptsize Data Flow Net};
\node[place,label={80:},minimum size=3mm] (i2) at (0,0.4) {\scriptsize $2$};
\node[place,label={80:},minimum size=3mm] (i2x) at (0,1.2) {\scriptsize $1$};
\node[transition,fill=black,minimum width=0.1mm,minimum height=10mm] (2) at (1,0.8) {};
\node[place,label={80:},minimum size=3mm] (o2) at (2,1.2) {\scriptsize $1$};
\node[place,label={80:},minimum size=3mm] (o2x) at (2,0.4) {\scriptsize $2$};

\node[transition,fill=black,minimum width=0.1mm,minimum height=10mm] (3) at (3,1.6) {};
\node[transition,fill=black,minimum width=0.1mm,minimum height=10mm] (0) at (3,0) {};

\node[place,label={80:},minimum size=3mm] (o) at (4,0.8) {\scriptsize $3$};

\node[transition,fill=black,minimum width=0.1mm,minimum height=10mm] (4) at (5,0.8) {};
\node[place,label={80:},minimum size=3mm] (o4) at (6,0.8) {\scriptsize $3$};

\draw[-latex,thick] (i2)--(2);\draw[-latex,thick] (i2x)--(2);\draw[-latex,thick] (2)--(o2);\draw[-latex,thick] (2)--(o2x);\draw[-latex,thick] (o2)--(0);draw[-latex,thick] (o2x)--(0);\draw[-latex,thick] (3)--(o2);\draw[-latex,thick] (3)--(o2x);
\draw[-latex,thick] (o2x)--(0);
\draw[-latex,thick] (o)--(3);\draw[-latex,thick] (0)--(o);
\draw[-latex,thick] (o)--(4);\draw[-latex,thick] (4)--(o4);
\end{scope}
\draw[opacity=0.3,line width=1.5pt, ->, -Latex] (-4,0) to node[pos=0.4,xshift=10]{\scriptsize $\mathcal{N}$} (-4,-0.8);
\begin{scope}[xshift=-6cm,yshift=-2.6cm]
\node[opacity=0.2] at (-1.2,1.1){\scriptsize Control and};\node[opacity=0.2] at (-1.2,0.8){\scriptsize Data Flow};\node[opacity=0.2] at (-1.2,0.5){\scriptsize Net};
\node[place,label={80:},minimum size=3mm] (i2c) at (0,1.4) {};
\node[place,label={80:},minimum size=3mm] (i2) at (0,0.9) {\scriptsize $2$};
\node[place,label={80:},minimum size=3mm] (i2x) at (0,0.4) {\scriptsize $1$};
\node[transition,fill=black,minimum width=0.1mm,minimum height=10mm] (2) at (1,0.8) {};
\node[place,label={80:},minimum size=3mm] (o2c) at (2,1.4) {};
\node[place,label={80:},minimum size=3mm] (o2) at (2,0.9) {\scriptsize $1$};
\node[place,label={80:},minimum size=3mm] (o2x) at (2,0.4) {\scriptsize $2$};

\node[transition,fill=black,minimum width=0.1mm,minimum height=10mm] (3) at (3,1.6) {};
\node[transition,fill=black,minimum width=0.1mm,minimum height=10mm] (0) at (3,0) {};

\node[place,label={80:},minimum size=3mm] (oc) at (4,1.1) {};
\node[place,label={80:},minimum size=3mm] (o) at (4,0.6) {\scriptsize $3$};

\node[transition,fill=black,minimum width=0.1mm,minimum height=10mm] (4) at (5,0.8) {};
\node[place,label={80:},minimum size=3mm] (o4c) at (6,1) {};
\node[place,label={80:},minimum size=3mm] (o4) at (6,0.6) {\scriptsize $3$};

\draw[-latex,thick] (i2)--(2);\draw[-latex,thick] (i2x)--(2);\draw[-latex,thick] (2)--(o2);\draw[-latex,thick] (2)--(o2x);\draw[-latex,thick] (o2)--(0);draw[-latex,thick] (o2x)--(0);\draw[-latex,thick] (3)--(o2);\draw[-latex,thick] (3)--(o2x);
\draw[-latex,thick] (o2x)--(0);
\draw[-latex,thick] (o)--(3);\draw[-latex,thick] (0)--(o);
\draw[-latex,thick] (o)--(4);\draw[-latex,thick] (4)--(o4);
\draw[-latex,thick] (i2c)--(2);\draw[-latex,thick] (2)--(o2c);\draw[-latex,thick] (o2c)--(0);\draw[-latex,thick] (0)--(oc);
\draw[-latex,thick] (oc)--(4);\draw[-latex,thick] (oc)--(3);\draw[-latex,thick] (3)--(o2c);\draw[-latex,thick] (4)--(o4c);
\end{scope}
\end{tikzpicture}
\caption{Cyclic control flow and cyclic data flow encapsulated by the tail-iterative computon from Figure \ref{fig:computon-iterative-tail-example}. We label some places for mapping purposes even though Petri nets are not labelled (see Section \ref{sec:operational-semantics}).}
\label{fig:encapsulation-tail}
\end{figure}
\section{Applications of the Proposed Model}\label{sec:applications}

Compositionality is not exclusive of a single domain, but it appears in many spheres, from physical \cite{coecke_compositionality_2023} to artificial systems \cite{arellanes_evaluating_2020}. This section describes the application of the computon model in two different domains: software engineering and artificial intelligence. Although we are not proposing an end-user modelling language but just an MHC to capture the essence of high-level computations, this section serves to demonstrate the suitability of our model for the compositional construction of high-level computations that separate data flow and control flow. 

For each case study, we describe the respective composite computons and show how the separation of concerns can be exploited to analyse control flow independently from data flow (and vice versa), in this case for model transformation. We particularly express control flow as a diagram in Business Process Model and Notation (BPMN) \cite{noauthor_business_2011} which is the standard language that has been widely used for many years, in both academia and industry, to canonically model workflow control flow. This conversion process does not consider data flow at all and is done through the graph transformation system described in Section \ref{sec:appendix-controlflow}. For data flows, we rely upon DFGs in standard notation \cite{johnston_advances_2004} where arrows represent data flow and circles denote consumer or producer computations. The conversion process is performed without considering control flow at all via the graph transformation system described in Section \ref{sec:appendix-dataflow}. Both BPMN diagrams and standard DFGs are far more expressive than Petri nets to express control and data flow within a system. For completeness, for each scenario, we display the Petri net under $\mathcal{C}\circ \mathfrak{E}$ that comprehensively captures system behaviour. Constructing the corresponding nets under $\mathcal{N}$ or $\mathcal{D}$ can easily be done using the mapping from \ref{sec:appendix-mapping}, which adheres to the functorial constructions from Definition \ref{def:functor-computon-to-petri} and Proposition \ref{prop:functor-data-petri}.

For clarity, we do not describe the corresponding colimiting diagrams of composite computons, but their definition is left to the reader as a matter of routine exercise. For each composite, we try to provide as much internal structure as possible. But when this is not possible due to space restrictions, we simply make reference to a previously created composite. Some composites are not shown independently so as to save space and minimise duplication. 

\subsection{Transformation System} \label{sec:transformation-system}

Before presenting our concrete case studies, we describe our model transformation system which consists of two different modules, one for transforming a computon into its corresponding BPMN diagram and another for retrieving the respective DFG in standard notation. To simplify transformation, we operate on computon CFGs and DFGs rather than Petri nets. This is because such constructs are multidirected labelled graphs that embed all the necessary information about computon flows, without adhering to specific operational semantics (see Definitions \ref{def:computon-cfg} and \ref{def:computon-dfg}).

\subsubsection{Transforming a Computon CFG into a BPMN Diagram} \label{sec:appendix-controlflow}

For this, we propose a simple graph transformation system \emph{ad hoc} to our specific case studies, whose aim is to realise the syntax mapping displayed in Figure \ref{fig:mapping-computon-bpmn}. As our purpose is not to provide general rewriting rules but to demonstrate how the separation of control and data flow can be used for model transformation purposes in two concrete scenarios, we do not investigate theoretical properties such as confluence. Instead, we just ensure that rewriting terminates in a correct state. For this, we implemented our system in Groove \cite{ghamarian_modelling_2012}, a well-established tool that supports state-space exploration, which enabled us to validate termination and confirm that, in each case study, all derivations lead to the intended terminal graph. In the future, we would like to investigate more general, application-independent rewriting rules as well as efficient rewriting algorithms. 
	
\begin{figure}[!h]
\centering
\begin{tabular}{ |c|c|c|c|c|c|c|c| } 
 \hline
 \scriptsize{Computon Syntax} & \tikz{\qmatch{q}{0}{0}{};\node (t1) at (0.7,-0.5){};\node (t2) at (0.7,0.5){};\flowdiag{q}{t1}{dashed}{}{pos=0.5,rotate=320};\flowdiag{q}{t2}{dashed}{}{pos=0.5,rotate=40};} & \tikz{\node at (0,0){};\forkplain{-1}{0.2};\flow{{0.3,0.52}}{{0.8,0.52}}{dashed}{};\flow{{0.3,0.35}}{{0.8,0.35}}{dashed}{};} & \tikz{\node at (0,0){};\joinplain{-0.5}{0.2};\flow{{0,0.52}}{{0.5,0.52}}{dashed}{};\flow{{0,0.35}}{{0.5,0.35}}{dashed}{};} & \tikz{\node at (0,0){};\qmatch{q}{0.55}{0.45}{};\flow{{0,0.45}}{q}{dashed}{};\flow{q}{{1.2,0.45}}{dashed}{};} & \tikz{\node at (0,0){};\computonPrimitive{0}{0.1}{0.3}{0.6}{};} & \tikz{\node at (0,0){};\qinplain{}{0}{0.35}{};} & \tikz{\node at (0,0){};\qoutplain{}{-0.8}{0.35}{};} \\
 \hline
 \scriptsize{BPMN Syntax} & \includegraphics[scale=0.3]{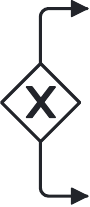} & \includegraphics[scale=0.3]{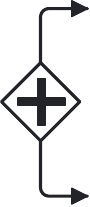} & \includegraphics[scale=0.3]{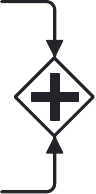} & \tikz{\node at (0,0){};\draw[-latex] (0,0.3)--(1.2,0.3);} & \includegraphics[scale=0.3]{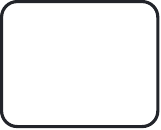} & \tikz{\node at (0,0){};\node at (0,0.3){\includegraphics[scale=0.3]{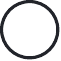}};} & \tikz{\node at (0,0){};\node at (0,0.3){\includegraphics[scale=0.3]{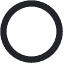}};} \\ 
 \hline
\end{tabular}
\caption{Mapping from computon syntax to BPMN syntax.}
\label{fig:mapping-computon-bpmn}
\end{figure}

To understand the purpose of our transformation system, let us consider the example depicted in Figure \ref{fig:transformation-example-comprehensive} in which we first convert a composite computon into its corresponding CFG via the functorial construction $\mathscr{C}$ described in Section \ref{sec:control-data-structures}. This CFG serves as the input to our system which, in turn, performs the mapping from Figure \ref{fig:mapping-computon-bpmn} to produce a BPMN diagram that captures the control flow structure of our composite. Here, for instance, we can notice that the unique fork construct has been replaced with an AND gateway, whereas ic-ports (together with their adjacent flows) have been substituted with a simple arrow, known as sequence flow in the context of BPMN. These two rewritings are just application instances of the mapping shown in Figure \ref{fig:mapping-computon-bpmn}. 

Although it might seem trivial, realising such a mapping cannot be done directly since technical considerations need to be taken into account to maintain graph integrity while ensuring semantic correctness (e.g., avoiding dangling edges during rewriting). For that reason, we propose a graph transformation system which converts a computon CFG $\mathscr{C}(\lambda)$ into a BPMN diagram $G$ via the sequential application of the injective rules from the set $R=\{r_1,\ldots,r_{12}\}$ displayed in Figure \ref{fig:rewriting-control}. 

\begin{figure}[!h]
\centering
\begin{tikzpicture}
\begin{scope}[yshift=5cm,scale=0.8]
\node[opacity=0.6] at (4,3.9){\scriptsize Composite computon};

\computonComposite{0.2}{-0.9}{7.6}{4.4};
\computonComposite{2.5}{-0.8}{5}{4.2};
\computonComposite{2.8}{0.6}{4.5}{2.7};

\computonPrimitive{0.8}{0.4}{0.5}{1.2}{};
\qin{0q1}{0}{1.1}{};

\qmatch{oq2}{1.9}{1.1}{};\flow{{1.3,1.1}}{oq2}{dashed}{};\flowdiag{oq2}{{2.95,1.95}}{dashed}{}{pos=0.4,rotate=45};\flowdiag{oq2}{{2.9,0.2}}{dashed}{}{pos=0.4,rotate=325};

\forkplain{1.95}{1.7};
\qmatch{fq1}{3.9}{2.1}{};\flow{{3.2,2.1}}{fq1}{dashed}{};\flow{fq1}{{4.7,2.1}}{dashed}{};
\qmatch{fq2}{3.9}{1.8}{};\flow{{3.2,1.8}}{fq2}{dashed}{};\flow{fq2}{{4.7,1.8}}{dashed}{};

\computonPrimitive{4.7}{2}{0.5}{1.2}{};
\computonPrimitive{4.7}{0.7}{0.5}{1.2}{};

\qmatch{jq1}{6}{2.1}{};\flow{{5.2,2.1}}{jq1}{dashed}{};\flow{jq1}{{6.8,2.1}}{dashed}{};
\qmatch{jq2}{6}{1.8}{};\flow{{5.2,1.8}}{jq2}{dashed}{};\flow{jq2}{{6.8,1.8}}{dashed}{};
\joinplain{5.8}{1.7};

\computonPrimitive{2.9}{-0.7}{0.5}{1.2}{};

\dmatch{1o}{8}{1.6}{$9$}{right};\flow{{5.2,2.5}}{1o}{}{bend left=20};\flowdiag{{3.4,-0.4}}{1o}{bend right=45}{}{pos=0.6,rotate=40};
\qmatch{jq3}{8}{1.3}{};\flowdiag{{7.1,1.94}}{jq3}{dashed}{}{pos=0.5,rotate=325};\flowdiag{{3.4,-0.1}}{jq3}{dashed,bend right=60}{}{rotate=20};
\end{scope}

\begin{scope}[yshift=5.3cm,xshift=7cm]
\draw[opacity=0.3,line width=1.5pt, ->, -Latex] (0,0) to node[pos=0.4,yshift=7]{\scriptsize $\mathscr{C}$} (1.5,0);
\end{scope}

\begin{scope}[yshift=5.4cm,xshift=8.9cm,scale=0.8]
\node[opacity=0.6] at (4.6,1.9){\scriptsize Corresponding CFG};
\node(n1) at (0,0.5){$0$};
\node(n2) at (1,0.5){$\kappa$};
\node(n3) at (2,0.5){$0$};

\node(n4) at (3,0){$\kappa$};

\node(n5) at (3,1){$\kappa$};
\node(n6) at (4,1.2){$0$};\node(n61) at (5,1.2){$\kappa$};\node(n62) at (6,1.2){$0$};
\node(n7) at (4,0.8){$0$};\node(n71) at (5,0.8){$\kappa$};\node(n72) at (6,0.8){$0$};
\node(n8) at (7,1){$\kappa$};

\node(n9) at (8,0.5){$0$};

\draw[-latex] (n1) -- (n2);
\draw[-latex] (n2) -- (n3);
\draw[-latex] (n3) -- (n4);
\draw[-latex] (n3) -- (n5);
\draw[-latex] (n5) -- (n6);\draw[-latex] (n6) -- (n61);\draw[-latex] (n61) -- (n62);\draw[-latex] (n62) -- (n8);
\draw[-latex] (n5) -- (n7);\draw[-latex] (n7) -- (n71);\draw[-latex] (n71) -- (n72);\draw[-latex] (n72) -- (n8);
\draw[-latex] (n4) to [bend right=20] (n9);
\draw[-latex] (n8) -- (n9);
\end{scope}

\begin{scope}[yshift=3.3cm,xshift=11cm]
\draw[opacity=0.3,line width=1.5pt, ->, -Latex] (1.4,1.5) to node[pos=0.5,xshift=60,align=left]{\scriptsize Our graph transformation \\\scriptsize system} (-0.1,0);
\end{scope}

\begin{scope}
\node[opacity=0.6,align=left] at (3.2,2){\scriptsize Corresponding BPMN diagram};
\node[anchor=south west,inner sep=0] at (6,0) {\includegraphics[scale=0.4]{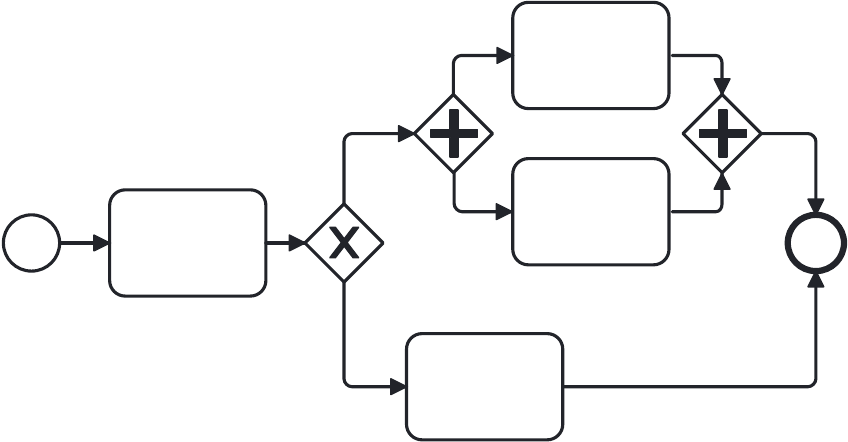}};
\end{scope}
\end{tikzpicture}
\caption{Example to illustrate the intended purpose of our graph transformation system.}
\label{fig:transformation-example-comprehensive}
\vspace{-4pt}
\end{figure}

\newpage

Each of the $R$-rules from Figure \ref{fig:rewriting-control} is applied individually until no further match is found, so our system produces a sequence of derivations of the form $\mathscr{C}(\lambda) \Rightarrow_{r_j} \cdots \Rightarrow_{r_k} G$ for ${j,k\in\{1,\ldots,12\}}$ and ${j \leq k}$. As we rely on the Double-Pushout Approach for graph transformation \cite{habel_double-pushout_2001}, a derivation ${G_1 \Rightarrow_{r_i} G_2}$ exists for ${i\in \{1,\ldots,12\}}$ if there is a context graph $G_0$ that makes the following two squares commute in $\textbf{Set}^{\textbf{Gr}}$:
\[
\begin{tikzcd}
 \mathcal{L}_i \arrow[d] & \mathcal{I}_i \arrow[l]\arrow[r]\arrow[d] & \mathcal{R}_i \arrow[d] \\
 G_1 & G_0 \arrow[l]\arrow[r] & G_2 
\end{tikzcd}
\]
Here, $\mathcal{L}_i$, $\mathcal{R}_i$ and $\mathcal{I}_i$ denote the left-hand side, right-hand side and interface of a rule $r_i$. 

The purpose of $r_1$ is to replace branching sources with XOR gateways (i.e., BPMN elements that denote alternative control flow) through the match of $0$-nodes with exactly two outgoing edges (i.e., control ports with two salient control flows). In our rule set, XOR join gateways are not taken into account to facilitate transformation processes, a consideration that is completely valid since such constructs are not strictly required in BPMN diagrams \cite{silver_bpmn_2011,oracle_corporation_developing_2016}. 

Rule $r_2$ removes intermediate branching sinks by matching $0$-nodes connected from two $\kappa$-nodes to exactly one $\kappa$-node (i.e., control ports linked from two computation units to some other unit). Applying $r_2$ results in the removal of a $0$-node (together with its adjacent edges) and in the addition of new edges from the sources of that node to its single target. Rules $r_3$ and $r_4$ rewrite parallel splits and synchronisation points, respectively. Splits correspond to $\kappa$-nodes with two outgoing edges (i.e., fork computons), whereas synchronisation points are $\kappa$-nodes with two incoming edges (i.e., join computons). Both of them are relabelled as AND gateways, i.e., BPMN constructs for splitting or merging concurrent execution paths. 

Rules $r_5$ to $r_8$ manage sequencing by removing all intermediate $0$-nodes (i.e., ic-ports) together with their adjacent edges. Whenever a node is removed, an arc is put instead, from the source of its left edge to the target of its right one. The purpose of rule $r_9$ is to replace $k$-nodes (i.e., computation units) with BPMN tasks which are atomic activities that represent a unit of computation performed by some computing device. Rules $r_{10}$, $r_{11}$ and $r_{12}$ simply relabel ec-inports and ec-outports as start and end events, respectively. In BPMN, start events denote control flow origin whilst end events represent control flow termination.\footnote{We are aware rules $r_5$, $r_6$, $r_7$ and $r_8$ can be merged into a single one by the use of node restrictions (the same for rules $r_{10}$ and $r_{11}$). Our purpose is not to provide a minimal set of rewriting rules but to demonstrate how the separation of control and data flow can be leveraged to convert a computon CFG into its equivalent BPMN diagram, without the need of analysing data flow at all. We believe that avoiding the use of node restrictions clarifies our transformation system and provides additional expressivity in terms of graph matching.}

\begin{figure}[!h]
\centering
\subcaptionbox{Rules for rewriting branching and parallel constructs.\label{fig:rewriting-control-1}}
{
\begin{tabular}{ |c|c|c|c| } 
 \hline
 \scriptsize{Rule} & \scriptsize{Left-Hand Side} & \scriptsize{Interface} & \scriptsize{Right-Hand Side} \\
 \hline 
 $r_1$ & \tikz{\node(q) at (0,0){\scriptsize$0$};\node (t1) at (0.55,-0.5){$\kappa$};\node (t2) at (0.55,0.5){$\kappa$};\draw[-latex] (q)--(0,-0.5)--(t1);\draw[-latex] (q)--(0,0.5)--(t2);} & \tikz{\node(q) at (0,0){\scriptsize\#};\node (t1) at (0.55,-0.5){$\kappa$};\node (t2) at (0.55,0.5){$\kappa$};\draw[-latex] (q)--(0,-0.5)--(t1);\draw[-latex] (q)--(0,0.5)--(t2);} & \tikz{\node at (0,0){\includegraphics[scale=0.35]{bpmn-branching-fork}};\node at (0.5,-0.5){$\kappa$};\node at (0.5,0.5){$\kappa$};} \\
 \hline
 $r_2$ & \tikz{\node(q) at (0.55,0){\scriptsize$0$};\node (t1) at (0,-0.5){$\kappa$};\node (t2) at (0,0.5){$\kappa$};\node(r) at (1.4,0){$\kappa$};\draw[-latex] (t1)--(0.55,-0.5)--(q);\draw[-latex] (t2)--(0.55,0.5)--(q);\draw[-latex] (q)--(r);} & \tikz{\node(q) at (0.55,0){};\node (t1) at (0,-0.5){$\kappa$};\node (t2) at (0,0.5){$\kappa$};\node(r) at (1.4,0){$\kappa$};} & \tikz{\node(q) at (0.55,0){};\node (t1) at (0,-0.5){$\kappa$};\node (t2) at (0,0.5){$\kappa$};\node(r) at (1.4,0){$\kappa$};\draw[-latex] (t1)--(1.4,-0.5)--(r);\draw[-latex] (t2)--(1.4,0.5)--(r);} \\
 \hline
 $r_3$ & \tikz{\node(q) at (0,0){$\kappa$};\node (t1) at (0.55,-0.5){\scriptsize$0$};\node (t2) at (0.55,0.5){\scriptsize$0$};\draw[-latex] (q)--(0,-0.5)--(t1);\draw[-latex] (q)--(0,0.5)--(t2);} & \tikz{\node(q) at (0,0){\scriptsize\#};\node (t1) at (0.55,-0.5){\scriptsize$0$};\node (t2) at (0.55,0.5){\scriptsize$0$};\draw[-latex] (q)--(0,-0.5)--(t1);\draw[-latex] (q)--(0,0.5)--(t2);} & \tikz{\node at (0,0){\includegraphics[scale=0.35]{bpmn-and-fork}};\node at (0.5,-0.5){\scriptsize$0$};\node at (0.5,0.5){\scriptsize$0$};} \\
 \hline
 $r_4$ & \tikz{\node(q) at (0.55,0){$\kappa$};\node (t1) at (0,-0.5){\scriptsize$0$};\node (t2) at (0,0.5){\scriptsize$0$};\draw[-latex] (t1)--(0.55,-0.5)--(q);\draw[-latex] (t2)--(0.55,0.5)--(q);} & \tikz{\node(q) at (0.55,0){\scriptsize\#};\node (t1) at (0,-0.5){\scriptsize$0$};\node (t2) at (0,0.5){\scriptsize$0$};\draw[-latex] (t1)--(0.55,-0.5)--(q);\draw[-latex] (t2)--(0.55,0.5)--(q);} & \tikz{\node at (0.5,0){\includegraphics[scale=0.35]{bpmn-and-join}};\node at (0,-0.5){\scriptsize$0$};\node at (0,0.5){\scriptsize$0$};} \\
 \hline
\end{tabular}
}
\subcaptionbox{Rules for rewriting sequential constructs.\label{fig:rewriting-control-2}}
{
\begin{tabular}{ |c|c|c|c| } 
 \hline
 \scriptsize{Rule} & \scriptsize{Left-Hand Side} & \scriptsize{Interface} & \scriptsize{Right-Hand Side} \\
 \hline
 $r_5$ & \tikz{\node(k1) at (0,0){$\kappa$};\node (n) at (1,0){\scriptsize$0$};\node(k2) at (2,0){$\kappa$};\draw[-latex] (k1)--(n);\draw[-latex] (n)--(k2);} & \tikz{\node(k1) at (0,0){$\kappa$};\node(k2) at (1.7,0){$\kappa$};} & \tikz{\node(k1) at (0,0){$\kappa$};\node(k2) at (1.5,0){$\kappa$};\draw[-latex] (k1)--(k2);} \\
 \hline
 $r_6$ & \tikz{\node(k1) at (0,0){\includegraphics[scale=0.35]{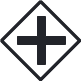}};\node (n) at (1,0){\scriptsize$0$};\node(k2) at (2,0){$\kappa$};\draw[-latex] (k1)--(n);\draw[-latex] (n)--(k2);} & \tikz{\node(k1) at (0,0){\includegraphics[scale=0.35]{bpmn-and}};\node(k2) at (1.7,0){$\kappa$};} & \tikz{\node(k1) at (0,0){\includegraphics[scale=0.35]{bpmn-and}};\node(k2) at (1.5,0){$\kappa$};\draw[-latex] (k1)--(k2);} \\
 \hline
 $r_7$ & \tikz{\node(k1) at (0,0){$\kappa$};\node (n) at (1,0){\scriptsize$0$};\node(k2) at (2,0){\includegraphics[scale=0.35]{bpmn-and}};\draw[-latex] (k1)--(n);\draw[-latex] (n)--(k2);} & \tikz{\node(k1) at (0,0){$\kappa$};\node(k2) at (1.7,0){\includegraphics[scale=0.35]{bpmn-and}};} & \tikz{\node(k1) at (0,0){$\kappa$};\node(k2) at (1.5,0){\includegraphics[scale=0.35]{bpmn-and}};\draw[-latex] (k1)--(k2);} \\
 \hline
 $r_8$ & \tikz{\node(k1) at (0,0){\includegraphics[scale=0.35]{bpmn-and}};\node (n) at (1,0){\scriptsize$0$};\node(k2) at (2,0){\includegraphics[scale=0.35]{bpmn-and}};\draw[-latex] (k1)--(n);\draw[-latex] (n)--(k2);} & \tikz{\node(k1) at (0,0){\includegraphics[scale=0.35]{bpmn-and}};\node(k2) at (1.7,0){\includegraphics[scale=0.35]{bpmn-and}};} & \tikz{\node(k1) at (0,0){\includegraphics[scale=0.35]{bpmn-and}};\node(k2) at (1.5,0){\includegraphics[scale=0.35]{bpmn-and}};\draw[-latex] (k1)--(k2);} \\
 \hline
\end{tabular}
}
\subcaptionbox{Rule for rewriting functional computons.\label{fig:rewriting-control-3}}
{
\begin{tabular}{ |c|c|c|c| } 
 \hline
 \scriptsize{Rule} & \scriptsize{Left-Hand Side} & \scriptsize{Interface} & \scriptsize{Right-Hand Side} \\
 \hline
 $r_9$ & \tikz{\node at (0,0){};\node(k1) at (0,0){$\kappa$};} & \tikz{\node at (0,0){};\node(k1) at (0,0){\scriptsize\#}}  & \includegraphics[scale=0.3]{bpmn-activity} \\
 \hline
\end{tabular}
}
\subcaptionbox{Rules for rewriting ec-inports and ec-outports.\label{fig:rewriting-control-4}}
{
\begin{tabular}{ |c|c|c|c| } 
 \hline
 \scriptsize{Rule} & \scriptsize{Left-Hand Side} & \scriptsize{Interface} & \scriptsize{Right-Hand Side} \\
 \hline
 $r_{10}$ & \tikz{\node(k1) at (1.3,0){\includegraphics[scale=0.3]{bpmn-activity}};\node (n) at (0,0){\scriptsize$0$};\draw[-latex] (n)--(k1);} & \tikz{\node at (0,0){};\node (i) at (0,0.3){\scriptsize\#};\node (a) at (1.3,0.3){\includegraphics[scale=0.3]{bpmn-activity}};\draw[-latex](i)--(a);} & \tikz{\node at (0,0){};\node (i) at (0,0.3){\includegraphics[scale=0.3]{bpmn-start}};\node (a) at (1.3,0.3){\includegraphics[scale=0.3]{bpmn-activity}};\draw[-latex](i)--(a);} \\
 \hline
 $r_{11}$ & \tikz{\node(k1) at (1.3,0){\includegraphics[scale=0.3]{bpmn-and}};\node (n) at (0,0){\scriptsize$0$};\draw[-latex] (n)--(k1);} & \tikz{\node at (0,0){};\node (i) at (0,0.3){\scriptsize\#};\node (a) at (1.3,0.3){\includegraphics[scale=0.3]{bpmn-and}};\draw[-latex](i)--(a);} & \tikz{\node at (0,0){};\node (i) at (0,0.3){\includegraphics[scale=0.3]{bpmn-start}};\node (a) at (1.3,0.3){\includegraphics[scale=0.3]{bpmn-and}};\draw[-latex](i)--(a);} \\
 \hline
 $r_{12}$ & \tikz{\node(k1) at (0,0){\includegraphics[scale=0.3]{bpmn-activity}};\node (n) at (1.3,0){\scriptsize$0$};\draw[-latex] (k1)--(n);} & \tikz{\node(k1) at (0,0){\includegraphics[scale=0.3]{bpmn-activity}};\node (n) at (1.3,0){\scriptsize\#};\draw[-latex] (k1)--(n);} & \tikz{\node(k1) at (0,0){\includegraphics[scale=0.3]{bpmn-activity}};\node (n) at (1.3,0){\includegraphics[scale=0.3]{bpmn-end}};\draw[-latex] (k1)--(n);} \\
 \hline
\end{tabular}
}
\caption{Rewriting rules to transform a computon CFG into a BPMN diagram. Here, a hash symbol \# is used as a wildcard to ensure that structure is preserved. This choice is arbitrary and any other label can be used instead, as long as it is different from the labels used for CFGs or BPMN diagrams. Labels are put in the place of nodes, and object mapping corresponds to graphical arrangement.}
\label{fig:rewriting-control}
\vspace{-11pt}
\end{figure}

To validate the twelve rules from Figure \ref{fig:rewriting-control} on the scenarios described in Sections \ref{sec:case1} and \ref{sec:case2}, we implemented our graph transformation system in Groove \cite{ghamarian_modelling_2012} which is a reference tool for specifying and simulating such kind of systems. With the help of Groove, we verified that our rule set $R$ satisfies dangling conditions and that it correctly produces BPMN diagrams for the CFG of the total sequential computon from Figure \ref{fig:application-AWS-workflow} and the memory cell from Figure \ref{fig:application-LSTM-composites}(e). Particularly, the correct BPMN diagram for Figure \ref{fig:application-AWS-workflow} is produced after exploring 37 states and 36 transitions, whereas the BPMN diagram for Figure \ref{fig:application-LSTM-composites}(e) is obtained after exploring 35 states and 34 transitions. For the simulation, we used linear state exploration which chooses one transition from each open state. For reproduction purposes, our source code is available at \url{https://github.com/damianarellanes/cfg-transformation}.\footnote{Groove implicitly creates interfaces, determines context graphs and computes pushouts, given the left- and right-hand side of a rule as well as a host graph.} 

\subsubsection{Transforming a Computon DFG into a DFG in Standard Notation} \label{sec:appendix-dataflow}

For this, we propose a functor $S$ from the category $\textbf{Set}^{\textbf{Gr}}$ to the category of graphs with labelled vertices and labelled edges, whose behaviour is formalised in Definition \ref{def:computon-dfg-standard}.

\begin{definition}\label{def:computon-dfg-standard}
Given a computon DFG $\mathscr{D}(\lambda)$, the functor $S$ constructs a graph $S(\mathscr{D}(\lambda))$ by letting:
\begin{itemize}
\item the set $L'$ of vertex labels be $L$, 
\item the set $M'$ of edge labels be $L \cup \{\epsilon\}$ ($\epsilon$ denotes the empty label),
\item the set of vertices $V'\subseteq V$ be ${\{v \mid v^-=0~\lor~v^+=0~\lor~l(v)=\kappa\}}$,
\item the set of edges $E' \subseteq V'\times V'\times L'$ be $\{(v,w,\epsilon)\mid (\exists e \in E)[\vec{s}(e)=v\land \vec{t}(e)=w]\}\cup$ $\{(v,w,x)\mid (\exists e_1,e_2 \in E)(\exists y \in V)[\vec{s}(e_1)=v\land \vec{t}(e_1)=y\land \vec{s}(e_2)=y\land \vec{t}(e_2)=w\land l(y)=x]\}$.
\item the source and target functions be mappings ${E'\rightarrow V'}$ given by ${s'(v,w,x)=v}$ and ${t'(v,w,x)=w}$, respectively,
\item the edge labelling function $m': E' \rightarrow M'$ and the vertex labelling function $l': V' \rightarrow L'$ be given by $m'(v,w,x)=x$ and $l'(v)=l(v)$, respectively.
\end{itemize}
For a graph homomorphism $\mathscr{D}(\alpha):\mathscr{D}(\lambda_1)\rightarrow \mathscr{D}(\lambda_2)$, the components of $S(\mathscr{D}(\alpha)):S(\mathscr{D}(\lambda_1))\rightarrow S(\mathscr{D}(\lambda_2))$ are:
\begin{itemize}
\item $S(\mathscr{D}(\alpha))_V:V'_1\rightarrow V'_2$ given by $S(\mathscr{D}(\alpha))_V(v) = \mathscr{D}(\alpha)_V(v)$,
\item $S(\mathscr{D}(\alpha))_E:E'_1\rightarrow E'_2$ given by $S(\mathscr{D}(\alpha))_E(v,w,x) = (\mathscr{D}(\alpha)_V(v),\mathscr{D}(\alpha)_V(w),x)$,
\item $S(\mathscr{D}(\alpha))_L:L'_1\rightarrow L'_2$ given by $S(\mathscr{D}(\alpha))_L(x)=\mathscr{D}(\alpha)_L(x)$, and
\item $S(\mathscr{D}(\alpha))_M:M'_1\rightarrow M'_2$ given by $S(\mathscr{D}(\alpha))_M(x) = x$.
\end{itemize}
Checking the functoriality of $S$ is routine and is analogous to that of Proposition \ref{prop:control-flow-functorial}.
\end{definition}

A glance at Definition \ref{def:computon-dfg-standard} reveals that $S$ is a functor that preserves all the boundary- and $\kappa$-vertices from a computon DFG $\mathscr{D}(\lambda)$, i.e., ed-inports, ed-outports and computation units. The only edges retained in $S(\mathscr{D}(\lambda))$ are those connected from a vertex with no incoming edges to a $\kappa$-node (i.e., data flows from ed-inports to computation units) and edges connected from a $\kappa$-node to a vertex with no outgoing edges (i.e., data flows from computation units to ed-outports). These preserved edges are all empty-labelled to meet the requirements of standard DFG notation. To fully satisfy such a notation, it suffices to replace $\kappa$ vertices from $S(\mathscr{D}(\lambda))$ with $\bigcirc$ symbols. For example, the DFGs in standard notation of the total sequential computon from Figure \ref{fig:application-AWS-workflow} and of the memory cell from Figure \ref{fig:application-LSTM-composites}(e) are displayed in Figures \ref{fig:application-AWS-separation}(b) and \ref{fig:application-LSTM-separation}(b), respectively. All these diagrams are created using the functorial construction $S$ from Definition \ref{def:computon-dfg-standard}. 

Basically, our functor $S$ leverages the separation of concerns of the computon model so as to optimise computon DFGs. On the one hand, $S$ compresses $\mathscr{D}(\lambda)$-paths of the form $v \rightarrow y \rightarrow w$ through the preservation of $v$ and $w$ in $S(\mathscr{D}(\lambda))$ and the creation of an edge $v \rightarrow w$ with the label of $y$. Apart from \emph{path compression}, $S$ performs \emph{multiplicity reduction}, i.e., edges with the same source, target and label are all collapsed onto a single edge.\footnote{As an edge is a triple $(v,w,x)$ and a set cannot contain repeated elements, it follows that there cannot be multiples edges with the same source, same target and same label in $S(\mathscr{D}(\lambda))$.} This behaviour is particularly important to simplify large and complex computon DFGs, leading to more efficient algorithms for data flow analysis and a clearer visualisation of data passing. Although multiplicity reduction is not relevant for any of the scenarios described in this section, we consider it to highlight the benefits of separating data and control towards data flow optimisations that do not consider control and preserve the order of data passing.  

\subsection{Case Study 1: Compositional AWS Infrastructure Deployment} \label{sec:case1}

\emph{AWS Step Functions} \cite{noauthor_aws_2023} is a serverless orchestration framework by Amazon Web Services (AWS), which allows software developers the implementation and management of (multi-step) serverless application workflows in the cloud, using visual and interactive programming constructs. A step function is a workflow that defines a high-level computation for the invocation of web-services in some pre-defined order, with the aim of automating a specific task such as database provisioning, release management or serverless deployment. In this section, we focus on a step function for automatic infrastructure deployment, provided as a use case by the AWS team \cite{mendonca_using_2017}, which follows different execution paths depending on the state of an AWS CloudFormation stack and intermediate processing results. If the stack does not exist, a new stack is created and deployment succeeds. Otherwise, a change set is created before inspecting its resources and deciding whether the change set will be executed or removed. If the change set is executed, then deployment succeeds; otherwise, deployment fails. The step function we consider is presented in Figure \ref{fig:application-AWS-concept}.
\vspace{-6pt}
\begin{figure}[!h]
\centering
\begin{tikzpicture}	
	\node[draw](s) at (0.3,0){\footnotesize Deployment Succeeded};\node[opacity=0.5] at (-1.7,0.15){\small $s_1$};
	\node[draw](b1) at (0.3,1){\footnotesize Stack Created?};
	\node[draw](g1) at (-1,2){\footnotesize Get Stack Creation Status};\node[opacity=0.5] at (-3.3,2.2){\small $g_1$};
	\node[draw](w1) at (0,3){\footnotesize Wait Stack Creation};\node[opacity=0.5] at (-1.9,3.2){\small $w_1$};
	\node[draw](c1) at (0,4){\footnotesize Create Stack};\node[opacity=0.5] at (-1.3,4.2){\small $c_1$};
	\draw[-latex] (b1) -- (s);
	\draw[-latex] (g1) to[bend right=15] (b1);\draw[-latex] ($(b1)+(1,0.25)$) to[bend right=40] (w1);\draw[-latex] (w1) to[bend right=15] (g1);
	\draw[-latex] (c1) -- (w1);
	
	\node[draw](s2) at (5,0){\footnotesize Deployment Succeeded};\node[opacity=0.5] at (3,0.2){\small $s_2$};
	\node[draw](f2) at (9,0){\footnotesize Deployment Failed};\node[opacity=0.5] at (10.75,0.2){\small $f_2$};
	\node[draw](x2) at (5,1){\footnotesize Execute Change Set};\node[opacity=0.5] at (3.1,1.2){\small $x_2$};
	\node[draw](d2) at (9,1){\footnotesize Delete Change Set};\node[opacity=0.5] at (10.75,1.2){\small $d_2$};
	\node[draw](b2) at (7,2){\footnotesize Safe to Update Infra?};
	\node[draw](i2) at (7,3){\footnotesize Inspect Change Set Changes};\node[opacity=0.5] at (4.6,3.2){\small $i_2$};
	\node[draw](b3) at (7,4){\footnotesize Change Set Created?};
	\node[draw](g2) at (5.5,5){\footnotesize Get Change Set Creation Status};\node[opacity=0.5] at (2.8,5.2){\small $g_2$};
	\node[draw](w2) at (7,6){\footnotesize Wait Change Set Creation};\node[opacity=0.5] at (9.35,6.2){\small $w_2$};
	\node[draw](c2) at (7,7){\footnotesize Create Change Set};\node[opacity=0.5] at (8.75,7.2){\small $c_2$};
	\draw[-latex] (x2) -- (s2);
	\draw[-latex] (d2) -- (f2);
	\draw[-latex] (b2) -- (x2);
	\draw[-latex] (b2) -- (d2);
	\draw[-latex] (i2) -- (b2);
	\draw[-latex] (b3) -- (i2);
	\draw[-latex] (g2) to[bend right=15] (b3);\draw[-latex] ($(b3)+(1,0.26)$) to[bend right=40] (w2);\draw[-latex] (w2) to[bend right=15] (g2);
	\draw[-latex] (c2) -- (w2);
	
	\node[draw](b4) at (3.5,8){\footnotesize Does Stack Exist?};
	\node[draw](k) at (3.5,9){\footnotesize Check Stack Existence};\node[opacity=0.5] at (1.55,9.2){\small $k$};
	\node[circle,draw](0) at (3.5,10.3){\footnotesize Start};
	\draw (2.08,8) -- (0,8);\draw[-latex] (0,8) -- (c1);\draw (4.91,8) -- (7,8);\draw[-latex] (7,8) -- (c2);	
	\draw[-latex] (k) -- (b4);
	\draw[-latex] (0) -- (k);
\end{tikzpicture}  
\caption{AWS step function to automate infrastructure deployment.}
\label{fig:application-AWS-concept}
\end{figure}

A step function does not allow the specification of data flow but data exchange is implicit in the processing of the steps involved in the workflow control flow being defined. Before executing the workflow, a JSON file is created to specify the initial input in the form of multiple properties and values. The JSON file is modified as the workflow execution progresses, by particularly appending the output of each intermediate web service invocation.

To compositionally construct the step function workflow, we use the model we propose in this paper by considering every service invocation as a functional computon which explicitly defines data required and produced (i.e., data flow and control flow are both explicit). The only processes we do not consider are those that branch control such as \emph{Stack Created?}. This is because those processes are built-in AWS functions which directly correspond to branching computons in our model. Although, strictly speaking, multiple computon ports can be coloured in the same way to represent the same data type (e.g., a boolean), for clarity concerns and demonstration purposes we treat every deployment parameter as a unique colour. The description of each port colour is presented in Figure \ref{fig:application-AWS-colours}.

\vspace{-2.9pt}

\begin{figure}[!h]
\centering
\resizebox{15.5cm}{!}{
\begin{tabular}{c|l}
 \hline
 Colour & Description \\\hline
 1 & Colour of an environment type on which the infrastructure code will be deployed \\
 & (e.g., development, testing or production). \\
 2 & Colour of a name of an AWS CloudFormation stack. \\
 3 & Colour of a path to an AWS CloudFormation template. \\
 4 & Colour of an identifier of a revision S3 bucket. \\
 5 & Colour of a revision S3 key. \\
 6 & Colour of a flag that specifies whether the AWS CloudFormation stack exists or not. \\
 7 & Colour of a stack creation status. \\
 8 & Colour of a state which can be either success or fail. \\
 9 & Colour of a change set name. \\
 10 & Colour of a change set creation status. \\
 11 & Colour of a change set action which determines whether the stack can be safely updated or not.
\end{tabular}}
\caption{Colours for the AWS step function shown in Figure \ref{fig:application-AWS-concept}.}
\label{fig:application-AWS-colours}
\vspace{-7.4pt}
\end{figure}

As compositionality enforces bottom-up construction, we start by defining the partial sequential computon $w_1 \rhd_{\rho_1} g_1$ where $w_1$ and $g_1$ respectively correspond to the processes \emph{wait stack creation} and \emph{get stack creation status} (see Figure \ref{fig:application-AWS-composites-left}(a)). This partial sequential computon, together with three functional computons ($e_1$, $e_2$ and $e_3$), serve as the basis to form the tail-iterative computon $(w_1 \rhd_{\rho_1} g_1)*_{\rho_2}$ which waits until the stack is created (see Figure \ref{fig:application-AWS-composites-left}(b)). The additional computons $e_1$, $e_2$ and $e_3$ echo data, remove data of colour $7$ and discard all data, respectively. The only branching structure in $(w_1 \unrhd_{\rho_1} g_1)*_{\rho_2}$ determines whether the stack has been created or not. If the stack has not been created, the loop continues; otherwise, $e_3$ is invoked to exit the tail-recursive composite and pass control to the external world.

The tail-iterative composite $(w_1 \rhd_{\rho_1} g_1)*_{\rho_2}$ is then used to construct the total sequential computon $c_1 \unrhd_{\rho_3} ((w_1 \rhd_{\rho_1} g_1)*_{\rho_2})$ where $c_1$ is a functional computon in charge of creating the stack. Such a sequential computon is then used as left operand to construct the (even more complex) total sequential computon $(c_1 \unrhd_{\rho_3} ((w_1 \rhd_{\rho_1} g_1)*_{\rho_2})) \unrhd_{\rho_4} s_1$ in which $s_1$ is a functional computon that marks deployment as successful. This complex total sequential composite, depicted in Figure \ref{fig:application-AWS-composites-left}(c), corresponds to the left path of the step function workflow shown in Figure \ref{fig:application-AWS-concept}. 

For the other path, we first construct the branching computon $(x_2 \unrhd_{\rho_5} s_2) ?_{\rho_7} (d_2 \unrhd_{\rho_6} f_2)$ whose only purpose is to succeed or fail deployment (see Figure \ref{fig:application-AWS-composites-right}(a)). In this composite, $x_2 \unrhd_{\rho_5} s_2$ corresponds to the total sequential computon which \emph{executes a change set} via the functional computon $x_2$, before marking deployment as \emph{successful} via the functional computon $s_2$. The other part of the branching composite triggers the total sequential computon $d_2 \unrhd_{\rho_6} f_2$ which interrupts deployment by first \emph{deleting the change set} via the functional computon $d_2$ and then using the functional computon $f_2$ to indicate that deployment has \emph{failed}. 

Returning to the innermost right part of the step function shown in Figure \ref{fig:application-AWS-concept}, we construct the partial sequential computon $w_2 \rhd_{\rho_8} g_2$ where $w_2$ and $g_2$ correspond to the proceses \emph{wait change set creation} and \emph{get change set creation status}, respectively. Such a sequential computon is then composed into a tail-iterative composite $(w_2 \rhd_{\rho_8} g_2)*_{\rho_9}$ which waits until the change set gets created, and whose structure is similar to that of $(w_1 \rhd_{\rho_1} g_1)*_{\rho_2}$ (see Figures \ref{fig:application-AWS-composites-left}(b) and \ref{fig:application-AWS-composites-right}(b)). 

\vspace{1pt}

\begin{figure}[!h]
\centering
\subcaptionbox{Partial sequential computon to wait stack creation before getting a stack creation status: ${w_1 \rhd_{\rho_1} g_1}$}
{\makebox[1\linewidth][c]{
\begin{tikzpicture}
	\computonComposite{0.2}{0.1}{4.5}{2.5};

	\computonPrimitive{0.8}{2}{0.8}{0.5}{$w_1$};  
	\qin{w1q0}{0}{2.2}{};
  \dinplain{w1i1}{0}{1.9}{$1$}{left};\flow{w1i1}{$(w1i1)+(3.3,0)$}{}{};
  \dinplain{w1i2}{0}{1.6}{$2$}{left};\flow{w1i2}{$(w1i2)+(3.3,0)$}{}{}; 
  \dinplain{w1i3}{0}{1.3}{$3$}{left};\flow{w1i3}{$(w1i3)+(3.3,0)$}{}{};
  \dinplain{w1i4}{0}{1}{$4$}{left};\flow{w1i4}{$(w1i4)+(3.3,0)$}{}{};
  \dinplain{w1i5}{0}{0.7}{$5$}{left};\flow{w1i5}{$(w1i5)+(3.3,0)$}{}{};

	\qmatch{wg0}{2.4}{2.2}{};\flow{{1.6,2.2}}{wg0}{dashed}{};\flow{wg0}{{3.4,2.2}}{dashed}{};
		
	\computonPrimitive{3.3}{0.2}{0.8}{2.2}{$g_1$};
	\qout{g1q1}{4.1}{2.2}{};
  \dout{g1o1}{4.1}{1.9}{$1$};
  \dout{g1o2}{4.1}{1.6}{$2$};  
  \dout{g1o3}{4.1}{1.3}{$3$};
  \dout{g1o4}{4.1}{1}{$4$};
  \dout{g1o5}{4.1}{0.7}{$5$};  
  \dout{g1o6}{4.1}{0.4}{$7$};
\end{tikzpicture}
}}\vspace{10pt}
\subcaptionbox{Tail-iterative computon over (a): ${(w_1 \rhd_{\rho_1} g_1)*_{\rho_2}}$}
{
\begin{tikzpicture}
	\computonComposite{-0.5}{-0.1}{10.4}{8.5};
  
  \computonPrimitive{2.3}{4.8}{0.8}{2}{$e_1$};
	\qinplain{e1q1}{-0.7}{6.5}{};\flowdiag{e1q1}{$(e1q1)+(3,0)$}{dashed}{}{pos=0.22};
  \dinplain{e1i1}{-0.7}{6.2}{$1$};\flowdiag{e1i1}{$(e1i1)+(3,0)$}{}{}{pos=0.22};
  \dinplain{e1i2}{-0.7}{5.9}{$2$};\flowdiag{e1i2}{$(e1i2)+(3,0)$}{}{}{pos=0.22};  
  \dinplain{e1i3}{-0.7}{5.6}{$3$};\flowdiag{e1i3}{$(e1i3)+(3,0)$}{}{}{pos=0.22};
  \dinplain{e1i4}{-0.7}{5.3}{$4$};\flowdiag{e1i4}{$(e1i4)+(3,0)$}{}{}{pos=0.22};
  \dinplain{e1i5}{-0.7}{5}{$5$};\flowdiag{e1i5}{$(e1i5)+(3,0)$}{}{}{pos=0.22};  
 	
  \computonPrimitive{2.3}{2.6}{0.8}{2}{$e_2$};  
  
  \computonPrimitive{2.3}{0}{0.8}{2.2}{$e_3$}; 
  \qoutplain{3q0}{9.4}{1.1}{};\flow{{3.1,1.1}}{3q0}{dashed}{};  
  
  \qmatch{g1q0}{1.5}{4.2}{};\flowdiag{g1q0}{{2.3,4.5}}{dashed}{}{pos=0.5,rotate=45};\flowdiag{g1q0}{{2.3,2}}{dashed}{}{pos=0.5,rotate=308};
  \dmatch{g1d1}{1.5}{3.6}{$1$}{above};\flowdiag{g1d1}{{2.3,4.2}}{}{}{pos=0.5,rotate=45};\flowdiag{g1d1}{{2.3,1.7}}{}{}{pos=0.5,rotate=308};
  \dmatch{g1d2}{1.5}{3}{$2$}{above};\flowdiag{g1d2}{{2.3,3.9}}{}{}{pos=0.5,rotate=45};\flowdiag{g1d2}{{2.3,1.4}}{}{}{pos=0.5,rotate=308};
  \dmatch{g1d3}{1.5}{2.4}{$3$}{above};\flowdiag{g1d3}{{2.3,3.6}}{}{}{pos=0.5,rotate=45};\flowdiag{g1d3}{{2.3,1.1}}{}{}{pos=0.5,rotate=308};
  \dmatch{g1d4}{1.5}{1.8}{$4$}{above};\flowdiag{g1d4}{{2.3,3.3}}{}{}{pos=0.5,rotate=45};\flowdiag{g1d4}{{2.3,0.8}}{}{}{pos=0.5,rotate=308};
  \dmatch{g1d5}{1.5}{1.2}{$5$}{above};\flowdiag{g1d5}{{2.3,3}}{}{}{pos=0.5,rotate=45};\flowdiag{g1d5}{{2.3,0.5}}{}{}{pos=0.5,rotate=308};
  \dmatch{g1d6}{1.5}{0.6}{$7$}{above};\flowdiag{g1d6}{{2.3,2.7}}{}{}{pos=0.5,rotate=45};\flowdiag{g1d6}{{2.3,0.3}}{}{}{pos=0.5,rotate=308};
	
	\qmatch{i1m0}{4.1}{6}{};\flowdiag{{3.2,6.5}}{i1m0}{dashed}{}{pos=0.5,rotate=308};\flowdiag{{3.1,4.5}}{i1m0}{dashed}{}{pos=0.5,rotate=45};
  \dmatch{i1m1}{4.1}{5.4}{$1$}{above};\flowdiag{{3.1,6.2}}{i1m1}{}{}{pos=0.5,rotate=308};\flowdiag{{3.1,4.2}}{i1m1}{}{}{pos=0.5,rotate=45};
  \dmatch{i1m2}{4.1}{4.8}{$2$}{above};\flowdiag{{3.1,5.9}}{i1m2}{}{}{pos=0.5,rotate=308};\flowdiag{{3.1,3.9}}{i1m2}{}{}{pos=0.5,rotate=45};
  \dmatch{i1m3}{4.1}{4.2}{$3$}{above};\flowdiag{{3.1,5.6}}{i1m3}{}{}{pos=0.5,rotate=308};\flowdiag{{3.1,3.6}}{i1m3}{}{}{pos=0.5,rotate=35};
  \dmatch{i1m4}{4.1}{3.6}{$4$}{above};\flowdiag{{3.1,5.3}}{i1m4}{}{}{pos=0.5,rotate=308};\flowdiag{{3.1,3.3}}{i1m4}{}{}{pos=0.5,rotate=25};
  \dmatch{i1m5}{4.1}{3}{$5$}{above};\flowdiag{{3.1,5}}{i1m5}{}{}{pos=0.5,rotate=308};\flowdiag{{3.1,3}}{i1m5}{}{}{pos=0.5};
	
	\begin{scope}[xshift=4.1cm,yshift=3.8cm]
	\computonComposite{0.2}{-1.7}{4.5}{4.3};

	\computonPrimitive{0.8}{2}{0.8}{0.5}{$w_1$};  
	\flow{i1m0}{$(i1m0)+(0.8,0)$}{dashed}{};
  \flow{i1m1}{$(i1m1)+(3.3,0)$}{}{};
  \flow{i1m2}{$(i1m2)+(3.3,0)$}{}{}; 
  \flow{i1m3}{$(i1m3)+(3.3,0)$}{}{};
  \flow{i1m4}{$(i1m4)+(3.3,0)$}{}{};
  \flow{i1m5}{$(i1m5)+(3.3,0)$}{}{};

	\qmatch{wg0}{2.4}{2.2}{};\flow{{1.6,2.2}}{wg0}{dashed}{};\flow{wg0}{{3.4,2.2}}{dashed}{};
		
	\computonPrimitive{3.3}{-1.6}{0.8}{4}{$g_1$};
	\end{scope}	
	
	\draw[dashed] (8.2,6) -- (8.5,6) -- (8.5,7) to node [pos=0.5] {\invertedarrowflow} (0.9,7) -- (0.9,4.2) -- (g1q0);
	\draw (8.2,5.4) -- (8.7,5.4) -- (8.7,7.2) to node [pos=0.5] {\invertedarrowflow} (0.7,7.2) -- (0.7,3.6) -- (g1d1);
	\draw (8.2,4.8) -- (8.9,4.8) -- (8.9,7.4) to node [pos=0.5] {\invertedarrowflow} (0.5,7.4) -- (0.5,3) -- (g1d2);
	\draw (8.2,4.2) -- (9.1,4.2) -- (9.1,7.6) to node [pos=0.5] {\invertedarrowflow} (0.3,7.6) -- (0.3,2.4) -- (g1d3);
	\draw (8.2,3.6) -- (9.3,3.6) -- (9.3,7.8) to node [pos=0.5] {\invertedarrowflow} (0.1,7.8) -- (0.1,1.8) -- (g1d4);
	\draw (8.2,3) -- (9.5,3) -- (9.5,8) to node [pos=0.5] {\invertedarrowflow} (-0.1,8) -- (-0.1,1.2) -- (g1d5);
	\draw (8.2,2.4) -- (9.7,2.4) -- (9.7,8.2) to node [pos=0.5] {\invertedarrowflow} (-0.3,8.2) -- (-0.3,0.6) -- (g1d6);
\end{tikzpicture}
}\vspace{10pt}
\subcaptionbox{Composite computon corresponding to the left path of the step function shown in Figure \ref{fig:application-AWS-concept}: ${(c_1 \unrhd_{\rho_3} ((w_1 \rhd_{\rho_1} g_1)*_{\rho_2})) \unrhd_{\rho_4} s_1}$}
{
\begin{tikzpicture}[scale=0.9]

  \computonComposite{-0.1}{-0.8}{15.5}{9.5};
  \computonComposite{0.2}{-0.6}{12.9}{9.1};
	\computonComposite{2.8}{-0.5}{10.2}{8.9};
	
	\computonPrimitive{0.8}{4.5}{0.8}{2.3}{$c_1$};
	\qinplain{c1q1}{-0.3}{6.5}{}{left};\flowdiag{c1q1}{$(c1q1)+(1.1,0)$}{dashed}{}{pos=0.24};
  \dinplain{c1i1}{-0.3}{6.2}{$1$}{left};\flowdiag{c1i1}{$(c1i1)+(1.1,0)$}{}{}{pos=0.24};
  \dinplain{c1i2}{-0.3}{5.9}{$2$}{left};\flowdiag{c1i2}{$(c1i2)+(1.1,0)$}{}{}{pos=0.24};
  \dinplain{c1i3}{-0.3}{5.6}{$3$}{left};\flowdiag{c1i3}{$(c1i3)+(1.1,0)$}{}{}{pos=0.24};
  \dinplain{c1i4}{-0.3}{5.3}{$4$}{left};\flowdiag{c1i4}{$(c1i4)+(1.1,0)$}{}{}{pos=0.24};
  \dinplain{c1i5}{-0.3}{5}{$5$}{left};\flowdiag{c1i5}{$(c1i5)+(1.1,0)$}{}{}{pos=0.24};
  \dinplain{c1i6}{-0.3}{4.7}{$6$}{left};\flowdiag{c1i6}{$(c1i6)+(1.1,0)$}{}{}{pos=0.24};
  
  \qmatch{c1q2}{2.4}{6.5}{};\flow{{1.6,6.5}}{c1q2}{dashed}{};\flowdiag{c1q2}{$(c1q2)+(3.1,0)$}{dashed}{}{pos=0.17};
  \dmatch{c1o1}{2.4}{6.2}{$1$}{right};\flow{{1.6,6.2}}{c1o1}{}{};\flowdiag{c1o1}{$(c1o1)+(3.1,0)$}{}{}{pos=0.17};
  \dmatch{c1o2}{2.4}{5.9}{$2$}{right};\flow{{1.6,5.9}}{c1o2}{}{};\flowdiag{c1o2}{$(c1o2)+(3.1,0)$}{}{}{pos=0.17};
  \dmatch{c1o3}{2.4}{5.6}{$3$}{right};\flow{{1.6,5.6}}{c1o3}{}{};\flowdiag{c1o3}{$(c1o3)+(3.1,0)$}{}{}{pos=0.17};
  \dmatch{c1o4}{2.4}{5.3}{$4$}{right};\flow{{1.6,5.3}}{c1o4}{}{};\flowdiag{c1o4}{$(c1o4)+(3.1,0)$}{}{}{pos=0.17};
  \dmatch{c1o5}{2.4}{5}{$5$}{right};\flow{{1.6,5}}{c1o5}{}{};\flowdiag{c1o5}{$(c1o5)+(3.1,0)$}{}{}{pos=0.17};

	\begin{scope}[xshift=3.2cm]	  
  \computonPrimitive{2.3}{4.8}{0.8}{2}{$e_1$};
 	
  \computonPrimitive{2.3}{2.6}{0.8}{2}{$e_2$};  
  
  \computonPrimitive{2.3}{-0.3}{0.8}{2.5}{$e_3$}; 
  
  \qmatch{g1q0}{1.5}{4.2}{};\flowdiag{g1q0}{{2.3,4.5}}{dashed}{}{pos=0.5,rotate=45};\flowdiag{g1q0}{{2.3,2}}{dashed}{}{pos=0.5,rotate=308};
  \dmatch{g1d1}{1.5}{3.6}{$1$}{above};\flowdiag{g1d1}{{2.3,4.2}}{}{}{pos=0.5,rotate=45};\flowdiag{g1d1}{{2.3,1.7}}{}{}{pos=0.5,rotate=308};
  \dmatch{g1d2}{1.5}{3}{$2$}{above};\flowdiag{g1d2}{{2.3,3.9}}{}{}{pos=0.5,rotate=45};\flowdiag{g1d2}{{2.3,1.4}}{}{}{pos=0.5,rotate=308};
  \dmatch{g1d3}{1.5}{2.4}{$3$}{above};\flowdiag{g1d3}{{2.3,3.6}}{}{}{pos=0.5,rotate=45};\flowdiag{g1d3}{{2.3,1.1}}{}{}{pos=0.5,rotate=308};
  \dmatch{g1d4}{1.5}{1.8}{$4$}{above};\flowdiag{g1d4}{{2.3,3.3}}{}{}{pos=0.5,rotate=45};\flowdiag{g1d4}{{2.3,0.8}}{}{}{pos=0.5,rotate=308};
  \dmatch{g1d5}{1.5}{1.2}{$5$}{above};\flowdiag{g1d5}{{2.3,3}}{}{}{pos=0.5,rotate=45};\flowdiag{g1d5}{{2.3,0.5}}{}{}{pos=0.5,rotate=308};
  \dmatch{g1d6}{1.5}{0.6}{$7$}{above};\flowdiag{g1d6}{{2.3,2.7}}{}{}{pos=0.5,rotate=45};\flowdiag{g1d6}{{2.3,0.3}}{}{}{pos=0.5,rotate=308};
	
	\qmatch{i1m0}{4.1}{6}{};\flowdiag{{3.2,6.5}}{i1m0}{dashed}{}{pos=0.5,rotate=308};\flowdiag{{3.1,4.5}}{i1m0}{dashed}{}{pos=0.5,rotate=45};
  \dmatch{i1m1}{4.1}{5.4}{$1$}{above};\flowdiag{{3.1,6.2}}{i1m1}{}{}{pos=0.5,rotate=308};\flowdiag{{3.1,4.2}}{i1m1}{}{}{pos=0.5,rotate=45};
  \dmatch{i1m2}{4.1}{4.8}{$2$}{above};\flowdiag{{3.1,5.9}}{i1m2}{}{}{pos=0.5,rotate=308};\flowdiag{{3.1,3.9}}{i1m2}{}{}{pos=0.5,rotate=45};
  \dmatch{i1m3}{4.1}{4.2}{$3$}{above};\flowdiag{{3.1,5.6}}{i1m3}{}{}{pos=0.5,rotate=308};\flowdiag{{3.1,3.6}}{i1m3}{}{}{pos=0.5,rotate=35};
  \dmatch{i1m4}{4.1}{3.6}{$4$}{above};\flowdiag{{3.1,5.3}}{i1m4}{}{}{pos=0.5,rotate=308};\flowdiag{{3.1,3.3}}{i1m4}{}{}{pos=0.5,rotate=25};
  \dmatch{i1m5}{4.1}{3}{$5$}{above};\flowdiag{{3.1,5}}{i1m5}{}{}{pos=0.5,rotate=308};\flowdiag{{3.1,3}}{i1m5}{}{}{pos=0.5};
	
	\begin{scope}[xshift=4.1cm,yshift=3.8cm]
	\computonComposite{0.2}{-1.7}{4.5}{4.3};

	\computonPrimitive{0.8}{2}{0.8}{0.5}{$w_1$};  
	\flow{i1m0}{$(i1m0)+(0.8,0)$}{dashed}{};
  \flow{i1m1}{$(i1m1)+(3.3,0)$}{}{};
  \flow{i1m2}{$(i1m2)+(3.3,0)$}{}{}; 
  \flow{i1m3}{$(i1m3)+(3.3,0)$}{}{};
  \flow{i1m4}{$(i1m4)+(3.3,0)$}{}{};
  \flow{i1m5}{$(i1m5)+(3.3,0)$}{}{};

	\qmatch{wg0}{2.4}{2.2}{};\flow{{1.6,2.2}}{wg0}{dashed}{};\flow{wg0}{{3.4,2.2}}{dashed}{};
		
	\computonPrimitive{3.3}{-1.6}{0.8}{4}{$g_1$};
	\end{scope}	
	
	\draw[dashed] (8.2,6) -- (8.5,6) -- (8.5,7) to node [pos=0.5] {\invertedarrowflow} (0.9,7) -- (0.9,4.2) -- (g1q0);
	\draw (8.2,5.4) -- (8.7,5.4) -- (8.7,7.2) to node [pos=0.5] {\invertedarrowflow} (0.7,7.2) -- (0.7,3.6) -- (g1d1);
	\draw (8.2,4.8) -- (8.9,4.8) -- (8.9,7.4) to node [pos=0.5] {\invertedarrowflow} (0.5,7.4) -- (0.5,3) -- (g1d2);
	\draw (8.2,4.2) -- (9.1,4.2) -- (9.1,7.6) to node [pos=0.5] {\invertedarrowflow} (0.3,7.6) -- (0.3,2.4) -- (g1d3);
	\draw (8.2,3.6) -- (9.3,3.6) -- (9.3,7.8) to node [pos=0.5] {\invertedarrowflow} (0.1,7.8) -- (0.1,1.8) -- (g1d4);
	\draw (8.2,3) -- (9.5,3) -- (9.5,8) to node [pos=0.5] {\invertedarrowflow} (-0.1,8) -- (-0.1,1.2) -- (g1d5);
	\draw (8.2,2.4) -- (9.7,2.4) -- (9.7,8.2) to node [pos=0.5] {\invertedarrowflow} (-0.3,8.2) -- (-0.3,0.6) -- (g1d6);
	\end{scope}
	
	\computonPrimitive{14}{-0.3}{0.8}{2.3}{$s_1$};
	\qmatch{s1q1}{13.2}{0.7}{};\flow{{6.3,0.7}}{s1q1}{dashed}{};\flow{s1q1}{$(s1q1)+(0.8,0)$}{dashed}{};
	\qout{s1q2}{14.8}{0.85}{};
  \dout{s1o1}{14.8}{0.55}{$8$};
\end{tikzpicture}
}
\caption{Constructing the composite computon for the left path of the step function shown in Figure \ref{fig:application-AWS-concept}.}
\label{fig:application-AWS-composites-left}
\end{figure}

\begin{figure}[!h]
\centering
\subcaptionbox{Branching computon that determines whether deployment succeeds or fails: ${(x_2 \unrhd_{\rho_5} s_2) ?_{\rho_7} (d_2 \unrhd_{\rho_6} f_2)}$}
{\makebox[1\linewidth][c]{
\begin{tikzpicture}
	\computonComposite{0.4}{0}{4.4}{6.1};
	\computonComposite{0.8}{3.2}{3.6}{2.7};
	\computonComposite{0.8}{0.2}{3.6}{2.7};
	
	\qinplain{i0}{0}{4.2}{};\flowdiag{i0}{{1,5.6}}{dashed}{}{pos=0.5,rotate=45};\flowdiag{i0}{{1,2.5}}{dashed}{}{pos=0.5,rotate=311};
	\dinplain{i1}{0}{3.9}{$1$};\flowdiag{i1}{{1,5.3}}{}{}{pos=0.5,rotate=45};\flowdiag{i1}{{1,2.2}}{}{}{pos=0.5,rotate=311};
	\dinplain{i2}{0}{3.6}{$2$};\flowdiag{i2}{{1,5}}{}{}{pos=0.5,rotate=45};\flowdiag{i2}{{1,1.9}}{}{}{pos=0.5,rotate=311};
	\dinplain{i3}{0}{3.3}{$3$};\flowdiag{i3}{{1,4.7}}{}{}{pos=0.5,rotate=45};\flowdiag{i3}{{1,1.6}}{}{}{pos=0.5,rotate=311};
	\dinplain{i4}{0}{3}{$4$};\flowdiag{i4}{{1,4.4}}{}{}{pos=0.5,rotate=45};\flowdiag{i4}{{1,1.3}}{}{}{pos=0.5,rotate=311};
	\dinplain{i5}{0}{2.7}{$5$};\flowdiag{i5}{{1,4.1}}{}{}{pos=0.5,rotate=45};\flowdiag{i5}{{1,1}}{}{}{pos=0.5,rotate=311};
	\dinplain{i6}{0}{2.4}{$9$};\flowdiag{i6}{{1,3.8}}{}{}{pos=0.5,rotate=45};\flowdiag{i6}{{1,0.7}}{}{}{pos=0.5,rotate=311};
	\dinplain{i7}{0}{2.1}{$11$};\flowdiag{i7}{{1,3.5}}{}{}{pos=0.5,rotate=45};\flowdiag{i7}{{1,0.4}}{}{}{pos=0.5,rotate=311};
	
	\computonPrimitive{1}{3.4}{0.8}{2.3}{$x_2$};
	\qmatch{x2o1}{2.6}{4.5}{};\flow{$(x2o1)+(-0.8,0)$}{x2o1}{dashed}{};\flow{x2o1}{$(x2o1)+(0.8,0)$}{dashed}{};	
	\computonPrimitive{3.4}{3.4}{0.8}{2.3}{$s_2$};
	
	\computonPrimitive{1}{0.4}{0.8}{2.3}{$d_2$};
	\qmatch{d2o1}{2.6}{1.5}{};\flow{$(d2o1)+(-0.8,0)$}{d2o1}{dashed}{};\flow{d2o1}{$(d2o1)+(0.8,0)$}{dashed}{};	
	\computonPrimitive{3.4}{0.4}{0.8}{2.3}{$f_2$};
	
	\qoutplain{o0}{4.4}{4.2}{};\flowdiag{{4.2,5.6}}{o0}{dashed}{}{pos=0.5,rotate=45};\flowdiag{{4.2,2.5}}{o0}{dashed}{}{pos=0.5,rotate=311};	
	\doutplain{o1}{4.4}{3.9}{$8$};\flowdiag{{4.2,5.3}}{o1}{}{}{pos=0.5,rotate=45};\flowdiag{{4.2,2.2}}{o1}{}{}{pos=0.5,rotate=311};		
\end{tikzpicture}
}}\vspace{10pt}
\subcaptionbox{Total sequential computon that creates a change set, waits until the set is created and then inspects change set changes: ${(c_2 \unrhd_{\rho_{10}} ((w_2 \rhd_{\rho_8} g_2)*_{\rho_9})) \unrhd_{\rho_{11}} i_2}$}
{
\begin{tikzpicture}[scale=0.865]

  \computonComposite{-0.3}{-1.1}{16.1}{10};
  \computonComposite{0}{-0.9}{13.3}{9.6};
	\computonComposite{2.6}{-0.8}{10.6}{9.4};
	
	\computonPrimitive{0.6}{4.5}{0.8}{2.3}{$c_2$};
	\qinplain{c1q1}{-0.5}{6.7}{}{left};\flowdiag{c1q1}{$(c1q1)+(1.1,0)$}{dashed}{}{pos=0.24};
  \dinplain{c1i1}{-0.5}{6.4}{$1$}{left};\flowdiag{c1i1}{$(c1i1)+(1.1,0)$}{}{}{pos=0.24};
  \dinplain{c1i2}{-0.5}{6.1}{$2$}{left};\flowdiag{c1i2}{$(c1i2)+(1.1,0)$}{}{}{pos=0.24};
  \dinplain{c1i3}{-0.5}{5.8}{$3$}{left};\flowdiag{c1i3}{$(c1i3)+(1.1,0)$}{}{}{pos=0.24};
  \dinplain{c1i4}{-0.5}{5.5}{$4$}{left};\flowdiag{c1i4}{$(c1i4)+(1.1,0)$}{}{}{pos=0.24};
  \dinplain{c1i5}{-0.5}{5.2}{$5$}{left};\flowdiag{c1i5}{$(c1i5)+(1.1,0)$}{}{}{pos=0.24};  
  \dinplain{c1i6}{-0.5}{4.9}{$6$}{left};\flowdiag{c1i6}{$(c1i6)+(1.1,0)$}{}{}{pos=0.24};  
  
  \qmatch{c1q2}{2}{6.7}{};\flow{{1.4,6.7}}{c1q2}{dashed}{};\flowdiag{c1q2}{$(c1q2)+(3.5,0)$}{dashed}{}{pos=0.12};
  \dmatch{c1o1}{2}{6.4}{$1$}{right};\flow{{1.4,6.4}}{c1o1}{}{};\flowdiag{c1o1}{$(c1o1)+(3.5,0)$}{}{}{pos=0.12};
  \dmatch{c1o2}{2}{6.1}{$2$}{right};\flow{{1.4,6.1}}{c1o2}{}{};\flowdiag{c1o2}{$(c1o2)+(3.5,0)$}{}{}{pos=0.12};
  \dmatch{c1o3}{2}{5.8}{$3$}{right};\flow{{1.4,5.8}}{c1o3}{}{};\flowdiag{c1o3}{$(c1o3)+(3.5,0)$}{}{}{pos=0.12};
  \dmatch{c1o4}{2}{5.5}{$4$}{right};\flow{{1.4,5.5}}{c1o4}{}{};\flowdiag{c1o4}{$(c1o4)+(3.5,0)$}{}{}{pos=0.12};
  \dmatch{c1o5}{2}{5.2}{$5$}{right};\flow{{1.4,5.2}}{c1o5}{}{};\flowdiag{c1o5}{$(c1o5)+(3.5,0)$}{}{}{pos=0.12};
  \dmatch{c1o6}{2}{4.9}{$9$}{right};\flow{{1.4,4.9}}{c1o6}{}{};\flowdiag{c1o6}{$(c1o6)+(3.5,0)$}{}{}{pos=0.12};

	\begin{scope}[xshift=3.2cm]	  
  \computonPrimitive{2.3}{4.8}{0.8}{2}{$e_4$};
 	
  \computonPrimitive{2.3}{2.3}{0.8}{2.3}{$e_5$};  
  
  \computonPrimitive{2.3}{-0.6}{0.8}{2.8}{$e_6$}; 
  
  \qmatch{g1q0}{1.5}{4.2}{};\flowdiag{g1q0}{{2.3,4.5}}{dashed}{}{pos=0.5,rotate=45};\flowdiag{g1q0}{{2.3,2}}{dashed}{}{pos=0.5,rotate=308};
  \dmatch{g1d1}{1.5}{3.6}{$1$}{above};\flowdiag{g1d1}{{2.3,4.2}}{}{}{pos=0.5,rotate=45};\flowdiag{g1d1}{{2.3,1.7}}{}{}{pos=0.5,rotate=308};
  \dmatch{g1d2}{1.5}{3}{$2$}{above};\flowdiag{g1d2}{{2.3,3.9}}{}{}{pos=0.5,rotate=45};\flowdiag{g1d2}{{2.3,1.4}}{}{}{pos=0.5,rotate=308};
  \dmatch{g1d3}{1.5}{2.4}{$3$}{above};\flowdiag{g1d3}{{2.3,3.6}}{}{}{pos=0.5,rotate=45};\flowdiag{g1d3}{{2.3,1.1}}{}{}{pos=0.5,rotate=308};
  \dmatch{g1d4}{1.5}{1.8}{$4$}{above};\flowdiag{g1d4}{{2.3,3.3}}{}{}{pos=0.5,rotate=45};\flowdiag{g1d4}{{2.3,0.8}}{}{}{pos=0.5,rotate=308};
  \dmatch{g1d5}{1.5}{1.2}{$5$}{above};\flowdiag{g1d5}{{2.3,3}}{}{}{pos=0.5,rotate=45};\flowdiag{g1d5}{{2.3,0.5}}{}{}{pos=0.5,rotate=308};
  \dmatch{g1d6}{1.5}{0.6}{$9$}{above};\flowdiag{g1d6}{{2.3,2.7}}{}{}{pos=0.5,rotate=45};\flowdiag{g1d6}{{2.3,0.2}}{}{}{pos=0.5,rotate=308};
  \dmatch{g1d7}{1.5}{0}{$10$}{above};\flowdiag{g1d7}{{2.3,2.4}}{}{}{pos=0.5,rotate=45};\flowdiag{g1d7}{{2.3,-0.1}}{}{}{pos=0.5,rotate=308};
	
	\qmatch{i1m0}{4.1}{6}{};\flowdiag{{3.2,6.6}}{i1m0}{dashed}{}{pos=0.5,rotate=308};\flowdiag{{3.1,4.5}}{i1m0}{dashed}{}{pos=0.5,rotate=45};
  \dmatch{i1m1}{4.1}{5.4}{$1$}{above};\flowdiag{{3.1,6.3}}{i1m1}{}{}{pos=0.5,rotate=308};\flowdiag{{3.1,4.2}}{i1m1}{}{}{pos=0.5,rotate=45};
  \dmatch{i1m2}{4.1}{4.8}{$2$}{above};\flowdiag{{3.1,6}}{i1m2}{}{}{pos=0.5,rotate=308};\flowdiag{{3.1,3.9}}{i1m2}{}{}{pos=0.5,rotate=45};
  \dmatch{i1m3}{4.1}{4.2}{$3$}{above};\flowdiag{{3.1,5.7}}{i1m3}{}{}{pos=0.5,rotate=308};\flowdiag{{3.1,3.6}}{i1m3}{}{}{pos=0.5,rotate=35};
  \dmatch{i1m4}{4.1}{3.6}{$4$}{above};\flowdiag{{3.1,5.4}}{i1m4}{}{}{pos=0.5,rotate=308};\flowdiag{{3.1,3.3}}{i1m4}{}{}{pos=0.5,rotate=25};
  \dmatch{i1m5}{4.1}{3}{$5$}{above};\flowdiag{{3.1,5.1}}{i1m5}{}{}{pos=0.5,rotate=308};\flowdiag{{3.1,3}}{i1m5}{}{}{pos=0.5};
  \dmatch{i1m6}{4.1}{2.4}{$9$}{above};\flowdiag{{3.1,4.8}}{i1m6}{}{}{pos=0.5,rotate=308};\flowdiag{{3.1,2.7}}{i1m6}{}{}{pos=0.5};
	
	\begin{scope}[xshift=4.1cm,yshift=3.8cm]
	\computonComposite{0.2}{-1.7}{4.5}{4.3};

	\computonPrimitive{0.8}{2}{0.8}{0.5}{$w_2$};  
	\flow{i1m0}{$(i1m0)+(0.8,0)$}{dashed}{};
  \flow{i1m1}{$(i1m1)+(3.3,0)$}{}{};
  \flow{i1m2}{$(i1m2)+(3.3,0)$}{}{}; 
  \flow{i1m3}{$(i1m3)+(3.3,0)$}{}{};
  \flow{i1m4}{$(i1m4)+(3.3,0)$}{}{};
  \flow{i1m5}{$(i1m5)+(3.3,0)$}{}{};
  \flow{i1m6}{$(i1m6)+(3.3,0)$}{}{};

	\qmatch{wg0}{2.4}{2.2}{};\flow{{1.6,2.2}}{wg0}{dashed}{};\flow{wg0}{{3.4,2.2}}{dashed}{};
		
	\computonPrimitive{3.3}{-1.6}{0.8}{4}{$g_2$};
	\end{scope}	
	
	\draw[dashed] (8.2,6) -- (8.5,6) -- (8.5,7) to node [pos=0.5] {\invertedarrowflow} (0.9,7) -- (0.9,4.2) -- (g1q0);
	\draw (8.2,5.6) -- (8.7,5.6) -- (8.7,7.2) to node [pos=0.5] {\invertedarrowflow} (0.7,7.2) -- (0.7,3.6) -- (g1d1);
	\draw (8.2,5.2) -- (8.9,5.2) -- (8.9,7.4) to node [pos=0.5] {\invertedarrowflow} (0.5,7.4) -- (0.5,3) -- (g1d2);
	\draw (8.2,4.8) -- (9.1,4.8) -- (9.1,7.6) to node [pos=0.5] {\invertedarrowflow} (0.3,7.6) -- (0.3,2.4) -- (g1d3);
	\draw (8.2,4.4) -- (9.3,4.4) -- (9.3,7.8) to node [pos=0.5] {\invertedarrowflow} (0.1,7.8) -- (0.1,1.8) -- (g1d4);
	\draw (8.2,4) -- (9.5,4) -- (9.5,8) to node [pos=0.5] {\invertedarrowflow} (-0.1,8) -- (-0.1,1.2) -- (g1d5);
	\draw (8.2,3.6) -- (9.7,3.6) -- (9.7,8.2) to node [pos=0.5] {\invertedarrowflow} (-0.3,8.2) -- (-0.3,0.6) -- (g1d6);
	\draw (8.2,3.2) -- (9.9,3.2) -- (9.9,8.4) to node [pos=0.5] {\invertedarrowflow} (-0.5,8.4) -- (-0.5,0) -- (g1d7);
	\end{scope}
	
	\computonPrimitive{14.4}{-0.6}{0.8}{2.6}{$i_2$};
	\qmatch{s1q1}{13.6}{1.6}{};\flow{{6.3,1.6}}{s1q1}{dashed}{};\flow{s1q1}{$(s1q1)+(0.8,0)$}{dashed}{};
	\dmatch{s1i1}{13.6}{1.3}{$1$}{right};\flow{{6.3,1.3}}{s1i1}{}{};\flow{s1i1}{$(s1i1)+(0.8,0)$}{}{};
	\dmatch{s1i2}{13.6}{1}{$2$}{right};\flow{{6.3,1}}{s1i2}{}{};\flow{s1i2}{$(s1i2)+(0.8,0)$}{}{};
	\dmatch{s1i3}{13.6}{0.7}{$3$}{right};\flow{{6.3,0.7}}{s1i3}{}{};\flow{s1i3}{$(s1i3)+(0.8,0)$}{}{};
	\dmatch{s1i4}{13.6}{0.4}{$4$}{right};\flow{{6.3,0.4}}{s1i4}{}{};\flow{s1i4}{$(s1i4)+(0.8,0)$}{}{};
	\dmatch{s1i5}{13.6}{0.1}{$5$}{right};\flow{{6.3,0.1}}{s1i5}{}{};\flow{s1i5}{$(s1i5)+(0.8,0)$}{}{};
	\dmatch{s1i6}{13.6}{-0.2}{$9$}{right};\flow{{6.3,-0.2}}{s1i6}{}{};\flow{s1i6}{$(s1i6)+(0.8,0)$}{}{};	
	\qout{s1q2}{15.2}{1.6}{};
  \dout{s1o1}{15.2}{1.3}{$1$};
  \dout{s1o2}{15.2}{1}{$2$};
  \dout{s1o3}{15.2}{0.7}{$3$};
  \dout{s1o4}{15.2}{0.4}{$4$};
  \dout{s1o5}{15.2}{0.1}{$5$};
  \dout{s1o6}{15.2}{-0.2}{$9$};
  \dout{s1o7}{15.2}{-0.5}{$11$};
\end{tikzpicture}
}\vspace{10pt}
\subcaptionbox{Composite corresponding to the right path of the step function shown in Figure \ref{fig:application-AWS-concept}: ${((c_2 \unrhd_{\rho_{10}} ((w_2 \rhd_{\rho_8} g_2)*_{\rho_9})) \unrhd_{\rho_{11}} i_2) \unrhd_{\rho_{12}} ((x_2 \unrhd_{\rho_5} s_2) ?_{\rho_7} (d_2 \unrhd_{\rho_6} f_2))}$}
{\makebox[1\linewidth][c]{
\begin{tikzpicture}
  \computonComposite{0.2}{-0.2}{6.8}{2.8};

	\computonComposite{0.8}{0}{2}{2.4};\node at (1.8,1.2){(b)}; 
	\qin{s1i0}{0}{2.2}{};
	\din{s1i1}{0}{1.9}{$1$};
	\din{s1i2}{0}{1.6}{$2$};
	\din{s1i3}{0}{1.3}{$3$};
	\din{s1i4}{0}{1}{$4$};
	\din{s1i5}{0}{0.7}{$5$};
	\din{s1i6}{0}{0.4}{$6$};
	
	\qmatch{s1q2}{3.6}{2.2}{};\flow{$(s1q2)+(-0.8,0)$}{s1q2}{dashed}{};\flowdiag{s1q2}{$(s1q2)+(0.8,0)$}{dashed}{}{pos=0.7};
  \dmatch{s1o1}{3.6}{1.9}{$1$}{right};\flow{$(s1o1)+(-0.8,0)$}{s1o1}{}{};\flowdiag{s1o1}{$(s1o1)+(0.8,0)$}{}{}{pos=0.7};
  \dmatch{s1o2}{3.6}{1.6}{$2$}{right};\flow{$(s1o2)+(-0.8,0)$}{s1o2}{}{};\flowdiag{s1o2}{$(s1o2)+(0.8,0)$}{}{}{pos=0.7};
  \dmatch{s1o3}{3.6}{1.3}{$3$}{right};\flow{$(s1o3)+(-0.8,0)$}{s1o3}{}{};\flowdiag{s1o3}{$(s1o3)+(0.8,0)$}{}{}{pos=0.7};
  \dmatch{s1o4}{3.6}{1}{$4$}{right};\flow{$(s1o4)+(-0.8,0)$}{s1o4}{}{};\flowdiag{s1o4}{$(s1o4)+(0.8,0)$}{}{}{pos=0.7};
  \dmatch{s1o5}{3.6}{0.7}{$5$}{right};\flow{$(s1o5)+(-0.8,0)$}{s1o5}{}{};\flowdiag{s1o5}{$(s1o5)+(0.8,0)$}{}{}{pos=0.7};
  \dmatch{s1o6}{3.6}{0.4}{$9$}{right};\flow{$(s1o6)+(-0.8,0)$}{s1o6}{}{};\flowdiag{s1o6}{$(s1o6)+(0.8,0)$}{}{}{pos=0.7};
  \dmatch{s1o7}{3.6}{0.1}{$11$}{right};\flow{$(s1o7)+(-0.8,0)$}{s1o7}{}{};\flowdiag{s1o7}{$(s1o7)+(0.8,0)$}{}{}{pos=0.7};
  
  \computonComposite{4.4}{0}{2}{2.4};\node at (5.4,1.2){(a)}; 
  \qout{s1o0}{6.4}{2.2}{};
	\dout{s1o1}{6.4}{1.9}{$8$};	
\end{tikzpicture}
}}
\caption{Constructing the composite computon for the right path of the step function shown in Figure \ref{fig:application-AWS-concept}.}
\label{fig:application-AWS-composites-right}
\vspace{-13pt}
\end{figure}
\newpage
The tail-iterative computon ${(w_2 \rhd_{\rho_8} g_2)*_{\rho_9}}$ is subsequently used as a right operand to define the total sequential computon ${c_2 \unrhd_{\rho_{10}} ((w_2 \rhd_{\rho_8} g_2)*_{\rho_9})}$ wherein $c_2$ is a functional computon in charge of \emph{creating the change set}. This newly constructed sequential composite is in turn used as left operand to construct the (even more complex) total sequential computon ${(c_2 \unrhd_{\rho_{10}} ((w_2 \rhd_{\rho_8} g_2)*_{\rho_9})) \unrhd_{\rho_{11}} i_2}$ wherein $i_2$ is a functional computon that \emph{inspects change set changes} to determine whether any of the existing resources need to be deleted or whether the existing stack can be safely updated. The whole structure of $(c_2 \unrhd_{\rho_{10}} ((w_2 \rhd_{\rho_8} g_2)*_{\rho_9})) \unrhd_{\rho_{11}} i_2$ is shown in Figure \ref{fig:application-AWS-composites-right}(b).

The most complex composite for the right path of the step function from Figure \ref{fig:application-AWS-concept} is constructed by taking the total sequential computon ${(c_2 \unrhd_{\rho_{10}} ((w_2 \rhd_{\rho_8} g_2)*_{\rho_9})) \unrhd_{\rho_{11}} i_2}$ and the branching computon ${(x_2 \unrhd_{\rho_5} s_2) ?_{\rho_7} (d_2 \unrhd_{\rho_6} f_2)}$ as left and right operands, respectively, in order to yield ${((c_2 \unrhd_{\rho_{10}} ((w_2 \rhd_{\rho_8} g_2)*_{\rho_9})) \unrhd_{\rho_{11}} i_2) \unrhd_{\rho_{12}} ((x_2 \unrhd_{\rho_5} s_2) ?_{\rho_7} (d_2 \unrhd_{\rho_6} f_2))}$ which is the total sequential computon shown in Figure \ref{fig:application-AWS-composites-right}(c).

Once the left and right paths of the intended step function have been constructed, we compose them into the branching composite $((c_1 \unrhd_{\rho_3} ((w_1 \rhd_{\rho_1} g_1)*_{\rho_2})) \unrhd_{\rho_4} s_1) ?_{\rho_{13}} (((c_2 \unrhd_{\rho_{10}} ((w_2 \rhd_{\rho_8} g_2)*_{\rho_9})) \unrhd_{\rho_{11}} i_2) \unrhd_{\rho_{12}} ((x_2 \unrhd_{\rho_5} s_2) ?_{\rho_7} (d_2 \unrhd_{\rho_6} f_2)))$ which checks whether a new stack needs to be created (via the left path) or whether a change set needs to be created and inspected (via the right path). This branching computon is ultimately composed with $k$ (i.e., a functional computon that determines whether a stack exists or not) into the total sequential computon $k \unrhd_{\rho_{14}} ((c_1 \unrhd_{\rho_3} ((w_1 \rhd_{\rho_1} g_1)*_{\rho_2})) \unrhd_{\rho_4} s_1) ?_{\rho_{13}} (((c_2 \unrhd_{\rho_{10}} ((w_2 \rhd_{\rho_8} g_2)*_{\rho_9})) \unrhd_{\rho_{11}} i_2) \unrhd_{\rho_{12}} ((x_2 \unrhd_{\rho_5} s_2) ?_{\rho_7} (d_2 \unrhd_{\rho_6} f_2)))$ which captures the whole behaviour of the step function for infrastructure deployment shown in Figure \ref{fig:application-AWS-concept}. The structure of such a complex sequential composite is depicted in Figure \ref{fig:application-AWS-workflow} and its behaviour is shown in Figure \ref{fig:application-AWS-behaviour}. 

\begin{figure}[!h]
\centering
{
\begin{tikzpicture}			
	\computonComposite{0.2}{0.1}{6.5}{6};

	\computonPrimitive{0.8}{2.2}{0.8}{2.3}{$k$};
	\qin{k1q1}{0}{4.2}{};\flow{{1.6,4.2}}{{2.6,4.2}}{dashed}{};
  \din{k1i1}{0}{3.9}{$1$};\flow{{1.6,3.9}}{{2.6,3.9}}{}{};
  \din{k1i2}{0}{3.6}{$2$};\flow{{1.6,3.6}}{{2.6,3.6}}{}{};
  \din{k1i3}{0}{3.3}{$3$};\flow{{1.6,3.3}}{{2.6,3.3}}{}{};
  \din{k1i4}{0}{3}{$4$};\flow{{1.6,3}}{{2.6,3}}{}{};
  \din{k1i5}{0}{2.7}{$5$};\flow{{1.6,2.7}}{{2.6,2.7}}{}{};
  \flow{{1.6,2.4}}{{2.6,2.4}}{}{};
	
	\begin{scope}[xshift=2.6cm]
	\computonComposite{0.4}{0.2}{3.6}{5.8};
	
	\qmatch{i0}{0}{4.2}{};\flowdiag{i0}{{1,5.6}}{dashed}{}{pos=0.5,rotate=45};\flowdiag{i0}{{1,2.5}}{dashed}{}{pos=0.5,rotate=311};
	\dmatch{i1}{0}{3.9}{$1$}{left};\flowdiag{i1}{{1,5.3}}{}{}{pos=0.5,rotate=45};\flowdiag{i1}{{1,2.2}}{}{}{pos=0.5,rotate=311};
	\dmatch{i2}{0}{3.6}{$2$}{left};\flowdiag{i2}{{1,5}}{}{}{pos=0.5,rotate=45};\flowdiag{i2}{{1,1.9}}{}{}{pos=0.5,rotate=311};
	\dmatch{i3}{0}{3.3}{$3$}{left};\flowdiag{i3}{{1,4.7}}{}{}{pos=0.5,rotate=45};\flowdiag{i3}{{1,1.6}}{}{}{pos=0.5,rotate=311};
	\dmatch{i4}{0}{3}{$4$}{left};\flowdiag{i4}{{1,4.4}}{}{}{pos=0.5,rotate=45};\flowdiag{i4}{{1,1.3}}{}{}{pos=0.5,rotate=311};
	\dmatch{i5}{0}{2.7}{$5$}{left};\flowdiag{i5}{{1,4.1}}{}{}{pos=0.5,rotate=45};\flowdiag{i5}{{1,1}}{}{}{pos=0.5,rotate=311};
	\dmatch{i6}{0}{2.4}{$6$}{left};\flowdiag{i6}{{1,3.8}}{}{}{pos=0.5,rotate=45};\flowdiag{i6}{{1,0.7}}{}{}{pos=0.5,rotate=311};
	
	\computonComposite{1}{3.4}{2.4}{2.4};\node at (2.2,4.6){Figure \ref{fig:application-AWS-composites-right}(c)}; 	
	\computonComposite{1}{0.4}{2.4}{2.4};\node at (2.2,1.6){Figure \ref{fig:application-AWS-composites-left}(c)}; 	
	
	\qoutplain{o0}{3.6}{3.6}{};\flowdiag{{3.4,5}}{o0}{dashed}{}{pos=0.5,rotate=45};\flowdiag{{3.4,1.9}}{o0}{dashed}{}{pos=0.5,rotate=311};	
	\doutplain{o1}{3.6}{3.3}{$8$};\flowdiag{{3.4,4.7}}{o1}{}{}{pos=0.5,rotate=45};\flowdiag{{3.4,1.6}}{o1}{}{}{pos=0.5,rotate=311};	
	\end{scope}
\end{tikzpicture}
}
\caption{Total sequential computon corresponding to the step function shown in Figure \ref{fig:application-AWS-concept}: \resizebox{0.999\textwidth}{!}{${k \unrhd_{\rho_{14}} ((c_1 \unrhd_{\rho_3} ((w_1 \rhd_{\rho_1} g_1)*_{\rho_2})) \unrhd_{\rho_4} s_1) ?_{\rho_{13}} (((c_2 \unrhd_{\rho_{10}} ((w_2 \rhd_{\rho_8} g_2)*_{\rho_9})) \unrhd_{\rho_{11}} i_2) \unrhd_{\rho_{12}} ((x_2 \unrhd_{\rho_5} s_2) ?_{\rho_7} (d_2 \unrhd_{\rho_6} f_2)))}$}}
\label{fig:application-AWS-workflow}
\end{figure}

\begin{figure}[!h]
\centering
{
\begin{tikzpicture}
\node[place,label={180:},minimum size=3mm] (ik) at (0,1) {};
\node[transition,fill=black,minimum width=0.1mm,minimum height=5mm] (k) at (0.5,1) {};
\node[place,label={180:},minimum size=3mm] (ok) at (1,1) {};

\node[transition,fill=black,minimum width=0.1mm,minimum height=5mm] (c1) at (1.5,1.5) {};
\node[place,label={180:},minimum size=3mm] (oc1) at (2,1.5) {};
\node[transition,fill=black,minimum width=0.1mm,minimum height=5mm] (e1) at (2.5,1.5) {};
\node[place,label={180:},minimum size=3mm] (oe1) at (3,1.5) {};
\node[transition,fill=black,minimum width=0.1mm,minimum height=5mm] (w1) at (3.5,1.5) {};
\node[place,label={180:},minimum size=3mm] (ow1) at (4,1.5) {};
\node[transition,fill=black,minimum width=0.1mm,minimum height=5mm] (g1) at (4.5,1.5) {};
\node[place,label={180:},minimum size=3mm] (og1) at (5,1.5) {};
\node[transition,fill=black,minimum width=0.1mm,minimum height=5mm] (e2) at (4,2) {};
\node[transition,fill=black,minimum width=0.1mm,minimum height=5mm] (e3) at (5.5,1.5) {};
\node[place,label={180:},minimum size=3mm] (oe3) at (6,1.5) {};
\node[transition,fill=black,minimum width=0.1mm,minimum height=5mm] (s1) at (6.5,1.5) {};

\node[transition,fill=black,minimum width=0.1mm,minimum height=5mm] (c2) at (1.5,0.5) {};
\node[place,label={180:},minimum size=3mm] (oc2) at (2,0.5) {};
\node[transition,fill=black,minimum width=0.1mm,minimum height=5mm] (e4) at (2.5,0.5) {};
\node[place,label={180:},minimum size=3mm] (oe4) at (3,0.5) {};
\node[transition,fill=black,minimum width=0.1mm,minimum height=5mm] (w2) at (3.5,0.5) {};
\node[place,label={180:},minimum size=3mm] (ow2) at (4,0.5) {};
\node[transition,fill=black,minimum width=0.1mm,minimum height=5mm] (g2) at (4.5,0.5) {};
\node[place,label={180:},minimum size=3mm] (og2) at (5,0.5) {};
\node[transition,fill=black,minimum width=0.1mm,minimum height=5mm] (e5) at (4,0) {};
\node[transition,fill=black,minimum width=0.1mm,minimum height=5mm] (e6) at (5.5,0.5) {};
\node[place,label={180:},minimum size=3mm] (oe6) at (6,0.5) {};
\node[transition,fill=black,minimum width=0.1mm,minimum height=5mm] (i2) at (6.5,0.5) {};
\node[place,label={180:},minimum size=3mm] (oi2) at (7,0.5) {};

\node[transition,fill=black,minimum width=0.1mm,minimum height=5mm] (x2) at (7.5,1) {};
\node[place,label={180:},minimum size=3mm] (ox2) at (8,1) {};
\node[transition,fill=black,minimum width=0.1mm,minimum height=5mm] (s2) at (8.5,1) {};
\node[transition,fill=black,minimum width=0.1mm,minimum height=5mm] (d2) at (7.5,0) {};
\node[place,label={180:},minimum size=3mm] (od2) at (8,0) {};
\node[transition,fill=black,minimum width=0.1mm,minimum height=5mm] (f2) at (8.5,0) {};

\node[place,label={180:},minimum size=3mm] (o) at (9,0.5) {};

\draw[-latex,thick](ik)--(k);\draw[-latex,thick](k)--(ok);\draw[-latex,thick](ok)--(c1);\draw[-latex,thick](ok)--(c2);

\draw[-latex,thick](c1)--(oc1);\draw[-latex,thick](oc1)--(e1);\draw[-latex,thick](e1)--(oe1);\draw[-latex,thick](oe1)--(w1);\draw[-latex,thick](w1)--(ow1);\draw[-latex,thick](ow1)--(g1);\draw[-latex,thick](g1)--(og1);\draw[-latex,thick](og1) to[bend right=30] (e2);\draw[-latex,thick](e2) to[bend right=30] (oe1);\draw[-latex,thick](og1)--(e3);\draw[-latex,thick](e3)--(oe3);\draw[-latex,thick](oe3)--(s1);

\draw[-latex,thick](c2)--(oc2);\draw[-latex,thick](oc2)--(e4);\draw[-latex,thick](e4)--(oe4);\draw[-latex,thick](oe4)--(w2);\draw[-latex,thick](w2)--(ow2);\draw[-latex,thick](ow2)--(g2);\draw[-latex,thick](g2)--(og2);\draw[-latex,thick](og2) to[bend left=30] (e5);\draw[-latex,thick](e5) to[bend left=30] (oe4);\draw[-latex,thick](og2)--(e6);\draw[-latex,thick](e6)--(oe6);\draw[-latex,thick](oe6)--(i2);\draw[-latex,thick](i2)--(oi2);\draw[-latex,thick](oi2)--(x2);\draw[-latex,thick](x2)--(ox2);\draw[-latex,thick](ox2)--(s2);\draw[-latex,thick](oi2)--(d2);\draw[-latex,thick](d2)--(od2);\draw[-latex,thick](od2)--(f2);

\draw[-latex,thick](s1) to[bend left=60] (o);\draw[-latex,thick](s2)--(o);\draw[-latex,thick](f2)--(o);
\end{tikzpicture}
}
\caption{Behaviour of the step function shown in Figure \ref{fig:application-AWS-concept}, expressed as the Petri net \resizebox{0.999\textwidth}{!}{$\mathcal{C}(\mathfrak{E}({k \unrhd_{\rho_{14}} ((c_1 \unrhd_{\rho_3} ((w_1 \rhd_{\rho_1} g_1)*_{\rho_2})) \unrhd_{\rho_4} s_1) ?_{\rho_{13}} (((c_2 \unrhd_{\rho_{10}} ((w_2 \rhd_{\rho_8} g_2)*_{\rho_9})) \unrhd_{\rho_{11}} i_2) \unrhd_{\rho_{12}} ((x_2 \unrhd_{\rho_5} s_2) ?_{\rho_7} (d_2 \unrhd_{\rho_6} f_2)))}))$}}
\label{fig:application-AWS-behaviour}
\end{figure}

Rather than presenting behaviour as a net under ${\mathcal{N}}$, we decide to use the functor ${\mathcal{C}\circ \mathfrak{E}}$ since control flow captures system behaviour comprehensively. Displaying the corresponding nets under ${\mathcal{N}}$ or ${\mathcal{D}}$ can be easily done using the mapping provided in \ref{sec:appendix-mapping}, which corresponds to the functorial descriptions from Definition \ref{def:functor-computon-to-petri} and Proposition \ref{prop:functor-data-petri}. Here, we show the equivalent BPMN diagram (for control flow) and the corresponding DFG in standard notation (for data flow), using the graph transformation system described in Section \ref{sec:transformation-system}. These models are far more expressive than Petri nets for our purpose which is just demonstrating the separation of control and data for model transformation.

By Proposition \ref{prop:computon-sequential-connected}, the total sequential computon from Figure \ref{fig:application-AWS-workflow} is connected because there is information flow from every non-e-outport to either the unique ec-outport or the unique ed-outport. A glance at this figure reveals that computons are modular by construction, a consequence of compositionality that allows hiding the internals of complex composite structures. For instance, Figure \ref{fig:application-AWS-workflow} hides the structure of the total sequential computon from Figure \ref{fig:application-AWS-composites-right}(c) which, in turn, hides the complexity of the composites \ref{fig:application-AWS-composites-right}(a) and \ref{fig:application-AWS-composites-right}(b). As per Proposition \ref{prop:computon-sequential-connected} and Corollaries \ref{cor:computon-branching-connected} and \ref{cor:computon-iterative-tail-connected}, all the computons we deal with in this example are connected, including primitives (such as $w_1$) and composites (such as ${w_1 \rhd_{\rho_1} g_1}$).

Apart from modularity, another semantic consequence of our model is the separation of data and control which can be leveraged to analyse these two dimensions independently. For instance, we use the functor from Definition \ref{def:computon-cfg} to extract the CFG of the computon from Figure \ref{fig:application-AWS-workflow} which is then converted into its corresponding BPMN diagram via the graph transformation system proposed in Section \ref{sec:appendix-controlflow} which, in turn, realises transformation without considering data flow at all (see Figure \ref{fig:application-AWS-separation}(a)). For data flow, we extract the DFG of the computon from Figure \ref{fig:application-AWS-workflow} through the functor from Definition \ref{def:computon-dfg}, which is then converted into its equivalent DFG in standard notation via the functor described in Section \ref{sec:appendix-dataflow} which do not consider control flow at all (see Figure \ref{fig:application-AWS-separation}(b)). Although the separation of concerns can be leveraged in other ways (e.g., to formally verify reachability of control flow only), the purpose of this section is just to demonstrate how control flow and data flow can be analysed independently for model transformation.

\begin{figure}[!h]
\centering
\subcaptionbox{BPMN Diagram (control flow).}
{
\includegraphics[scale=0.488]{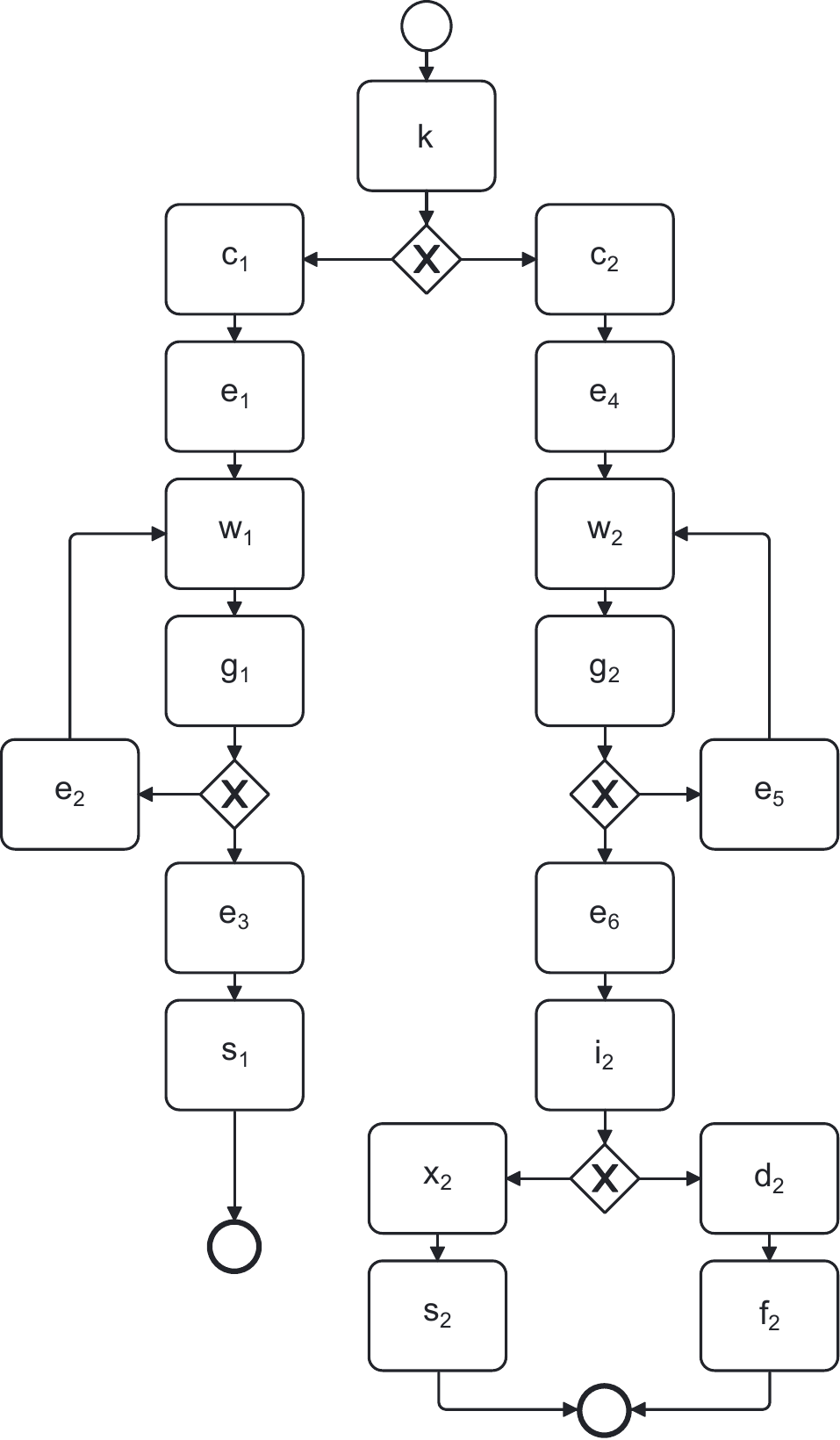}
}
\subcaptionbox{DFG in standard notation (data flow). For readability, we use a dash symbol between numbers $n$ and $m$ to indicate the presence of data flows $\xrightarrow{n}$, $\xrightarrow{n+1}$, $\ldots$, $\xrightarrow{m}$ such that $n \geq 1$ and $m > n$.}
{
\begin{tikzpicture}[scale=0.982]

\node[circle,draw](s1) at (2,3){$s_1$};
\node[circle,draw](e3) at (2,4.5){$e_3$};
\node[circle,draw](e2) at (0,7.5){$e_2$};
\node[circle,draw](g1) at (2,7.5){$g_1$};
\node[circle,draw](w1) at (0,9.5){$w_1$};
\node[circle,draw](e1) at (2,10.5){$e_1$};
\node[circle,draw](c1) at (2,12){$c_1$};
\node (output1) at (1,3){$8$};

\node[circle,draw](s2) at (3,0){$s_2$};
\node[circle,draw](f2) at (5,0){$f_2$};
\node[circle,draw](x2) at (3,1.5){$x_2$};
\node[circle,draw](d2) at (5,1.5){$d_2$};
\node[circle,draw](i2) at (4,3){$i_2$};
\node[circle,draw](e6) at (4,4.5){$e_6$};
\node[circle,draw](e5) at (6,7.5){$e_5$};
\node[circle,draw](g2) at (4,7.5){$g_2$};
\node[circle,draw](w2) at (6,9.5){$w_2$};
\node[circle,draw](e4) at (4,10.5){$e_4$};
\node[circle,draw](c2) at (4,12){$c_2$};
\node (output2) at (2,0){$8$};
\node (output3) at (6,0){$8$};

\node[circle,draw](k) at (3,13.5){$k$};
\node (input1) at (1,13.5){$1-5$};

\draw[-latex] (input1) -- (k);
\draw[-latex] (k) -- node[left]{$1-6$} (c1);
\draw[-latex] (c1) -- node[left]{$1-5$} (e1);
\draw[-latex] (e1) -- node[right,yshift=0.5cm]{$1-5$} (g1);
\draw[-latex] (g1) to[bend right=25] node[above]{$1-5$} (e2);\draw[-latex] (g1) to[bend left=25] node[above]{$7$} (e2);\draw[-latex] (e2) to[bend right=60] node[below]{$1-5$} (g1);
\draw[-latex] (g1) -- node[right,yshift=0.5cm]{$1-5$} (e3);\draw[-latex] (g1) to[bend right=25] node[left]{$7$} (e3);
\draw[-latex] (s1) -- (output1);
\draw[-latex] (k) -- node[right]{$1-6$} (c2);
\draw[-latex] (c2) -- node[left]{$1-5$} (e4);\draw[-latex] (c2) to[bend left=25] node[right]{$9$} (e4);
\draw[-latex] (e4) -- node[left,yshift=-0.5cm]{$1-5$} (g2);\draw[-latex] (e4) to[bend left=25] node[right,yshift=0.5cm]{$9$} (g2);
\draw[-latex] (g2) to[bend right=25] node[left,yshift=-0.5cm]{$1-5$} (e6);\draw[-latex] (g2) -- node[left]{$9$} (e6);\draw[-latex] (g2) to[bend left=25] node[right,yshift=-0.5cm]{$10$} (e6);
\draw[-latex] (g2) to[bend right=25] node[above]{$9$} (e5);\draw[-latex] (g2) to[bend left=25] node[above]{$1-5$} (e5);\draw[-latex] (g2) to[bend right=60] node[above]{$10$} (e5);
\draw[-latex] (e5) to[bend left=75] node[below]{$1-5$} (g2);\draw[-latex] (e5) to[bend right=75] node[above]{$9$} (g2);
\draw[-latex] (e6) -- node[left]{$1-5$} (i2);\draw[-latex] (e6) to[bend left=25] node[right]{$9$} (i2);
\draw[-latex] (i2) to[bend left=25] node[left]{$11$} (x2);\draw[-latex] (i2) to[bend right=25] node[left]{$9$} (x2);\draw[-latex] (i2) to[bend right=60] node[left,yshift=-0.5cm,xshift=-0.2cm]{$1-5$} (x2);
\draw[-latex] (i2) to[bend left=25] node[right]{$1-5$} (d2);\draw[-latex] (i2) to[bend right=25] node[right]{$9$} (d2);\draw[-latex] (i2) to[bend right=60] node[above,yshift=-0.2cm,xshift=0.2cm]{$11$} (d2);
\draw[-latex] (s2) -- (output2);
\draw[-latex] (f2) -- (output3);
\end{tikzpicture}
}
\caption{BPMN diagram and DFG in standard notation to respectively express the control flow and data flow structures of the composite shown in Figure \ref{fig:application-AWS-workflow}.}
\label{fig:application-AWS-separation}
\end{figure}

In the actual implementation of the step function from Figure \ref{fig:application-AWS-concept}, data is appended at every step of the workflow execution \cite{mendonca_using_2017}. This issue is derived from the fact that, apart from being non-compositional, AWS step functions do not separate control flow and data flow, so it is neccesary to pass a bundle of parameters as a single data item (i.e., as a JSON object). In other words, data flow is implicitly defined in the explicit workflow control flow. 

Enabling separation of concerns through our model allows us to remove the data modelling issue of AWS step functions so as to pass only relevant data among computation units. For example, the functional computon $w_1$, which is just in charge of delaying computation, can be implemented as a sleep function to wait for a fixed amount of time, without processing any data at all. Also, the change set creation status can only be used to terminate the loop of the right path, without the need of passing it onto subsequent computations. Certainly, it is still possible to pass all data parameters at every step of the computation in the form of a single port colour (e.g., the JSON object colour), just as in the actual implementation. However, doing this could not be as expressive as the way we model data flow in our scenario. 

\subsection{Case Study 2: Compositional LSTMs} \label{sec:case2}

A Long Short-Term Memory (LSTM) \cite{hochreiter_long_1997} is a Recurrent Neural Network which has been widely used in the field of Deep Learning to learn long-term data dependencies. Typical applications of it include handwriting recognition, automatic language translation and writing generation. The key idea of a LSTM is to use a container to store and process information for an extended period of time, in order to allow for constant error flow during training. Such a container, known as a memory cell, is controlled by three types of interacting gates: an input gate, a forget gate and an output gate, which respectively decide what information can be added to, removed from and sent out of the cell. Particularly, the outcome of the forget and input gates is added so as to produce an updated cell state that can be further consumed by other memory cells. The abstract, high-level schematic representation of an individual memory cell is shown in Figure \ref{fig:application-LSTM-concept}(a). Such a scheme is abstract because the specific behaviour of the components involved is not provided and it is high-level because each component can contain further internal components which are not exposed to the outside world. 
\vspace{-2pt}
\begin{figure}[!h]
  \centering
  \subfloat[Abstract, high-level architecture.]{
\begin{tikzpicture}	
  \draw[dashed] (-0.6,0.7) rectangle ++(5.3,4.6);\node(w) at (3.9,5.1){\scriptsize \emph{Cell State}};
	\draw[dashed] (2.5,0.9) rectangle ++(2,2.5);\node[align=left](x) at (4.1,3.2){\scriptsize\emph{Input}};\node[align=left](xx) at (4.1,2.9){\scriptsize\emph{Gate}};
	\draw[dashed] (0.4,0.9) rectangle ++(1.6,4);\node(y) at (1.5,3.6){\scriptsize\emph{Forget}};\node(yy) at (1.5,3.3){\scriptsize\emph{Gate}};
	\draw[dashed] (5.1,0.7) rectangle ++(2.2,3.5);\node(z) at (6.8,4){\scriptsize\emph{Output}};\node(zz) at (6.8,3.7){\scriptsize\emph{Gate}};		

	\node[draw](m1) at (1,4.5){$m_1$};
	\node[draw](m2) at (3.5,2.5){$m_2$};	
	\node[draw](m3) at (6,2.5){$m_3$};
	\node[draw](s1) at (1,1.5){$s_1$};  
	\node[draw](s2) at (3,1.5){$s_2$};	
	\node[draw](s3) at (6,1.5){$s_3$};
	\node[draw](t1) at (4,1.5){$t_1$};
	\node[draw](t2) at (6,3.5){$t_2$};
	\node[draw](a1) at (3.5,4.5){$a_1$};	
	
	\draw[-to] (-1.1,0.5) -- node[below] {\scriptsize ${h_{t-1}}$}(1,0.5);
	\draw[-to] (1,-0.2) -- node[right] {\scriptsize $x_t$}(1,0.5);
	\draw[-to] (-1.1,4.5) -- node[above] {\scriptsize ${c_{t-1}}$}(m1);
	\draw[-to] (1,0.5) -- (s1);
	\draw[-to] (3,0.5) -- (s2);
	\draw[-to] (4,0.5) -- (t1);
	\draw (1,0.5) -- (6,0.5);\draw[-to] (6,0.5) -- (s3);
	\draw (s2) -- node[left] {\scriptsize $i_t$} (3,2.5);\draw[-to] (3,2.5) -- (m2);
	\draw (t1) -- node[right] {\scriptsize $z_t$} (4,2.5);\draw[-to] (4,2.5) -- (m2);
	\draw[-to] (m2) -- node[left,yshift=0.2cm] {\scriptsize $g_t$} (a1);\draw[-to] (m1) -- node[above] {\scriptsize $j_t$} (a1);
	\draw[-to] (s3) -- node[right] {\scriptsize $o_t$} (m3);\draw[-to] (t2) -- node[right] {\scriptsize $k_t$} (m3);
	\draw[-to] (a1) -- node[above] {\scriptsize $c_t$} (7.6,4.5);\draw[-to] (6,4.5) -- (t2);
	\draw[-to] (m3) -- node[above] {\scriptsize $h_t$} (7.6,2.5);
	
	\draw[-to] (s1) -- node[left] {\scriptsize $f_t$}(m1);
\end{tikzpicture}  
  }%
  \qquad
  \subfloat[Mapping from colours to variables.]{%
  \resizebox{3.3cm}{!}{
    \begin{tabular}{c|c}
      \hline
      Colour & Variable \\ \hline
      $1$ & $h_{t-1}$ \\
      $2$ & $x_t$ \\
      $3$ & $f_t$ \\
      $4$ & $i_t$ \\
      $5$ & $o_t$ \\
      $6$ & $z_t$ \\
      $7$ & $c_t$ \\
      $8$ & $k_t$ \\
      $9$ & $c_{t-1}$ \\
      $10$ & $j_t$ \\
      $11$ & $g_t$ \\
      $12$ & $h_t$     
    \end{tabular}}   
    \label{subtbl:the-table}
  }
  \caption{Conceptual representation of a memory cell.}
  \label{fig:application-LSTM-concept}
  \vspace{-10pt}
\end{figure}

A memory cell can be constructed compositionally using the computon model. Such a construction is done in a bottom-up manner, starting from the most elementary units of computation of the memory cell, which are displayed as non-dashed squares in Figure \ref{fig:application-LSTM-concept}(a). As we are dealing with a model of high-level computation, the specific computation details of such squares are not required, so non-dashed squares can naturally be defined as functional computons, i.e., as black boxes that can perform any kind of computation (see Figure \ref{fig:application-LSTM-functionals}). Typically, for a concrete memory cell, the functional computons $s_i$, $t_j$, $m_i$ and $a_1$ would be required to compute sigmoid, hyperbolic tangent, Hadamard and element-wise summation functions, respectively. For the sake of ``high-levelness'', we just say $s_i$ and $t_j$ are activation computons, while $m_i$ and $a_1$ are multiplication and summation computons, respectively. 

\vspace{-1pt}

\begin{figure*}[!h]
\centering
{
\begin{tikzpicture}
  \computonPrimitive{0.8}{0}{0.8}{1}{$s_1$}
  \qin{s1q0}{0}{0.8}{}
  \din{s1i1}{0}{0.5}{$1$};
  \din{s1i2}{0}{0.2}{$2$};
  \qout{s1q1}{1.6}{0.8}{}
  \dout{s1o1}{1.6}{0.5}{$3$};
  
  \computonPrimitive{4.2}{0}{0.8}{1}{$s_2$}
  \qin{s2q0}{3.4}{0.8}{}
  \din{s2i1}{3.4}{0.5}{$1$};
  \din{s2i2}{3.4}{0.2}{$2$};
  \qout{s2q1}{5}{0.8}{}
  \dout{s2o1}{5}{0.5}{$4$};
  
  \computonPrimitive{7.6}{0}{0.8}{1}{$s_3$}
  \qin{s3q0}{6.8}{0.8}{}
  \din{s3i1}{6.8}{0.5}{$1$};
  \din{s3i2}{6.8}{0.2}{$2$};
  \qout{s3q1}{8.4}{0.8}{}
  \dout{s3o1}{8.4}{0.5}{$5$};
\end{tikzpicture}
}
\subcaptionbox{Activation computons.}
{
\begin{tikzpicture}[scale=0.9]
  \computonPrimitive{0.8}{0}{0.8}{1}{$t_1$}
  \qin{t1q0}{0}{0.8}{}
  \din{t1i1}{0}{0.5}{$1$};
  \din{t1i2}{0}{0.2}{$2$};
  \qout{t1q1}{1.6}{0.8}{}
  \dout{t1o1}{1.6}{0.5}{$6$};
  
  \computonPrimitive{4.2}{0}{0.8}{1}{$t_2$}
  \qin{t2q0}{3.4}{0.8}{}
  \din{t2i1}{3.4}{0.5}{$7$};
  \qout{t2q1}{5}{0.8}{}
  \dout{t2o1}{5}{0.5}{$8$};
\end{tikzpicture}
}
\subcaptionbox{Multiplication computons.}
{
\begin{tikzpicture}[scale=0.9]
  \computonPrimitive{0.6}{0}{0.8}{1}{$m_1$}
  \qin{m1q0}{-0.2}{0.8}{}
  \din{m1i1}{-0.2}{0.5}{$3$};
  \din{m1i2}{-0.2}{0.2}{$9$};
  \qout{m1q1}{1.4}{0.8}{}
  \dout{m1o1}{1.4}{0.5}{$10$};
  
  \computonPrimitive{4.2}{0}{0.8}{1}{$m_2$}
  \qin{m2q0}{3.4}{0.8}{}
  \din{m2i1}{3.4}{0.5}{$4$};
  \din{m2i1}{3.4}{0.2}{$6$};
  \qout{m2q1}{5}{0.8}{}
  \dout{m2o1}{5}{0.5}{$11$};
  
  \computonPrimitive{7.8}{0}{0.8}{1}{$m_3$}
  \qin{m3q0}{7}{0.8}{}
  \din{m3i1}{7}{0.5}{$5$};
  \din{m3i2}{7}{0.2}{$8$};
  \qout{m3q1}{8.6}{0.8}{}
  \dout{m3o1}{8.6}{0.5}{$12$};
\end{tikzpicture}
}
\subcaptionbox{Summation computon.}
{
\begin{tikzpicture}[scale=0.9]
  \computonPrimitive{0.8}{0}{0.8}{1}{$a_1$}
  \qin{p1q0}{0}{0.8}{}
  \din{p1i1}{0}{0.5}{$10$};
  \din{p1i2}{0}{0.2}{$11$};
  \qout{p1q1}{1.6}{0.8}{}
  \dout{p1o1}{1.6}{0.5}{$7$};
  \dout{p1o1}{1.6}{0.2}{$7$};
\end{tikzpicture}
}
\caption{Functional computons for the compositional construction of a memory cell.}
\label{fig:application-LSTM-functionals}
\vspace{-5pt}
\end{figure*}

Strictly speaking, every ed-port shown in Figure \ref{fig:application-LSTM-functionals} must have the same colour (i.e., the vector colour) since data being moved within a cell corresponds to a vector of the same type. But, for clarity and demonstration purposes, we consider 12 different colours, each corresponding to the type of each variable shown in Figure \ref{fig:application-LSTM-concept}(a). The mapping from colours to variables is presented in Figure \ref{fig:application-LSTM-concept}(b). 

Using a long-term memory vector $c_{t-1}$ (of colour $9$), a short-term memory vector $h_{t-1}$ (of colour $1$) and an external input/predictor vector $x_t$ (of colour $2$), the whole memory cell performs some computation and then returns a new long-term memory vector $c_{t}$ (of colour $7$) and a new short-term memory vector $h_{t}$ (of colour $12$) which, in turn, can used by other memory cells to perform subsequent computations. 

For the (high-level) computation of a memory cell, the inner forget gate computes the activation computon $s_1$ in terms of $h_{t-1}$ and $x_t$ to yield the vector $f_t$ (of colour $3$). The result $f_t$ and the long-term memory vector $c_{t-1}$ are then multipled via $m_1$ to obtain $j_t$ (of colour $10$). To capture this computation, Figure \ref{fig:application-LSTM-composites}(a) shows that the forget gate is characterised as the partial sequential computon $s_1 \rhd_{\rho_1} m_1$. 

Constructing the input gate is done differently since it requires two copies of $h_{t-1}$ and two copies of $x_t$ to simultaneously compute the activation computons $s_2$ and $t_1$, in order to produce the vector $i_t$ (of colour $4$) and the vector $z_t$ (of colour $6$). Such results are then multiplied through $m_2$ to produce the vector $g_t$ (of colour $11$). As $s_2$ and $t_1$ are computed in parallel before $m_2$, the input gate is naturally characterised as the total sequential computon $(s_2 \mid_{\rho_2} t_1) \unrhd_{\rho_3} m_2$ whose structure is depicted in Figure \ref{fig:application-LSTM-composites}(b). Figure \ref{fig:application-LSTM-composites}(c) shows that the structure of the output gate is similar to that of $(s_2 \mid_{\rho_2} t_1) \unrhd_{\rho_3} m_2$ so the output gate is precisely the total sequential computon $(s_3 \mid_{\rho_4} t_2) \unrhd_{\rho_5} m_3$ whose left and right operands are the p-sync computon $s_3 \mid_{\rho_4} t_2$ and the multiplication computon $m_3$, respectively. As data is not shared within a p-sync computon, $s_3 \mid_{\rho_4} t_2$ has the same ed-outports as $s_3$ and $t_2$, namely the vector $o_t$ (of colour $5$) and the vector $k_t$ (of colour $8$). The only ed-outport of $(s_3 \mid_{\rho_4} t_2) \unrhd_{\rho_5} m_3$ is $h_{t}$ (of colour $12$) which represents the new short-term memory value to be passed onto subsequent memory cells.

As they have no dependencies among them, the forget gate composite ${s_1 \rhd_{\rho_1} m_1}$ and the input gate composite ${(s_2 \mid_{\rho_2} t_1) \unrhd_{\rho_3} m_2}$ are composed into the p-sync computon $(s_1 \rhd_{\rho_1} m_1) \mid_{\rho_6} ((s_2 \mid_{\rho_2} t_1) \unrhd_{\rho_3} m_2)$ whose ed-ports are inherited from ${s_1 \rhd_{\rho_1} m_1}$ and ${(s_2 \mid_{\rho_2} t_1) \unrhd_{\rho_3} m_2}$. To enable the functional computon $a_1$ to element-wisely add $j_t$ and $g_t$, the newly constructed p-sync and $a_1$ are composed into the total sequential computon $((s_1 \rhd_{\rho_1} m_1) \mid_{\rho_6} ((s_2 \mid_{\rho_2} t_1) \unrhd_{\rho_3} m_2)) \unrhd_{\rho_7} a_1$ which, by Proposition \ref{prop:computon-sequential-inports-outports-1}, has the same ed-outports as $a_1$, i.e., the next cell state $c_t$ (of colour $7$) -- see Figure \ref{fig:application-LSTM-composites}(d). As $c_t$ is sent outside the memory cell and is further required by the output gate, Figure \ref{fig:application-LSTM-functionals}(c) shows that $a_1$ produces two copies of it. \footnote{\samepage We decide to model the summation computon $a_1$ in this way in order to avoid the unnecessary burden of introducing extra functional computons to express data replication.}

\begin{figure}[!h]
\centering
\subcaptionbox{Forget gate: ${s_1 \rhd_{\rho_1} m_1}$}
{
\begin{tikzpicture}
	\computonComposite{0.2}{0.1}{4.5}{1.3};

	\computonPrimitive{0.8}{0.2}{0.8}{1}{$s_1$};  
  \din{s1i1}{0}{1}{$1$};
  \din{s1i2}{0}{0.7}{$2$};
  \qin{s1q0}{0}{0.4}{};

  \dmatch{3d}{2.4}{0.7}{$3$}{above};\flow{{1.6,0.7}}{3d}{}{};\flow{3d}{{3.3,0.7}}{}{};
	\qmatch{3q}{2.4}{0.4}{};\flow{{1.6,0.4}}{3q}{dashed}{};\flow{3q}{{3.3,0.4}}{dashed}{};	
	
	\computonPrimitive{3.3}{0.2}{0.8}{1}{$m_1$};
  \dinplain{m1i1}{0}{1.3}{$9$};\flow{m1i1}{{3.3,1}}{}{bend left=7};  
  \dout{m1o1}{4.1}{0.7}{$10$};
  \qout{m1q1}{4.1}{0.4}{};
\end{tikzpicture}
}
\subcaptionbox{Input gate: ${(s_2 \mid_{\rho_2} t_1) \unrhd_{\rho_3} m_2}$}
{
\begin{tikzpicture}
\computonComposite{0.2}{-0.1}{8.1}{2.5};
	\computonComposite{0.6}{0}{4.8}{2.3};
  
  \dinplain{s2i2}{0}{1.9}{$1$}{left};\flow{s2i2}{$(s2i2)+(2.6,0)$}{}{};
  \dinplain{s2i1}{0}{1.6}{$2$}{left};\flow{s2i1}{$(s2i1)+(2.6,0)$}{}{};  
  \qinplain{f1}{0}{1.2}{};\flow{f1}{$(f1)+(0.8,0)$}{dashed}{};
  \dinplain{t1i1}{0}{0.8}{$1$}{left};\flow{t1i1}{$(t1i1)+(2.6,0)$}{}{};
  \dinplain{t1i2}{0}{0.5}{$2$}{left};\flow{t1i2}{$(t1i2)+(2.6,0)$}{}{};
	
	\forkplain{-0.25}{0.95};
	\qmatch{f2}{1.8}{1.3}{};\flow{$(f2)+(-0.8,0)$}{f2}{dashed}{};\flow{f2}{$(f2)+(0.8,0)$}{dashed}{};
	\qmatch{f3}{1.8}{1}{};\flow{$(f3)+(-0.8,0)$}{f3}{dashed}{};\flow{f3}{$(f3)+(0.8,0)$}{dashed}{};
	
	\computonPrimitive{2.6}{1.2}{0.8}{1}{$s_2$};
  \computonPrimitive{2.6}{0.1}{0.8}{1}{$t_1$}

	\joinplain{4}{0.95};	
	\qmatch{j2}{4.2}{1.3}{};\flow{$(j2)+(-0.8,0)$}{j2}{dashed}{};\flow{j2}{$(j2)+(0.8,0)$}{dashed}{};
	\qmatch{j3}{4.2}{1}{};\flow{$(j3)+(-0.8,0)$}{j3}{dashed}{};\flow{j3}{$(j3)+(0.8,0)$}{dashed}{};
	
	\dmatch{s2o1}{6.1}{1.75}{$4$}{above};\flow{$(s2o1)+(-2.7,0)$}{s2o1}{}{};
	\qmatch{j1}{6.1}{1.2}{};\flow{$(j1)+(-0.8,0)$}{j1}{dashed}{};
	\dmatch{t1o1}{6.1}{0.65}{$6$}{below};\flow{$(t1o1)+(-2.7,0)$}{t1o1}{}{};
	
	\computonPrimitive{6.9}{0.45}{0.8}{1.4}{$m_2$}
	\flow{s2o1}{$(s2o1)+(0.8,0)$}{}{};
	\flow{j1}{$(j1)+(0.8,0)$}{dashed}{};
	\flow{t1o1}{$(t1o1)+(0.8,0)$}{}{};  
  \qout{m2q1}{7.7}{1.45}{}
  \dout{m2o1}{7.7}{0.9}{$11$};  
\end{tikzpicture}
}
\subcaptionbox{Output gate: $(s_3 \mid_{\rho_4} t_2) \unrhd_{\rho_5} m_3$}
{
\begin{tikzpicture}
\computonComposite{0.2}{-0.1}{8.1}{2.5};
	\computonComposite{0.6}{0}{4.8}{2.3};
  
  \dinplain{s2i2}{0}{1.9}{$1$}{left};\flow{s2i2}{$(s2i2)+(2.6,0)$}{}{};
  \dinplain{s2i1}{0}{1.6}{$2$}{left};\flow{s2i1}{$(s2i1)+(2.6,0)$}{}{};  
  \qinplain{f1}{0}{1.2}{};\flow{f1}{$(f1)+(0.8,0)$}{dashed}{};
  \dinplain{t1i1}{0}{0.8}{$7$}{left};\flow{t1i1}{$(t1i1)+(2.6,0)$}{}{};
	
	\forkplain{-0.25}{0.95};
	\qmatch{f2}{1.8}{1.3}{};\flow{$(f2)+(-0.8,0)$}{f2}{dashed}{};\flow{f2}{$(f2)+(0.8,0)$}{dashed}{};
	\qmatch{f3}{1.8}{1}{};\flow{$(f3)+(-0.8,0)$}{f3}{dashed}{};\flow{f3}{$(f3)+(0.8,0)$}{dashed}{};
	
	\computonPrimitive{2.6}{1.2}{0.8}{1}{$s_3$};
  \computonPrimitive{2.6}{0.1}{0.8}{1}{$t_2$}

	\joinplain{4}{0.95};	
	\qmatch{j2}{4.2}{1.3}{};\flow{$(j2)+(-0.8,0)$}{j2}{dashed}{};\flow{j2}{$(j2)+(0.8,0)$}{dashed}{};
	\qmatch{j3}{4.2}{1}{};\flow{$(j3)+(-0.8,0)$}{j3}{dashed}{};\flow{j3}{$(j3)+(0.8,0)$}{dashed}{};
	
	\dmatch{s2o1}{6.1}{1.75}{$5$}{above};\flow{$(s2o1)+(-2.7,0)$}{s2o1}{}{};
	\qmatch{j1}{6.1}{1.2}{};\flow{$(j1)+(-0.8,0)$}{j1}{dashed}{};
	\dmatch{t1o1}{6.1}{0.65}{$8$}{below};\flow{$(t1o1)+(-2.7,0)$}{t1o1}{}{};
	
	\computonPrimitive{6.9}{0.45}{0.8}{1.4}{$m_3$}
	\flow{s2o1}{$(s2o1)+(0.8,0)$}{}{};
	\flow{j1}{$(j1)+(0.8,0)$}{dashed}{};
	\flow{t1o1}{$(t1o1)+(0.8,0)$}{}{};  
  \dout{m2o1}{7.7}{1.45}{$12$};  
  \qout{m2q1}{7.7}{0.9}{}
\end{tikzpicture}
}
\subcaptionbox{Cell state: $((s_1 \rhd_{\rho_1} m_1) \mid_{\rho_6} ((s_2 \mid_{\rho_2} t_1) \unrhd_{\rho_3} m_2)) \unrhd_{\rho_7} a_1$}
{
\begin{tikzpicture}
\computonComposite{0.2}{-0.2}{14}{4.7};
\computonComposite{0.6}{0}{10.9}{4.3};
\computonComposite{2}{2.8}{4.5}{1.35};
\computonComposite{2}{0.1}{8}{2.5};
\computonComposite{2.1}{0.2}{5}{2.3};

  \dinplain{x0}{0}{4}{$9$}{left};\flow{x0}{$(x0)+(5.2,-0.3)$}{}{bend left=7};  
  \dinplain{x1}{0}{3.7}{$1$}{left};\flow{x1}{$(x1)+(2.6,0)$}{}{};
  \dinplain{x2}{0}{3.4}{$2$}{left};\flow{x2}{$(x2)+(2.6,0)$}{}{};  
  
  \qin{x3}{0}{2.25}{};
  
  \dinplain{y0}{0}{1.9}{$1$}{left};\flow{y0}{$(y0)+(5,0)$}{}{};
  \dinplain{y1}{0}{1.6}{$2$}{left};\flow{y1}{$(y1)+(5,0)$}{}{};  
  \dinplain{y2}{0}{0.8}{$1$}{left};\flow{y2}{$(y2)+(5,0)$}{}{};
  \dinplain{y3}{0}{0.5}{$2$}{left};\flow{y3}{$(y3)+(5,0)$}{}{};
	
	\forkplain{-0.25}{2};	
	\flowdiag{{1,2.3}}{{1.8,3.1}}{dashed}{}{pos=0.5,rotate=27};
	\flowdiag{{1,2.1}}{{1.7,1.4}}{dashed}{}{pos=0.5,rotate=311};

\begin{scope}[xshift=1.8cm,yshift=2.7cm]	

	\computonPrimitive{0.8}{0.2}{0.8}{1}{$s_1$};  
	\qmatch{s1i1}{0}{0.4}{};\flow{s1i1}{$(s1i1)+(0.8,0)$}{dashed}{};

  \dmatch{3d}{2.4}{0.7}{$3$}{above};\flow{{1.6,0.7}}{3d}{}{};\flow{3d}{{3.3,0.7}}{}{};
	\qmatch{3q}{2.4}{0.4}{};\flow{{1.6,0.4}}{3q}{dashed}{};\flow{3q}{{3.3,0.4}}{dashed}{};	
	
	\computonPrimitive{3.3}{0.2}{0.8}{1}{$m_1$};
  \dmatch{m1o1}{10.4}{0.7}{$10$}{above};\flow{$(m1o1)+(-6.3,0)$}{m1o1}{}{};  
  \qmatch{m1q1}{8.5}{0.4}{};\flow{$(m1q1)+(-4.3,0)$}{m1q1}{dashed}{};
\end{scope}

\begin{scope}[xshift=1.8cm,yshift=0.2cm]	
	
  \qmatch{f1}{0}{1.2}{};\flow{f1}{$(f1)+(0.8,0)$}{}{dashed};
	
	\forkplain{-0.25}{0.95};
	\qmatch{f2}{1.8}{1.3}{};\flow{$(f2)+(-0.8,0)$}{f2}{dashed}{};\flow{f2}{$(f2)+(0.8,0)$}{dashed}{};
	\qmatch{f3}{1.8}{1}{};\flow{$(f3)+(-0.8,0)$}{f3}{dashed}{};\flow{f3}{$(f3)+(0.8,0)$}{dashed}{};
	
	\computonPrimitive{2.6}{1.2}{0.8}{1}{$s_2$};
  \computonPrimitive{2.6}{0.1}{0.8}{1}{$t_1$}

	\joinplain{4}{0.95};	
	\qmatch{j2}{4.2}{1.3}{};\flow{$(j2)+(-0.8,0)$}{j2}{dashed}{};\flow{j2}{$(j2)+(0.8,0)$}{dashed}{};
	\qmatch{j3}{4.2}{1}{};\flow{$(j3)+(-0.8,0)$}{j3}{dashed}{};\flow{j3}{$(j3)+(0.8,0)$}{dashed}{};
	
	\dmatch{s2o1}{6.1}{1.75}{$4$}{above};\flow{$(s2o1)+(-2.7,0)$}{s2o1}{}{};
	\qmatch{j1}{6.1}{1.2}{};\flow{$(j1)+(-0.8,0)$}{j1}{dashed}{};
	\dmatch{t1o1}{6.1}{0.65}{$6$}{below};\flow{$(t1o1)+(-2.7,0)$}{t1o1}{}{};
	
	\computonPrimitive{6.9}{0.45}{0.8}{1.4}{$m_2$}
	\flow{s2o1}{$(s2o1)+(0.8,0)$}{}{};
	\flow{j1}{$(j1)+(0.8,0)$}{dashed}{};
	\flow{t1o1}{$(t1o1)+(0.8,0)$}{}{};    
  \qmatch{m2q1}{8.5}{1.45}{};\flow{$(m2q1)+(-0.8,0)$}{m2q1}{dashed}{};
  \dmatch{m2o1}{10.4}{0.9}{$11$}{below};\flow{$(m2o1)+(-2.7,0)$}{m2o1}{}{};  
\end{scope}

	\joinplain{10.1}{2};	
	\flowdiag{m1q1}{$(m1q1)+(0.8,-0.7)$}{dashed}{}{pos=0.5,rotate=311};
	\flowdiag{m2q1}{$(m2q1)+(0.8,0.5)$}{dashed}{}{pos=0.5,rotate=27};
	\qmatch{jx}{12.2}{2.25}{};\flow{$(jx)+(-0.8,0)$}{jx}{dashed}{};	
	
	\flowdiag{m1o1}{$(m1o1)+(0.8,-0.8)$}{}{}{pos=0.5,rotate=311};
	\flow{jx}{$(jx)+(0.8,0)$}{dashed}{};
	\flowdiag{m2o1}{$(m2o1)+(0.8,0.7)$}{}{}{pos=0.5,rotate=27};

	\computonPrimitive{12.8}{1.5}{0.8}{1.4}{$a_1$}
  \qout{p1q1}{13.6}{2.6}{}
  \dout{p1o1}{13.6}{2.3}{$7$};
  \dout{p1o2}{13.6}{2}{$7$};
\end{tikzpicture}
}
\subcaptionbox{Memory cell: $(((s_1 \rhd_{\rho_1} m_1) \mid_{\rho_6} ((s_2 \mid_{\rho_2} t_1) \unrhd_{\rho_3} m_2)) \unrhd_{\rho_7} a_1) \unrhd_{\rho_{8}} ((s_3 \mid_{\rho_4} t_2) \unrhd_{\rho_5} m_3)$}
{
\begin{tikzpicture}
	\computonComposite{0.2}{0}{5.8}{3.3};

	\computonComposite{0.8}{0.2}{1.5}{2.4};\node[align=center] at (1.55,1.5){Cell\\State};  

	\dinplain{gi0}{0}{3.1}{$1$}{left};
	\dinplain{gi1}{0}{2.8}{$2$}{left};	
	\din{ci0}{0}{2.5}{$9$};
	\din{ci1}{0}{2.2}{$1$};	
	\din{ci2}{0}{1.9}{$2$};
	\qin{cq0}{0}{1.6}{};
  \din{ci3}{0}{1.3}{$1$};
  \din{ci4}{0}{1}{$2$};  
  \din{ci5}{0}{0.7}{$1$};
  \din{ci6}{0}{0.4}{$2$};
  
  \qmatch{lq}{3.1}{1.6}{};\flow{$(lq)+(-0.8,0)$}{lq}{dashed}{};\flow{lq}{$(lq)+(0.8,0)$}{dashed}{};		
  \dmatch{ld}{3.1}{1.3}{$7$}{below};\flow{$(ld)+(-0.8,0)$}{ld}{}{};\flow{ld}{$(ld)+(0.8,0)$}{}{};
  
  \computonComposite{3.9}{1}{1.5}{1.6};\node[align=center] at (4.6,1.7){Output\\Gate}; 
  \flow{gi0}{{3.9,2.2}}{}{bend left=20};
  \flow{gi1}{{3.9,1.9}}{}{bend left=20};
  
  \qout{go0}{5.4}{1.6}{};
  \dout{go1}{5.4}{1.3}{$12$};
  \doutplain{co0}{5.4}{0.7}{$7$}{right};\flow{$(co0)+(-3.9,0)$}{co0}{}{};

  \node at (11,0){};
\end{tikzpicture}
}
\caption{Subfigures (a)-(d) are composite computons that serve as building blocks to construct the memory cell depicted in (e). Each of them model a dotted box with the same name in Figure \ref{fig:application-LSTM-concept}(a). For example, the composite (a) models the forget gate by providing ed-inports of colour $1$, $2$ and $9$ which respectively correspond to $h_{t-1}$, $x_t$ and $c_{t-1}$ (according to Figure \ref{fig:application-LSTM-concept}(b)). The functional computon $s_1$ uses $h_{t-1}$ and $x_t$ to produce a $3$-coloured data item (i.e., $f_t$) which is sequentially passed to the functional computon $m_1$. The output of $m_1$ is the output of (a), viz., a $10$-coloured data item corresponding to $j_t$. Evidently, (a) has a single ec-inport to allow the forget gate to be invoked from outside, and a ec-outport to pass control to other computons. For example, in (d), (a) passes control to the join computon used to synchronise the composites (a) and (b).}
\label{fig:application-LSTM-composites}
\end{figure}

The final step in our compositional construction is forming the composite structure that represents the whole memory cell. Figure \ref{fig:application-LSTM-composites}(e) shows that such a structure is precisely the partial sequential computon $(((s_1 \rhd_{\rho_1} m_1) \mid_{\rho_6} ((s_2 \mid_{\rho_2} t_1) \unrhd_{\rho_3} m_2)) \unrhd_{\rho_7} a_1) \rhd_{\rho_{8}} ((s_3 \mid_{\rho_4} t_2) \unrhd_{\rho_5} m_3)$ in which the composite $((s_1 \rhd_{\rho_1} m_1) \mid_{\rho_6} ((s_2 \mid_{\rho_2} t_1) \unrhd_{\rho_3} m_2)) \unrhd_{\rho_7} a_1$ and the output gate $(s_3 \mid_{\rho_4} t_2) \unrhd_{\rho_5} m_3$ are the left and right operands, respectively. For visualisation purposes, we treat such operands as black boxes. 

The resulting memory cell composite can also be treated as a black box for defining even more complex composites such as a complete Recurrent Neural Network. Black boxing is possible because computons are modular by construction so that their internals can be hidden without any side effects. In this scenario, the memory cell hides the complexity of the total sequential computons from Figures \ref{fig:application-LSTM-composites}(c) and \ref{fig:application-LSTM-composites}(d). By Proposition \ref{prop:computon-sequential-connected} and Corollary \ref{cor:computon-parallel-sync-connected}, all the composites and primitive computons we deal with are connected in the sense of Definition \ref{def:computon-connected}. This can be easily verified by observing in Figure \ref{fig:application-LSTM-composites} that there is information flow from every non-e-outport to either an ec-outport or an ed-outport. For example, there is information flow from the ec-inport of the p-sync computon of the output gate to the unique ec-outport of $m_3$. Although a purely sequentially-driven construction could have been used instead, we decide to rely upon p-sync computons so as to emphasise the parallel nature of computations occuring within a memory cell. Using p-sync composites instead of p-async constructions allows us to semantically express the fact that data needs to be synchronised (via control) before being consumed by subsequent computations.

The behaviour of the whole memory cell is displayed in Figure \ref{fig:application-LSTM-behaviour}, which corresponds to the net under $\mathcal{C}\circ \mathfrak{E}$ of the composite from Figure \ref{fig:application-LSTM-composites}(e). Like in the previous example, we decide to just display the net that encapsulates control flow since it comprehensively captures the computational behaviour of the memory cell. The corresponding nets under $\mathcal{N}$ and under $\mathcal{D}$ can be easily constructed using the mapping from \ref{sec:appendix-mapping}, which captures the mapping given by the functorial constructions from Definition \ref{def:functor-computon-to-petri} and Proposition \ref{prop:functor-data-petri}. In Figure \ref{fig:application-LSTM-separation}, we show the equivalent BPMN diagram (for control flow) and the DFG in standard notation (for data flow) using the graph transformation system described in Section \ref{sec:transformation-system}.

\begin{figure}[!h]
\centering
{
\begin{tikzpicture}
\node[place,label={180:},minimum size=3mm] (if1) at (0,1) {};
\node[transition,fill=black,minimum width=0.1mm,minimum height=5mm] (f1) at (0.5,1) {};
\node[place,label={180:},minimum size=3mm] (o1f1) at (1,1.5) {};
\node[place,label={180:},minimum size=3mm] (o2f1) at (1,0.5) {};

\node[transition,fill=black,minimum width=0.1mm,minimum height=5mm] (s1) at (1.5,1.5) {};
\node[place,label={180:},minimum size=3mm] (os1) at (2,1.5) {};
\node[transition,fill=black,minimum width=0.1mm,minimum height=5mm] (m1) at (3.5,1.5) {};
\node[place,label={180:},minimum size=3mm] (om1) at (5,1.5) {};

\node[transition,fill=black,minimum width=0.1mm,minimum height=5mm] (f2) at (1.5,0.5) {};
\node[place,label={180:},minimum size=3mm] (o1f2) at (2,1) {};
\node[place,label={180:},minimum size=3mm] (o2f2) at (2,0) {};
\node[transition,fill=black,minimum width=0.1mm,minimum height=5mm] (s2) at (2.5,1) {};
\node[place,label={180:},minimum size=3mm] (os2) at (3,1) {};
\node[transition,fill=black,minimum width=0.1mm,minimum height=5mm] (t1) at (2.5,0) {};
\node[place,label={180:},minimum size=3mm] (ot1) at (3,0) {};
\node[transition,fill=black,minimum width=0.1mm,minimum height=5mm] (j2) at (3.5,0.5) {};
\node[place,label={180:},minimum size=3mm] (oj2) at (4,0.5) {};
\node[transition,fill=black,minimum width=0.1mm,minimum height=5mm] (m2) at (4.5,0.5) {};
\node[place,label={180:},minimum size=3mm] (om2) at (5,0.5) {};
\node[transition,fill=black,minimum width=0.1mm,minimum height=5mm] (j1) at (5.5,1) {};

\node[place,label={180:},minimum size=3mm] (oj1) at (6,1) {};
\node[transition,fill=black,minimum width=0.1mm,minimum height=5mm] (a1) at (6.5,1) {};
\node[place,label={180:},minimum size=3mm] (oa1) at (7,1) {};
\node[transition,fill=black,minimum width=0.1mm,minimum height=5mm] (f3) at (7.5,1) {};
\node[place,label={180:},minimum size=3mm] (o1f3) at (8,1.5) {};
\node[place,label={180:},minimum size=3mm] (o2f3) at (8,0.5) {};
\node[transition,fill=black,minimum width=0.1mm,minimum height=5mm] (s3) at (8.5,1.5) {};
\node[place,label={180:},minimum size=3mm] (os3) at (9,1.5) {};
\node[transition,fill=black,minimum width=0.1mm,minimum height=5mm] (t2) at (8.5,0.5) {};
\node[place,label={180:},minimum size=3mm] (ot2) at (9,0.5) {};
\node[transition,fill=black,minimum width=0.1mm,minimum height=5mm] (j3) at (9.5,1) {};
\node[place,label={180:},minimum size=3mm] (oj3) at (10,1) {};
\node[transition,fill=black,minimum width=0.1mm,minimum height=5mm] (m3) at (10.5,1) {};
\node[place,label={180:},minimum size=3mm] (om3) at (11,1) {};

\draw[-latex,thick](if1)--(f1);\draw[-latex,thick](f1)--(o1f1);\draw[-latex,thick](f1)--(o2f1);
\draw[-latex,thick](o1f1)--(s1);\draw[-latex,thick](s1)--(os1);\draw[-latex,thick](os1)--(m1);\draw[-latex,thick](m1)--(om1);\draw[-latex,thick](om1)--(j1);
\draw[-latex,thick](o2f1)--(f2);\draw[-latex,thick](f2)--(o1f2);\draw[-latex,thick](o1f2)--(s2);\draw[-latex,thick](s2)--(os2);\draw[-latex,thick](os2)--(j2);
\draw[-latex,thick](f2)--(o2f2);\draw[-latex,thick](o2f2)--(t1);\draw[-latex,thick](t1)--(ot1);\draw[-latex,thick](ot1)--(j2);
\draw[-latex,thick](j2)--(oj2);\draw[-latex,thick](oj2)--(m2);\draw[-latex,thick](m2)--(om2);\draw[-latex,thick](om2)--(j1);

\draw[-latex,thick](j1)--(oj1);\draw[-latex,thick](oj1)--(a1);\draw[-latex,thick](a1)--(oa1);\draw[-latex,thick](oa1)--(f3);
\draw[-latex,thick](f3)--(o1f3);\draw[-latex,thick](o1f3)--(s3);\draw[-latex,thick](s3)--(os3);\draw[-latex,thick](os3)--(j3);
\draw[-latex,thick](f3)--(o2f3);\draw[-latex,thick](o2f3)--(t2);\draw[-latex,thick](t2)--(ot2);\draw[-latex,thick](ot2)--(j3);
\draw[-latex,thick](j3)--(oj3);\draw[-latex,thick](oj3)--(m3);\draw[-latex,thick](m3)--(om3);
\end{tikzpicture}
}
\caption{Behaviour of the memory cell computon presented in Figure \ref{fig:application-LSTM-composites}(e), expressed as the Petri net \resizebox{0.69\textwidth}{!}{$\mathcal{C}(\mathfrak{E}({(((s_1 \rhd_{\rho_1} m_1) \mid_{\rho_6} ((s_2 \mid_{\rho_2} t_1) \unrhd_{\rho_3} m_2)) \unrhd_{\rho_7} a_1) \unrhd_{\rho_{8}} ((s_3 \mid_{\rho_4} t_2) \unrhd_{\rho_5} m_3)}))$}}
\label{fig:application-LSTM-behaviour}
\end{figure}

\begin{figure}[!h]
\centering
\subcaptionbox{BPMN Diagram (control flow).}
{
\includegraphics[scale=0.4]{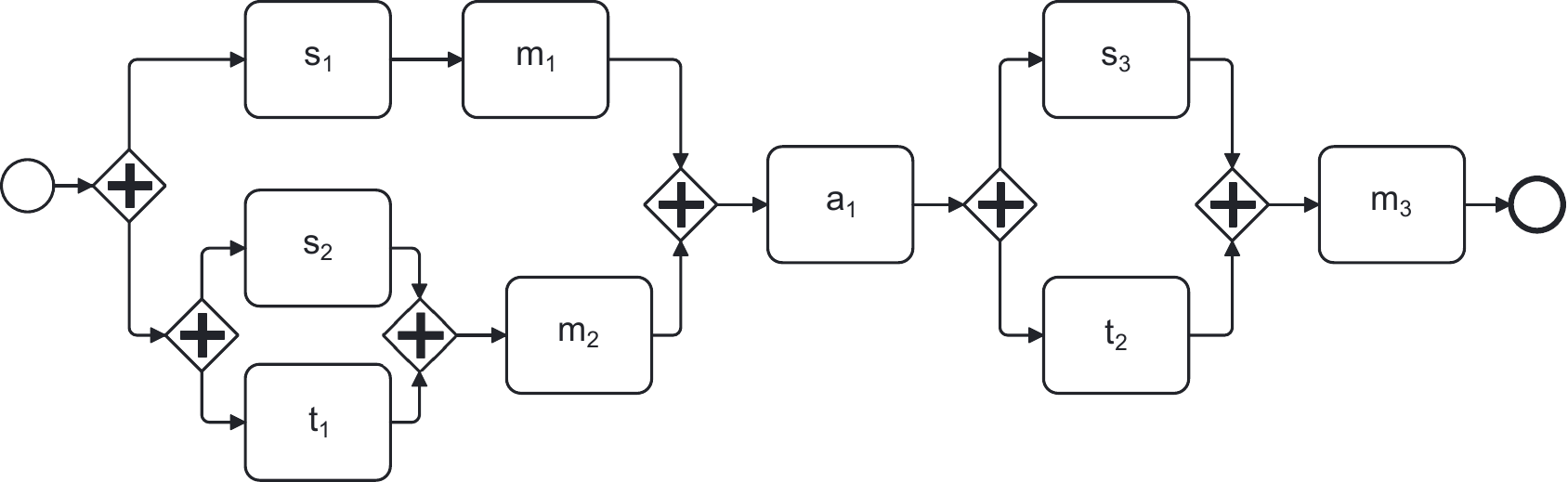}
}
\subcaptionbox{DFG in standard notation (data flow). For readability and since they do not produce or receive any data, we omit the nodes corresponding to fork or join computons.}
{
\begin{tikzpicture}[scale=0.82]
  \node(d1) at (0,5){$1$};
  \node(d2) at (0,4){$2$};
  \node(d3) at (0,3){$1$};
  \node(d4) at (0,2){$2$};
	\node(d5) at (0,1){$1$};	
	\node(d6) at (0,0){$2$};	
	
	\node[circle,draw=black](s1) at (1,4.5){$s_1$};
	\node[circle,draw=black](s2) at (1,2.5){$s_2$};
	\node[circle,draw=black](t1) at (1,0.5){$t_1$};
	
	\node(d7) at (2,5){$9$};
	\node[circle,draw=black](m1) at (3,4.5){$m_1$};
	\node[circle,draw=black](m2) at (3,1.5){$m_2$};
	
	\node[circle,draw=black](a1) at (5,3){$a_1$};
	
	\node(d8) at (6,5){$1$};
	\node(d9) at (6,4){$2$};
	\node[circle,draw=black](s3) at (7,4.5){$s_3$};
	\node(d10) at (6.5,3){$7$};
	\node[circle,draw=black](t2) at (7,1.5){$t_2$};
	
	\node[circle,draw=black](m3) at (9,3){$m_3$};
	\node(d11) at (10.5,3){$12$};
	
	\draw[-latex] (d1) -- (s1);\draw[-latex] (d2) -- (s1);
	\draw[-latex] (d3) -- (s2);\draw[-latex] (d4) -- (s2);
	\draw[-latex] (d5) -- (t1);\draw[-latex] (d6) -- (t1);
	\draw[-latex] (d7) -- (m1);\draw[-latex] (s1) -- node[below] {$3$} (m1);
	\draw[-latex] (s2) -- node[above] {$4$} (m2);
	\draw[-latex] (t1) -- node[below] {$6$} (m2);
	\draw[-latex] (m1) -- node[above, yshift=0.5mm] {$10$} (a1);
	\draw[-latex] (m2) -- node[below, yshift=-0.5mm] {$11$} (a1);
	\draw[-latex] (d8) -- (s3);\draw[-latex] (d9) -- (s3);
	\draw[-latex] (a1) -- (d10);\draw[-latex] (a1) -- node[below] {$7$} (t2);
	\draw[-latex] (s3) -- node[above] {$5$} (m3);
	\draw[-latex] (t2) -- node[below] {$8$} (m3);
  \draw[-latex] (m3) -- (d11);
\end{tikzpicture}
}
\caption{BPMN diagram and DFG in standard notation to respectively express the control flow and data flow structures of the memory cell $(((s_1 \rhd_{\rho_1} m_1) \mid_{\rho_6} ((s_2 \mid_{\rho_2} t_1) \unrhd_{\rho_3} m_2)) \unrhd_{\rho_7} a_1) \rhd_{\rho_{8}} ((s_3 \mid_{\rho_4} t_2) \unrhd_{\rho_5} m_3)$ depicted in Figure \ref{fig:application-LSTM-composites}(e).}
\label{fig:application-LSTM-separation}
\end{figure}

The process described in Section \ref{sec:transformation-system} for transforming a CFG into a BPMN diagram does not consider data flow at all, whereas the process for transforming a DFG does not consider control flow at all. Although the separation of concerns can be leveraged for other purposes (e.g., to formally verify termination of control flow only), the purpose of this section is just to demonstrate how control flow and data flow can be analysed independently for model transformation.

Figure \ref{fig:application-LSTM-separation} demonstrates that, as every composite computon captures both control flow and data flow within a single compositional structure, the proposed model can be perceived as a unification (or conciliation) of workflow control flow languages \cite{van_der_aalst_application_1998} and data flow languages \cite{johnston_advances_2004} within a single compositional setting. 
\section{Related Work} \label{sec:related-work}

In this section, we present the related work of our proposal, namely related compositional approaches and component models that separate concerns.

The CommUnity component model \cite{fiadeiro_categories_2005} separates computation from interaction and provides categorical semantics for architectural configurations. At its core, software components (called designs) are objects in a category, interrelated through morphisms. The most well-known structure preserving relationships are the so-called superposition morphisms, which allow embedding designs into more complex ones, similar to the action of computon morphisms. The difference is that our morphisms force computons to interact exclusively at the boundaries and that our colimit constructions define explicit control flow structures. CommUnity colimits serve to merge (or amalgamate) designs compositionally à la Definition \ref{def:pushout-computation}, leaving control flow implicit without any support to separate it from data. 

Prosave \cite{bures_procom_2008} is a design language built on top of the ProCom component model, which was inspired on \cite{hanninen_rubus_2008} to allow the definition of nested structures of interconnected components. Like computons, Procom components are passive units of computation with explicit separation of control ports and data ports. Despite of this similarity, Procom is not compositional since it does not provide algebraic operators to perform control-based composition, but just informal programming constructs for connecting ports either directly or indirectly. Indirect connection is done through so-called connectors which establish control or data flow interaction between components via message passing. As the model is not compositional, Procom composites do not offer a clear separation of concerns like their internal components. They rather have data ports only where both control flow and data flow terminate. 

SCADE \cite{colaco_scade_2017} is a similar component model which integrates an imperative language (i.e., Esterel \cite{berry_esterel_1992}) and a functional language (i.e., Lustre \cite{halbwachs_synchronous_1991}) to define control flow and data flow, respectively. Particularly, so-called Safe-State Machines (SSMs) model the discrete control part of a system, whereas Lustre blocks serve to continuously process data. Like Prosave, SCADE does not provide formal operators for defining control-based composite blocks, but just programming constructs to non-compositionally assemble a system.

In the same line of work, \cite{arellanes_decentralized_2023,lau_component_2011,stepan_design_2011} describe a component model that provides two orthogonal dimensions to manage control flow and data flow separately. The model encapsulates control since it offers composition operators to define sequential, parallel or branching composites in a hierarchical, bottom-up manner. Unfortunately, the semantics of the model is semi-formal \cite{arellanes_evaluating_2020,arellanes_exogenous_2017} so it is not possible to precisely determine whether the model is fully compositional or not. Also, components do not have separate ports for data and control, but just control ports. Consequently, the data dimension is implicitly defined in the underlying composition mechanism whose goal is to build complex workflows from simpler ones.

Workflow Nets (WF-nets) \cite{van_der_aalst_application_1998,van_der_aalst_soundness_2011} provide support for modelling workflow processes in the form of control-driven computations. As they offer well-founded semantics built upon Petri nets, WF-nets formalise the notion of workflow graphs which are traditionally specified through industry-oriented languages such as UML diagrams, Event-driven Process Chains or the BPMN notation. WF-nets do not separate control from data and do not provide formal operators for explicitly and compositionally defining sequential, parallel, branching or iterative composites. The issue of the separation of concerns is resolved by RWFN-nets \cite{prisecaru_resource_2008} which unify extended WF-nets and so-called resource nets for separating the process and resource perspectives of a workflow. Although the model provides a clear separation of concerns, there is not a clear distinction between input and output data, and composition is not algebraically defined. Therefore, RWFN-nets do not separate data and control compositionally. Other Petri net based approaches, for workflow construction, that non-compositionally separate control flow and data flow are the functor model \cite{ohba_functor_1981}, extended-time nets \cite{peng_automated_1994}, the FunState model \cite{thiele_funstate-internal_1999} and dual flow nets \cite{varea_dual_2006}. 

Existing compositional approaches built upon Petri net foundations rely on the notion of open interfaces to the external world. Specially designated open places are particularly used by open Petri nets (ONets) \cite{baez_open_2020} to construct complex behaviours from simpler ones. In this framework, ONet composition is realised by gluing the output places of one net with the input places of another. As this composition mechanism is characterised as a pushout in a categorical setting \cite{baldan_compositional_2005}, ONets are compositional. An ONet morphism resembles a computon morphism in the sense that input and output places can be preserved upon transformation (see Proposition \ref{prop:computon-morphism-inports-outports}). Nevertheless, like Definition \ref{def:pushout-computation}, a pushout operation just serves for merging two ONets via a common object so that there are no specific operators for explicitly defining sequential, parallel, branching or iterative composites (i.e., ONets do not encapsulate explicit control flow). Petri box calculus \cite{best_box_2002}, Open WF-nets \cite{wolf_does_2009}, Petri nets with interface \cite{baldan_modular_2015}, nets with boundaries \cite{bruni_connector_2013} and Petri net components \cite{kindler_compositional_1997} also rely on the notion of open interfaces. Like ONets, all these Petri-net-based approaches do not separate data from control.

Although they do not separate concerns, Whole-grain Petri nets \cite{kock_whole-grain_2022} deserve a mention since, unlike classical Petri net theory and like computons, they abolish the traditional notion of multisets of places, typically expressed as a free commutative monoid $S^\oplus$ on a set $S$ of places (cf. Definition \ref{def:petri-net-cat}). Accordingly, they also work upon a similar categorical scheme to \textbf{Comp}, in order to define concrete instances of Whole-grain nets (cf. \cite{kock_whole-grain_2022,patterson_categorical_2022}). The difference is that \textbf{Comp} has objects that enable computons to have a clear distinction between control and data ports. Another difference is that our theory identifies particular classes of computon objects that can be used as building blocks to define more complex computons through sequencing, parallelising, branching or iteration operations. Although primitive computons are isomorphic to Whole-grain corollas, Whole-grain Petri nets do not distinguish between different types of corollas (e.g., join or fork corollas).

Within the realm of related compositional models, we also find string diagrams \cite{piedeleu_introduction_2025,hinze_introducing_2023,pavlovic_programs_2023}, whose origins trace back some 50 years to Penrose's graphical notation for tensor networks \cite{penrose_applications_1971}, later given a sound and complete categorical foundation by Joyal--Street \cite{joyal_geometry_1991}. To date, such diagrams have been widely applied across several domains, including quantum mechanics \cite{coecke_picturing_2017}, electrical engineering \cite{baez_compositional_2018}, control theory \cite{baez_categories_2015} and natural language processing \cite{coecke_mathematical_2010}, just to name a few. String diagrams are becoming increasingly popular because they offer well-founded syntax to graphically represent symmetric monoidal categories, where boxes represent processes/morphisms and wires express inputs or outputs for those processes \cite{selinger_survey_2011,coecke_interacting_2011,bonchi_full_2015,coecke_mathematical_2016,fong_categorical_2016,bonchi_diagrammatic_2019,fong_string_2020,piedeleu_string_2021,bonchi_string_2022,bonchi_string_2022-2,ghica_rewriting_2023}. In this model, string diagrams can be composed sequentially via morphism composition \cite{bonchi_string_2022}, in parallel via monoidal product \cite{selinger_survey_2011} and into iterative structures via traces \cite{joyal_traced_1996,ghica_rewriting_2023}. Extensions to support probabilistic branching have also been developed \cite{villoria_enriching_2025,piedeleu_complete_2025}. Since sequential composition is done by totally matching outputs with inputs (or domain with codomain) and there are not distinguished wires for representing control, it follows that, unlike computons, not every string diagram can be composed sequentially with one another (cf. Theorem \ref{th:always-sequentiable}). Moreover, there is no distinction between control flow and data flow, as evidenced by the uniform treatment of wires to carry data only \cite{bonchi_string_2022-1}. Some efforts have been made to colouring wires but not to explicitly separate concerns \cite{earnshaw_string_2023}. The separation of concerns has more recently been addressed through tape diagrams \cite{bonchi_deconstructing_2023,bonchi_diagrammatic_2025,bonchi_tape_2025}, which are similar in spirit to \cite{comfort_sheet_2020}, but that allow the formation of nested string diagrams where an inner layer represents data flow and an outer part expresses control flow. Technically, the control/data interplay is captured by the laws of rig categories with finite biproducts, featuring two symmetric monoidal structures on relations, ${(Rel,\oplus,0)}$ and ${(Rel,\otimes,1)}$, which respectively serve to model control and data flow. Both monoidal structures can admit traces to model iteration \cite{bonchi_diagrammatic_2025} and, like string diagrams, total sequencing, asynchronous parallelising and branching are supported. Unfortunately, the notions of synchronous parallelising and partial sequencing are not considered, since it is assumed that data always follows control \cite{bonchi_tape_2025}. As a result, tape diagrams lack the expressive power to model scenarios where asynchronous data can delay process execution independently of control. Moreover, no independent extraction of control flow and data flow graphs from tape diagrams has yet been attempted in a categorical setting. The computon model accommodates this through the functors presented in Section \ref{sec:control-data-structures}. 

The colimit-based composition operators described in Section \ref{sec:composite-computons} might resemble works from category-theoretic cybernetics. For example, \cite{ehresmann_memory_2007} proposes the use of general colimits to bind components of complex systems. Although binding is hierarchical and bottom-up, as in the computon model, the proposed operations do not define explicit control flow structures capable of operationally enacting what the authors describe as causal and informational interactions, which correspond to those enabled by computons. Composition machines \cite{arellanes_composition_2024} attempt to address this limitation by composing components into total sequential structures. Rather than using colimits, composition is realised by morphism composition, applied dynamically to generate spaces of sequential composites in discrete time. Unlike computons, there are no operations for forming branching, iterative or parallel composites. A glance at Figure \ref{fig:application-LSTM-composites}(d) in Section \ref{sec:applications} reveals that the structure of a composite computon is like a membrane in which other computons reside and that can be part of another membrane/composite. An edge connected to/from a composite's e-port is akin to a fibre which can traverse other membranes, as long as the e-port it is connected to/from does not become an i-port. This analogy resembles the structural organisation of a P-system \cite{paun_computing_2000} where membranes are delimiting compartments of multisets of objects that evolve according to bio-inspired rules. Other models resembling this structural analogy include Architectural Design Rewriting \cite{bruni_formal_2011}, Fractal \cite{bruneton_fractal_2006} and Robin Milner's bigraphs \cite{milner_space_2009}. Unfortunately, all these models do not separate data and control, and some of them just consider these dimensions implicitly.
\section{Conclusions and Future Directions} \label{sec:conclusions}

In this paper, we presented a model of high-level computation in which computons are first-class semantic entities which structurally possess a number of computation units that can be connected to/from two types of ports: control ports and data ports. Computons are objects in a functor category, denoted ${\textbf{Set}^\textbf{Comp}}$, where two major classes of objects reside. The first class is that of trivial computons which have just ports and no computation units. The second class pertains to primitive computons which are fully connected entities in the sense they have a unique computation unit to which all ports are attached. These two classes serve as building blocks to define complex computons via category-theoretic operations. We particularly presented separate operations to inductively form sequential, parallel, branching and iterative composite computons. In Section \ref{sec:composite-computons}, we proved that all of them are connected in the sense of Definition \ref{def:computon-connected}, meaning there is a sequence of information flows from every non e-outport to some e-outport. As the model is compositional, composites exhibit the same properties as their constituents, i.e., they have the same structure with a clear separation of control flow and data flow. 

Generally speaking, both control flow and data flow are inextricably present in any classical high-level computation (e.g., a workflow process), so it is crucial to separately reason about them for verification, maintainability and optimisation purposes. For example, in Section \ref{sec:applications}, we leveraged the separation of concerns of the proposed model to show how control flow can be transformed into a BPMN diagram without analysing data flow at all, and how data flow can be converted into a DFG in standard notation without considering control flow at all. Evidently, model transformation is not the only way of exploiting the separation of concerns of our proposal. By leveraging the fact that the behaviour of a computon can be expressed as a token game, it is also possible to use standard Petri net tools or relevant graph-based analysis techniques to separately verify computing properties, such as reachability of control flow only or data flow only. Taking advantage of graph-based techniques can also enable an optimal implementation in which functional computons exchange data decentrally while composites coordinate control flow hierarchically \cite{arellanes_decentralized_2023}. Although our model does not consider explicit structures for data processing (e.g., map-reduce or filter constructs), because data flow is ultimately governed by control flow, we acknowledge that introducing them is important to increase the expressivity of composite computons. However, doing this in a compositional manner requires further investigation. 

Enabling compositionality is also important to induce modularity which is a well-known feature for reusing computations at scale. Modularity does not imply compositionality because modules can be constructed in many different ways (not necessarily algebraically). When an algebraic composition mechanism is used to realise this feature, computation properties are preserved across all composition levels. In our proposal, the separation of control flow and data flow is one of such properties. Thus, as computons only interact through their respective e-ports, composite computons can be perceived as modular black-boxes that encapsulate control and data flow structures. Although branching is supported by our theory, not every pair of computons is a candidate for defining a branching structure, as described in Section \ref{sec:branching-computons}. 

In Section \ref{sec:iterative-computons}, we showed that a head- or a tail-iterative structure can be formed for any connected computon (see Theorems \ref{th:computon-iterative-head-always} and \ref{th:computon-iterative-tail-always}). Likewise, in Sections \ref{sec:sequential-computons} and \ref{sec:parallel-computons}, we proved that any pair of connected computons can always be composed sequentially (see Theorem \ref{th:always-sequentiable}) or in parallel (see Theorems \ref{th:computon-parallel-async-always} and \ref{th:always-parallelisable-sync}) regardless of the data they require or produce. This is because, intuitively, a computon has at least one ec-outport that can always be matched with at least one ec-inport of another. Matching all the e-outports of one computon with all the e-inports of another one gives rise to total sequential composition which, to the best of our knowledge, is the \emph{de facto} way of sequencing computations nowadays (cf. \cite{fong_algebra_2016}). 

In this paper, we argue that sequencing is a particular form of merging because the former can be expressed in terms of the latter. Particularly, in our proposal, merging corresponds to a pushout operation in $\textbf{Set}^\textbf{Comp}$ (see Definition \ref{def:pushout-computation}), while sequencing is characterised as a pushout with restrictions in the same category (see Definition \ref{def:computon-sequential}). As sequencing cannot only be done totally but also partially, our sequencing mechanism is more general than those prevailing in the existing literature. Partial composition entails that non-matching e-ports are preserved across every composition level (e.g., see Figures \ref{fig:computon-sequential-example} and \ref{fig:computon-parallel-sync-example}).  

If computons are seen as relations from e-inports to e-outports, our composition mechanism provides the basis to redefine the current notion of composition of relations which states that the composite ${S \circ R}$ of ${R \subseteq X \times Y}$ and ${S \subseteq Y \times Z}$ is given by ${\{(x,z) \mid \exists y[R(x,y)~\land~S(y,z)]\}}$. Since ${S \circ R}$ is a subset of ${X \times Z}$, it is evident that some relations in ${X \times Y}$ and in ${Y \times Z}$ are lost. By resorting to the foundations laid in this paper, a more structure-preserving definition emerges: ${(S \circ R) \cup [(X \times Y) \setminus (S \circ R)] \cup [(Y \times Z) \setminus (S \circ R)]}$. Thus, rather than being a subset of ${X \times Z}$, a composite relation would be a subset of ${(X \cup Y) \times (Y \cup Z)}$. In the future, we would like to further investigate this new notion derived from the foundations of partial sequential composition.

Defining computons as preorders in a categorical setting can be achieved by borrowing ideas from resource theories \cite{coecke_mathematical_2016}. We hypothesise there are symmetric monoidal categories in which computons are morphisms and ports are objects. Defining categories of this sort can be helpful to study the operational semantics of composite computons through the arrow of time. Particularly, \emph{v-categories} and \emph{v-profunctors} can provide theoretical underpinnings for formally answering specific questions about the execution of computons. Another potential direction is to study the operational semantics of computons from the lenses of polynomial-style finite-set configurations and etale maps in the context of Whole-grain Petri nets and processes. Studying operational semantics from different angles is possible due to the separation between composition and execution semantics of the proposed model. 

\appendix

\section{Mapping Between Petri Net Syntax and Computon Syntax} \label{sec:appendix-mapping}
\setcounter{table}{0}

Table \ref{tab:syntax-mapping} presents a mapping from Petri net syntax to computon syntax, given by any of the three functorial constructions presented in Section \ref{sec:operational-semantics}, which is useful to discuss the operational semantics of computons. A glance at this table reveals that, in general, places with no incoming arrows correspond to e-inports, whereas places with no outgoing arrows correspond to e-outports. This reflects the fact that e-inports and e-outports receive and send information from and to the external world, respectively. 

\begin{table}[!h]
  \begin{center}
    \caption{Mapping from computon syntax to Petri net syntax where $n_j$ is a natural number greater than zero for all ${j=1,\ldots,14}$.}
    \label{tab:syntax-mapping}
    \begin{tabular}{ |c|c|c|c| } 
 		 \hline
		 & & Computon syntax & Petri net syntax \\
		 \hline
		 \multirow{3}{*}{Control} & ec-inport &
		 \begin{tikzpicture}		    
		 	\qin{q0}{0}{0}{};
		 \end{tikzpicture} &
		 \begin{tikzpicture}
		     \node at (-0.25,0){};
		 	 \node[place,label={180:},minimum size=3mm] (q0) at (0.1,0) {}; 
		 	 \draw[-latex,thick] (q0) -- (0.6,0);			 
		 \end{tikzpicture}
		 \\ \cline{2-4}
		 & ec-outport &		  
		 \begin{tikzpicture}		 
		     \node at (0,0){};
			 \qout{q1}{0.1}{0}{};
		 \end{tikzpicture} &
		 \begin{tikzpicture}
		     \node at (0.9,0){};
		 	 \node[place,label={360:},minimum size=3mm] (q1) at (0.5,0) {};
			 \draw[-latex,thick] (0,0) -- (q1);		     
		 \end{tikzpicture}
		 \\ \cline{2-4}
		 & ec-inoutport &		  
		 \begin{tikzpicture}
		 	 \node at (0,0){};		 
		 	 \qmatch{q1}{-0.1}{0}{};
		 \end{tikzpicture} &
		 \begin{tikzpicture}
		 	 \node[place,label={},minimum size=3mm] (q1) at (0.5,0) {};
		 \end{tikzpicture}
		 \\ \cline{2-4}
		 & ic-port &
		 \hspace*{0em}		  
		 \begin{tikzpicture}
		   \qmatch{qx}{0.4}{0}{};
		 	 \flow{{-0.2,0}}{qx}{dashed}{};\flow{qx}{{1,0}}{dashed}{};
		 \end{tikzpicture} &
		 \begin{tikzpicture}
			 \node[place,label={},minimum size=3mm] (p1) at (0.5,0) {};
			 \draw[-latex,thick] (0,0) -- (p1);\draw[-latex,thick] (p1) -- (1,0);
		 \end{tikzpicture}
		 \\ \hline
		 \multirow{3}{*}{Data} & ed-inport &		  
		 \begin{tikzpicture}
			 \din{o1}{0}{0}{$n_1$};
		 \end{tikzpicture} &
		 \begin{tikzpicture}
		     \node at (0.1,0){};
			 \node[place,label={180:},minimum size=3mm] (i1) at (0.5,0) {};
		 	 \draw[-latex,thick] (i1) -- (1,0);
		 \end{tikzpicture}
		 \\ \cline{2-4}
		 & ed-outport &		  
		 \begin{tikzpicture}		 	
		    \node at (0,0){};
		 	\dout{i1}{0.3}{0}{$n_2$};
		 \end{tikzpicture} &
		 \begin{tikzpicture}
		    \node at (0.47,0){};
		    \node[place,label={360:},minimum size=3mm] (o1) at (0.1,0) {};
			\draw[-latex,thick] (-0.4,0) -- (o1);		 	
		 \end{tikzpicture}
		 \\ \cline{2-4}
		 & ed-inoutport &		  
		 \begin{tikzpicture}
		 	\node at (0,0){};		 	
		 	\dmatch{i1}{0.2}{0}{$n_3$}{above};
		 \end{tikzpicture} &
		 \begin{tikzpicture}		 	
		    \node[place,label={},minimum size=3mm] (o1) at (0.7,0) {};
		 \end{tikzpicture}		 		  
		 \\ \cline{2-4}
		 & id-port &
		 \hspace*{0em}		  
		 \begin{tikzpicture}
		 	 \dmatch{dx}{0.4}{0}{$n_4$}{above};
		 	 \flow{{-0.2,0}}{dx}{}{};\flow{dx}{{1,0}}{}{};
		 \end{tikzpicture} &
		 \begin{tikzpicture}
			 \node[place,label={},minimum size=3mm] (p1) at (0.5,0) {};
			 \draw[-latex,thick] (0,0) -- (p1);\draw[-latex,thick] (p1) -- (1,0);
		 \end{tikzpicture}
		 \\ \hline
		 \multicolumn{2}{|c|}{Trivial Computon} &		 		     	
    	\begin{tikzpicture}    	 
\qmatch{q1}{1.2}{2.7}{}
\node[label={$\vdots$}] (dots1) at (1.2,1.85) {};
\qmatch{qi}{1.2}{1.8}{}
\dmatch{d1}{1.2}{1.4}{$n_5$}{left};
\node[label={$\vdots$}] (dots3) at (1.2,0.51) {};
\dmatch{dj}{1.2}{0.5}{$n_6$}{left};
\node(del) at (1.55,0.5){};
    	\end{tikzpicture} &
    	\begin{tikzpicture}
\node(del) at (0,2.9){};
\node[place,label={[xshift=-0.6cm,yshift=-0.45cm]:},minimum size=3mm] (q0) at (1.2,2.7) {};
\node[label={$\vdots$}] (dots1) at (1.2,1.85) {};
\node[place,label={[xshift=-0.6cm,yshift=-0.45cm]:},minimum size=3mm] (qi) at (1.2,1.8) {};
\node[place,label={[xshift=-0.6cm,yshift=-0.45cm]:},minimum size=3mm] (d1) at (1.2,1.4) {};
\node[label={$\vdots$}] (dots3) at (1.2,0.51) {};
\node[place,label={[xshift=-0.6cm,yshift=-0.45cm]:},minimum size=3mm] (dj) at (1.2,0.5) {};
\node(del2) at (2.4,0.5){};
    	\end{tikzpicture}
    	\\ \hline
		 \multicolumn{2}{|c|}{Functional Computon} &		 		     	
    	\begin{tikzpicture}
	    	\computonPrimitive{1}{0}{1}{1.9}{$\lambda$}
	    	
      \qin{1q0}{0.2}{1.7}{}
			\din{1i1}{0.2}{1.1}{$n_7$};
			\node[label={$\vdots$}] (1i2) at (0.6,0.25) {};
			\din{1in}{0.2}{0.2}{$n_8$};

			\qout{1q1}{2}{1.7}{}
			\dout{1o1}{2}{1.1}{$n_9$};
			\node[label={$\vdots$}] (1o2) at (2.4,0.25) {};
			\dout{1om}{2}{0.2}{$n_{10}$};
    	\end{tikzpicture} &
    	\begin{tikzpicture}
			\node[transition,fill=black,minimum width=0.5mm,minimum height=10mm,label=$ $] (t1) at (1,0.7) {};

			\node[place,label={180:},minimum size=3mm] (q0) at (0,1.7) {};
			\node[place,label={180:},minimum size=3mm] (i1) at (0,1.1) {};
			\node[label={$\vdots$}] (i2) at (0,0.15) {};
			\node[place,label={180:},minimum size=3mm] (in) at (0,0) {};

			\node[place,label={0:},minimum size=3mm] (q1) at (2,1.7) {};
			\node[place,label={0:},minimum size=3mm] (o1) at (2,1.1) {};
			\node[label={$\vdots$}] (o2) at (2,0.15) {};
			\node[place,label={0:},minimum size=3mm] (om) at (2,0) {};

			\draw[-latex,thick] (q0) -- (t1);
			\draw[-latex,thick] (i1) -- (t1);
			\draw[-latex,thick] (in) -- (t1);
      		\draw[-latex,thick] (t1) -- (q1);
      		\draw[-latex,thick] (t1) -- (o1);
      		\draw[-latex,thick] (t1) -- (om);
    	\end{tikzpicture}
    	\\ \hline
		 \multicolumn{2}{|c|}{Fork Computon} &		 		     	
    	\begin{tikzpicture}	    
			\fork{0}{0}{}{}{}
    	\end{tikzpicture} &
    	\begin{tikzpicture}
			\node[transition,fill=black,minimum width=0.5mm,minimum height=10mm,label=$ $] (t1) at (1,0.7) {};

			\node[place,label={180:},minimum size=3mm] (q0) at (0,0.7) {};			
			\node[place,label={0:},minimum size=3mm] (q1) at (2,0.5) {};			
			\node[place,label={0:},minimum size=3mm] (q2) at (2,1) {};

			\draw[-latex,thick] (q0) -- (t1);
 	    \draw[-latex,thick] (t1) -- (q1);
      \draw[-latex,thick] (t1) -- (q2);
    	\end{tikzpicture}
    	\\ \hline
    	\multicolumn{2}{|c|}{Join Computon} &		 		     	
    	\begin{tikzpicture}	    
			\join{0}{0}{}{}{}
    	\end{tikzpicture} &
    	\begin{tikzpicture}
			\node[transition,fill=black,minimum width=0.5mm,minimum height=10mm,label=$ $] (t1) at (1,0.7) {};

			\node[place,label={180:},minimum size=3mm] (q0) at (0,0.5) {};			
			\node[place,label={180:},minimum size=3mm] (q1) at (0,1) {};			
			\node[place,label={0:},minimum size=3mm] (q2) at (2,0.7) {};

			\draw[-latex,thick] (q0) -- (t1);
 	    \draw[-latex,thick] (q1) -- (t1);
      \draw[-latex,thick] (t1) -- (q2);
    	\end{tikzpicture}
    	\\ \hline
    	\multicolumn{2}{|c|}{Composite Computon} &		 		 
    	\begin{tikzpicture}
	    	\node at (-0.3,0){};
	    	\computonComposite{1.14}{-0.35}{1}{2.25}
	    	
    		\qin{1q0}{0.34}{1.7}{}
    		\node[label={$\vdots$}] (1oy) at (0.74,0.9) {};
    		\qin{1q1}{0.34}{0.9}{}
			\din{1i1}{0.34}{0.6}{$n_{11}$};
			\node[label={$\vdots$}] (1i2) at (0.74,-0.2) {};
			\din{1in}{0.34}{-0.15}{$n_{12}$};

			\qout{1q2}{2.14}{1.7}{}
			\node[label={$\vdots$}] (1ox) at (2.54,0.9) {};
			\qout{1q3}{2.14}{0.9}{}
			\dout{1o1}{2.14}{0.6}{$n_{13}$};
			\node[label={$\vdots$}] (1o2) at (2.54,-0.2) {};
			\dout{1om}{2.14}{-0.15}{$n_{14}$};
    	\end{tikzpicture} &
    	\begin{tikzpicture}
			\node[place,label={180:},minimum size=3mm] (q0) at (0,1.7) {};
			\node[label={$\vdots$}] (ix) at (0,0.8) {};
			\node[place,label={180:},minimum size=3mm] (q1) at (0,0.8) {}; 			
			\node[place,label={180:},minimum size=3mm] (i1) at (0,0.4) {};
			\node[label={$\vdots$}] (i2) at (0,-0.5) {};
			\node[place,label={180:},minimum size=3mm] (in) at (0,-0.5) {};
			
			\node at (1,1.2){$\cdots$};
			\node at (1,0.6){$\cdots$};			
			\node at (1,-0.1){$\cdots$};

			\node[place,label={0:},minimum size=3mm] (q2) at (2,1.7) {};
			\node[label={$\vdots$}] (iy) at (2,0.8) {};
			\node[place,label={0:},minimum size=3mm] (q3) at (2,0.8) {};
			\node[place,label={0:},minimum size=3mm] (o1) at (2,0.4) {};
			\node[label={$\vdots$}] (o2) at (2,-0.5) {};
			\node[place,label={0:},minimum size=3mm] (om) at (2,-0.5) {};

			\draw[-latex,thick] (q0) -- (0.5,1.7);
			\draw[-latex,thick] (q1) -- (0.5,0.8);
			\draw[-latex,thick] (i1) -- (0.5,0.4);
			\draw[-latex,thick] (in) -- (0.5,-0.5);
			\draw[-latex,thick] (1.5,1.7) -- (q2);
			\draw[-latex,thick] (1.5,0.8) -- (q3);
			\draw[-latex,thick] (1.5,0.4) -- (o1);
			\draw[-latex,thick] (1.5,-0.5) -- (om);
    	\end{tikzpicture}
    	\\ \hline
	\end{tabular}
  \end{center}
\end{table}

\bibliographystyle{unsrtnat} 
\bibliography{refs}

\end{document}